\documentclass[sigplan,nonacm,screen]{acmart}

\input{appendixtoc}

\usepackage{tikz}
\usetikzlibrary{cd}

\setcopyright{none}

\usepackage{amsmath}
\usepackage{amsfonts}
\usepackage{color}
\usepackage{proof}
\usepackage{mathrsfs}  
\usepackage{bm}
\usepackage{stmaryrd}
\usepackage{subcaption}

\usepackage{thmtools}
\usepackage{thm-restate}
\usepackage{varwidth}

\usepackage{csquotes}
\usepackage[scr=boondoxo,scrscaled=1.05]{mathalfa}
\usepackage{bussproofs}
\EnableBpAbbreviations
\usepackage{adjustbox}

\usepackage[toc,page,header]{appendix}

\usepackage{titletoc}

\hypersetup{
           breaklinks=true}

\newtheorem{remark}{Remark}

\renewcommand{\labelitemii}{$\bullet$}

\newsavebox{\mypti}
\newsavebox{\mypto}
\newsavebox{\mypta}
\newsavebox{\myptu}
\newsavebox{\mypte}

\newcommand{\FF}[1]{\mathfrak{#1}}
\newcommand{\B}[1]{\mathbf{#1}}
\newcommand{\BB}[1]{\mathbb{#1}}
\newcommand{\CC}[1]{\mathscr{#1}}
\newcommand{\TT}[1]{\mathtt{#1}}

\newcommand{\CPL}{\mathsf{CPL}}
\newcommand{\ICPL}{\mathsf{iCPL}}
\newcommand{\ICPLL}{\mathsf{iCPL}_{0}}

\newcommand{\NDCPLLV}{\B{ND}_{\mathsf{iCPL0}}^{\mathrm{CbN}}}
\newcommand{\NDCPL}{\B{ND}_{\mathsf{iCPL}}}
\newcommand{\NDCPLL}{\B{ND}_{\mathsf{iCPL}_{0}}}

\newcommand{\STLC}{\lambda_{\to}^{\{\}}}
\newcommand{\STLCINT}{\lambda_{\to,\cap}}

\newcommand{\TCI}{\B C\lambda_{\to}}

\newcommand{\LCPL}{\B C\lambda^{\{\}}_{\to}}

\newcommand{\Pto}{\rightarrowtail}
\newcommand{\To}{\Rightarrow}

\newcommand{\FN}{\mathsf{FN}}

\newcommand{\bone}{{\color{blue}\mathscr b}}
\newcommand{\btwo}{{\color{blue}\mathscr c}}
\newcommand{\bthree}{{\color{blue}\mathscr d}}
\newcommand{\bfour}{{\color{blue}\mathscr e}}
\newcommand{\bfive}{{\color{blue}\mathscr f}}

\newcommand{\TTH}{{\color{blue}\mathsf{Th}}}

\newcommand{\TOP}{{\color{blue}\top}}
\newcommand{\BOT}{{\color{blue}\bot}}

\newcommand{\bvar}{{\color{blue}\mathscr x}}

\newcommand{\BOX}{\mathbf{C}}

\newcommand{\TCINT}{\B C\lambda_{\to,\cap}}

\newcommand{\srank}[1]{\lceil #1\rceil}

\newcommand{\GHRED}{\mathrm{H^{*}Red}}
\newcommand{\HRED}{\mathrm{HRed}}
\newcommand{\RED}{\mathrm{NRed}}
\newcommand{\RRED}{\mathrm{Red}}

\newcommand{\SNRED}{\mathrm{SNRed}}
\newcommand{\NNRED}{\mathrm{NNRed}}
\newcommand{\HNRED}{\mathrm{HNRed}}

\newcommand{\Lnu}{\Lambda_{\oplus,\nu}}
\newcommand{\HLnu}{\Lambda_{\oplus,\nu, \{\}}}

\newcommand{\HNN}{\mathrm{HN}}

\newcommand{\NNN}{\mathrm{NN}}

\newcommand{\SN}{\mathrm{SN}}
\newcommand{\sn}{\mathrm{sn}}

\newcommand{\EVL}{\Lambda_{\mathsf{{PE}}}}
\newcommand{\EVLL}{\Lambda^{\{\}}_{\mathsf{{PE}}}}

\newcommand{\redbeta}{\rightarrowtriangle_{\beta}}
\newcommand{\redperm}{\rightarrowtriangle_{\mathsf p}}
\newcommand{\redall}{\rightarrowtriangle}
\newcommand{\redpermm}{\rightarrowtriangle_{\mathsf p\{\}}}
\newcommand{\redalll}{\rightarrowtriangle_{\{\}}}
\newcommand{\redallh}{\rightarrowtriangle_{\{\}h}}

\newcommand{\HNF}{\mathrm{HNV}}
\newcommand{\NF}{\mathrm{NF}}

\newcommand{\NEUT}{\mathrm{Neut}}
\newcommand{\GNEUT}{\mathrm{Neut}_{\{\}}}

\newcommand{\BAL}{\mathrm{Bal}}

\newcommand{\NNF}{\mathrm{NF}_{\redall}}
\newcommand{\NHNF}{\mathrm{HNV}_{\redall}}

\newcommand{\NORM}{\mathsf{n}}
\newcommand{\HNORM}{\mathsf{hn}}

\newcommand{\TU}{(\vee)}
\newcommand{\TLA}{({\lambda})}
\newcommand{\TA}{({@})}
\newcommand{\TL}{({\oplus} l)}
\newcommand{\TR}{(\oplus r)}
\newcommand{\TN}{(\mu)}

\newcommand{\Borel}{\B B}

\newcommand{\model}[1]{\llbracket #1\rrbracket}

\newcommand{\hide}[1]{}

\RequirePackage{ifthen}
\newcommand{\typeof}{0}
\newcommand{\longv}[1]{\ifthenelse{\equal{\typeof}{0}}{}{#1}}
\newcommand{\shortv}[1]{\ifthenelse{\equal{\typeof}{0}}{#1}{}}

\renewcommand{\terms}{\mathscr{T}}

\newcommand{\midd}{\;\;\mbox{\Large{$\mid$}}\;\;}

\newcommand{\NN}{\mathbb N}

\newenvironment{varitemize}
{
	\begin{list}{\longv{\labelitemi}\shortv{\labelitemii}}
		{\setlength{\itemsep}{0pt}
			\setlength{\topsep}{0pt}
			\setlength{\parsep}{0pt}
			\setlength{\partopsep}{0pt}
			\setlength{\leftmargin}{15pt}
			\setlength{\rightmargin}{0pt}
			\setlength{\itemindent}{0pt}
			\setlength{\labelsep}{5pt}
			\setlength{\labelwidth}{10pt}
	}}
	{
	\end{list} 
}

\newcounter{numberone}

\begin{document}


\title{Curry and Howard Meet Borel} 

\author{Melissa Antonelli}
\affiliation{{University of Bologna}\country{Italy}}
\email{}

\author{Ugo Dal Lago}
\affiliation{{University of Bologna}\country{Italy}}
\email{}

\author{Paolo Pistone}
\affiliation{{University of Bologna}\country{Italy}}
\email{}


\begin{abstract}
We show that an intuitionistic version of counting propositional logic 
	corresponds, in the sense of Curry and Howard, to an expressive type system 
	for the probabilistic event $\lambda$-calculus, a vehicle calculus 
	in which both call-by-name and call-by-value evaluation of discrete 
	randomized functional programs can be simulated. 
	Remarkably, proofs (respectively, types) do not only guarantee that validity (respectively, termination) holds, but also \emph{reveal} the underlying probability. We finally show that by endowing the type system with an intersection operator, one obtains a system precisely capturing the probabilistic behavior of $\lambda$-terms.
%
\end{abstract}

\maketitle

\section{Introduction}\label{section1}

Among the many ways in which mathematical logic influenced programming language 
theory, the so-called Curry-Howard correspondence is certainly among the most 
intriguing and meaningful ones. 
Traditionally, the correspondence identified by 
Curry~\cite{Curry1958} and formalized by Howard~\cite{Howard80} (CHC in the following) relates
propositional intuitionistic logic and the simply-typed 
$\lambda$-calculus. As is well-known, though, this correspondence holds in 
other contexts, too. Indeed, in the last fifty years more sophisticated type 
systems have been put in relation with logical formalisms: from 
polymorphism~\cite{Girard72, Girard1989} to various forms of session typing~\cite{Wadler2012, Caires2016}, from control 
operators~\cite{Parigot1992} to dependent types~\cite{Mlof75, CH}.

Nevertheless, there is a class of programming languages and type systems for
which a correspondence in the style of 
the CHC has not yet been found. We are talking about languages with probabilistic effects, for which
 type-theoretic accounts have recently been put forward in various ways, e.g. type 
systems based on sized types~\cite{DalLagoGrellois}, intersection types~\cite{Breuvart2018} 
or type systems in the style 
of so-called amortized analysis~\cite{Wang2020}. 
In all the aforementioned cases, 
a type system was built by modifying well-known type systems for deterministic 
languages \emph{without} being guided by logic, and instead incepting
 inherently quantitative
 concepts directly from 
probability theory.
So, one could naturally wonder if there is any logical 
system behind all this, or 
what kind of logic could possibly play the role of 
propositional logic in suggesting meaningful and expressive type systems for a 
$\lambda$-calculus endowed with probabilistic choice effects.

A tempting answer is to start from modal logic, 
which is known to correspond in the Curry-Howard sense
 to staged computation
and to algebraic effects \cite{DaviesPfennig, Nanevski,Curien2016}. 
Nevertheless, there is one aspect of randomized 
computation that modal logic fails to capture\footnote{With a few notable exceptions, 
e.g.~\cite{lmcs:6054}.}, namely the \emph{probability} of 
certain notable events, typically termination. 
Indeed, several other properties like reachability and safety can be reduced to the termination property (see, e.g., the discussion in \cite{dallago2019}).  
In many of the probabilistic 
type systems mentioned above, for example, a term $t$ receives a type that 
captures the fact that $t$ has a certain probability $q$, perhaps strictly smaller 
than $1$, of reducing to a value. This probability is an essential part of what we want to observe about the dynamics of $t$ and, as such, has to be captured by its 
type, at least if one wants the type system to be expressive.

Recently, the authors have proposed to use \emph{counting quantifiers} \cite{AntonelliCie, Antonelli2021} as a means to express probabilities within a logical language.
These quantifiers, 
unlike standard ones, determine not only \emph{the 
existence} of an assignment of values to variables with certain
characteristics, but rather count \emph{how many} of those assignments exist. 
Recent works show that classical propositional logic
enriched with counting quantifiers
corresponds to Wagner's hierarchy of {counting complexity classes}~\cite{Wagner} (itself intimately linked to 
probabilistic complexity), 
and that Peano Arithmetics enriched with similar quantitative quantifiers 
yields a system capable of speaking of randomized computation in the same 
sense in which standard PA models deterministic computation~\cite{AntonelliCie}. 
One may now wonder whether all this scale 
to something like a proper CHC.

The aim of this work is precisely to give a positive answer to the aforementioned 
question, at the same time highlighting a few 
remarkable consequences of our study of probabilistic computation through the lens of logic.
 More specifically, 
the contributions of this paper are threefold:
\begin{varitemize}

\item\label{itemii}
First of all, we introduce an intuitionistic version of 
%
 Antonelli et al.'s counting propositional 
logic~\cite{Antonelli2021}, together with a Kripke-style semantics based on the Borel $\sigma$-algebra 
of the Cantor space and a sound and complete proof-theory.  
Then, we identify a ``computational fragment'' of this logic and we 
design for it a natural deduction system in which proof normalization simulates 
well-known evaluation strategies for probabilistic programs. 
This is in 
Section~\ref{section3}.
\item
We then show that derivations can be decorated with terms of the 
\emph{probabilistic event $\lambda$-calculus}, a
 $\lambda$-calculus for randomized computation 
introduced by Dal Lago et al. in~\cite{DLGH}. This gives rise to a type-system called 
$\LCPL$. Remarkably, the correspondence scales to the underlying dynamics, 
i.e.~proof normalization relates to reduction in the probabilistic event
$\lambda$-calculus. This is in Section~\ref{section4} and 
Section~\ref{section5}.
\item
We complete the picture by giving an intersection type assignment system, 
derived from $\LCPL$, and proving that it precisely captures the normalization 
probability of terms of the quoted $\lambda$-calculus. This is in Section~\ref{section6}.
\end{varitemize}
Space limits prevent us from being comprehensive, meaning that many technical 
details will be unfortunately confined to the Appendix.

\section{From Logic to Counting and Probability: a Roadmap}\label{section2}

In this section we provide a first overview of our probabilistic CHC, 
at the same time sketching the route we will follow in the rest of 
the paper.

\subsection{Randomized Programs and Counting Quantifiers}

The first question we should ask ourselves concerns the kind of programs we 
are dealing with. 
What is a probabilistic functional program? 
Actually, it is a functional program with the additional ability of sampling from some 
distributions, or to proceed by performing some form of discrete probabilistic 
choice.\footnote{Here we are not at all concerned with sampling from continuous 
distributions, nor with capturing any form of conditioning.} 
This has a couple of crucial consequences: program evaluation becomes an 
essentially stochastic process, and programs satisfy a given specification 
\emph{up to} a certain probability. As an example, consider the $\lambda$-term,
$$
\Xi_{\mathsf{half}} :=\lambda x.\lambda y. x\oplus y
$$
where $\oplus$ is a binary infix operator for fair probabilistic choice.
When applied to two arguments $t$ and $u$,
the evaluation on $\Xi_{\mathsf{half}}$ 
results in either $t$, with probability one half, or in $u$, again with 
probability one half.

But now, if we try to take $\Xi_{\mathsf{half}}$ as a proof of some logical formula, we see that standard propositional logic is simply not 
rich enough to capture the behavior above. Indeed, given that 
$\Xi_{\mathsf{half}}$ is a function of two arguments, it is natural to see it 
as a proof of an implication $A\to B\to C$, namely (following the BHK 
interpretation) as a function turning a proof of 
$A$ and a proof of $B$ into a proof of $C$. What is $C$, then? Should it be $A$ 
or should it be $B$? Actually, it could be both, with some degree of 
uncertainty, but propositional logic \emph{is not} able to express all this.
At this point, the recent work on \emph{counting quantifiers} in propositional 
logic~\cite{Antonelli2021} 
comes to the rescue.

Another way to look at discrete probabilistic programs is 
as programs which are allowed to sample an element $\omega$ from the 
\emph{Cantor space} $2^{\BB N}$. For example, a probabilistic Turing machine 
$M$ can be described as a 2-tape machine where the second tape is read-only and 
sampled from the Cantor space at each run.
Similarly, the execution of $\Xi_{\mathsf{half}}$ applied to programs $t$ and $u$ can be described as the result of sampling $\omega$ and returning either $t$ or $u$ depending on some read value, e.g.~$\omega(0)$.
Crucially, for each input $x$ and each possible output $y$ of a probabilistic 
program $f$, the set $S_{x,y}$ of elements in $2^{\BB N}$, making
$f(x)$ produce $y$, is a \emph{Borel} set, so it makes sense to say that the 
probability that $f(x)$ yields $y$ coincides with the \emph{measure} of $S_{x,y}$ 
(more on Borel sets and measures below).

Yet, what has all this to do with logic and counting? The fundamental 
observation is that formulas of classical propositional logic 
provide ways of 
denoting Borel sets. 
Let $\bvar_{0},\bvar_{1},\dots$ indicate a countable set of propositional variables 
(which can be seen as i.i.d.~and uniformly distributed random variables $\bvar_{i}\in \{0,1\}$); 
each variable $\bvar_{i}$ can be associated with the \emph{cylinder set} (indeed, a Borel set of measure $\frac{1}{2}$) 
formed by all $\omega\in 2^{\BB N}$ such that $\omega(i)=1$ 
(i.e., in logical terms, such that $\omega \vDash \bvar_{i}$). 
Then, any propositional formula $\bone, \btwo,\dots$ formed by way of $\lnot,\land,\lor$ can be associated 
with the Borel sets $\{\omega \in 2^{\BB N}\mid \omega \vDash\bone\}$ which can be constructed starting from cylinder sets 
by means of complementation, finite intersection and finite unions.
In this way, for any Boolean formula $\bone$, it makes sense to define
\emph{new} formulas like the formula $\B C^{\frac{1}{2}}\bone$, which is true when the Borel set associated with $\bone$ has measure greater than $\frac{1}{2}$ (i.e.~when $\bone$ is satisfied by \emph{at least 
$\frac{1}{2}$} of its models). 
For instance, for any two \emph{distinct} indexes $i,j\in \BB N$ (so that the variables $\bvar_{i},\bvar_{j}$ correspond to two independent events), the formula $\B C^{\frac{1}{2}}(\bvar_{i} \lor \bvar_{j})$ is 
true, since the Borel set formed by the $\omega\in 2^{\BB N}$ such that either $\omega(i)=1$ or $\omega(j)=1$ has measure greater than $\frac{1}{2}$
(equivalently, 
 $\bvar_{i} \lor \bvar_{j}$ is satisfied by at least $\frac{1}{2}$ of its models). 
Instead, $\B C^{\frac{1}{2}}(\bvar_{i} \land \bvar_{j})$ is false, since, $\bvar_{i}$ and $\bvar_{j}$ being independent,
the Borel set associated with $\bvar_{i}\land \bvar_{j}$ has only measure $\frac{1}{4}$.

Counting propositional logic (from now on $\CPL$) is defined by enriching classical propositional logic with counting 
quantifiers $\B C^{q}$,
where $q\in (0,1]\cap\BB Q$. Following our sketch, $\CPL$ admits a natural semantics in the 
Borel $\sigma$-algebra of the Cantor space, together with a sound and complete 
sequent calculus \cite{Antonelli2021}. 
Notice that measuring a Boolean formula actually amounts at \emph{counting} its models, 
that is, to a purely recursive operation, 
albeit one which needs not be doable in polynomial time 
(indeed, $\CPL$ is deeply related to Wagner's counting hierarchy \cite{Antonelli2021}).

Going back to the term $\Xi_{\mathsf{half}}$, what seems to be lacking in 
intuitionistic logic is precisely a way 
to express that $C$ could be $\B C^{\frac{1}{2}}A$, that is, that it could be 
$A$ \emph{with probability at least} $\frac{1}{2}$, and, similarly, that it could be  
by $\B C^{\frac{1}{2}}B$.
%
In Section \ref{section3} we introduce an intuitionistic version of $\CPL$, called $\ICPL$, which enriches 
intuitionistic logic with Boolean variables
 as well as the counting quantifier $\B C^{q}$. 
 Intuitively, if a proof 
of a formula $A$ can be seen as a deterministic program satisfying the
specification $A$, a proof of $\B C^{q}A$ 
will correspond 
to a probabilistic program that satisfies the specification $A$ with 
probability $q$. 
Our main result consists in showing that proofs in $\ICPL$ 
correspond to functional probabilistic programs and, most importantly, that 
normalization in this logic describes probabilistic evaluation.

\subsection{Can CbN and CbV Evaluation Coexist?}

When studying probabilistic extensions of the $\lambda$-calculus, two different evaluation strategies are generally considered: 
the \emph{call-by-name} (CbN) strategy, which (possibly) duplicates choices before evaluating 
them, and the \emph{call-by-value} (CbV) strategy, which instead evaluates choices before (possibly) duplicating their outcomes. 
Importantly, the probability of termination of a program might differ depending on the chosen strategy.
For example, consider the application of the term $\B 2:=\lambda y x.y(yx)$ (i.e.~the second Church numeral)
 to
  $I\oplus \Omega$, 
where $I=\lambda x.x$ and $\Omega$ is the diverging term $(\lambda x.xx)\lambda x.xx$.
Under CbN, the redex $\B 2(I\oplus\Omega)$ first produces $\lambda x.( I\oplus \Omega)(( I\oplus \Omega)x)$, 
then reduces to any of the terms $\lambda x. u(vx)$, with $u,v$ 
chosen from $\{I,\Omega\}$, each with probability $\frac{1}{4}$. 
Since $\lambda x. u(vx)$ converges only when $u=v=I$, 
the probability of convergence is thus $\frac{1}{4}$.
Under CbV,  in $\B 2(I\oplus\Omega)$ one first has to evaluate $I\oplus\Omega$, 
and then passes the result to $\B 2$, hence  returning
either the converging term $I(Ix)$ or the diverging term $\Omega(\Omega x)$, 
each with probability $\frac{1}{2}$.

If we now try to think of the Church numeral $\B 2$ as a proof of some counting quantified formula, we see that, depending on the reduction strategy we have in mind, it must prove a \emph{different} formula. Indeed, 
%
given that $I\oplus\Omega$ proves $\B C^{\frac{1}{2}}(A\to A)$, in the CbN case, $\B 2$ proves  
$\B C^{\frac{1}{2}}(A\to A) \to A\to  \B C^{\frac{1}{4}}A$, since only in one over four cases it yields a proof of $A$, while in the CbV case, $\B 2 $ proves the formula  
$\B C^{\frac{1}{2}}(A\to A) \to  A\to \B C^{\frac{1}{2}}A$, as it yields a proof of $A$ in one case over two.


In the literature on probabilistic $\lambda$-calculi, the apparent 
incompatibility of CbN and CbV evaluation is usually resolved by 
restricting to calculi with one or the other strategy.
However, the observation above suggests that, if functional programs are typed using counting quantifiers, it should become possible 
to make the two evaluation strategies coexist, by assigning them different 
types.

Actually,  a few 
recent approaches \cite{DalLago2017, Ehrhard2018, Faggian2019, DLGH} already suggest ways to make CbN and CbV evaluation live together. In particular, in the 
\emph{probabilistic event $\lambda$-calculus} \cite{DLGH} the choice operator $\oplus$ is decomposed into two different operators, yielding a confluent calculus: a 
\emph{choice operator} $t\oplus_{a}u$, depending on some 
probabilistic event $a\in\{0,1\}$, 
and a \emph{probabilistic event generator} $\nu a.t$, which 
actually ``flips the coin''.
In this language, the CbN and CbV applications of $\B 2$ to $I\oplus 
\Omega$ are encoded by two \emph{distinct} terms
$\B 2(\nu a.I\oplus_{a}\Omega)$ and 
$\nu a. \B 2 ( I\oplus_{a}\Omega)$,
crucially distinguishing between generating a probabilistic choice \emph{before} or \emph{after} a duplication takes place.
%
%

This calculus constitutes then an ideal candidate for our CHC. Indeed, as we'll see, the logical rules for the counting quantifier $\B C^{q}$ naturally give rise to typing rules for the event generator $\nu a.$
In Section \ref{section4} we introduce a variant $\EVLL$ of the calculus from \cite{DLGH}, 
with the underlying probability space $\{0,1\}$  replaced by 
the Cantor space, and in 
%
%
%
%
Section \ref{section5} we introduce a type system $\LCPL$ for simple types with counting quantifiers, showing that natural deduction derivations translate into typing derivations in $\LCPL$, with 
normalization precisely corresponding to reduction in $\EVLL$.

\subsection{Capturing Probability of Normalization via Types}

As observed in the Introduction, a fundamental quantitative property we would like to observe using types is the probability of termination. 
Nevertheless, given that the reduction of $\EVLL$ is purely deterministic, what notions of probabilistic termination should we actually observe?
%
%

Rather than evaluating programs by implementing probabilistic choices,
reduction in the probabilistic event $\lambda$-calculus has the effect of progressively generating the \emph{full tree} of outcomes of (sequences of) probabilistic choices, giving rise to a \emph{distribution} of values.
%
Therefore, given a term $t$, rather than asking whether 
\emph{some} or \emph{all} reductions of $t$ terminate, it makes sense to ask \emph{with what probability} a normal form is found by generating all probabilistic outcomes of $t$. 

In Section \ref{section6} we will first show that when the type $\B C^{q}\sigma$ is assigned to a program $t$, 
the value $q\in \mathbb{Q}\cap (0,1]$ provides a lower bound for the actual probability of finding a (head) normal form in the development of $t$. 
Then we show that, by extending the type system with an intersection operator, one can attain an upper bound and thus fully characterize the distribution of values associated with a term.

\subsection{Preliminaries on the Cantor Space}\label{sub:Cantor}

Throughout the paper, we exploit some basic facts about the Cantor space, 
its Borel $\sigma$-algebra, and their connections with Boolean logic, that we briefly recall here.

We consider a countably infinite set $\CC A$ of \emph{names}, noted $a,b,c,\dots$. For any finite subset $X\subseteq \CC A$, we let $\Borel_{X}$ (resp.~$\Borel_{\CC A}$) indicate the \emph{Borel $\sigma$-algebra} on the $X$-th product of the Cantor space $(2^{\BB N})^{X}$ (resp.~on the $\CC A$-th product $(2^{\BB N})^{\CC A}$), that is, the smallest $\sigma$-algebra containing all open sets under the product topology.
%
There exists a unique measure $\mu$ of $\Borel_{\CC A}$ such that $\mu(C_{a,i})=\frac{1}{2}$ for all \emph{cylinders} $C_{a,i}=\{\omega \mid \omega(a)(i)=1\}$. The measure $\mu$ restricts to a measure $\mu_{X}$ on $\Borel_{X}$ by 
letting $\mu_{X}(S)=\mu(S\times (2^{\BB N})^{\CC A-X})$.
%

\emph{Boolean formulas with names in $\CC A$} are defined by:
$$
\bone::= \TOP\mid\BOT \mid \bvar_{a}^{i}\mid\lnot \bone \mid  \bone \land \bone \mid \bone \lor \bone
$$
where $a\in \CC A$ and $i\in \BB N$. We let $\FN(\bone)\subseteq \CC A$ be the set of names occurring in $\bone$.
For all Boolean formula $\bone$ and $X\supseteq\FN(\bone)$, 
we let $\model{\bone}_{X}$ indicate the Borel set $\{\omega\in (2^{\BB  N})^{X}\mid \omega \vDash \bone\}$
(in particular, $\model{\bvar^i_a}_{\{a\}}=\{\omega \ | \ \omega(a)(i)=1\}$.
The value  $\mu_{X}(\model{\bone}_{X})\in [0,1]\cap\BB Q$ is independent from the choice of $X\supseteq \FN(\bone)$, and we will note it simply as $\mu(\bone)$.

\section{Intuitionistic Counting Propositional Logic}\label{section3}

%

In this section we introduce a constructive version of $\CPL$, that we call $\ICPL$, which extends standard intuitionistic logic with Boolean variables and counting quantifiers.
The logic $\ICPL$ combines constructive reasoning (corresponding, under CHC, to functional programming) with  semantic reasoning on Boolean formulas and their models (corresponding, as we have seen, to discrete probabilistic reasoning). The formulas of $\ICPL$ are \emph{hybrid}, that is, comprise both a countable set $\CC P=\{\CC p, \CC q, \dots\}$ of intuitionistic propositional variables, and named Boolean propositional variables $\bvar_{a}^{i}$. 
For example, consider the formula below:
$$
A:= \CC p \to \CC q \to (\bvar_{a}^{0}\land \CC p)\lor (\lnot \bvar_{a}^{0}\land \CC q).
$$
Intuitively, proving $A$ amounts to showing that, 
whenever $\CC p$ and $\CC q$ hold, 
given that either $\bvar_{a}^{0}$ or $\lnot\bvar_{a}^{0}$ holds,  
in the first case $\CC p$ holds, and in the second case $\CC q$ holds. 

Suppose we test $A$ against some element $\omega$ from the Cantor space.
A way of proving $A$, in the ``environment'' $\omega$, could then be as follows: 
given assumptions $\CC p$ and $\CC q$, conclude $\CC p$ and $\bvar_{a}^{0}$ if $\omega(0)=1$ (i.e.~if $\omega$ satisfies $\bvar_{a}^{0}$), 
and conclude $\CC q$ and $\lnot \bvar_{a}^{0}$ if $\omega(0)=0$ (i.e.~if $\omega$ does not satisfy $\lnot\bvar_{a}^{i}$). In other words, a proof of $A$ under $\omega$ could be something like $\lambda xy. x\oplus_{a}y$. 
Assuming that $\omega$ is uniformly sampled, what are the chances that our strategy will actually yield a proof of $A$?
Well, there are two possible cases, and in both cases we get a proof of one of the disjuncts 
$\bvar_{a}^{0}\land\CC p$
and $\lnot\bvar_{a}^{0}\land \CC q$, thus a proof of $A$. 
We thus obtain a proof of $A$ with probability 1 or, in other words,   
%
a proof of $\B C^{1}_{a}A$ (no more depending on any ``environment''). As a term, this proof looks then precisely as the closed term $\nu a.\lambda xy.x\oplus_{a}y$.  

Consider now another formula:
$$
B:= \CC p \to \CC q \to \bvar_{a}^{0}\land \CC p.
 $$
 Given some $\omega$, a proof $B$ should conclude, from assumptions $\CC p$ and $\CC q$, both $\CC p$ and $\bvar_{a}^{0}$ when $\omega(0)=1$, but when $\omega(0)=0$, it needs not conclude anything, since we have nothing to prove in this case. A proof of $B$ should then be something like $\lambda xy.x\oplus_{a}?$, where $?$ can be any program. By uniformly sampling $\omega$, we get then a proof of $B$ only in one half of the cases, i.e.~we get a proof of $\B C^{\frac{1}{2}}_{a}B$.

\subsection{The Semantics and Proof-Theory of $\ICPL$.}

The formulas of the logic just sketched are defined by:
$$
A::= \TOP\mid \BOT\mid \bvar_{a}^{i}\mid \CC p\mid A\land A\mid A\lor A\mid A\to A\mid \B C^{q}_{a}A
$$
where $\CC p\in \CC P$ and $q\in (0,1]\cap \BB Q$.
A natural semantics for $\ICPL$-formulas can be given in terms of Kripke-like structures:

\begin{definition}\label{def:icpl}
An \emph{$\ICPL$-structure} is a triple $\CC M=(W, \preceq , \B i)$ where $W$ is a countable set, $\preceq$ is a preorder on $W$,
and $\B i: \CC P\to W^{\uparrow}$, where $W^{\uparrow}$ is the set of upper subsets of $W$.

\end{definition}

The interpretation of formulas in 
$\ICPL$-structures combines a set $W$ of worlds (for the interpretation of intuitionistic propositional variables) with the choice of  an element of the Cantor space (for the interpretation of Boolean variables):
for any $\ICPL$-structure $\CC M=(W,\preceq,  \B i)$ and finite set $X$, we define the relation $w,\omega \Vdash^{X}_{\CC M}A$ (where $w\in W$, $\omega \in (2^{\BB N})^{X}$ and $\FN(A)\subseteq X$) by induction as follows:
\begin{varitemize}

\item $w, \omega \not\Vdash_{\CC M}^{X} \BOT$ and $w, \omega \Vdash_{\CC M}^{X} \TOP$;

\item $w, \omega \Vdash^{X}_{\CC M}   \bvar_{a}^{i}$ iff $\omega(a)(i)=1$;

\item $w, \omega \Vdash^{X}_{\CC M}   \CC p$ iff  $w\in \B i(\CC p)$;

\item $w,\omega \Vdash^{X}_{\CC M}   A\land B$ iff $w,\omega \Vdash^{X}_{\CC M}   A$ and $w,\omega \Vdash^{X}_{\CC M}    B$;

\item $w,\omega \Vdash^{X}_{\CC M}   A\lor B$ iff $w,\omega \Vdash^{X}_{\CC M}   A$ or $w,\omega \Vdash^{X}_{\CC M}    B$;

\item $w,\omega \Vdash^{X}_{\CC M}   A\to B$ iff for all $w'\succeq w$, 
$w', \omega \Vdash^{X}_{\CC M}  A$ implies $w',\omega \Vdash^{X}_{\CC M}   B$;

\item $w,\omega \Vdash^{X}_{\CC M}  \B C^{q}_{a}A$ iff 
$\mu\left (\left \{ \omega'\in 2^{\BB N} \mid w, \omega+ \omega'\Vdash^{X\cup\{a\}}_{\CC M}A\right\}\right)\geq q
$, where $\omega+\omega'\in (2^{\BB N})^{X\cup\{a\}}$ is given by $
(\omega+\omega')(b)(n)=\omega(b)(n)$, for $b\in X$, and
$(\omega+\omega')(a)(n)=\omega'(n)$.

\end{varitemize}
Using properties of the Borel $\sigma$-algebra $\Borel_{X}$, it can be shown that for any $w\in W$ and formula $A$, the set $\{\omega\in (2^{\BB N})^{X}\mid w,\omega \Vdash^{X}_{\CC M}A\}$ is a Borel set, and thus measurable.
We write $\Gamma \Vdash^{X}_{\CC M}  A$ when for all $w\in W$ and $\omega \in (2^{\BB N})^{X}$, 
whenever $w,\omega \Vdash^{X}_{\CC M}  \Gamma$ holds, also
$w,\omega \Vdash^{X}_{\CC M}  A$ holds.
We write $\Gamma \vDash A$ when for any $\ICPL$-structure $\CC M$ and $X\supseteq\FN(A)$, $\Gamma \Vdash^{X}_{\CC M}  A$ holds.

A sound and complete proof system for $\ICPL$ can be defined. 
Indeed, starting from usual natural deduction for intuitionistic logic, one obtains a calculus $\NDCPL$ for $\ICPL$ 
by adding the excluded middle $\bvar_{a}^{i}\vee \lnot \bvar_{a}^{i}$ as an axiom for Boolean variables, 
as well as suitable rules and axioms for counting quantifiers (details can be found in App.~\ref{app:ProofTheoryiCPL}). 
\begin{theorem}\label{thm:soundcomplete}
$\Gamma \vDash A$ iff $\Gamma\vdash_{\NDCPL}A$.
\end{theorem}
\noindent
For example, $\B C^{q}$ admits an introduction rule as below
\begin{align}\label{rule0}
\AXC{$\Gamma, \bthree \vdash A$}
\AXC{$\mu(\bthree)\geq q$}
\BIC{$\Gamma \vdash \B C^{q}_{a}A$}
\DP
\tag{$\B C$I}
\end{align}
with the proviso that $a$ does not occur in $\Gamma$, and is the only name occurring in $\bthree$. Intuitively, this rule says that if we can prove $A$ under assumptions $\Gamma$ and $\bthree$, and if a randomly chosen valuation of $a$ has chance $q$ of being a model of $\bthree$, then we can build a proof of $\B C^{q}_{a}A$ from $\Gamma$. 
Observe that the rule ($\B C$I) has a semantic premise $\mu(\bthree)\geq q$: we ask to an oracle to count the models of $\bthree$ for us (this is similar to what happens in sequent calculi for $\CPL$, see \cite{Antonelli2021}).

\subsection{The Computational Fragment of $\ICPL$.}

From the perspective of the CHC, however, the natural deduction system just sketched is not what we are looking for. As $\ICPL$ contains Boolean logic, in order to relate proofs and programs, one should first choose among the several existing constructive interpretations of classical logic. 
Yet, in our previous examples Boolean formulas were never those to be proved; rather, they were used as \emph{semantic constraints} that programs may or may not satisfy (for example, when saying that a program, depending on some event $\omega$, yields a proof of $A$ \emph{when $\omega$ satisfies $\bone$}, or that 
a program has $q$ chances of yielding a proof of $A$ \emph{when $\bone$ has measure at least $q$}).

Would it be possible, then, to somehow \emph{separate} purely constructive reasoning from Boolean semantic reasoning within formulas and proofs of $\ICPL$? The following lemma suggests that this is indeed possible.

\begin{lemma}[Decomposition Lemma]\label{lemma:decomposition}
For any formula $A$ of $\ICPL$ there exist Boolean formulas $\bone_{v}$ and purely intuitionistic formulas $A_{v}$ (i.e.~formulas containing no Boolean variables), where $v$ varies over all possible  valuations of the Boolean variables in $A$, 
 such that 
$$\vDash A \leftrightarrow \bigvee_{v} \bone_{v}\land A_{v}$$.
\end{lemma}
\begin{proof}[Proof sketch]
The idea is to let 
$\bone_{v}$ be the formula characterizing $v$, i.e.~the 
 conjunction of all variables true in $v$ and of all negations of variables false in $v$, and let  $A_{v}$ be obtained from $A$ by replacing each Boolean variable by either $\TOP$ or $\BOT$, depending on its value under $v$. 
\end{proof}

By Lemma \ref{lemma:decomposition}, any sequent $\Gamma\vdash A$ of $\ICPL$ can be associated with a family of intuitionistic sequents of the form 
$\Gamma_{v}, \bone_{v}\vdash A_{v}$ where $v$ ranges over all valuations of the Boolean variables of $\Gamma$ and $A$, and $\Gamma_{v},\bone_{v},A_{v}$ are as in Lemma \ref{lemma:decomposition}. 
We will note such special sequents as $\Gamma_{v}\vdash \bone_{v}\Pto A_{v}$, in order to highlight the special role played by the Boolean formula $\bone_{v}$.

The sequents $\Gamma \vdash \bone\Pto A$ have a natural computational interpretation: they express program specifications of the form ``$\Pi$ yields a proof of $A$ from $\Gamma$ whenever its sampled function satisfies $\bone$''. 
By the way, Lemma \ref{lemma:decomposition}, 
ensures that, \emph{modulo} Boolean reasoning, 
logical arguments in $\ICPL$ can be reduced to (families of) arguments of this kind.


Let then $\ICPLL$ be the fragment formed by the purely intuitionistic formulas of $\ICPL$, which are defined by:
$$
A::= \CC p\midd A\to A \midd \BOX^{q}A
$$
where $q\in (0,1]\cap \BB Q$. 
For simplicity, and since this is enough for the CHC, we take here implication as the only connective, 
yet all other propositional connectives could be added.
As formulas do not contain Boolean variables, counting quantifiers in $\ICPLL$ are not named.
We use $\B C^{q_{1}* \dots * q_{n}}A$ (or even $\B C^{\vec q}$, for simplicity) as an abbreviation for $\B C^{q_{1}}\dots \B C^{q_{n}}A$.
%
 

The natural deduction system $\NDCPLL$ for $\ICPLL$ is defined by the rules illustrated in Fig.~\ref{fig:icplrules} 
(where in ($\B C$I) it is assumed that $\FN(\bone)\cap \FN(\bthree)=\emptyset$).
A few rules involve \emph{semantic premises} of the forms $\bone \vDash \btwo$ and $\mu(\bone)\geq q$. 
Beyond standard intuitionistic rules, $\NDCPLL$ comprises structural rules to manipulate Boolean formulas, and introduction and elimination rules for the counting quantifier. 
The rule ($\BOT$) yields \emph{dummy} proofs of any formula, i.e.~proofs which are correct for \emph{no} possible event; the rule ($\mathsf m$) 
combines two proofs $\Pi_{1},\Pi_{2}$ of the same formula into one single proof $\Pi'$, 
with the choice depending on the value of some Boolean variable $\bvar_{a}^{i}$ ($\Pi'$ is thus something like $\Pi_{1}\oplus_{a}\Pi_{2}$). 
The introduction rule for $\B C^{q}$ is similar to rule \eqref{rule0}. It is explained as follows: if $\Pi$, in the ``environment'' $\omega+\omega'\in(2^{\BB N})^{X\cup\{a\}}\simeq(2^{\BB N})^{X} \times 2^{\BB N}$, yields a proof of $A$ whenever $\omega+\omega'$ 
satisfies the two \emph{independent} constraints $\bone$ and $\bthree$ (i.e.~$\omega\vDash \bone$ and $\omega'\vDash \bthree$), then by randomly choosing $\omega'\in 2^{\BB N}$, we have at least $q\geq \mu(\bthree)$ chances of getting a proof of $A$
($\Pi'$  is thus something like $\nu a.\Pi$). 
Finally, the elimination rule ($\B C$E) essentially turns a proof of $A\to B$ into a proof of $\B C^{q}A\to \B C^{qs}B$. As we'll see, this rule captures CbV function application.

%

\begin{figure}[t]
\fbox{
\begin{minipage}{0.48\textwidth}
\adjustbox{scale=0.85}{
\begin{minipage}{\textwidth}
\begin{center}
Identity Rule 
$$
\AXC{}
\RL{(id)}
\UIC{$ \Gamma,A \vdash \bone \Pto A$}
\DP
$$

\bigskip

Structural Rules
$$
\AXC{$\bone \vDash \BOT$}
\RL{($\BOT$)}
\UIC{$\Gamma \vdash \bone \Pto A$}
\DP
$$

$$
\AXC{$\Gamma \vdash \btwo \Pto A$}
\AXC{$\Gamma \vdash \bthree \Pto A$}
\AXC{$\bone \vDash (\btwo \land \bvar_{a}^{i})\lor (\bthree \land \lnot \bvar_{a}^{i})$}
\RL{($\mathsf m$)}
\TIC{$\Gamma \vdash \bone \Pto A$}
\DP
$$

\bigskip

Logical Rules
$$
\AXC{$ \Gamma, A \vdash \bone \Pto B$}
\RL{($\to$I)}
\UIC{$\Gamma \vdash \bone \Pto (A\to B)$}
\DP
$$

$$
\AXC{$ \Gamma \vdash \bone \Pto (A\to B)$}
\AXC{$ \Gamma \vdash \bone \Pto A$}
\RL{($\to$E)}
\BIC{$\Gamma \vdash \bone \Pto B$}
\DP
$$

\bigskip

Counting Rules
$$
\AXC{$\Gamma \vdash \bone \land \bthree\Pto A$}
\AXC{$\mu(\bthree)\geq q$}
\RL{($\BOX$I)}
\BIC{$\Gamma \vdash \bone \Pto \BOX^{q}A$}
\DP
$$

$$
\AXC{$\Gamma \vdash \bone \Pto \BOX^{q}A$}
\AXC{$ \Gamma, A \vdash \bone \Pto B$}
\RL{($\BOX$E)}
\BIC{$\Gamma \vdash \bone \Pto \BOX^{qs}B$}
\DP
$$
%

%

\end{center}
\end{minipage}
}
\end{minipage}
}
\caption{Rules of $\NDCPLL$.}
\label{fig:icplrules}
\end{figure}

\begin{example}
Fig.~\ref{fig:ex0} illustrates a proof $\Pi_{\frac{1}{2}\mathsf{id}}$ of $\B C^{\frac{1}{2}}(A\to A)$ obtained by first ``mixing'' an exact proof of $A\to A$ with a dummy one, and then introducing the counting quantifier. 
Fig.~\ref{fig:ex1} illustrates a derivation of 
$ \B C^{q}(A\to A) \to A \to \B C^{q*q}B$.

\end{example}

%
%

\begin{figure*}
\fbox{
\begin{subfigure}{0.98\textwidth}
\begin{center}
\adjustbox{scale=0.66}
{$
\Pi_{\frac{1}{2}\mathsf{id}}\quad  =   \quad 
\AXC{$A\vdash \bvar_{a}^{i}\Pto A$}
\RL{($\to$I)}
\UIC{$\vdash \bvar_{a}^{i}\Pto (A\to A)$}
\AXC{}
\RL{($\BOT$)}
\UIC{$\vdash \BOT \Pto (A\to A)$}
\AXC{$\bvar_{a}^{i}\vDash (\bvar_{a}^{i}\land \bvar_{a}^{i})\lor (\lnot\bvar_{a}^{i}\land \BOT)$}
\RL{($\mathsf m$)}
\TIC{$\vdash \bvar_{a}^{i}\Pto (A\to A)$}
\AXC{$\mu(\bvar_{a}^{i})\geq \frac{1}{2}$}
\RL{($\B C$I)}
\BIC{$\vdash \B C^{\frac{1}{2}}(A\to A)$}
\DP
$}
\end{center}
\caption{}
\label{fig:ex0}
\end{subfigure}
}

\medskip

\fbox{
\begin{subfigure}{0.98\textwidth}
\adjustbox{scale=0.66}{
\begin{minipage}{\textwidth}
\begin{tabular}{c c c }
$\Pi
$ & $=$ & $ 
\AXC{$\B C^{q}(A\to A), A \vdash \B C^{q}(A\to A)$}
\AXC{$A\to A, A \vdash A\to A$}
\AXC{$A\to A, A \vdash A$}
\RL{($\to$E)}
\BIC{$A\to A, A \vdash A$}
\RL{($\B C$E)}
\BIC{$\B C^{q}(A\to A), A \vdash \B C^{q}A$}
\AXC{$\B C^{q}(A\to A), A,A \vdash \B C^{q}(A\to A)$}
\AXC{$A\to A, A,A \vdash A\to A$}
\AXC{$A\to A, A,A \vdash A$}
\RL{($\to$E)}
\BIC{$A\to A,A, A \vdash A$}
\RL{($\B C$E)}
\BIC{$\B C^{q}(A\to A), A,A \vdash \B C^{q}A$}
\RL{($\B C$E)}
\BIC{$\B C^{q}(A\to A),A \vdash \B C^{q*q}A$}
\doubleLine
\RL{($\to$I)}
\UIC{$\vdash \B C^{q}(A\to A)\to  A\to \B C^{q*q}A$}
\DP
%
%
%
%
%
$
\end{tabular}
\end{minipage}
}
\caption{}
\label{fig:ex1}
\end{subfigure}
}
\caption{Examples of derivations in $\NDCPLL$.}
\end{figure*}
A precise connection between the proof systems for 
$\ICPL$ and $\ICPLL$ is discussed in Appendix~\ref{app:relatingICPLICPLL}.


\subsection{Normalization in $\ICPLL$}

From the CHC perspective, natural deduction proofs correspond to programs, and normalization corresponds to execution. Let us look at normalization in $\ICPLL$, then.

The two main normalization steps are ($\to$I/$\to$E) and ($\B C$I/$\B C$E).
The cuts ($\to$I/$\to$E) are eliminated, as usual, 
by means of the admissible substitution rule
$$
\AXC{$\Gamma \vdash\bone\Pto A$}
\AXC{$\Gamma , A\vdash \bone\Pto B$}
\RL{(subst)}
\BIC{$\Gamma \vdash \bone\Pto B$}
\DP
$$
The normalization step ($\B C$I/$\B C$E), illustrated in Fig.~\ref{fig:cutelim2},
also applies the rule (subst) to the premiss of ($\B C$I) and the minor premiss of ($\B C$E), and permutes the rule ($\B C$I) downwards.
For example, if we cut the proof $\Pi$ from Fig.~\ref{fig:ex1} with $\Pi_{\frac{1}{2}\mathsf{id}}$ from Fig.~\ref{fig:ex0} (by letting $q=\frac{1}{2}$), normalization duplicates $\Pi$, that is, it duplicates the choice between the correct and the dummy proof of $A\to A$, yielding a normal proof of $A\to \B C^{\frac{1}{2}*\frac{1}{2}}A$.
%
%
%
%
The other normalization steps (illustrated in App.~\ref{app:permu}) permute ($\mathsf m$) with other rules, and  correspond to the \emph{permuting reductions} for the probabilistic $\lambda$-calculus introduced in the next section. 
To make the study of normalization as simple as possible, we did not consider a ``multiplication rule'' to pass from $\B C^{q}\B C^{s}A$ to $\B C^{qs}A$, as this would introduce other normalization steps. Nevertheless, in App.~\ref{app:CbNProofSystem} we study an alternative  ``CbN'' proof-system $\NDCPLL$ also comprising this rule.
%
As a by-product of the CHC developed in the following sections, we will obtain a strong normalization theorem for $\ICPLL$ (Corollary \ref{thm:cutelim}).

In Section \ref{section5} it will be shown (Prop.~\ref{prop:stab}) that, once proofs in $\NDCPLL$ are seen as probabilistic programs, the two normalization steps ($\to$I/$\to$E) and 
($\B C$I/$\B C$E) simulate CbN and CbV evaluation, respectively.
%

\begin{figure*}[t]
%
\fbox{
\begin{minipage}{0.98\textwidth}
\adjustbox{scale=0.8}{
\begin{tabular}{c  c  c}
\AXC{$\Sigma$}
\noLine
\UIC{$\Gamma \vdash \bone \land \bthree  \Pto A$}
\AXC{$\mu(\bthree  )\geq q$}
\RL{($\B C$I)}
\BIC{$\Gamma \vdash \bone \Pto \B C^{q}A$}
\AXC{$\Pi$}
\noLine
\UIC{$\Gamma,A \vdash \bone \Pto B$}
\RL{($\B C$E)}
\BIC{$\Gamma \vdash \bone \Pto \B C^{qs}B$}
\DP
&
$\leadsto$
&
\AXC{$\Sigma$}
\noLine
\UIC{$\Gamma \vdash \bone \land \bthree  \Pto A$}
\AXC{$\Pi${\small$[\bone\mapsto \bone\land \bthree]$}}
\noLine
\UIC{$\Gamma,A \vdash \bone  \land \bthree\Pto B$}
\RL{(subst)}
\BIC{$\Gamma \vdash \bone \land \bthree  \Pto B$}
\AXC{$\mu(\bthree)\geq qs$}
\RL{($\B C$I)}
\BIC{$\Gamma \vdash \bone\Pto \B C^{qs}B$}
\DP
\end{tabular}
}
\end{minipage}
}
\caption{Normalization step ($\B C$I)/($\B C$E) of $\NDCPLL$.}
\label{fig:cutelim2}
%
\end{figure*}

\section{The Probabilistic Event Lambda Calculus}\label{section4}
In this section we introduce the computational side of the CHC, that is, a variant of the probabilistic event $\lambda$-calculus $\EVL$ from \cite{DLGH}, with choices depending on events from the Cantor space. We will then discuss how terms of $\EVL$ yield distributions of values, and we define two notions of probabilistic normalization for such distributions. We finally introduce a further variant $\EVLL$ of $\EVL$ which provides a smoother representation of CbV functions.

\subsection{A $\lambda$-Calculus Sampling from the Cantor Space}

The terms of $\EVL$ are defined by the grammar below:
$$
t::= x\midd \lambda x.t\midd tt\midd t\oplus_{a}^{i}u \midd \nu a.t
$$
with $a\in \CC A$, and $i\in \mathbb N$. The intuition is that
 $\nu a.$ samples some function $\omega$ from the Cantor space, and  
 $t\oplus_{a}^{i}u$ yields either $t$ or $u$ depending on the value $\omega(a)(i)\in\{0,1\}$.
 In the following we let $t\oplus^{i} u$ be an abbreviation for $\nu a.t\oplus_{a}^{i}u$ (supposing $a$ does not occur free in either $t$ or $u$).
  For any term $t$, finite set $X$ and $\omega \in (2^{\BB N})^{X}$, let 
 $\pi^{\omega}_{X}(t)$, the ``application of $\omega$ to $t$ through $X$'', be defined as: 
 \begin{center}
 \adjustbox{scale=0.9}{
 \begin{minipage}{0.5\textwidth}
 \begin{align*}
 \pi^{\omega}_{X}(x)& =x \\
  \pi^{\omega}_{X}(\lambda x.t)& =\lambda x.\pi^{\omega}_{X}(t) \\
   \pi^{\omega}_{X}(tu)& =\pi^{\omega}_{X}(t)\pi^{\omega}_{X}(u) \\
 \pi^{\omega}_{X}(t\oplus_{a}^{i}u) &= \begin{cases}
 \pi^{\omega}_{X}(t) & \text{ if }a\in X \text{ and } \omega(a)(i)=1 \\
  \pi^{\omega}_{X}(u) & \text{ if }a\in X \text{ and } \omega(a)(i)=0 \\
  \pi^{\omega}_{X}(t) \oplus_{a}^{i}\pi^{\omega}_{X}(u)& \text{ otherwise}
  \end{cases}\\
  \pi^{\omega}_{X}(\nu b.t) & = \nu b. \pi^{\omega}_{X}(t)
 \end{align*}
 \end{minipage}
 }
 \end{center} 
 
In usual randomized $\lambda$-calculi, program execution is defined so as to be inherently probabilistic: for example a term $t\oplus u$ can reduce to either $t$ or $u$, with probability $\frac{1}{2}$. In this way, chains of reduction can be described as \emph{stochastic Markovian sequences} \cite{Puterman1994}, leading to formalize the idea of \emph{normalization with probability} $r\in[0,1]$ (see \cite{Bournez2002}).

By contrast, reduction in $\EVL$ is fully deterministic: beyond the usual (and un-resticted) $\beta$-rule $(\lambda x.t)u \redbeta t[u/x]$, it comprises a \emph{permutative reduction} $t\redperm u$ defined by the rules in Fig.~\ref{fig:permutations} (where $(a,i)<(b,j)$ if either $\nu b$ occurs in the scope of $\nu a$, or $a=b$ and $i<j$).  
Intuitively, permutative reductions implement probabilistic choices by computing the full tree of possible choices. 
For example, given terms $t_{1},t_{2}, u_{1},u_{2}$, one can see that the term 
$\nu a.(t_{1}\oplus_{a}^{0}t_{2})(u_{1}\oplus_{a}^{1}u_{2})$ reduces to $
\nu a.  (t_{1}u_{1}\oplus_{a}^{1} t_{1}u_{2}) \oplus_{a}^{0}( t_{2}u_{1}\oplus_{a}^{1} t_{2}u_{2})
$, hence displaying all possible alternatives.

\begin{figure}[t]
\fbox{
\begin{minipage}{0.48\textwidth}
\adjustbox{scale=0.9}{
\begin{minipage}{1.1\textwidth}
\begin{align*}
 t\oplus_{a}^{i}t & \redperm t  \tag{$\mathsf{i}$} \\
 (t\oplus_{a}^{i}u)\oplus_{a}^{i} v & \redperm t\oplus_{a}^{i}v \tag{$\mathsf{c}_{1}$}\label{eq:c1} \\
 t\oplus_{a}^{i}(u\oplus_{a}^{i}v) & \redperm t\oplus_{a}^{i}v  \tag{$\mathsf{c}_{2}$} \\
 \lambda x.(t\oplus_{a}^{i}u) & \redperm (\lambda x.t)\oplus_{a}^{i} (\lambda x.u)  \tag{$\oplus\lambda$} \\
 (t\oplus_{a}^{i}u)v & \redperm (tu)\oplus_{a}^{i}(uv)  \tag{$\oplus\mathsf{f}$} \\
 t(u\oplus_{a}^{i}v) & \redperm (tu)\oplus_{a}^{i}(tv)  \tag{$\oplus\mathsf{a}$}\label{eq:oa} \\
 (t\oplus_{a}^{i}u)\oplus_{b}^{j} v & \redperm (t\oplus_{b}^{j}v)\oplus_{a}^{i}(u\oplus_{b}^{j}v) & ((a,i)< (b,j)) \tag{$\oplus\oplus_{1}$}  \\
 t\oplus_{b}^{j}(u\oplus_{a}^{i} v) & \redperm (t\oplus_{b}^{j}u)
\oplus_{a}^{i}(t\oplus_{b}^{j} v) & ((a,i)<(b,j))
\tag{$\oplus\oplus_{2}$}
 \\
\nu b. (t\oplus_{a}^{i} u) & \redperm (\nu b.t)\oplus_{a}^{i}(\nu b.u) &  (a\neq b)
\tag{$\oplus\nu$}\label{eq:cnu} \\
\nu a.t & \redperm t &  (a\notin \FN(t)) 
\tag{$\lnot\nu$}\label{eq:notnu}\\
\lambda x.\nu a.t & \redperm \nu a. \lambda x.t 
\tag{$\nu\lambda$}\label{eq:nulambda}\\
(\nu a.t)u  & \redperm  \nu a.(tu)\tag{$\nu\mathsf{f}$}\label{eq:nuapp}
 \end{align*}
 \end{minipage}
 }
\end{minipage}
}
\caption{Permutative reductions.}
\label{fig:permutations}
\end{figure}

The fundamental properties of $\EVL$ are the following:
\begin{theorem}[\cite{DLGH}]\label{thm:confluence}
$\redperm$ is confluent and strongly normalizing. 
 \emph{Full reduction}  $\redall\ := \ \redbeta\cup \redperm$ is confluent.
\end{theorem}


The existence and unicity of normal forms for $\redperm$ (that we call \emph{permutative normal forms}, PNFs for short) naturally raises the question of what these normal forms represent.

Let $\CC T$ indicate the set of PNFs containing no free name occurrence.
For any $t\in \CC T$, the PNF of $t$ can be of two forms: 
either $t$ starts with a generator, i.e.~$t=\nu a.t'$, and $t'$ is a tree of $a$-labeled choices $\oplus_{a}^{i}$ whose leaves form a finite set of $\CC T$ (the \emph{support} of $t'$, $\mathrm{supp}(t)$). 
Otherwise, $t$ is of the form 
$\lambda x_{1}.\dots.\lambda x_{n}.t't_{1}\dots t_{p}$, where $t'$ is either a variable or a $\lambda$. We call these last terms \emph{pseudo-values}, and we let $\CC V\subseteq \CC T$ indicate the set formed by them.

%
%
%
%
%
%
Using this decomposition, any $t\in \CC T$ can be associated in a unique way with a \emph{(sub-)distribution} of pseudo-values $\CC D_{t}: \CC V\to [0,1]$ by letting $\CC D_{t}(v)=\delta_{t}$ when $t\in \CC V$, and 
$$
\CC D_{t}(v) =  \sum_{u\in \mathrm{supp}(t')}\CC D_{u}(v) \cdot \mu(\{\omega\in 2^{\BB N}\mid \pi^{\omega}_{\{a\}}(t')=u\}) 
$$
if $t=\nu a.t'$. Intuitively, $\CC D_{t}(v)$ measures the probability of finding $v$ by iteratively applying to $t$ random choices of events from the Cantor space any time a $\nu$ is found. 

%
%

\subsection{Probabilistic (Head) Normalization.}

Given a term $t\in \CC T$, the questions ``is $t$ in normal form?'' and ``does $t$ reduce to a normal form?'' have univocal yes/no answers, because $\redall$ is deterministic. However, if we think of $t$ rather as $\CC D_{t}$, the relevant questions become ``with what probability is $t$ in normal form?'' and ``with what probability does $t$ reduce to normal form?''.

To answer this kind of questions we will introduce functions $\NHNF(t)$, $\NNF(t)$ measuring the probability that $t$ reduces to a normal form.
%
%
%

%

%
Let us consider head-normal forms, first. 
A \emph{head-reduction} $t\redall_{\mathsf h} u$ is either a $\redperm$-reduction or a $\redbeta$-reduction 
of the form 
$\TT R[\lambda \vec x.(\lambda y.t)uu_{1}\dots u_{n}]\redbeta \TT R[\lambda \vec x.t[u/x]u_{1}\dots u_{n}]$, 
where $\TT R$ is a \emph{randomized context}, defined by the grammar 
$$
\TT R[\ ]::= [\ ]\mid \TT R[\ ]\oplus_{a}^{i}u \mid t\oplus_{a}^{i}\TT R[\ ]\mid  \nu a. \TT R[\ ].
$$
A \emph{head normal value} (in short, HNV) is a $\redall_{\mathsf h}$-normal term which is also a pseudo-value, i.e.~it is of the form $\lambda \vec x.yu_{1}\dots u_{n}$. We let $\HNF$ indicate the set of such terms. 
\begin{definition}
For any $t\in \CC T$,  
$\HNF(t) :=\sum_{v\in \HNF}\CC D_{t}(v)$ and 
$\NHNF(t):= \sup\{ \HNF(u)\mid t\redall_{\mathsf h}^{*} u\}$.
When $\NHNF(t)\geq q$, we say that $t$ \emph{yields a HNV with probability at least $q$}.
\end{definition} 
For example, if $t= \nu a. (\lambda x\lambda y.(y\oplus_{a}^{i}I )x)u$, where $u= I \oplus^{j}\Omega$, then $\NHNF(t)= \frac{3}{4}$. Indeed,  we have 
$$t\redall^{*}_{\mathsf h} \nu a.( \lambda y. y(\nu b.I\oplus_{b}^{j}\Omega)) \oplus_{a}^{i} (\nu b'.I\oplus_{b'}^{j} \Omega)$$
and 
 three over the four possible choices (corresponding to choosing between either left or right for both $\nu a$ and $\nu b'$) yield a HNV. Observe that the choice about $\nu b$ does not matter, since $\lambda y.yu$ is already a HNV.

Let us now consider normal forms. The first idea might be to define a similar function $\NF(t)=\sum_{v \text{ normal form}}\CC D_{t}(v)$. 
However, with this definition a term like $t=\lambda x.x(\nu a. I \oplus_{a}^{0}\Omega)$ would have probability 0 of yielding a normal form. Instead, our guiding intuition here is that $t$ should yield a normal form with probability $\frac{1}{2}$, i.e.~depending on a choice for $a$. 
This leads to the following definition:
\begin{definition}
For any $t\in \CC T$, $\NF(t)$ is defined by:
\begin{varitemize}
\item if $t=\lambda \vec x.yu_{1}\dots u_{n}\in \HNF$, then $\NF(t):= \prod_{i=1}^{n}\NF(u_{i})$;
\item otherwise $\NF(t):= \sum_{u\in \HNF}\NF(u)\cdot \CC D_{t}(u)$.
\end{varitemize}
We let $\NNF(t)=\sup\{\NF(u)\mid t\redall^{*}u\}$ and, if $\NNF(t)\geq q$, we say that $t$ \emph{yields a normal form with probability at least $q$}.\end{definition} 
For example, for the term $t$ considered above, $\NNF(t)=\frac{4}{8}=
\frac{1}{2}$: four over the eight possible choices for $\nu a$, $\nu b$ and $\nu b'$ yield a normal form (i.e.~either choose left for $\nu a$ and $\nu b$ and choose anything for $\nu b'$, or 
choose right for $\nu a$, left for $\nu b'$, and choose anything for $\nu b$).


%

\subsection{Extending $\EVL$ with CbV Functions}

$\EVL$ makes it possible to encode a CbV redex like $\nu a. \B 2 (I\oplus_{a}^{0}\Omega)$, as we have seen. However, 
in view of the functional interpretation of $\ICPLL$, it would be convenient also to be able to represent the CbV \emph{functions} mapping $\nu a. (u\oplus_{a}v) $ onto $\nu a.\B 2 (u\oplus_{a}^{0}v)$. 
%
%
A simple way to do this is by enriching the language of $\EVL$ with a ``CbV application'' operator $\{t\}u$, with suitable permutative rules. Let $\EVLL$ indicate this extension of terms of $\EVL$ with $\{\}$. While $\beta$-reduction for $\EVLL$ is the same as for $\EVL$, permutative reduction $\redpermm$ is defined by all rules in Fig.~\ref{fig:permutations} except for \eqref{eq:notnu}, together with the three permutations below
\begin{align*}
\{t\} \nu a. u & \ \redpermm \ \nu a. tu \tag{$\{\}\nu$} \label{eq:e0}\\
\{t\oplus_{a}^{i}u\} v & \ \redpermm \ \{t\}v \oplus_{a}^{i}\{t\} v \tag{$\{\}\oplus_{1}$}\label{eq:e1} \\ 
\{t\} (u \oplus_{a}^{i}v) & \ \redpermm \ \{t\}u \oplus_{a}^{i}\{t\} v\tag{$\{\}\oplus_{2}$}\label{eq:e2} 
\end{align*}
For instance, a CbV Church numeral can be encoded in $\EVLL$ as $ {\B 2}^{\mathrm{CbV}}:=\lambda f. \{\B 2 \}f 
$, since one has ${\B 2}^{\mathrm{CbV}}(\nu a.u\oplus_{a}^{i}v) \redbeta
\{\B 2\} \nu a.u\oplus_{a}^{i}v \redpermm
\nu a. \B 2 (u\oplus_{a}^{i}v)$.

The fundamental properties of $\EVL$ extend to $\EVLL$:
\begin{proposition}
$\redpermm$ is confluent and strongly normalizing. 
 {Full reduction}  $\redalll\ := \ \redbeta\cup \redpermm$ is confluent.
\end{proposition}
In the Appendix~\ref{app:PElambda} it is shown how the definitions of $\CC D_{t}$ and $\NHNF(t)$ extend to terms of $\EVLL$. 

\section{The Correspondence, Statically and Dynamically}\label{section5}

In this section we present the core of the CHC. First, we introduce two type systems $\TCI$ and $\LCPL$ for $\EVL$ and $\EVLL$, respectively, which extend the simply typed $\lambda$-calculus with counting quantifiers. Then, we show that each proof $\Pi$ in $\ICPLL$ can be  associated with a typing derivation in $\LCPL$ of some probabilistic term $t^{\Pi}$, in such a way that normalization of $\Pi$ corresponds to reduction of $t^{\Pi}$. 
Observe that translating the rule ($\B C$E) requires the CbV application operator $\{\}$. In the Appendix~\ref{app:CbNProofSystem} we show that an alternative ``CbN'' proof-system for $\ICPLL$ can be translated without $\{\}$, and thus into $\TCI$.

\subsection{Two Type Systems with Counting Quantifiers}

Both type systems $\TCI$ and $\LCPL$ extend the simply typed $\lambda$-calculus with counting quantifiers, but in a slightly different way: in $\TCI$ types are of the form $\B C^{q}\FF s$, i.e.~prefixed by \emph{exactly one} counting quantifier, while in $\LCPL$ types are of the form $\B C^{\vec q}\FF s$, i.e.~prefixed by a (possibly empty) list of counting quantifiers. 

More precisely, the types (denoted $\FF s, \FF t$) of $\TCI$ and $\LCPL$ are generated by the grammars below:
\begin{align*}
 \FF s &  :: = \B C^{ q}\sigma  &  \sigma &::=o\midd \FF s\To \sigma
\tag{$\TCI$}\\
 \FF s &  :: = \B C^{\vec q}\sigma  &  \sigma &::=o\midd \FF s\To \sigma
 \tag{$\LCPL$}
\end{align*}
where, in both cases, the $q$s are chosen in $(0,1]\cap \BB Q$. 
%
%
%
%

\emph{Judgements} in both systems are of the form $\Gamma \vdash^{X} t: \bone \Pto \FF s$, where $\Gamma$ is a set of type declarations $x_{i}: \FF s_{i}$ of pairwise distinct variables, $t$ is a term of $\EVL$ (resp.~of $\EVLL$), $\bone$ a Boolean formula, and $X$ a finite set of names with $\FN(t), \FN(\bone)\subseteq X$. 
The intuitive reading of $\Gamma \vdash^{X} t: \bone \Pto \FF s$ is that, whenever $\omega\in (2^{\BB N})^{X}$ satisfies $\bone$, then $\pi^{\omega}_{X}(t)$ is a program from $\Gamma$ to $\FF s$.
%

The \emph{typing rules} of $\LCPL$, which are essentially derived from those of $\ICPLL$, are illustrated in Fig.~\ref{fig:rulesLCPL}, where in the rule ($\mu$) it is assumed that $\FN(\bone)\subseteq X$, $\FN(\bthree)\subseteq \{a\}$ and $a\notin X$.
The rule ($\vee$) allows one to merge $n$ typing derivations for the same term; in particular, with $n=0$, one has that $\Gamma \vdash^{X}t: \BOT \Pto \FF s$ holds for \emph{any} term $t$.
The rule ($\oplus$) is reminiscent of the rule ($\mathsf m$) of $\ICPLL$,
the rule ($\mu$) is reminiscent of the rule ($\B C$I) of $\ICPLL$, while the rule ($\{\}$) for CbV application is reminiscent of the rule ($\B C$E).

The typing rules of $\TCI$ coincide with those of $\LCPL$ (with $\B C^{\vec q}$ replaced everywhere by $\B C^{q}$), except for the rule ($\{\}$), which is obviously absent, and for the rule ($\mu$), which is adapted as follows to the presence of exactly one counting quantifier:
$$
\AXC{$\Gamma \vdash^{X\cup\{a\}}t:\bone\land \bthree \Pto \B C^{q}\sigma$}
\AXC{$\mu(\bthree)\geq s$}
\RL{($\mu'$)}
\BIC{$\Gamma \vdash^{X}\nu a.t:\bone \Pto \B C^{qs}\sigma$}
\DP
$$


\begin{figure}[t]
\fbox{
\begin{minipage}{0.48\textwidth}
\adjustbox{scale=0.73}{
\begin{minipage}{\textwidth}
\begin{center}
Identity rule:
$$
\AXC{$\mathrm{FN}(\bone)\subseteq X$}
\RL{(id)}
\UIC{$\Gamma, x:\FF s \vdash^{X} x: \bone \Pto \FF s $}
\DP
$$

\bigskip

Structural rule:
$$
\AXC{$
\Big\{ \Gamma \vdash^{X}t: \bone_{i}\Pto \FF s\Big\}_{i=1,\dots,n}$}
\AXC{$\bone \vDash^{X}\bigvee_{i}\bone_{i}$}
\RL{$(\lor)$}
\BIC{$\Gamma \vdash^{X}t: \bone\Pto \FF s$}
\DP
$$

%

%

\bigskip

Plus rule:
$$
\AXC{$ \Gamma \vdash^{X\cup\{a\}}t: \btwo\Pto \FF s$}
\AXC{$ \Gamma \vdash^{X\cup\{a\}}u: \bthree\Pto \FF s$}
\AXC{$\bone \vDash (\bvar_{a}^{i}\land \btwo)\lor(\lnot \bvar_{a}^{i}\land \bthree)$}
\RL{$(\oplus)$}
\TIC{$\Gamma \vdash^{X\cup\{a\}} t\oplus^{i}_{a}u :  \bone\Pto \FF s$}
\DP
$$
%
%
%

\bigskip

Arrow rules:
$$
\AXC{$\Gamma, x:\FF s \vdash^{X} t: \bone\Pto \B C^{\vec q}\tau$}
\RL{$(\lambda)$}
\UIC{$\Gamma \vdash^{X} \lambda x.t : \bone\Pto \B C^{\vec q}(\FF s\To \tau)$}
\DP
$$

$$\AXC{$\Gamma \vdash^{X}t:\btwo\Pto  \B C^{\vec q}(\FF s \To \tau)$}
\AXC{$ \Gamma \vdash^{X} u: \bthree\Pto \FF s$}
\AXC{$\bone\vDash^{X}\btwo\land \bthree$}
\RL{$(@)$}
\TIC{$\Gamma \vdash^{X} tu: \bone\Pto \B C^{\vec q}\tau$}
\DP
$$

$$\AXC{$\Gamma \vdash^{X}t:\btwo\Pto \B C^{\vec q}(\FF s \To \tau)$}
\AXC{$ \Gamma \vdash^{X} u: \bthree\Pto \B C^{r}\FF s$}
\AXC{$\bone\vDash^{X}\btwo\land \bthree$}
\RL{$(\{ \})$}
\TIC{$\Gamma \vdash^{X} \{t\}u: \bone\Pto \B C^{rs}\B C^{\vec q} \tau$}
\DP
$$

\bigskip

Counting rule:
$$
\AXC{$\Gamma \vdash^{X\cup \{a\}} t: \bone\land \bthree\Pto \FF s$}
\AXC{$\mu(\bthree)\geq q$}
\RL{$(\mu)$}
\BIC{$\Gamma \vdash^{X} \nu a.t: \bone\Pto \B C^{q}\FF s$}
\DP
$$

\

\end{center}
\end{minipage}
}
\end{minipage}
}

\caption{Typing rules of $\LCPL$.}
\label{fig:rulesLCPL}
\end{figure}
Both systems enjoy the subject reduction property:

\begin{restatable}[Subject Reduction]{proposition}{subje}\label{prop:subje}
If $\Gamma \vdash^{X}t: \bone \Pto \FF s$ in $\TCI$ (resp.~in $\LCPL$) and $t \redall u$ (resp.~$t\redalll u$), then 
$\Gamma \vdash^{X}u: \bone \Pto \FF s$.
\end{restatable}

The choice of considering arrow types of the form $\B C^{\vec q}(\FF s\To \sigma)$ (hence of never having a counting quantifier \emph{to the right} of $\To$, as in e.g.~ $\FF s\To \B C^{q}\sigma$) was made to let the rule ($\lambda$) be permutable over ($\mu$), as required by the permuting rule \eqref{eq:nulambda}.

 \begin{example}
Fig.~\ref{fig:exm1} illustrates typing derivations in $\LCPL$ for a version of the CbN Church numeral $\B 2^{\mathrm{CbN}}= \lambda y.\lambda x.\{y\}(yx)$, with type
$\B C^{q*q}(\B C^{q}(o\To o)\To (o\To o))$, and for the CbV numeral
$\B 2^{\mathrm{CbV}}$ with type $\B C^{q}(\B C^{q}(o\To o)\To (o\To o))$.

%
 
 \end{example}
 

Both $\TCI$ and $\LCPL$ can type \emph{non-normalizable} terms. For example, one can type all terms of the form 
 $I\oplus^{i}\Omega$ with  $\B C^{\frac{1}{2}}(o\To o)$ in $\LCPL$ and with 
  $\B C^{\frac{1}{2}}(\B C^{1}o\To o)$ in $\TCI$.
%
%
Actually, the failure of normalization for typable programs can be ascribed to the rule ($\vee$), as shown by the result below, proved in App.~\ref{app:detnorm} (where we let  
%
 $\Gamma \vdash_{\lnot \vee}t:\bone \Pto \FF s$ indicate that  
$\Gamma \vdash t:\bone \Pto \FF s$ is deduced without using the rule ($\vee$)).

%
%
%

\begin{restatable}[Deterministic Normalization]{theorem}{detnormalization}\label{thm:detnormalization}
In both $\TCI$ and $\LCPL$, if $\Gamma \vdash_{\lnot \vee} t: \bone \Pto \FF s$, then $t$ is strongly normalizable.
\end{restatable}

As observed in the previous section, restricting to terms of $\EVL$ having a normal form 
excludes 
the most interesting part of the calculus, which is made of terms for which normalization is inherently probabilistic. Similarly, restricting to type derivations without ($\vee$) trivializes the most interesting features of $\TCI$ and $\LCPL$, 
that is, their ability to estimate probabilities of termination. We will explore the expressiveness of these systems in this sense 
in the next section.

%
%
%


 \begin{figure*}
 \fbox{
 \begin{minipage}{0.98\textwidth}
 \begin{center}
 \adjustbox{scale=0.75}{
 $
 \AXC{$y: \B C^{q}(o\To o), x: \B C^{q}o \vdash y:\TOP \Pto \B C^{q}(o\To o)$}
 \AXC{$y: \B C^{q}(o\To o), x: \B C^{q}o \vdash y:\TOP \Pto \B C^{q}(o\To o)$}
  \AXC{$y: \B C^{q}(o\To o), x: \B C^{q}o \vdash x:\TOP \Pto o$} 
  \RL{($@$)}
 \BIC{$y: \B C^{q}(o\To o), x: \B C^{q}o \vdash yx:\TOP\Pto \B C^{q}o$}
 \RL{($\{\}$)}
 \BIC{$y: \B C^{q}(o\To o), x: \B C^{q}o \vdash \{y\}(yx): \TOP\Pto \B C^{q*q}o$}
 \doubleLine
 \RL{($\lambda$)}
 \UIC{$\vdash \B 2^{\mathrm{CbN}}: \TOP \Pto \B C^{q*q}(\B C^{q}(o\To o)\To (o\To o))$}
 \DP
 $
   }
%
%

%
 \adjustbox{scale=0.75}{
  $
  \AXC{$\vdots$}
  \noLine
  \UIC{$y: \B C^{q}(o\To o) \vdash \B 2:\TOP \Pto (o\To o)\To (o\To o)$}
 \AXC{$y: \B C^{q}(o\To o) \vdash y:\TOP \Pto \B C^{q}(o\to o)$}
 \RL{($\{\}$)}
 \BIC{$y: \B C^{q}(o\To o) \vdash \{\B 2\}y: \TOP\Pto \B C^{q}(o\To o)$}
%
%
 \RL{($\lambda$)}
 \UIC{$\vdash \B 2^{\mathrm{CbV}}: \TOP \Pto \B  C^{q}(\B C^{q}(o\To o)\To (o\To o))$}
 \DP
 $
 }
%
%
%
\end{center}
 \end{minipage}
 }
 \caption{Typing of CbN and CbV Church numerals in $\LCPL$.}
 \label{fig:exm1}
 \end{figure*}

\subsection{Translating $\ICPLL$ into $\LCPL$.}

We now show how derivations in $\ICPLL$ translate into typing derivations in $\LCPL$. 
In the Appendix~\ref{app:CbNProofSystem} a similar translation of a ``CbN'' proof-system for $\ICPLL$ into $\TCI$ is shown.

For any formula $A$ of $\ICPLL$, let us define a corresponding type $\FF s_{A}$ by letting 
$\FF s_{\CC p}:= o$, $\FF s_{A\to \B C^{\vec q}B}:=\B C^{\vec q}(\FF s_{A}\To \FF s_{B})$ and $\FF s_{\B C^{q}A}:= \B C^{q}\FF s_{A}$. 
Fig.~\ref{fig:ch} shows how a derivation $\Pi$ of $\Gamma \vdash \bone \Pto A$ in $\ICPL$ translates into a typing derivation $D^{\Pi}$ of $\FF s_{\Gamma} \vdash t^{\Pi}: \bone \Pto \FF s_{A}$ in $\LCPL$, with $\FN(t^{\Pi})\subseteq \FN(\bone)$, by induction on $\Pi$.
Notice that we exploit a special constant $\B c$ to translate the rule ($\BOT$).
%
%
%
%
$$
\AXC{$ \Gamma \vdash^{X} t: \btwo \Pto \FF s$}
\AXC{$\bone \vDash \btwo$}
\RL{($\vDash$)}
\BIC{$\Gamma \vdash^{X} t: \bone \Pto \FF s$}
\DP
$$
Moreover, the rule ($\to$E) translates as CbN application $tu$, while the rule ($\B C$E) translates as CbV application $\{t\}u$.

As required by the CHC, normalization steps of $\ICPLL$ are simulated by $\redalll$-reductions:

\begin{proposition}[Stability Under Normalization]\label{prop:stab}
If $\Pi \leadsto\Pi'$, then $t^{\Pi} \redalll^{*}t^{\Pi'}$.
\end{proposition}
\begin{proof}[Proof sketch]
The normalization steps in Fig.~\ref{fig:cutelim} translate into the following reductions:
\begin{align*}
(\lambda x^{\B C^{q}\FF s_{A}}.t^{\Pi}) \nu a. t^{\Sigma} &  \ \redbeta \ 
t^{\Pi}[ \nu a.t^{\Sigma} /x ]  \\
\{\lambda x^{\FF s_{A}}.t^{\Pi}\} \nu a. t^{\Sigma} &  \ \redpermm \ 
\nu a.(\lambda x^{\FF s_{A}}.t^{\Pi}) t^{\Sigma} \ \redbeta \
\nu a. t^{\Pi}[ t^{\Sigma}/x ] 
\end{align*}
All other normalization steps translate into $\redpermm$-reductions.
\end{proof}

Notice that the normalization step ($\to$I/$\to$E) translates into CbN reduction (i.e.~plain $\beta$-reduction): the ``choice'' $\nu a.t^{\Sigma}$ is directly substituted, and thus possibly duplicated; instead the 
($\B C$I/$\B C$E) translates into CbV reduction (i.e.~\eqref{eq:nuapp} followed by $\beta$-reduction): the generator $\nu a$ is first permuted down and only $t^{\Sigma}$ is substituted.

By observing that the only use of ($\vee$) coming from the translation introduces a constant $\B c$, from Theorem \ref{thm:detnormalization} we deduce, as promised, the following:

\begin{corollary}\label{thm:cutelim}
$\ICPLL$ is strongly normalizing.
\end{corollary}

%
%

%
%
%

\begin{figure*}
\fbox{
\begin{minipage}{0.98\textwidth}
\adjustbox{scale=0.8}{
\begin{minipage}{\textwidth}
\begin{center}
\begin{tabular}{c c c}
$
\AXC{}
\RL{(id)}
\UIC{$\Gamma, A \vdash \bone\Pto A$}
\DP$
 &
 $\mapsto$
  & 
 $
 \AXC{}
 \RL{(id)}
 \UIC{$\FF s_{\Gamma}, x: \FF s_{A} \vdash x:\bone\Pto \FF s_{A}$}
 \DP
 $
  \\ & & \\
  $
\AXC{$\bone \vDash \BOT$}
\RL{($\BOT$)}
\UIC{$\Gamma \vdash \bone\Pto A$}
\DP$
 &
 $\mapsto$
  & 
 $
 \AXC{$\bone \vDash \BOT$}
 \RL{($\vee$)}
 \UIC{$\FF s_{\Gamma} \vdash \B c:\bone\Pto \FF s_{A}$}
 \DP
 $
  \\ & & \\
\AXC{$\Pi$}
\noLine
\UIC{$\Gamma \vdash \btwo \Pto A$}
\AXC{$\Pi'$}
\noLine
\UIC{$\Gamma \vdash \bthree \Pto A$}
\AXC{$\bone \vDash (\btwo \land \bvar_{a}^{i})\lor (\bthree \land \lnot \bvar_{a}^{i})$}
\RL{($\mathsf m$)}
\TIC{$ \Gamma \vdash \bone \Pto A$}
\DP
& 
$\mapsto $
&
\AXC{$D^{\Pi}$}
\noLine
\UIC{$ \FF s_{\Gamma}\vdash t^{\Pi}: \btwo \Pto A$}
\AXC{$D^{\Pi'}$}
\noLine
\UIC{$\FF s_{\Gamma}\vdash t^{\Pi'}: \bthree \Pto A$}
\AXC{$\bone \vDash (\btwo \land \bvar_{a}^{i})\lor (\bthree \land \lnot \bvar_{a}^{i})$}
\RL{($\oplus$)}
\TIC{$ \FF s_{\Gamma}\vdash t^{\Pi}\oplus_{a}^{i}t^{\Pi'}: \bone \Pto \FF s_{A}$}
\DP
\\ & & \\
\AXC{$\Pi$}
\noLine
\UIC{$\Gamma,A \vdash \bone \Pto B$}
\RL{($\to$I)}
\UIC{$ \Gamma \vdash \bone \Pto (A\to B)$}
\DP
& 
$\mapsto $
&
\AXC{$D^{\Pi}$}
\noLine
\UIC{$ \FF s_{\Gamma},x:\FF s_{A}\vdash t^{\Pi}: \bone \Pto \FF s_{B}$}
\RL{($\lambda$)}
\UIC{$ \FF s_{\Gamma}\vdash \lambda x.t^{\Pi}: \bone \Pto (\FF s_{A\to B})$}
\DP
\\ & & \\
\AXC{$\Pi$}
\noLine
\UIC{$\Gamma\vdash \bone \Pto (A\to \B C^{\vec q}B)$}
\AXC{$\Sigma$}
\noLine
\UIC{$\Gamma\vdash \bone \Pto A$}
\RL{($\to$E)}
\BIC{$ \Gamma \vdash \bone \Pto  \B C^{\vec q}B$}
\DP
& 
$\mapsto $
&
\AXC{$D^{\Pi}$}
\noLine
\UIC{$\FF s_{\Gamma}\vdash t^{\Pi}: \bone \Pto\B C^{\vec q} (\FF s_{A}\to \FF s_{B})$}
\AXC{$D^{\Sigma}$}
\noLine
\UIC{$\FF s_{\Gamma}\vdash t^{\Sigma}: \bone \Pto \FF s_{A}$}
\RL{($@$)}
\BIC{$ \FF s_{\Gamma}\vdash t^{\Pi}t^{\Sigma}: \bone \Pto \B C^{\vec q} \FF s_{B}$}
\DP
\\ & & \\
\AXC{$\Pi$}
\noLine
\UIC{$\Gamma\vdash \bone\land \bthree \Pto A$}
\AXC{$\mu(\bthree)\geq q$}
\RL{($\B C$I)}
\BIC{$\Gamma \vdash \bone \Pto \B C^{q}A$}
\DP
& 
$\mapsto $
&
\AXC{$D^{\Pi}$}
\noLine
\UIC{$ \FF s_{\Gamma}\vdash t^{\Pi}: \bone \land \bthree \Pto \FF s_{A}$}
\AXC{$\mu(\bthree)\geq q$}
\RL{($\mu$)}
\BIC{$\FF s_{\Gamma}\vdash \nu a.t^{\Pi}: \bone \Pto \FF s_{\B C^{q}A}$}
\DP
\\ & & \\
\AXC{$\Pi$}
\noLine
\UIC{$\Gamma\vdash \bone \Pto \B C^{q}A$}
\AXC{$\Sigma$}
\noLine
\UIC{$\Gamma,A\vdash \bone \Pto \B C^{\vec s}B$}
\RL{($\B C$E)}
\BIC{$\Gamma \vdash \bone \Pto  \B C^{qs}\B C^{\vec s}B$}
\DP
& 
$\mapsto $
&
\AXC{$D^{\Pi}$}
\noLine
\UIC{$ \FF s_{\Gamma}\vdash t^{\Pi}: \bone \Pto \B C^{q}\FF s_{A}$}
\AXC{$D^{\Sigma}$}
\noLine
\UIC{$\FF s_{\Gamma},x:\FF s_{A}\vdash t^{\Sigma}: \bone \Pto \B C^{\vec s} \FF s_{B}$}
\RL{($\lambda$)}
\UIC{$\FF s_{\Gamma } \vdash \lambda x.t^{\Sigma}: \bone \Pto \B C^{\vec s}(\FF s_{A}\To \FF s_{B})$}
\RL{($\{ \}$)}
\BIC{$ \FF s_{\Gamma}\vdash \{\lambda x.t^{\Sigma}\}t^{\Pi}: \bone \Pto   \FF s_{\B C^{qs}\B C^{\vec s}B}$}
\DP
\end{tabular}
\end{center}
\end{minipage}
}
\end{minipage}
}
\caption{Translation $\Pi\mapsto D^{\Pi}$ from $\NDCPLL$ to $\LCPL$.}
\label{fig:ch}
\end{figure*}

\section{From Type Soundness to Type Completeness: Intersection Types}\label{section6}

In this section we first show that derivations in $\TCI$ and $\LCPL$ provide sound approximations of $\HNF(t)$ and $\NF(t)$.
In order to achieve completeness, we then introduce an extension $\TCINT$ of $\TCI$ with intersection types and show that this system fully captures both deterministic and probabilistic notions of termination for $\EVL$.

\subsection{From Types to Probability}

We already observed, through examples, that if a term $t$ has a type like, e.g.~$\B C^{\frac{1}{2}}(o\To o)$, then $t$ has one chance over two of yielding a ``correct'' program for $o\To o$. 
The result below makes this intuition precise, by showing that the probabilities derived in $\TCI$ and $\LCPL$ are lower bounds for the function $\NHNF(t)$, that is, for the actual probability of finding a head normalizable term in the distribution $\CC D_{t}$. 

%
%
%
%
\begin{theorem}\label{thm:normalization}
 If $ \vdash_{\TCI} t: \TOP \Pto \B C^{q}\sigma$, then $\NHNF(t)\geq q$.
If $ \vdash_{\LCPL} t: \TOP \Pto \B C^{q_{1}*\dots *q_{n}}\sigma$, then $\NHNF(t)\geq \prod_{i=1}^{n}q_{i}$.
\end{theorem}

What about reduction to normal form, i.e.~the function $\NNF(t)$? 
A result like Theorem \ref{thm:normalization} cannot hold in this case. Indeed, consider the term $t=\lambda y.y(I\oplus^{i}\Omega)$. While $\NF(t)=\frac{1}{2}$, $\LCPL$ types $t$ with $\FF s=\B C^{1}(\B C^{1}(\B C^{\frac{1}{2}}\sigma\To \sigma)\To \sigma)$, with $\sigma=o\To o$. 
The problem in this example is that the type $\FF s$ contains the ``unbalanced'' assumption
$\B C^{1}(\B C^{\frac{1}{2}}\sigma\To \sigma)$ (in logical terms, $\B C^{\frac{1}{2}}A\to \B C^{1}A$), i.e.~it exploits the assumption of the existence of a function turning a $\frac{1}{2}$-correct input into a $1$-correct output. 
Notice that such a function $f$ can only be one that erases its input, and these are the only functions such that $tf$ can reduce to a normal form.

Nevertheless, soundness with respect to $\NF(t)$ can be proved, for $\TCI$, by restricting to those types not containing such ``unbalanced'' assumptions, i.e.~to types whose programs cannot increase probabilities.

\begin{definition}\label{def:balanced}
A type $\B C^{q}(\FF s_{1}\To \dots \To \FF s_{n}\To o)$ of $\TCI$ is \emph{balanced} if
all $\FF s_{i}$ are balanced and 
$q\leq \prod_{i=1}^{n}\srank{\FF s_{i}}$ 
(where for any type $\FF s=\B C^{q}\sigma$, $\srank{\FF s}=q$).
\end{definition}


\begin{theorem}\label{thm:normalization2}
If $ \vdash_{} t: \TOP \Pto \B C^{q}\sigma$ holds in $\TCI$, where $\sigma$ is balanced, then $\NNF(t)\geq q$.
\end{theorem}
Theorems \ref{thm:normalization} and \ref{thm:normalization2} are proved by adapting the standard technique of \emph{reducibility predicates} to the quantitative notion of probabilistic normal form (for further details, see App.~\ref{app:normalization}).
%

\subsection{From Probability to (Intersection) Types}

%
%
%


To achieve a type-theoretic characterization of $\NHNF(t)$ and $\NNF(t)$, we  introduce an extension of $\TCI$ with intersection types. 
Like those of $\TCI$, the types of $\TCINT$ are of the form $\FF s=\B C^{q}\sigma$, but $\sigma$ is now defined as:
\begin{align*}
\sigma & ::= o\midd \NORM\midd \HNORM \midd \FF M\To \sigma &
 \FF M& ::=  [\FF s, \dots, \FF s]
\end{align*}
where $[a_{1},\dots, a_{n}]$ indicates a finite set. 
While $\FF M$ intuitively stands for a finite intersection of types, the new ground types $\NORM$ and $\HNORM$ basically correspond to the types of normalizable and head-normalizable programs.




We introduce a preorder $\sigma \preceq \tau$ over types by $\alpha\preceq \alpha$, for $\alpha=o,\NORM, \HNORM$, $\B C^{q}\sigma \preceq \B C^{s}\tau$ if $q\leq s$ and $\sigma \preceq \tau$,  and 
$(\FF M\To\sigma) \preceq (\FF N\To \tau)$ if $\sigma\preceq \tau$ and $\FF N\preceq^{*} \FF M$, where 
$[\FF s_{1},\dots, \FF s_{n}]\preceq^{*} [ \FF t_{1},\dots, \FF t_{m}]$ holds is there exists an injective function $f: \{1,\dots, m\}\to \{1,\dots, n\}$ such that $ \FF s_{f(i)}\preceq \FF t_{i}$.

A type judgment of $\TCINT$ is of the form $\Gamma \vdash^{X}t:\bone \Pto \FF s$, where $\Gamma$ is made of declarations of the form $x_{i}:\FF M_{i}$. 
The typing rules of $\TCINT$ are illustrated in Fig.~\ref{fig:typerulesint}. We omit the ($\vee$)- and ($\oplus$)-rules, which are as in $\TCI$.
In the rule ($\mu_{\Sigma}$) it is assumed that $a$ does not occur in $\bone$, is the only name in the $\bthree_{i}$, and that for $i\neq j$, $\bthree_{i}\land \bthree_{j}\vDash \BOT$.

The two rules ($\HNORM$) and ($\NORM$) are justified by Proposition \ref{prop:detcomple} and Theorem \ref{thm:completenessa} below. As rule ($\NORM$) must warrant a bound on normal forms, following Theorem \ref{thm:normalization2}, $\sigma$ has to be 
 \emph{safe}, i.e.~balanced\footnote{Definition~\ref{def:balanced} extends to the types of $\TCINT$ by letting
 $\srank{\HNORM}=\srank{[]}=0$, $\srank{\NORM}=1$ and $\srank{[\FF s_{1},\dots, \FF s_{n+1}]}=\max\{\srank{\FF s_{i}}\}$.} and $\{[],\HNORM\}$-free.
The rule ($@_{\cap}$) is a standard extension of rule ($@$) of $\TCI$ to finite intersections.

The counting rule ($\mu_{\Sigma}$) requires some discussion. The rule admits $n+1$ major premisses expressing typings for $t$ which depend on pairwise disjoint events (the Boolean formulas $\bthree_{i}$).
This is needed to cope with situations as follows: let 
\begin{align*}
t[a,b] & =\Big ((I \oplus_{b}^{1}\Omega)\oplus_{b}^{0}\Omega\Big ) \oplus_{a}^{0}\Big (\Omega \oplus_{b}^{0} I\Big)
\end{align*}
$t[a,b]$ can be given type $\sigma=(\B C^{1}o)\To o$ under either of the two \emph{disjoint} Boolean constraints
$\bthree_{1}:= \bvar_{a}^{0}\land (\bvar_{b}^{0}\land \bvar_{b}^{1})$ and $ \bthree_{2}:=\lnot\bvar_{a}^{0}\land \lnot\bvar_{b}^{0}$. 
In $\TCI$, the best we can achieve to measure the probability that $\nu a.\nu b.t[a,b]$ has type $\sigma$ is $\frac{1}{4}$. Indeed, the rule ($\mu'$) forces us  
to approximate $\mu(\bvar_{b}^{0}\land \bvar_{b}^{1})$ and $\mu(\lnot \bvar_{b}^{1})$ to a \emph{common} lower bound, i.e.~$\frac{1}{4}$, in order to apply a ($\vee$)-rule as illustrated in Fig.~\ref{fig:exa3}. Instead, using ($\mu_{\Sigma}$) we can compute (again, see Fig.~\ref{fig:exa3}), the \emph{actual} probability $
\mu(\bvar_{a}^{0})\mu(\bvar_{b}^{0}\land \bvar_{b}^{1})+ \mu(\lnot \bvar_{a}^{0})\mu(\lnot\bvar_{b}^{0})=
\frac{1}{2}\cdot \frac{1}{4} + \frac{1}{2}\cdot \frac{1}{2} = \frac{3}{8}$.

%

Thanks to the rule ($\mu_{\Sigma}$), the generalized counting rule ($\mu^{*}$) below becomes admissible in $\TCINT$:
$$
\AXC{$\Gamma\vdash^{\{a_{1},\dots, a_{n}\}}t: \bone \Pto  \B C^{q}\sigma$}
\RL{($\mu^{*}$)}
\UIC{$\Gamma \vdash^{\emptyset}\nu a_{1}.\dots. \nu a_{n}.t: \TOP \Pto \B C^{ q\cdot \mu(\bone)} \sigma$}
\DP
$$
This rule plays an essential role in the completeness results below, together with the standard result that both \emph{subject reduction} and \emph{subject expansion} hold for intersection types. 

\begin{figure}
\fbox{
\begin{minipage}{0.48\textwidth}
\adjustbox{scale=0.8}{
\begin{minipage}{\textwidth}
\begin{center}
Identity rule:
$$
\AXC{$\FF s_{i}\preceq \FF t$}
\AXC{$\mathrm{FN}(\bone)\subseteq X$}
\RL{(id$_{\preceq}$)}
\BIC{$\Gamma, x:[\FF s_{1},\dots, \FF s_{n}] \vdash^{X} x: \bone \Pto \FF t  $}
\DP
$$

\bigskip


%

%

\bigskip

Ground types rules:
$$
\AXC{$\Gamma \vdash^{X}t: \bone \Pto \B C^{q}\sigma$}
\RL{($\HNORM$)}
\UIC{$\Gamma \vdash^{X}t: \bone \Pto \B C^{q}\HNORM$}
\DP
\qquad
\AXC{$\Gamma \vdash^{X}t: \bone \Pto \B C^{q}\sigma$}
\AXC{$\sigma$ safe}
\RL{($\NORM$)}
\BIC{$\Gamma \vdash^{X}t: \bone \Pto \B C^{q}\NORM$}
\DP
$$

\bigskip

%
%
%

Arrow rules:
$$
\AXC{$\Gamma, x:\FF M \vdash^{X} t: \bone\Pto \B C^{q}\tau$}
\RL{$(\lambda)$}
\UIC{$\Gamma \vdash^{X} \lambda x.t : \bone\Pto \B C^{q} (\FF M\To \tau)$}
\DP
$$

$$\AXC{$\Gamma \vdash^{X}t:\bone\Pto  \B C^{q}([\FF s_{1},\dots, \FF s_{n}] \To \tau)$}
\AXC{$ \Big\{\Gamma \vdash^{X} u: \bone\Pto \FF s_{i}\Big\}_{i=1,\dots,n}$}
\RL{$(@_{\cap})$}
\BIC{$\Gamma \vdash^{X} tu: \bone\Pto \B C^{q} \tau$}
\DP
$$

\bigskip

Counting rule:
$$
\AXC{$\Big\{\Gamma \vdash^{X\cup \{a\}} t: \bone\land \bthree_{i}\Pto\B C^{q_{i}} \sigma\Big\}_{i=1,\dots,n+1}$}
\AXC{$\mu(\bthree_{i})\geq s_{i}$}
\RL{$(\mu_{\Sigma})$}
\BIC{$\Gamma \vdash^{X} \nu a.t: \bone\Pto\B C^{\sum_{i}q_{i}s_{i}} \sigma$}
\DP
$$

\

\end{center}
\end{minipage}
}
\end{minipage}
}
\caption{Typing rules of $\TCINT$.}
\label{fig:typerulesint}
\end{figure}

\begin{figure*}
\fbox{
\begin{minipage}{0.98\textwidth}
\begin{center}
\adjustbox{scale=0.68}{
$
\AXC{$\Gamma\vdash^{\{a,b\}} t[a,b]:  \bvar_{a}^{0}\land (\bvar_{b}^{0}\land \bvar_{b}^{1})\Pto\B C^{1}\sigma$}
\AXC{$\mu(\bvar_{b}^{0}\land \bvar_{b}^{1})\geq 1/4$}
\RL{$(\mu')$}
\BIC{$\Gamma \vdash^{\{a\}}\nu b.t[a,b]: \bvar_{a}^{0}\Pto \B C^{\frac{1}{4}}\sigma$}
\AXC{$\Gamma\vdash^{\{a,b\}} t[a,b]: \lnot \bvar_{a}^{0}\land \lnot\bvar_{b}^{0}\Pto\B C^{1} \sigma$}
\AXC{$\mu(\lnot\bvar_{b}^{0})\geq 1/4$}
\RL{$(\mu')$}
\BIC{$\Gamma \vdash^{\{a\}}\nu b.t[a,b]: \lnot\bvar_{a}^{0}\Pto \B C{\frac{1}{4}}\sigma$}
\RL{($\vee$)}
\BIC{$\Gamma \vdash^{\{a\}}\nu b.t[a,b]: \bvar_{a}^{0}\lor\lnot\bvar_{a}^{0}\Pto \B C{\frac{1}{4}}\sigma$}
\AXC{$\mu( \bvar_{a}^{0}\lor \lnot \bvar_{a}^{0})\geq 1$}
\RL{$(\mu')$}
\BIC{$\Gamma \vdash^{\emptyset}\nu a.\nu b.t[a,b]:\TOP \Pto \B C^{ \frac{1}{4}} \sigma$}
\DP
$}

\bigskip
\adjustbox{scale=0.68}{
$
\AXC{$\Gamma\vdash^{\{a,b\}} t[a,b]:  \bvar_{a}^{0}\land (\bvar_{b}^{0}\land \bvar_{b}^{1})\Pto\B C^{1}\sigma$}
\AXC{$\mu(\bvar_{b}^{0}\land \bvar_{b}^{1})\geq 1/4$}
\RL{$(\mu_{\Sigma})$}
\BIC{$\Gamma \vdash^{\{a\}}\nu b.t[a,b]: \bvar_{a}^{0}\Pto \B C^{\frac{1}{4}}\sigma$}
\AXC{$\Gamma\vdash^{\{a,b\}} t[a,b]: \lnot \bvar_{a}^{0}\land \lnot\bvar_{b}^{0}\Pto\B C^{1} \sigma$}
\AXC{$\mu(\lnot\bvar_{b}^{0})\geq 1/2$}
\RL{$(\mu_{\Sigma})$}
\BIC{$\Gamma \vdash^{\{a\}}\nu b.t[a,b]: \lnot\bvar_{a}^{0}\Pto \B C{\frac{1}{2}}\sigma$}
\AXC{$\mu( \bvar_{a}^{0}), \mu( \lnot \bvar_{a}^{0})\geq 1/2$}
\RL{$(\mu_{\Sigma})$}
\TIC{$\Gamma \vdash^{\emptyset}\nu a.\nu b.t[a,b]:\TOP \Pto \B C^{ \frac{1}{4}\cdot \frac{1}{2}+ \frac{1}{2} \cdot \frac{1}{2}} \sigma$}
\DP
$}
\end{center}
\end{minipage}
}
\caption{Comparing probabilities derived with the rules ($\mu'$) and ($\mu_{\Sigma}$).}
\label{fig:exa3}
\end{figure*}

%
%
%

\begin{proposition}[Subject Reduction/Expansion]\label{prop:subexp}
If $\Gamma \vdash^{X}t: \bone \Pto \FF s$ and either $t\redall u$ or $u\redall t$, then
 $\Gamma \vdash^{X}t: \bone \Pto \FF s$.
\end{proposition}

%
%
%

 We will now discuss how typings in $\TCINT$ capture both deterministic and probabilistic properties of terms. 
 First we have the following facts, which show that the types $\HNORM$ and $\NORM$ capture deterministic termination. 
%
%
%
  \begin{restatable}[Deterministic Completeness]{proposition}{detcomple}\label{prop:detcomple}
For any closed term $t$,
\begin{varitemize}
 \item[(i.)] $t$ is head-normalizable iff $\vdash_{\lnot \vee} t: \TOP \Pto \B C^{1}\HNORM$;
  \item[(ii.)] $t$ is normalizable iff $\vdash_{\lnot \vee} t: \TOP \Pto \B C^{1}\NORM$.
    \item[(iii.)] $t$ is strongly normalizable iff $\vdash_{\lnot \vee} t: \TOP \Pto \B C^{1}\NORM$ and all types in the derivation are safe.

 \end{varitemize}
 \end{restatable}
\begin{proof}[Proof sketch]
Using standard intersection types arguments, it is shown that $\vdash_{\lnot \vee} t: \TOP \Pto \B C^{1}\HNORM$ holds for any head-normal $t$. The first half of (i.) is deduced then using Proposition~\ref{prop:subexp}. 
The second half follows from a normalization argument similar to that of Theorem \ref{thm:detnormalization}.
Cases (ii.) and (iii.) are similar.
\end{proof}

The probabilistic normalization theorems \ref{thm:normalization} and \ref{thm:normalization2} (which extend smoothly to $\TCINT$) ensure that if $t$ has type $\B C^{q}\HNORM$ (resp.~$\B C^{q}\NORM$), then $\NHNF(t)\geq q$ (resp.~$\NNF(t)\geq q$). 
%
%
Conversely, $\NHNF(t)$ and $\NNF(t)$ can be bounded by means of derivations in $\TCINT$, in the following sense:
%
%
%

\begin{restatable}[Probabilistic Completeness]{theorem}{completenessa}\label{thm:completenessa}
For any closed term $t$,
\begin{align*}
\NHNF(t) & = \sup \{ q\mid\  \vdash t: \TOP \Pto \B C^{q}\HNORM\} \\
\NNF(t) & = \sup \{ q\midd \  \vdash t: \TOP \Pto\B C^{q} \NORM\} 
\end{align*}

\end{restatable}
\begin{proof}[Proof sketch.]
Suppose w.l.o.g.~that $t=\nu a_{1}.\dots.\nu a_{k}.t'$. 
For any $u\in \HNF$ such that $\CC D_{t}(u)>0$, we can deduce $\vdash u:\TOP\Pto \HNORM$. The sequence of probabilistic choices leading to $u$ is finite, and thus captured by a Boolean formula $\bone_{t\mapsto u}$. 
Using subject reduction/expansion we thus deduce $\vdash t': \bone_{t\mapsto u}\Pto \HNORM$.
Hence, for any finite number of head normal forms $u_{1},\dots, u_{n}$, we deduce $\vdash t': \bone_{t\mapsto u_{i}}\Pto \HNORM$. Using ($\vee$) and the generalized counting rule ($\mu^{*}$) we deduce then $\vdash t: \TOP \Pto \B C^{s}\HNORM$, where $s= \sum_{i=1}^{n}\mu(\bone_{t\mapsto u_{i}})=\mu(\bigvee_{i=1}^{n}\bone_{t\mapsto u_{i}})$. 
The argument for $\NF(t)$ is similar.\end{proof}

\section{Related Work}\label{section7}

As discussed in the introduction, our results provide the first clear correspondence between a proof system and a probabilistic extension of the $\lambda$-calculus. This is not to say that our logic and calculi come from nowhere. 

Different kinds of measure-theoretic quantifiers have been investigated in the literature, with the intuitive meaning of ``$A$ is true \emph{for almost all $x$}'',
see~\cite{Morgenstern,steinhorn} and more recently~\cite{MichalewskiMio,MSM}, or
``$A$ is true for the \emph{majority of $x$}''~\cite{Papadimitriou, ZachosHeller, Zachos88}.
 Our use of the term ``counting quantifier'' comes from~\cite{Antonelli2021}, where an extension of classical propositional logic with such quantifiers is studied and related to Wagner's counting operator on classes of languages~\cite{Wagner84,Wagner,Wagner86}. 
 To our knowledge, the present work is the first to apply some form of measure quantifier to typed probabilistic functional programs.

Despite the extensive literature on logical systems enabling 
(in various ways and for different purposes) some forms of probabilistic
reasoning, there is not much about logics tied to computational aspects, 
as $\ICPL$ is. 
Most of the recent logical formalisms have been
developed in the realm of modal logic, like e.g.~\cite{Nilsson86,Nilsson93,Bacchus,Bacchus90a, Bacchus90b,FH94,Halpern90,Halpern03}.
Another class of probabilistic modal logics have been designed to model Markov chains and similar structures~\cite{KOZEN1981328,Hansson1994,LEHMANN1982165}.
With the sole exception of \emph{Riesz modal logic} \cite{lmcs:6054}, we are
not aware of sequent calculi for probability logic.

Intuitionistic modal logic has been related in the Curry-Howard sense to monadic extensions of the $\lambda$-calculus~\cite{DaviesPfennig, Benton1998, AlechinaDePaiva, Nanevski,Curien2016}. 
Nevertheless, in these correspondences 
modal operators are related to \emph{qualitative} properties of programs (typically, tracing algebraic effects), as opposed to the \emph{quantitative} properties expressed by counting quantifiers. 
Our Kripke structures for $\ICPL$ can be related to standard Kripke structures for intuitionistic modal logic \cite{Plotkin1986, Simpson1994}. These are based on a set $W$ with \emph{two} pre-order relations $\leq$ and $R$ enjoying a suitable ``diamond'' property $R; \leq  \ \subseteq\  \leq ; R$. We obtain a similar structure by taking worlds to be pairs $w,\omega$ made of a world and an outcome from the Cantor space,
with $(w,\omega)\preceq (w',\omega)$ whenever $w\leq w'$, and $(w,\omega)R(w,\omega+\omega')$. The clause for $\B C^{q}A$ can then be seen as a quantitative variant of the corresponding clause for $\Diamond A$. This is not very surprising, given the similarity between the introduction and elimination rules for $\B C^{q}$ and those for $\Diamond$, see e.g.~\cite{Benton1998,AlechinaDePaiva}.


On the other hand, quantitative semantics arising from linear logic have been largely used in the study of probabilistic $\lambda$-calculi,
e.g.~\cite{DalLago2017, Ehrhard2018, Faggian2019}. 
Notably, \emph{probabilistic coherence spaces}~\cite{Girard2004, Ehrhard2018, Ehrhard2014} have been shown to provide a fully abstract model of probabilistic PCF.
While we are not aware of correspondences relating probabilistic programs with proofs in linear logic, it seems that the proof-theory of counting quantifiers could be somehow related to that of \emph{bounded exponentials}~\cite{Girard1992, DalLago2009} and, more generally, to the theory of \emph{graded} monads and comonads~\cite{Katsumata2014, Brunel2014, Ghica2014, Katsumata2018}.


%
%
%

The calculus $\EVL$ derives from \cite{DLGH}, which also introduces a simple type system ensuring strong normalization, although the typings do not provide any quantitative information.
Beyond this, several type systems for probabilistic $\lambda$-calculi have been introduced in the recent literature. 
Among these, systems based on \emph{type distributions} \cite{DalLagoGrellois}, i.e.~where a single derivation assigns several types to a term, each with some probability, and systems based on \emph{oracle intersection types} \cite{Breuvart2018}, where type derivations capture single evaluations as determined by an oracle. 
Our type systems sit in between these two approaches: like the first (and unlike the second), typing derivations can capture a finite number of different evaluations, although without using distributions of types; like the second, typings reflect the dependency of evaluation on oracles, although the latter are manipulated in a collective way by means of Boolean constraints. 

Finally, in 
\cite{Warrell2018} dependent type theory is enriched with a probabilistic choice operator, yielding a calculus with both term and type distributions. Interestingly, a fragment of this system enjoys a sort of CHC with 
so-called \emph{Markov Logic Networks} \cite{Richardson2006}, a class of probabilistic graphical models specified by means of first-order logic formulas.  


\section{Conclusions}\label{section8}
The main contribution of this work consists in defining a Curry-Howard 
correspondence between a logic with counting quantifiers and a 
type system that expresses probability of 
termination. Moreover, in analogy with what happens in the 
deterministic case, extending 
the type system with an intersection operator leads to a full characterization of probability of termination.
Even though intersection types do not have a clear logical 
counterpart, the existence of this extension convinces us that the 
correspondence introduced is meaningful. The possibility of defining a Curry-Howard correspondence 
relating algebraic effects, on the program side, with a modal operator, on the logic 
side, is certainly not surprising. 
Instead, it seems to us that the new and surprising contribution of this work is showing that 
the peculiar 
features of probabilistic effects can be managed in an elegant way
using ideas coming from logic.

Among the many avenues of research that this work opens, the study of the 
problem of type inference must certainly be mentioned, as well as the extension 
of the correspondence to polymorphic types or to control operators. Particularly 
intriguing, then, is the possibility of studying the system of intersection 
types introduced to support the synthesis, always in analogy with what is 
already known in deterministic calculations.

\bibliographystyle{plain}
\bibliography{main}

\onecolumn
\appendix
\appendixpage

\DoToC
%

\section{Details about Intuitionistic Counting Propositional Logic}

In this section we introduce a proof-system $\NDCPL$ for $\ICPL$ and describe the correspondence between proofs in $\ICPL$ and families of proofs $\ICPLL$; then, we establish the soundness and completeness of $\NDCPL$ with respect to $\ICPL$-structures. 
Finally, we provide some more details about normalization in the ``computational fragment'' $\ICPLL$, and we describe an alternative ``CbN'' proof-system for $\ICPLL$, whose derivations can be decorated with terms of $\EVL$, yielding a CHC with the type system $\TCI$.

\subsection{The Proof-Theory of $\ICPL$.}\label{app:ProofTheoryiCPL}

The natural deduction system $\NDCPL$  for $\ICPL$ is formed by the rules illustrated in Fig.~\ref{fig:logicrules}, with the following proviso:
\begin{varitemize}
\item it is everywhere assumed that in a sequent $\Gamma \vdash  A$, $\FN(\Gamma),  \FN(A)
\subseteq X$;

\item in the rule ($\B C$I) and ($\B C$E$_{3}$) it is assumed that $ \FN(\bthree)\subseteq \{a\}$.

\end{varitemize}
together with all instances of the two axiom schema below:
 \begin{align} 
\B C^{q}_{a}(A\lor B) & \to A \lor (\B C^{q}_{a}A) & (a\notin \FN(A)) \tag{$\B C\vee$} \label{ax1}\\
&  \lnot\B C^{q}_{a}\bone & (\FN(\bone)\subseteq\{a\}, \mu(\bone)<q) \tag{$\B C\BOT$}\label{ax2}
\end{align}

As usual, we take $\lnot A$ as an abbreviation for $A\supset \BOT$.

We let $\NDCPL^{-}$ indicate $\NDCPL$ without \eqref{ax1} and \eqref{ax2}.

%
%


\begin{remark}
The reason for distinguishing between the systems $\NDCPL$ and $\NDCPL^{-}$ is somehow analogous to what happens in intuitionistic modal logic (IML). Indeed, standard axiomatizations of IML include two axioms
 \begin{align} 
\Diamond(A\lor B) & \to\Diamond A \lor \Diamond B \tag{$\Diamond\vee$}\\
&\lnot \Diamond\BOT \tag{$\Diamond\BOT$}
\end{align}
which do not have a clear computational interpretation. Instead, a Curry-Howard correspondence can be defined for an axiomatization of IML (usually referred to as \emph{constructive modal logic}) which does not include these two axioms. 

In a similar way, we will show that provability in $\NDCPL^{-}$ corresponds, under the decomposition provided by Lemma \ref{lemma:decomposition} to provability in the Curry-Howard proof-system $\NDCPLL$, while axioms \eqref{ax1} and \eqref{ax2} cannot be similarly interpreted.
\end{remark}

\begin{figure}[t]
\begin{center}
\fbox{
\resizebox{0.7\textwidth}{!}{
\begin{minipage}{0.9\textwidth}
\begin{center}
{Classical Identity}
$$
\AXC{}
\RL{(Cid)}
\UIC{$\Gamma\vdash \bvar_{a}^{i}\lor \lnot \bvar_{a}^{i}$}
\DP
$$

\bigskip

{Intuitionistic Identity}
$$
\AXC{}
\RL{(Iid)}
\UIC{$\Gamma, A \vdash  A$}
\DP
$$


\bigskip

{Logical Rules}

\begin{center}
\begin{tabular}{c c c}
$
\AXC{}
\RL{($\TOP$I)}
\UIC{$\Gamma \vdash   \TOP$}
\DP
$ &  &
$
\AXC{$\Gamma\vdash  \BOT$}
\RL{($\BOT$E)}
\UIC{$\Gamma \vdash  A$}
\DP
$ \\
&& \\
$
\AXC{$\Gamma \vdash  A$}
\AXC{$\Gamma \vdash  A$}
\RL{($\land$)I}
\BIC{$\Gamma \vdash  A\land B$}
\DP
$ & &
$
\AXC{$\Gamma \vdash  A\land B$}
\RL{($\land$E$_{1}$)}
\UIC{$\Gamma \vdash  A$}
\DP
\qquad
\AXC{$\Gamma \vdash A\land B$}
\RL{($\land$E$_{2}$)}
\UIC{$\Gamma \vdash  B$}
\DP
$ 
\\ && \\ 
$
\AXC{$\Gamma \vdash  A$}
\RL{($\lor$I$_{1}$)}
\UIC{$\Gamma \vdash  A\lor B$}
\DP
\qquad
\AXC{$\Gamma \vdash  A$}
\RL{($\lor$I$_{2}$)}
\UIC{$\Gamma \vdash  A\lor B$}
\DP
$ & &
$
\AXC{$\Gamma \vdash  A\lor B$}
\AXC{$\Gamma, A \vdash  C$}
\AXC{$\Gamma,B \vdash  C$}
\RL{($\lor$E)}
\TIC{$\Gamma \vdash  C$}
\DP
$\\
&& \\
$
\AXC{$\Gamma, A \vdash  B$}
\RL{($\to$I)}
\UIC{$\Gamma \vdash  A\to B$}
\DP
$ & &$
\AXC{$\Gamma \vdash A\to B$}
\AXC{$\Gamma \vdash  A$}
\RL{($\to$E)}
\BIC{$\Gamma \vdash  B$}
\DP
$
\end{tabular}
\end{center}

\bigskip

{Counting Rules}
$$
\AXC{$ \Gamma, \bthree \vdash  A$}
\AXC{$\mu(\bthree)\geq q$}
\RL{($\B C$I) {\small$a\notin \Gamma$}}
\BIC{$\Gamma \vdash\B C^{q}_{a}A$}
\DP
$$
$$
\AXC{$\Gamma \vdash\B C^{q}_{a}A$}
%
\RL{($\B C$E$_{1}$)  {\small$a\notin A$}}
\UIC{$\Gamma \vdash A$}
\DP
\qquad\qquad
\AXC{$\Gamma \vdash\B C^{q}_{a}A$}
\AXC{$\Gamma, A \vdash C$}
\RL{($\B C$E$_{2}$)  {\small$a\notin \Gamma$}}
\BIC{$\Gamma \vdash \B C^{qs}_{a}C$}
\DP
$$
%

\end{center}
\end{minipage}
}
}
\end{center}

\caption{Rules of $\NDCPL$.}
\label{fig:logicrules}
\end{figure}

A formula $A$ of $\ICPL$ is \emph{purely classical} if $A$ contains no intuitionistic propositional variable, and 
\emph{purely intuitionistic} if $A$ contains no classical propositional variable.

Let us list a few properties of $\ICPL$.

\begin{lemma}
For any purely Boolean formula $\bone$, $\vdash \bone$ is derivable in $\NDCPL$ iff $\bone$ is a tautology.
\end{lemma}


\begin{lemma}\label{lemma:props}
The following are theorems of $\ICPL$:
\begin{varitemize}
\item[(i.)] $\B C^{q}_{a}\BOT \leftrightarrow \BOT$ and $\B C^{q}_{a}\TOP \leftrightarrow \TOP$;
\item[(ii.)] if $a\notin \FN(A)$, $\B C^{q}_{a}(A\to B) \to (A\to \B C^{q}_{a}B)$.
\item[(iii.)] if $a\notin \FN(A)$, $\B C^{q}_{a}(A\land B) \leftrightarrow A\land \B C^{q}_{a}B$.
\item[iv.] if $a\notin\FN(A)$, then $\B C^{q}_{a}A\leftrightarrow A$.

\end{varitemize}
\end{lemma}

\begin{lemma}
For any Boolean formula $\bone$ with $\FN(\bone)\subseteq \{a\}$,  $\B C^{q}_{a}\bone \lor (\lnot \B C^{q}_{a}\bone)$ is provable.
%
%

\end{lemma}
\begin{proof}

Two possibilities arise:
if $\mu(\bone)\geq q$, then from $\bone \vdash \bone$ we deduce, by ($\B C$I), $\vdash\B C^{q}_{a}\bone$ as well as 
$\vdash \B C^{q}_{a}\bone \lor (\lnot \B C^{q}_{a}\bone)$.
if $\mu(\bone)> q$, then we deduce, using axiom \ref{ax2}, 
$\vdash \lnot \B C^{q}_{a} \bone$, as well as 
$\vdash \B C^{q}_{a}\bone \lor (\lnot \B C^{q}_{a}\bone)$.
\end{proof}

%
%
%
%
%

\subsection{Relating $\ICPL$ and $\ICPLL$.}\label{app:relatingICPLICPLL}

Let us first provide a complete proof of the decomposition lemma (Lemma \ref{lemma:decomposition}).

For any formula $C$ and $b\in\{0,1\}$, let $\lnot^{b}C$ indicate the formula $C$ if $b=0$ and $\lnot C$ if $b=1$. 
For any $\omega \in (2^{\mathbb N})^{\CC A}$, let the \emph{theory of $\omega$} be the following set of Boolean formulas:
$$\TTH(\omega)=\{\lnot^{(1-\omega(a)(i))}\bvar_{a}^{i} \mid a,i\in \mathbb N\}$$
Moreover, for all $K\in \mathbb N$ and finite set $X$, let 
$\TTH^{K}_{X}(\omega)$ be the Boolean formula below
$$
\TTH^{K}_{X}(\omega)=\bigwedge_{a\in A,i\leq K}
\lnot^{(1-\omega(a)(i))}\bvar_{a}^{i}
%
$$
Observe that the formula $\TTH^{K}_{X}(\omega)$ only depends on a finite amount of information of $\omega$, namely, on the unique $v\in (2^{\{1,\dots,K\}})^{X}$ such that $v(a)(i)=\omega(a)(i)$. 
When $v$ is clear from the context, we will indicate $\TTH^{K}_{X}(\omega)$ simply as $\TTH(v)$.

The decomposition lemma can be reformulated then as follows:
\begin{lemma}[decomposition lemma, syntactic formulation]\label{lemma:decompositionbis}
For any formula $A$ of $\ICPL$ there exist intuitionistic formulas $A_{v}$, where $v$ varies over all possible  valuations of the Boolean variables in $A$, 
 such that 
$\vdash A \leftrightarrow \bigvee_{v}\TTH(v) \land A_{v}$.
\end{lemma}
\begin{proof}
The formula $A_{v}$ is defined by induction on $A$ as follows:
\begin{varitemize}
\item if $A=\TOP$, $A=\BOT$, or $A=\CC p$, then $A_{v}=A$;
\item if $A=\bvar_{a}^{i}$, then $A_{v}=\TOP$ if $v(a)(i)=1$, and $A_{v}=\BOT$ if $v(a)(i)=0$;
\item if $A=B c C$, where $c=\land, \lor, \to$, then $A_{v}= B_{v} c C_{v}$;
\item if $A=\B C^{q}_{a}B $, then we consider two cases: 
\begin{varitemize}
\item if $a$ does not occur in $B$, then we let $A_{v}:=B_{v}$;
\item if $a$ does occur in $B$, then for all valuation $w$ of the variables of name $a$ in $B$, we can suppose that the formulas $B_{v+w}$ are well-defined. We let then 
$
A_{v}= \B C^{q}_{a} \bigvee_{w}B_{v+w}
$.

\end{varitemize}
%
%
%

\end{varitemize}

Let us show that for any formula $A$ and valuation $v$ of its Boolean variables, 
$\TTH (v)\vdash_{ \ICPL}A\leftrightarrow A_{v}$. 
We argue by induction on $A$. Indeed, the only non-trivial case is that of $A=\B C^{q}_{a}B$.
If $a$ does not occur in $B$, then we can conclude by the IH.
Otherwise, by the IH we have that for any valuation $w$ of the Boolean variables of $B$ of name $a$, 
$\TTH (v+w)\vdash B\leftrightarrow B_{v+w}$. Using the fact that $\bigvee_{w}\TTH(w)$ is provable and that 
 $\TTH(v+w)\leftrightarrow (\TTH(v)\land \TTH(w))$, we deduce then
$ \TTH(v) \vdash B\leftrightarrow \bigvee_{w}B_{v+w}$, from which we can deduce 
$\TTH(v)\vdash A  \leftrightarrow A_{v}$.

Using $\TTH(v)\vdash A\leftrightarrow A_{v}$ and $\vdash \bigvee_{v}\TTH(v)$, we can conclude 
$\vdash A \leftrightarrow \bigvee_{v}\TTH(v)\land A_{v}$.

\end{proof}

If $A$ is a formula of $\ICPL$ containing no Boolean variable, we let $|A|$ indicate the corresponding formula of $\ICPLL$ obtained by deleting names from counting quantifiers, i.e.~replacing $\B C^{q}_{a}$ by $\B C^{q}$.

In the following, we consider as $\NDCPLL$ the proof-system obtained by enriching the rules from Section \ref{section3} with all standard rules for intuitionistic connectives (straightforwardly adapted to the sequents of $\NDCPLL$ by adding everywhere a fixed Boolean formula $\dots\vdash\bone \Pto \dots$).

\begin{lemma}\label{lemma:nothv}
In $\NDCPLL$, if $\Gamma \vdash \bone \land \btwo \Pto A$, where $\bone \land \btwo$ is satisfiable, then
$\Gamma \vdash \bone \Pto A$.
\end{lemma}
\begin{proof}
By induction on the rules of $\NDCPLL$. 

\end{proof}

The following result describes the correspondence between the proof-theories of $\ICPL$ and $\ICPLL$.

\begin{proposition}\label{thm:decom}
$\Gamma \vdash_{\NDCPL^{-}} A$ iff for all valuations $v$, $|\Gamma_{v}| \vdash_{\NDCPLL}\TTH(v)\Pto | A_{v}|$.

\end{proposition}
\begin{proof}
The if part follows from the decomposition lemma (Lemma \ref{lemma:decomposition}) and from the fact that 
$|\Gamma| \vdash_{\NDCPLL} \bone \Pto |A|$ implies $\Gamma, \bone \vdash_{\NDCPL^{-}}A$ (easily checked by induction on the rules of $\NDCPLL$).

%

For the only-if part we argue by induction on the rules of $\NDCPL$. All propositional cases are straightforward, so we only focus here on the rules for counting quantifiers.

\begin{description}
\item[($\B C$I)] 

$$
\begingroup\makeatletter\def\f@size{10}
\begin{lrbox}{\mypti}
\begin{varwidth}{\linewidth}
$
\AXC{$|\Gamma_{v}|\vdash\TTH(v)\land \TTH(w)\Pto| A_{v+w}|$}
\doubleLine
\RL{($\vee$I)}
\UIC{$|\Gamma_{v}|\vdash\TTH(v)\land \TTH(w)\Pto\bigvee_{w} |A_{v+w}|$}
\DP$
\end{varwidth}
\end{lrbox}
\AXC{$\Gamma, \bthree \vdash A$}
\AXC{$\mu(\bthree)\geq q$}
\BIC{$\Gamma \vdash  \B C^{q}_{a}A$}
\DP
\qquad\mapsto \qquad
\AXC{$\left\{ \usebox{\mypti}\right\}_{w\vDash \bthree}$}
\doubleLine
\RL{($\mathsf m$)}
\UIC{$|\Gamma_{v}|\vdash\TTH(v)\land \bthree \Pto \bigvee_{w}|A_{v+w}|$}
\AXC{$\mu(\bthree)\geq q$}
\RL{($\B C$I)}
\BIC{$| \Gamma_{v}|\vdash\TTH(v)\Pto |(\B C^{q}_{a}A)_{v}|$}
\DP
\endgroup
$$

\item[($\B C$E)$_{1}$] 
Since $a\notin \FN(A)$, it follows that $(\B C^{q}_{a}A)_{v}=A_{v}$, so we can conclude by the IH.
%
%
%

\item[($\B C$E)$_{2}$]\ \\

\adjustbox{scale=0.85}{$
\begingroup\makeatletter\def\f@size{10}
\begin{lrbox}{\mypti}
\begin{varwidth}{\linewidth}
$
\AXC{$|\Gamma_{v}| ,|A_{v+w}|\vdash\TTH(v)\land \TTH(w)\Pto |B_{v+w}|$}
\RL{[Lemma \ref{lemma:nothv}]}
\UIC{$|\Gamma_{v}| ,|A_{v+w}|\vdash\TTH(v)\Pto |B_{v+w}|$}
\doubleLine
\RL{($\vee$I)}
\UIC{$|\Gamma_{v} |,|A_{v+w}|\vdash\TTH(v)\Pto (\bigvee_{w}|B_{v+w}|)$}
\DP$
\end{varwidth}
\end{lrbox}
\AXC{$\Gamma \vdash \B C^{q}_{a}$}
\AXC{$\Gamma, A \vdash B$}
\BIC{$\Gamma \vdash \B C^{qs}_{a}B$}
\DP
\qquad \mapsto \qquad 
\AXC{$|\Gamma_{v}| \vdash\TTH(v)\Pto |(\B C_{a}^{q}A)_{v}|$}
\AXC{$\left\{ \usebox{\mypti}\right\}_{w}$}
\doubleLine
\RL{($\vee$E)}
\UIC{$|\Gamma_{v} |,\bigvee_{w}|A_{v+w}|\vdash \TTH (v)\Pto (\bigvee_{w}|B_{v+w}|)$}
\RL{($\B C$E)}
\BIC{$|\Gamma_{v} |\vdash 
\TTH(v)\Pto 
|(\B C^{qs}_{a}B)_{v}|$}
\DP
\endgroup
$}

%
%
%
%
%
\end{description}

\end{proof}

\subsection{Normalization in $\NDCPL^{-}$.}

It is possible to define a normalization procedure for $\NDCPL^{-}$. Cuts formed by propositional intro-elim rules are reduced as in standard intuitionistic logic. The cuts formed by counting rules are reduced as follows:

\begin{description}
\item[($\B C$I/$\B C$E$_{1}$)]
$$
\AXC{$\Pi$}
\noLine
\UIC{$\Gamma,\bthree \vdash A$}
\AXC{$\mu(\bthree)\geq q$}
\RL{$(\B C$I)}
\BIC{$\Gamma \vdash \B C^{q}_{a}A$}
\RL{($\B C$E)$_{1}$}
\UIC{$\Gamma \vdash A$}
\DP
\qquad \leadsto \qquad 
\AXC{$\Pi^{*}$}
\noLine
\UIC{$\Gamma \vdash A$}
\DP
$$
where the derivation $\Pi^{*}$ is obtained from Lemma \ref{lemma:noBool} below, easily established by induction.
\begin{lemma}\label{lemma:noBool}
If $\Gamma, \bone \vdash A$ is derivable in $\NDCPL^{-}$, where $\FN(\bone)\subseteq \{a\}$, and $a$ does not occur in either $\Gamma$ nor $A$, then $\Gamma \vdash A$ is derivable with a derivation of same length and using the same rules. 
\end{lemma}

\item[($\B C$I/$\B C$E$_{2}$)]

$$
\AXC{$\Pi$}
\noLine
\UIC{$\Gamma,\bthree \vdash A$}
\AXC{$\mu(\bthree)\geq q$}
\RL{$(\B C$I)}
\BIC{$\Gamma \vdash \B C^{q}_{a}A$}
\AXC{$\Sigma$}
\noLine
\UIC{$\Gamma, A \vdash B$}
\RL{($\B C$E)$_{2}$}
\BIC{$\Gamma \vdash \B C^{qs}_{a}B$}
\DP
\qquad \leadsto \qquad 
\AXC{$\Pi$}
\noLine
\UIC{$\Gamma,\bthree \vdash A$}
\AXC{$\Sigma$}
\noLine
\UIC{$\Gamma, A \vdash B$}
\RL{(subst)}
\BIC{$\Gamma, \bthree \vdash B$}
\AXC{$\mu(\bthree)\geq qs$}
\RL{$(\B C$I)}
\BIC{$\Gamma \vdash \B C^{qs}_{a}B$}
\DP
$$
\end{description}
where the admissibility of the rule (subst) is easily checked by induction.

The proof of 
Theorem \ref{thm:decom} yields a way to associate each proof $\Pi$ of $\Gamma \vdash A$ in $\NDCPL^{-}$ with a family of proofs $\Pi_{v}$ of $|\Gamma_{v}|\vdash \TTH(v)\Pto |A_{v}|$ in $\NDCPLL^{-}$. 
This association preserves normalization in the following sense:

\begin{lemma}\label{lemma:sat}
For any two proofs $\Pi, \Sigma$ of $\Gamma \vdash A$ in $\NDCPL^{-}$ and any valuation $v$ of the Boolean variables in $\Gamma$ and $A$, if $\Pi \leadsto^{*}\Sigma$, then $\Pi_{v}\leadsto^{*}\Sigma_{v}$.
\end{lemma}
The lemma above, easily checked, can be used to deduce the strong normalization of $\NDCPL^{-}$ from that of $\NDCPLL$:

\begin{theorem}
$\NDCPL^{-}$ is strongly normalizing. 
\end{theorem}

\subsection{Soundness and Completeness of $\NDCPL$.}

Let us first establish a few properties of Kripke Semantics.
 For this we first need to recall a fundamental property of Borel sets.

For any Borel set $S\in \Borel_{X\cup Y}$ and $\omega\in (2^{\mathbb N})^{X}$, let 
$$
\Pi^{\omega}(S)=\{\omega'\in (2^{\mathbb N})^{Y}\mid \omega+\omega'\in S\}\subseteq (2^{\mathbb N})^{X}
$$
Notice that $\Pi^{\omega}(S)$ is an \emph{analytic} set and needs not be Borel. However, since the Lebesgue measure is defined on all analytic sets, the values $\mu(\Pi^{\omega}(S))$, for $S$ Borel (or more generally, analytic), are always defined. Moreover the following holds:

\begin{lemma}\label{lemma:borel}[\cite{Kechris}, Theorem 14.11 + Theorem 29.26]
For any $S\in \Borel_{X\cup Y}$, with $X\cap Y=\emptyset$, and $r\in [0,1]$,
$
\{ \omega\in (2^{\mathbb N})^{X}\mid \mu(\Pi^{f}(S))\geq r\}\in \Borel_{X}
$.
\end{lemma}

For any $S\subseteq \Borel_{X}$ and $Y\supseteq X$, let $S^{\uparrow Y}:= S\times (2^{\BB N})^{Y-X}$.

Using Lemma \ref{lemma:borel} we can show that for any $\ICPL$-structure $\CC M=(W,\preceq, \B i)$ and world $w\in W$, the set of functions $\omega\in (2^{\mathbb N})^{X}$ such that $w,\omega \Vdash^{X}_{\CC M}A$ is a Borel set.

\begin{lemma}
Given a $\ICPL$-structure $\CC M=(W,\preceq,  \B i)$, for any finite set $X$, $w\in W$ and formula $A$ with $\mathrm{FN}(A)\subseteq X$, the set
$$
\mathrm{mod}_{\CC M}(A,X,w)=\{\omega \in (2^{\mathbb N})^{X}\mid w,\omega \Vdash^{X}_{\CC M}  A\}
$$
is Borel. 
\end{lemma}
\begin{proof}
We argue by induction on $A$:
\begin{varitemize}
\item if $A=\BOT$, then $\mathrm{mod}_{\CC M}(A,X,w)=\emptyset$ is Borel;
similarly, if $A=\TOP$, then $\mathrm{mod}_{\CC M}(A,X,w)=(2^{\mathbb N})^{X}$ is Borel;

\item if $A=\bvar_{a}^{i}$, then $\mathrm{mod}_{\CC M}(A,X,w)$ is the cylinder $\{f\mid f(a)(i)=1\}$, so it is Borel;

\item if $A=\CC p$, then $\mathrm{mod}_{\CC M}(A,X,w)$ is either $(2^{\mathbb N})^{X}$ or the empty set, which are both Borel;

\item if $A=B\land C$, then $\mathrm{mod}_{\CC M}(A,X,w)=\mathrm{mod}_{\CC M}(B,X,w)\cap \mathrm{mod}_{\CC M}(C,X,w)$, so we conclude by the I.H.;

\item if $A=B\lor C$, then $\mathrm{mod}_{\CC M}(A,X,w)=\mathrm{mod}_{\CC M}(B,X,w)\cup \mathrm{mod}_{\CC M}(C,X,w)$, so we conclude by the I.H.;

\item if $A=B\to C$, then $\mathrm{mod}_{\CC M}(A,X,w)=
\bigcap_{w'\succeq w} \overline{\mathrm{mod}_{\CC M}(B,X,w')}\cup \mathrm{mod}_{\CC M}(C,X,w')
$ is, by the I.H., a countable intersection of Borel sets, so it is Borel;

\item if $A=\B C^{q}_{a}B$, then by the I.H.~, for all $w'\in W$,  
$\mathrm{mod}_{\CC M}(B, X\cup\{a\}, w')$ is Borel; using Lemma \ref{lemma:borel} we then have that the set
$$\mathrm{mod}_{\CC M}(A,X,w)=\bigcap_{w'\succeq w}
\left\{\omega\in (2^{\mathbb N})^{X}\mid
\mu\left( \Pi^{\omega}\left( \mathrm{mod}_{\CC M}(B, X\cup\{a\}, w') \right)\right) \geq q\right\}$$
is Borel.
%
%
%
\end{varitemize}
\end{proof}

\begin{lemma}[monotonicity]\label{lemma:monotonicity}
If $w,\omega \Vdash^{X}_{\CC M}  A$ and $w\preceq w'$, then $w',\omega \Vdash^{X}_{\CC M}  A$.
\end{lemma}
\begin{proof}
We argue by induction on $A$.

If $A=\CC p$ the claim follows from the fact that $\B i(\CC p)$ is an upper set;

If $A=\bvar_{a}^{i}$ the claim is immediate. 

If $A=B\land C$, the claim follows from the I.H.

If $A=B\to C$ and 
$w,\omega \Vdash^{X}_{\CC M}  A$, then for all $w''\succeq w$, 
$w'', \omega \Vdash^{X}_{\CC M}  B$ implies $w'',\omega \Vdash^{X}_{\CC M}  C$. Hence, in particular, for all $w''\succeq w'\succeq w$, $w'', \omega \Vdash^{X}_{\CC M}  B$ implies $w'',\omega \Vdash^{X}_{\CC M}  C$, and thus
$w', \omega \Vdash^{X}_{\CC M}  A$.

If $A=\B C^{q}_{a}B$ then by the I.H.~for all $\omega'\in 2^{\mathbb N}$, $w,\omega+\omega'\Vdash^{X\cup\{a\}}_{\CC M}B $ implies that for all $w'\succeq w$, $w',\omega+\omega'\Vdash^{X\cup\{a\}}_{\CC M}B $; we deduce then that for for all $w'\succeq w$, 
for all $w''\succeq w'\succeq w$, the set $
 \left \{ \omega'\in 2^{\mathbb N}\mid   w'',\omega+\omega'\Vdash^{X\cup\{a\}}_{\CC M}B \right \}
$ has measure $\geq q$, which implies  $w',\omega\Vdash^{X}_{\CC M}  A$.
%
%
%
%
\end{proof}

\begin{lemma}\label{lemma:alpha}
Let $\CC M=\langle W,\preceq, \alpha\rangle$ be a $\ICPL$-structure. 
For any finite set $X$ and $a\notin X$, for all 
$u\in W$ and $\bone\in \Borel_{X}$, and for any formula $A\in\mathrm{Formulae}_{X}$,
if $u,\omega \Vdash^{X}_{\CC M}  A$ holds, then $u, \omega'\Vdash^{X\cup\{a\}}_{\CC M}A$ holds for any $\omega'\in\{\omega\}^{\uparrow_{X\cup\{a\}}}$.
\end{lemma}
\begin{proof}
By induction on the rules.
\end{proof}
\begin{lemma}\label{lemma:ctriv}
Let $\CC M=\langle W,\preceq, \B i\rangle$ be a $\ICPL$-structure. 
For any finite set $X$, $a\notin X$, formula $A\in\mathrm{Formulae}_{X}$, $q\in(0,1]$ and $\omega \in \Borel_{X}$,
$$
w, \omega \Vdash^{X}_{\CC M}   \B C^{q}_{a}A \quad \To \quad 
w, \omega \Vdash^{X}_{\CC M}  A
$$
\end{lemma}
\begin{proof}
$w', \omega \Vdash^{X}_{\CC M}   \B C^{q}_{a}A $ implies that 
$\mu(\{\omega'\in 2^{\mathbb N}\mid u, \omega+\omega' \Vdash^{X\cup\{a\}}_{\CC M} A\}) \geq q>0$. From $u,\omega+\omega'\Vdash^{X\cup\{a\}}_{\CC M}A$ one can deduce by induction on $A$ that $u,\omega\Vdash^{X}_{\CC M}  A$. Since there exists at least one such $f$, we conclude that $u,\omega\Vdash^{X}_{\CC M}  A$ holds. 

\end{proof}
%
%
%

We are now in a position to establish, by induction on the rules, the soundness of $\NDCPL$.

\begin{proposition}[soundness]
If $\Gamma \vdash_{\NDCPL} A$, then $\Gamma \vDash A$.
%

\end{proposition}
\begin{proof}
Let $\CC M=\langle W, \preceq, \B i\rangle$ be a $\ICPL$-structure. 
We argue by induction on the rules of $\NDCPL$ (selected rules):

\begin{varitemize}
%
%
\item if the last rule is $\AXC{}\UIC{$\Gamma, A\vdash A$}\DP$, then for all $w\in W$ and $\omega \in \Borel_{X}$, if $w, \omega \Vdash^{X}_{\CC M}  \Gamma,A$, then 
$w, \omega \Vdash^{X}_{\CC M}  A$.

\item if the last rule is
$$
\AXC{$\Gamma, \bthree \vdash^{X\cup\{a\}} A$}
\AXC{$\mu(\bthree)\geq q$}
\BIC{$\Gamma \vdash \B C^{q}_{a}A$}
\DP
$$
then by the I.H.~together with Lemma \ref{lemma:monotonicity}, for all $w\in W$, for all $\omega\in \Borel_{X}$ and $\omega'\in 2^{\mathbb N}$, if $w,\omega \Vdash^{X}_{\CC M}  \Gamma$ and $w,\omega'\Vdash^{\{a\}} \bthree$ (where the latter condition only depends on $\omega'$), then for all $w'\succeq w$, we have 
$w',\omega+\omega' \Vdash^{X\cup\{a\}}_{\CC M}A$. 

%
Since $\mu(\bthree)\geq q$, we deduce that for all $w'\succeq w$, 
$$
\left \{ \omega'\mid w',\omega+\omega' \Vdash^{X\cup\{a\}}_{\CC M}A\right \}
$$
has measure greater than $q$, and we conclude then that $w,\omega \Vdash^{X}_{\CC M}  \B C^{q}_{a}A$.

\item if the last rule is 
$$
\AXC{$\Gamma \vdash \B C^{q}_{a}A$}
\AXC{$\Gamma, A\vdash^{X\cup\{a\}} C$}
\RL{($\B C$E$_{1}$) \ $a\notin \mathrm{FN}(\Gamma,C)$}
\BIC{$\Gamma \vdash  C$}
\DP
$$
let $w\in W$, $\omega\in (2^{\mathbb N})^{X}$ and $w,\omega \Vdash^{X}_{\CC M}  \Gamma$. 
By I.H.~we deduce that the set 
$$
S:=\left \{ \omega'\mid 
w, \omega+\omega' \Vdash^{X\cup\{a\}}_{\CC M}A\right \}
$$
has measure greater than $q>0$, and it is thus non-empty.
Let $\omega'$ be an element of $S$; again, by the I.H.~we deduce that $w,\omega+\omega'\Vdash^{X\cup\{a\}}_{\CC M}C$; 
since $a\notin \FN(C)$, by Lemma \ref{lemma:ctriv}, we conclude then that $w,\omega\Vdash^{X}_{\CC M}  C$.

\item if the last rule is 
$$
\AXC{$\Gamma \vdash \B C^{q}_{a}A$}
\AXC{$\Gamma, A\vdash^{X\cup\{a\}} C$}
\RL{($\B C$E$_{2}$) \ $a\notin \mathrm{FN}(\Gamma)$}
\BIC{$\Gamma \vdash  \B C^{qs}_{a}C$}
\DP
$$
let $w\in W$, $\omega\in (2^{\mathbb N})^{X}$ and $w,\omega \Vdash^{X}_{\CC M}\Gamma$. 
By I.H.~we deduce that the set 
$$
S:=\left \{ \omega'\mid 
w, \omega+\omega' \Vdash^{X\cup\{a\}}_{\CC M}A\right \}
$$
has measure greater than $q\geq qs$ and is contained in the set 
$
S':=\left \{ \omega'\mid 
w, \omega+\omega' \Vdash^{X\cup\{a\}}_{\CC M}C\right \}
$. 
We can conclude then that $w,\omega\Vdash^{X}_{\CC M}  \B C^{qs}_{a}C$.

\item Axiom \eqref{ax1} is valid: suppose $w,\omega \Vdash^{X}_{\CC M}\B C^{q}_{a} (A\lor B)$, where $a\notin \FN(A)$. Then the set $\{\omega'\mid w,\omega+\omega'\Vdash^{X\cup\{a\}}_{\CC M}A\lor B\}$ has measure greater than $q$.
Observe that $w,\omega+\omega'\Vdash^{X\cup\{a\}}_{\CC M}A$ holds iff $w,\omega\Vdash^{X}_{\CC M}A$; so suppose $w,\omega\Vdash^{X}_{\CC M}A$ does not hold; then for any $\omega'$, 
 $w,\omega+\omega'\Vdash^{X\cup\{a\}}_{\CC M}A$ does not hold, and thus the set 
$\{\omega'\mid w,\omega+\omega'\Vdash^{X\cup\{a\}}_{\CC M} B\}$ must have measure greater than $q$. 
We have thus proved that either $w,\omega\Vdash^{X}_{\CC M}A$ holds
 or $w, \omega \Vdash^{X}_{\CC M}\B C^{q}_{a}B$ holds, and thus that 
 $w,\omega \Vdash^{X}A\lor \B C^{q}_{a}B$.

\item Axiom \eqref{ax2} is valid:
suppose  $\FN(\bone)\subseteq\{a\}$ and $\mu(\bone)<q$; for any $w\in W$, the set $\{\omega'\mid w,\omega' \vdash^{\{a\}}\bone\}$ coincides with $\model{\bone}$, and has thus measure $<q$ by hypothesis; we deduce that for any $w'\geq w$, $w',*\Vdash^{\emptyset}_{\CC M}\B C^{q}_{a}\bone$ does not hold, and we conclude that $w,*\Vdash^{\emptyset}_{\CC M} \lnot \B C^{q}_{a}\bone$.
\end{varitemize}
\end{proof}


To prove the completeness theorem we first need to study provability in $\ICPL$ a bit further.

\begin{definition}
For any formula $A$ and finite set $X$, $|A|_{X}:=\max\{i\mid \exists a\in X\cap \FN(A) \text{ s.t. }\bvar_{a}^{i} \text{ occurs in }A\}$. We use $|A|$ as a shorthand for $|A|_{\FN(A)}$.
For any $\omega\in (2^{\mathbb N})^{\CC A}$,  
$A^{\omega}:=\TTH^{|A|}_{\FN(A)}(\omega) \to A$. 
\end{definition}

We will exploit the following relativized version of the decomposition lemma (Lemma \ref{lemma:decomposition}), which is proved in a similar way:
\begin{lemma}\label{lemma:adecompo}
Let $A$ be a formula with $\FN(A)\subseteq X\cup Y$, where $X\cap Y=\emptyset$. Then there exist formulas $A^{X}_{w}$, with $\FN(A^{X}_{w})\subseteq X$, where $w$ ranges over the valuations of the variables of name in $Y$ in $A$, 
 such that 
$\vdash A \leftrightarrow ( \bigvee_{w}\TTH(w)\land A^{X}_{w} )$.
\end{lemma}
%
%


\begin{definition}
For each formula $A$, with $\FN(A)\subseteq \{a\}$, let $\bone_{A}:=\bigvee_{v}\{\TTH(v)\mid A_{v}\not\vdash_{\NDCPL}\BOT\}$. 
\end{definition}

From $A\vdash \bigvee_{v}\TTH(v) \land A_{v}$ and $A\vdash\lnot \TTH(w)$ for all $\TTH (w)$ not occurring in $\bone_{A}$, we can deduce $A\vdash \bone_{A}$. 
The following result shows that measuring the formula $\bone_{A}$ provides a test to know if $\B C^{q}_{a}A$ is consistent:

\begin{corollary}
Let $A$ be a formula with $\FN(A)\subseteq \{a\}$. Then for all $q\in (0,1]$, 
$$
\mu (\bone_{A})< q \quad \To \quad \B C^{q}_{a}A \vdash \BOT
$$
\end{corollary}
\begin{proof}
The claim follows from the observation that $A\vdash \bone_{A}$ by applying Axiom \eqref{ax2} and the rule ($\B C$E$_{2}$).\end{proof}

\begin{remark}[relativizing $\bone_{A}$ to $\omega$]\label{rem:henkinrelative}
Suppose $A$ is a formula with $\FN(A)\subseteq X\cup\{a\}$ (with $a\notin X$). Then for each $\omega \in (2^{\mathbb N})^{\CC A}$, let $A^{X}_{\omega}$ be $A^{X}_{w}$, where $w$ is the restriction of $\omega$ to $(2^{|A|})^{X}$.
 Then, by applying the construction from Lemma \ref{lemma:adecompo} to $A^{X}_{\omega}$ we can define a Boolean formula $\bone_{A}^{\omega}$ with $\FN(\bone_{A}^{\omega})\subseteq\{a\}$ such that $\TTH^{|A|}_{X}(\omega)\vdash A\to \bone_{A}^{\omega}$ and such that 
 $\mu(\bone_{A}^{\omega})< q \To  \TTH^{|A|}_{X}(\omega)\vdash \lnot \B C^{q}_{a}A$.
\end{remark}

\bigskip

The proof of the completeness theorem relies on the construction of a suitable ``canonical model'' based on sets of formulas. Let us first introduce some terminology.
Let $\Gamma$ be a (possibly infinite) set of formulas. We say that $\Gamma$ is 
\begin{varitemize}

\item \emph{$A$-consistent} if $\Gamma \not \vdash A$;
\item \emph{consistent} if it is $\BOT$-consistent;
\item \emph{super-consistent} if for all $\omega\in (2^{\mathbb N})^{\CC A}$, $\TTH(\omega)\cup \Gamma$ is consistent;
\item \emph{closed} if $A_{1},\dots, A_{n}\in \Gamma$ and $A_{1},\dots, A_{n}\vdash_{\NDCPL}A$ implies $A\in \Gamma$, and if $(\bthree\to A)^{\omega}\in \Gamma$ and $\mu(\bthree)\geq q$ then $(\B C^{q}_{a}A)^{\omega}\in \Gamma$.

\end{varitemize}

Let $\Theta$ be the set of all super-consistent and closed sets of formulas $\Gamma$ satisfying the following conditions, for all $\omega\in (2^{\mathbb N})^{\CC A}$:
\begin{align*}
(A\lor B)^{\omega} \in \Gamma \quad & \To \quad  A^{\omega}\in \Gamma \text{ or }B^{\omega}\in \Gamma
\tag{$\vee$-closure}\\
(\B C^{q}_{a}A)^{\omega}\in \Gamma \quad & \To \quad 
\exists\bthree \text{ s.t. }\mu(\bthree)\geq q \text{ and }
(\bthree \to A)^{\omega}\in \Gamma
\tag{$\B C$-closure}
\end{align*}

%

%

The fundamental ingredient for the completeness theorem is the lemma below, which will be used to lift any consistent set of formulas to an element of $\Theta$.

\begin{lemma}[saturation lemma]\label{lemma:saturation}
Let $\Gamma$ be a super-consistent set and let $\omega_{0}\in (2^{\mathbb N})^{\CC A}$ be such that 
$\TTH(\omega_{0})\cup \Gamma$ is $A$-consistent.
Then there exists a set $\Gamma^{*}\supseteq \Gamma$ such that $\Gamma^{*}\in\Theta$ and 
$\TTH_{X}(\omega_{0})\cup \Gamma^{*}$ is $A$-consistent.
\end{lemma}
\begin{proof}
Let us fix an enumeration $(a_{i})_{i\in \mathbb N}$ of all names.
For any $p,q\in \mathbb N$, Let $\CC F_{p,q}$ be the set of all formulas $B$ such that $\FN(B)\subseteq \{a_{0},\dots, a_{q-1}\}$ and $|B|\leq p$. Let us fix, for all $p,q\in \mathbb N$ an enumeration $(C_{n}^{p,q})_{n\in \mathbb N}$ of $\CC F_{p,q}$.

For any natural number $N$, let $[N]:=\{0,\dots, N-1\}$. Given $p\leq p'$ and $q\leq q'$, and 
finite matrices $s \in (2^{[p+1]})^{[q+1]}$
$s'\in (2^{[p'+1]})^{[q'+1]}$, let $s\sqsubseteq s'$ if for all $i\leq p$ and $j\leq q$ $s(j)(i)=s'(j)(i)$. 
We will often indicate as $p_{s}$ and $q_{s}$ the (unique) natural numbers such that $s\in (2^{[p+1]})^{[q+1]}$.

Moreover, for all $\omega\in (2^{\mathbb N})^{X}$, let $s\sqsubseteq \omega$ hold if for all $i\leq p_{s}$ and $j\leq q_{s}$, $s(j)(i)=\omega(a_{j})(i)$.

Observe that if $s\sqsubseteq \omega$, $\TTH(\omega)\vdash \TTH(s)$.

For all $p,q$ and $s\in (2^{[p+1]})^{[q+1]}$, we define a set of formulas $\Gamma^{\langle p,q,s\rangle}\subseteq \CC F_{p,q}$ such that for all $p\leq p'$ and $q\leq q'$, $s \in (2^{[p+1]})^{[q+1]}$ and $s' \in (2^{[p'+1]})^{[q'+1]}$, 
 $s\sqsubseteq s'$ implies $\Gamma^{\langle p,q,s\rangle}\subseteq \Gamma^{\langle p',q',s'\rangle}$.

We let $\Gamma^{\langle p,q,s\rangle}:=\bigcup_{n}\Gamma^{\langle p,q,s\rangle}_{n}$, where the sets 
$\Gamma^{\langle p,q,s\rangle}_{n}$ are defined by a triple induction on $p,q$ and $n$ as follows:
\begin{varitemize}

\item if $p=q=0$ and $b= \omega_{0}(0)$, 
	\begin{varitemize}
	\item $\Gamma^{\langle 0,0,b\rangle}_{0}:= \Gamma \cup\{\lnot^{(1-b)}\bvar_{0}^{0}\}$; 
 
 \item $\Gamma^{\langle 0,0,b\rangle}_{n+1}:= \Gamma_{n} \cup \{C^{0,0}_{n}\}$, if this is $A$-consistent, and 
 $\Gamma^{\langle 0,0,b\rangle}_{n+1}:= \Gamma_{n}^{\langle 0,0,b\rangle} $ otherwise.
 	\end{varitemize}
If $b\neq \omega_{0}(0)$, the definition is the same, with $\BOT$ in place of $A$; 

\item if $p>0$ and $q=0$, $\Gamma^{\langle p,0,s\rangle}= \bigcup_{n}\Gamma^{\langle p,0,s\rangle}_{n}$, where if 
$s\sqsubseteq \omega_{0}$, 
	\begin{varitemize}
	\item $\Gamma^{\langle p,0,s\rangle}_{0}:=
	\Gamma^{\langle p-1,0, s|_{p-1,0}\rangle}  
\cup
	\{\lnot^{(1-s(p)(0))}\bvar_{p}^{0}\}$;
	\item $\Gamma^{\langle p,0,s\rangle}_{n+1}:=
	\Gamma^{\langle p,0, s\rangle}_{n}\cup \{C^{p,0}_{n}\}$ if this is $A$-consistent, and 
	$\Gamma^{\langle p,0,s\rangle}_{n+1}:=
	\Gamma^{\langle p,0, s\rangle}_{n}$ otherwise;

\end{varitemize}
and if $s\not\sqsubseteq \omega_{0}$, the definition is similar, with $\BOT$ in place of $A$;

\item if $p=0$ and $q>0$, $\Gamma^{\langle 0,q,s\rangle}= \bigcup_{n}\Gamma^{\langle 0,q,s\rangle}_{n}$, where if 
$s\sqsubseteq \omega_{0}$, 
	\begin{varitemize}
	\item $\Gamma^{\langle0,q,s\rangle}_{0}:=
	\Gamma^{\langle 0,q-1, s|_{0,q-1}\rangle}  
\cup
	\{\lnot^{(1-s(0)(q))}\bvar_{0}^{q}\}$;
	\item $\Gamma^{\langle 0,q,s\rangle}_{n+1}:=
	\Gamma^{\langle 0,q, s\rangle}_{n}\cup \{C^{0,q}_{n}\}$ if this is $A$-consistent, and 
	$\Gamma^{\langle 0,q,s\rangle}_{n+1}:=
	\Gamma^{\langle 0,q, s\rangle}_{n}$ otherwise;	

\end{varitemize}
and if $s\not\sqsubseteq \omega_{0}$, the definition is similar, with $\BOT$ in place of $A$;

\item if both $p>0$ and $q>0$, $\Gamma^{\langle p,q,s\rangle}= \bigcup_{n}\Gamma^{\langle p,q,s\rangle}_{n}$, where if 
$s\sqsubseteq \omega_{0}$, 
	\begin{varitemize}
	\item $\Gamma^{\langle p,q,s\rangle}_{0}:=
	\Gamma^{\langle p-1,q, s|_{p-1,q}\rangle}  
\cup 	\Gamma^{\langle p,q-1, s|_{p,q-1}\rangle}  
\cup
	\{\lnot^{(1-s(p)(q))}\bvar_{p}^{q}\}$;
	\item $\Gamma^{\langle p,q,s\rangle}_{n+1}:=
	\Gamma^{\langle p,q, s\rangle}_{n}\cup \{C^{p,q}_{n}\}$ if this is $A$-consistent, and 
	$\Gamma^{\langle p,q,s\rangle}_{n+1}:=
	\Gamma^{\langle p,q, s\rangle}_{n}$ otherwise;	

\end{varitemize}
and if $s\not\sqsubseteq \omega_{0}$, the definition is similar, with $\BOT$ in place of $A$.

\end{varitemize}
In the following, whenever this creates no confusion, we will indicate $\Gamma^{\langle p,q,s\rangle}$ simply as $\Gamma^{s}$.

For all finite matrices $s,s'$, the following hold:
\begin{varitemize}
\item[a.]  $\TTH(s)\subseteq \Gamma^{s}$;
\item[b.] if $s\sqsubseteq \omega_{0}$, $\Gamma^{s}$ is $A$-consistent;
\item[c.] if $s\not\sqsubseteq \omega_{0}$, $\Gamma^{s}$ is consistent;
\item[d.] if $s\sqsubseteq s'$, $\Gamma^{s}\subseteq \Gamma^{s'}$;
\item[e.] if $\Gamma^{s}\vdash B$, then for all $t\sqsupseteq s$ such that $B\in \CC F_{p_{t},q_{t}}$, $B\in \Gamma^{t}$.
\end{varitemize}
Facts a.-d. are verified by construction, so we only prove e.: suppose there exist $B_{1},\dots, B_{n}\in \Gamma^{s}$ such that $B_{1},\dots, B_{n}\vdash_{\NDCPL}B$ and let $t\sqsupseteq s$ be such that $B \in \CC F_{p_{t},q_{t}}$. Suppose $B\notin \Gamma^{t}$: then for some $k\in \mathbb N$, $\Gamma^{t}_{k}\cup\{B\}\vdash \BOT$; yet, from d.~it follows that $B_{1},\dots,B_{n}\in \Gamma^{t}$, and thus $\Gamma^{t}\vdash \BOT$, contradicting c.

Given matrices $s_{1},\dots, s_{n}\sqsubseteq \omega$, let $\bigvee^{\omega}\{s_{1},\dots, s_{n}\}$ indicate the smallest sub-matrix of $\omega$ extending all $s_{1},\dots, s_{n}$ (i.e.~the restriction of $\omega$ to $p=\max\{p_{s_{1}},\dots, p_{s_{n}}\}$ and $q=\max\{q_{s_{1}},\dots, q_{s_{n}}\}$. 

For any matrix $s$, let $\Gamma^{\dag s}=\{A^{s}\mid A\in \Gamma^{s}\}=\{\TTH(s)\to A\mid A\in \Gamma^{s}\}$. 
Let $\Gamma^{\dag}=\bigcup_{s}\Gamma^{\dag s}$ and let $\Gamma^{*}$ be the deductive closure of $\Gamma^{\dag}$.
We will establish the following properties of $\Gamma^{*}$:
\begin{varitemize}
\item[$\alpha.$] $\Gamma^{*}$ is super-consistent;
\item[$\beta.$] $\TTH(\omega_{0})\cup \Gamma^{*}$ is $A$-consistent;
\item[$\gamma.$] $\Gamma\subseteq\Gamma^{*}$;

\item[$\delta.$] if $B^{\omega} \notin \Gamma^{*}$ and $B^{\omega}=B^{s}$, then 
$\TTH(s)\cup\Gamma^{*}\cup\{B\}\vdash \BOT$;
moreover, if $B^{\omega_{0}} \notin \Gamma^{*}$ and $B^{\omega_{0}}=B^{s}$, then 
$\TTH(s)\cup\Gamma^{*}\cup\{B\}\vdash A$;

\item[$\epsilon.$] $\Gamma^{*}$ is $\vee$-closed;

\item[$\eta.$] $\Gamma^{*}$ is $\B C$-closed. 

\end{varitemize}
This will conclude the proof of the theorem.

Let us preliminarily observe that any formula $B^{\omega}$ can be written as $B^{s}$ for a unique $s\sqsubseteq \omega$ such that $p_{s}=|B|$ and $q_{s}$ is minimum with the property that $\FN(B)\subseteq\{a_{0},\dots, a_{q_{s}-1}\}$.

\begin{varitemize}
\item[$\alpha.$] 
Let us show that $\Gamma^{\dag}$ is super-consistent. This immediately implies that $\Gamma^{*}$ is super-consistent too.
 Suppose $\TTH(\omega)\cup\Gamma^{\dag}\vdash \BOT$; then there exist $s_{1},\dots, s_{n}$ and $B_{1}\in \Gamma^{ s_{1}},\dots, B_{n}\in\Gamma^{ s_{n}}$, such that 
$\TTH(\omega) , B_{1}^{s_{1}},\dots, B_{n}^{s_{n}}\vdash \BOT$; we can suppose w.l.o.g.~that $s_{1},\dots, s_{n}\sqsubseteq \omega$, since if $s_{i}\not\sqsubseteq \omega$, $\TTH(\omega)\vdash \lnot \TTH(s_{i})$, and thus $\TTH(\omega)\vdash B_{i}^{s_{i}} $.

 We have then $B_{i}^{s_{i}}= B_{i}^{\omega}$; 
by letting $s=\bigvee^{\omega}\{s_{1},\dots, s_{n}\}$ we thus have 
$B_{i}^{s_{i}}=B_{i}^{s}$ and 
$\TTH(s),  B_{1}^{s},\dots, B_{n}^{s}\vdash \BOT$, 
which implies $\TTH(s),  B_{1},\dots, B_{n}\vdash \BOT$, 
 and since $\TTH(s)\cup\{ B_{1},\dots, B_{n}\}\subseteq \Gamma^{ s}$, we deduce 
$\BOT \in \Gamma^{  s}$, contradicting a.

\item[$\beta.$] The argument is similar to the one for $\alpha$.

\item[$\gamma.$]
Let $B\in \Gamma$ and $s$ be such that $B^{\omega}=B^{s}$, $p_{s}=|B|$ and $\FN(B)\subseteq \{a_{0},\dots, a_{q_{s}-1}\}$. 
For any $s'\in (2^{[p_{s}+1]})^{[q_{s}+1]}$, $B^{s'}\in \Gamma^{\dag s'}\subseteq \Gamma^{\dag}$. Hence, using the fact that
$\vdash \bigvee_{s\in (2^{[p_{s}+1]})^{[q_{s}+1]}}\TTH(s)$, we deduce that 
$\{ B^{s}\mid s\in (2^{[p_{s}+1]})^{[q_{s}+1]}\} \vdash B$, and thus that $\Gamma^{\dag}\vdash B$, which implies $B\in \Gamma^{*}$.

\item[$\delta.$]
Let $s$ be such that $B^{\omega}=B^{s}$. We consider the case of $\omega=\omega_{0}$; the case $\omega \neq \omega_{0}$ is proved similarly with $\BOT $ in place of $A$.
We will prove the contrapositive, i.e.~that if $\TTH(s)\cup \Gamma^{\dag}\cup \{B\}$ is $A$-consistent, then $B^{\omega_{0}}=B^{s}\in \Gamma^{*}$, from which $\delta$.~follows. Suppose that $\TTH(\omega)\cup \Gamma^{\dag}\cup \{B\}$ is $A$-consistent. Suppose $B\notin \Gamma^{s}$: this implies $B\notin \Gamma$ and that for some $k\in \mathbb N$, $\Gamma^{s}_{k}\cup\{B\}\vdash A$. But this forces then $\Gamma^{s}\cup\{B\}\vdash A$; 
since for all $F\in \Gamma^{s}$, $F^{s}\in \Gamma^{\dag}$, and $\TTH(s), F^{s} \vdash F$, we deduce then 
$\TTH(s)\cup \Gamma^{\dag}\cup \{B\}\vdash A$, which is absurd. We conclude then that $B\in \Gamma^{s}$, and thus $B^{s}=B^{\omega}\in \Gamma^{\dag}\subseteq \Gamma^{*}$.

\item[$\epsilon.$] 
Again, we consider the case of $\omega=\omega_{0}$; the case $\omega \neq \omega_{0}$ is proved similarly with $\BOT $ in place of $A$.
Let $s\sqsubseteq \omega_{0}$ be such that  $(B\lor C)^{\omega}=(B \lor C)^{s}$ and suppose $(B\lor C)^{s}\in \Gamma^{*}$ but neither $B^{s}\in \Gamma^{*}$ nor $C^{s}\in \Gamma^{*}$; then by $\delta.$, $\TTH(s)\cup \Gamma^{*}\cup\{B\}\vdash A$ and $\TTH(s)\cup \Gamma^{*}\cup\{C\}\vdash A$, hence
$\TTH(s)\cup \Gamma^{*}\cup\{B\lor C\}\vdash A$;
since $(B\lor C)^{s}\in \Gamma^{*}$ we have $\TTH(s)\cup \Gamma^{*}\vdash B\lor C$, so we deduce  
$\TTH(s)\cup \Gamma^{*}\vdash A$, and since $\TTH(s)\subseteq \TTH(\omega_{0})$, we have
$\TTH(\omega_{0})\cup \Gamma^{*}\vdash A$, contradicting $\beta$.

\item[$\eta.$] 
Once more, we consider the case of $\omega=\omega_{0}$; the case $\omega \neq \omega_{0}$ is proved similarly with $\BOT $ in place of $A$.

Let $s\sqsubseteq \omega_{0}$ be such that $(\B C^{q}_{a}B)^{\omega_{0}}= (\B C^{q}_{a}B)^{s}$ and suppose $(\B C^{q}_{a}B)^{s}\in \Gamma^{*}$. 

We can suppose w.l.o.g.~that $a=a_{q}$ with $q> q_{s}$.
Hence we can suppose that the formula $B^{\omega_{0}}$ is of the form $B^{s+s'+s''}$, for some finite matrices $s',s''$, with $s''\in 2^{[|B|+1]}$. 
 
%

By Lemma \ref{lemma:adecompo} and Remark \ref{rem:henkinrelative} we have that $\Gamma^{s}\vdash B\leftrightarrow( \bigvee_{v\in 2^{[|B|+1]} }\TTH(v)\land B_{s+v})$, where $B_{s+v}$ has no free name.
Let $S=\{ v\in 2^{[|B|+1]}\mid B_{s+v}\in \Gamma^{s} \}$ and $\bthree= \bigvee_{v\in S}\TTH(v)$. 
Observe that for any $v\in S$, $\Gamma^{s+s'}\vdash \TTH(v)\to B$.
Let $v\in S$ and $\omega_{v}$ be such that $v\subseteq \omega_{v}$. Then 
$\Gamma^{s+s'}\vdash \TTH(\omega_{v})\to B$.
Using e.~we have then that for all $t\in 2^{[|B|+1]}$, 
$\TTH(\omega_{v})\to B\in \Gamma^{s+s'+t}$, so in particular, 
$\TTH(\omega_{v})\to B\in \Gamma^{s+s'+s''}$, so
$(\TTH(\omega_{v})\to B)^{\omega_{0}}\in \Gamma^{\dag}\subseteq \Gamma^{*}$. 
%
%
We thus have that for all $v\in S$, $\Gamma^{*}\vdash (\TTH(v)\to B)^{\omega_{0}}$, which implies
$\Gamma^{*}\vdash (\bthree\to B)^{\omega_{0}}$ and thus $(\bthree\to B)^{\omega_{0}}\in \Gamma^{*}$. 

At this point, if $\mu(\bthree)\geq q$ we are done.
Otherwise, suppose $\mu(\bthree)<q$; if $v\notin S$, then $B_{s+v}\notin \Gamma^{s}$, which implies that for some $k$, 
$\Gamma^{s}_{k}\cup\{B_{s+v}\}\vdash A$; this implies that $\Gamma^{s}\vdash (\TTH(v)\land B_{s+v}) \to A$; we deduce then that $\Gamma^{s}\vdash B\to ((\bigvee_{v\in S}\TTH(v)\land B_{s+v})\lor A)$, and in particular that $\Gamma^{s}\vdash B\to(\bthree\lor A)$;
this means that there exist formulas $B_{1},\dots, B_{n}\in \Gamma^{s}\subseteq \CC F_{p_{s},q_{s}}$, hence containing no occurrences of the name $a$, such that $B_{1},\dots, B_{n} \vdash B\to  (\bthree\lor A)$;
from this we deduce first $\TTH(s), B_{1}^{s},\dots, B_{n}^{s}\vdash B\to (\bthree\lor A)$, and then, using ($\B C$E$_{2}$), $\TTH(s),\B C^{q}_{a}B, B^{s}_{1},\dots, B^{s}_{n}\vdash \B C^{q}_{a}(\bthree\lor A)$. Now, since $s\subseteq\omega_{0}$, $\B C^{q}_{a}A, B_{1}^{s},\dots, B_{n}^{s}\in \Gamma^{*}$ and 
$  \B C^{q}_{a}(\bthree\lor A)\vdash  (\B C^{q}_{a}\bthree)\lor A$ (Axiom \eqref{ax1}, since we can suppose w.l.o.g.~that $a\notin \FN(A)$), we have 
$\TTH(\omega_{0})\cup\Gamma^{*}\vdash (\B C^{q}_{a}\bthree)\lor A$; moreover, from $\mu(\bthree)< q$ we get (by Axiom \eqref{ax2}) 
$\vdash \lnot \B C^{q}_{a}\bthree$ and thus $(\B C^{q}_{a}\bthree)\lor A\vdash A$, so we can conclude
$\TTH(\omega_{0})\cup\Gamma^{*}\vdash A$, which contradicts $\beta$.

\end{varitemize}
\end{proof}

\begin{lemma}\label{lemma:closure}
Let $\Gamma$ be an $A$-consistent set of formulas. Then there exists a super-consistent set $\Delta$ and $\omega\in (2^{\mathbb N})^{\CC A}$ such that ${\TTH(\omega)\cup \Delta}$ is $A$-consistent and its closure contains $\Gamma$. 
\end{lemma}
\begin{proof}
Let us define sets of formulae $\Gamma_{a,i}$ together with $\omega(\langle a,i\rangle)$ as follows:
\begin{varitemize}
\item let $\omega(\langle 0,0\rangle)=1$ if $\Gamma\cup\{\bvar_{0}^{0}\}$ is $A$-consistent, and $\omega(\langle 0,0\rangle)=0$ otherwise;
moreover, let
 $\Gamma_{0,0}= \Gamma$;

\item let $\omega(\langle a, i+1\rangle)=1$ if  $\Gamma_{a,i}\cup\{\bvar_{a}^{i+1}\}$ is $A$-consistent, and $\omega(\langle a,i+1\rangle)=0$ otherwise;
moreover, let
$\Gamma_{a,i+1}=\Gamma_{a,i}\cup\{\lnot^{(1-\omega(\langle a,i+1\rangle))}\bvar_{a}^{i+1}\}$;
\item let $\omega(\langle a+1, 0\rangle)=1$ if  $\bigcup_{i}\Gamma_{a,i}\cup\{\bvar_{a+1}^{0}\}$ is $A$-consistent, and $\omega(\langle a,i+1\rangle)=0$ otherwise;
moreover, let
$\Gamma_{a+1,0}=\bigcup_{i}\Gamma_{a,i}\cup\{(\lnot^{1-\omega(\langle a+1,0\rangle))}\bvar_{a+1}^{0}\}$.

\end{varitemize}


By construction we then have that $\TTH(\omega)\cup \Gamma=\bigcup_{a,i} \Gamma_{a,i}$;
let us show that  
$\bigcup_{a,i} \Gamma_{a,i}$ is $A$-consistent:
\begin{varitemize}
\item $\Gamma_{0,0}=\Gamma$ is $A$-consistent by hypothesis;

\item suppose $\Gamma_{a,i}$ is $A$-consistent; if $\Gamma_{a,i+1}= \Gamma_{a,i}\cup\{\bvar_{a}^{i+1}\}$ then by construction $\Gamma_{a,i+1} $ is $A$-consistent;
if $\Gamma_{a,i+1}= \Gamma_{a,i}\cup\{\lnot\bvar_{a}^{i+1}\}$ then by construction $\Gamma_{a,i} \cup\{\bvar_{a}^{i}\} \vdash A$; if moreover
$\Gamma_{a,i} \cup\{\lnot\bvar_{a}^{i}\} \vdash A$, the 
$\Gamma_{a,i}\cup\{\bvar_{a}^{i}\lor\lnot \bvar_{a}^{i}\}\vdash A$ and thus
$\Gamma_{a,i}\vdash A$, which is absurd.
We conclude then that $\Gamma_{a,i+1}$ is $A$-consistent.

\item suppose $\Gamma_{a,i}$ is $A$-consistent for all $i\in \mathbb N$; first observe that $\Gamma_{a}:=\bigcup_{i}\Gamma_{a,i}$ is $A$-consistent: if $\Gamma_{a}\vdash A$, then there exist $B_{1}\in \Gamma_{a,i_{1}},\dots, B_{n}\in \Gamma_{a,i_{n}}$ such that $B_{1},\dots, B_{n}\vdash  A$;  since for $j\leq k$,  $\Gamma_{a,j}\subseteq \Gamma_{a,k}$, we deduce that $\Gamma_{a,\max\{i_{1},\dots, i_{n}\}}\vdash A$, which is absurd.

Now, if $\Gamma_{a+1,0}=\Gamma_{a}\cup\{\bvar_{a+1}^{0}\}$ then by construction $\Gamma_{a+1,0}$ is $A$-consistent; if  $\Gamma_{a+1,0}=\Gamma_{a}\cup\{\lnot\bvar_{a+1}^{0}\}$ then 
$\Gamma_{a}\cup\{\bvar_{a+1}^{0}\}\vdash A$; hence, if $\Gamma_{a+1,0}\vdash A$, we deduce
$\Gamma_{a}\cup\{\bvar_{a+1}^{0}\lor \lnot \bvar_{a+1}^{0}\} \vdash A$, so we conclude
$\Gamma_{a}\vdash A$, which is absurd. 
\end{varitemize}

Now, for each formula $B\in \Gamma$ let us define a formula $B^{*}$ as follows:
$$
B^{*} := B \lor \left ( \bigvee_{a\in \FN(B), i\leq |B|} \lnot^{\omega(\langle a,i\rangle)}\bvar_{a}^{i}   \right)
$$
Observe that $\TTH(\omega)\cup\{B^{*}\}\vdash B$ and that for all $\omega'\in(2^{\mathbb N})^{\CC A}$, 
$\TTH(\omega')\cup\{B^{*}\}$ is consistent.

Let then $\Delta=\{B^{*}\mid B\in \Gamma\}$. It is clear that $\TTH(\omega)\cup \Delta \vdash B$ for all $B\in \Gamma$.
Suppose $\TTH(\omega)\cup \Delta\vdash A$, then there exists Boolean formulas $
\bfour_{1},\dots, \bfour_{k}\in \TTH(\omega)$ and formulas $B_{1},\dots, B_{n}\in \Gamma$ such that 
$\bfour_{1},\dots, \bfour_{k}, B_{1}^{*},\dots, B_{n}^{*}\vdash A$. Since $B_{i}\vdash B^{*}$, this implies then
$\bfour_{1},\dots, \bfour_{k}, B_{1},\dots, B_{n}\vdash A$, and thus $\TTH(\omega)\cup \Gamma$ is not $A$-consistent, which is absurd. 
\end{proof}

We now have all elements to proceed to the proof of the completeness theorem.
\begin{theorem}[completeness]
If $\Gamma \vDash A$, then $\Gamma \vdash _{\NDCPL}A$.
\end{theorem}
\begin{proof}
Let $\CC T= \langle \Theta  , \subseteq, \B i\rangle$, where
$\B i( \CC p)=\{\Gamma\mid \CC p\in  \Gamma \}$. 
%
%
%

We will prove that for all $\Gamma \in \Theta$ and $\omega\in \Borel_{X}$, 
$
\Gamma, \omega \vdash_{\CC T}^{X} A$ iff $A^{\omega}\in \Gamma
$.
From this claim the theorem is proved as follows: suppose $\Gamma \not\vdash  A$; by Lemma \ref{lemma:closure} we obtain a super-consistent set $\Delta$ and $\omega_{0}\in (2^{\mathbb N})^{\CC A}$ such that $\TTH(\omega_{0})\cup \Delta$ is $A$-consistent and its closure contains $\Gamma$. 
By Lemma \ref{lemma:saturation} $\Delta$ extends to $\Gamma^{*}\in \Theta$ such that $\Gamma^{*}\supseteq \Gamma$ and $\TTH(\omega_{0})\cup\Gamma^{*}$ is $A$-consistent. 
%
%
%
%
Then using the claim, from $\Gamma \subseteq \Gamma^{*}$, we deduce that $\Gamma^{*},\omega_{0} \vdash_{\CC T}^{X}\Gamma$, and from $A\notin \Gamma^{*}$ we deduce that 
$\Gamma^{*},\omega_{0}\not\vdash_{\CC T}^{X}A$; we can conclude then that $\Gamma \not\vdash_{\CC T}^{X} A$, and thus that $\Gamma \not\vDash A$.

%

Let us now prove the claim, by arguing by induction on $A$: 
\begin{varitemize}

\item if $A=\BOT$, from the vacuous assumption $\Gamma, \omega \vdash_{\CC T}^{X}\BOT$ we can freely deduce $\BOT^{\omega}\in \Gamma$; conversely, 
 from the vacuous assumption 
 $\BOT^{\omega}\in \Gamma$ we can freely deduce 
 $\Gamma, \omega \vdash_{\CC T}^{X}\BOT$;
 
 \item if $A=\TOP$, since $\Gamma, \omega \vdash_{\CC T}^{X}\TOP$ always holds and $\TOP^{\omega}\in \Gamma$ always holds (since $\Gamma$ is closed), we can conclude;

 \item if $A=\CC p$, $\Gamma, \omega \vdash_{\CC T}^{X}\CC p$ iff $\CC p\in  \Gamma$ holds by definition (since $\CC p\notin \TTH(\omega)$);

\item if $A=\bvar_{a}^{i}$, if $\Gamma, \omega \vdash_{\CC T}^{X}A$ holds then $\omega(a)(i)=1$, hence $\TTH(\omega)\vdash \bvar_{a}^{i}$, and thus $A^{\omega}$ is logically valid, which implies $A^{\omega}\in \Gamma$ by closure;
conversely, if $A^{\omega}\in \Gamma$,  since $\TTH(\omega)\cup \Gamma$ is consistent, the only possibility is that $\omega(a)(i)=1$, which forces $\Gamma, \omega \vdash_{\CC T}^{X}A$.

 \item if $A=B\land C$, then $\Gamma, \omega \vdash_{\CC T}^{X}A$ iff 
$\Gamma, \omega \vdash_{\CC T}^{X}B$ and $\Gamma, \omega \vdash_{\CC T}^{X}C$, which by the I.H.~is equivalent to
$B^{\omega}\in   \Gamma$ and $C^{\omega}\in \Gamma$, which is in turn equivalent to $(B\land C)^{\omega}\in \Gamma$.

 \item if $A=B\lor C$, and $\Gamma, \omega \vdash_{\CC T}^{X}A$, then either 
$\Gamma, \omega \vdash_{\CC T}^{X}B$ or $\Gamma, \omega \vdash_{\CC T}^{X}C$ holds; 
if the first holds, then by the I.H.~$B^{\omega}\in \Gamma$, which implies $(B\lor C)^{\omega}\in \Gamma$, and one can argue similarly if the second holds;
conversely, if $(B\lor C)^{\omega}\in \Gamma$, by $\vee$-closure either $B^{\omega}\in \Gamma$ or $C^{\omega}\in \Gamma$, so by the I.H.~in each case we deduce $\Gamma, \omega \vdash_{\CC T}^{X}A$.

\item if $A=B\to C$, then suppose that $\Gamma, \omega \vdash_{\CC T}^{X}A$; then for all $\Gamma'\in\Theta$ such that $\Gamma'\supseteq \Gamma$ and $\Gamma', \omega \vdash_{\CC T}^{X}B$, also $\Gamma', \omega \vdash_{\CC T}^{X}C$ holds.

Suppose $\Gamma\cup\{B^{\omega}\}$ is super-consistent.

Furthermore, suppose
 $\TTH(\omega)\cup\Gamma\cup\{ B^{\omega}\} \not \vdash  C^{\omega}$; then, by Lemma \ref{lemma:saturation} (with $\omega_{0}=\omega$ and $A=C$) there exists $\Gamma'\in \Theta$ with $\Gamma'\supseteq\Gamma\cup\{B^{\omega}\}$ such that $C^{\omega}\notin \Gamma'$.
 Using the I.H.~we deduce then that $\Gamma', \omega\vdash_{\CC T}^{X}B$ but $\Gamma', \omega\not\vdash_{\CC T}^{X}C$, against the assumption. We conclude then that 
$\TTH(\omega)\cup\Gamma\cup\{ B^{\omega}\}  \vdash  C^{\omega}$, and thus that
$\Gamma \vdash (B\to C)^{\omega}$, which implies by closure $(B\to C)^{\omega}\in \Gamma$. 

Suppose now that $\Gamma\cup\{B^{\omega}\}$ is not super-consistent. There exists $\omega'$ such that 
$\TTH(\omega') \cup \Gamma\cup\{B^{\omega}\}\vdash \BOT$.
Since $\TTH(\omega') \cup \Gamma$ is consistent the only possibility is that $\TTH(\omega')\vdash \TTH^{|B|}_{\FN(B)}(\omega)$ (which forces $\omega|_{|B|,\FN(B)}$ to coincide with $\omega'|_{|B|,\FN(B)}$) and $\TTH(\omega')\cup\Gamma \cup \{B\} \vdash \BOT$; 
but then we also have that $\TTH(\omega)\cup\Gamma \cup \{B\} \vdash \BOT$, which implies 
$ \TTH(\omega)\cup\Gamma\vdash B\to C$ and thus $\Gamma\vdash (B\to C)^{\omega}$; by closure 
this implies then $(B\to C)^{\omega}\in \Gamma$.

Conversely, suppose $(B\to C)^{\omega}\in\Gamma$ and let $\Gamma'\in \Theta$ be such that $\Gamma'\supseteq \Gamma$ and $\Gamma',\omega \vdash_{\CC T}^{X}B$; by the I.H.~this implies $B^{\omega}\in \Gamma'$, and since 
$(B\to C)^{\omega}\in\Gamma\subseteq \Gamma'$ we also deduce by closure 
$C^{\omega}\in \Gamma'$, which by the I.H.~implies $\Gamma', \omega\vdash_{\CC T}^{X}C$.

\item if $A=\B C^{q}_{a}B$, then suppose that $\Gamma, \omega \vdash_{\CC T}^{X} A$; then by the I.H.~the Borel set 
$$
S=\Big\{ \omega'\mid  B^{\omega+\omega'}\in \Gamma\Big\}
\supseteq
\Big\{ \omega'\mid \Gamma, \omega+\omega' \Vdash^{X\cup\{a\}}_{\CC T} B\Big\}
$$
has measure greater than $q$. Observe that for all $\omega'\in 2^{\mathbb N}$, 
$B^{\omega+\omega'}$ is equivalent to 
$$ \TTH^{|B|_{X}}_{X}(\omega|_{X}) \to \Big(
\TTH^{|B|_{\{a\}}}_{\{a\}}(\omega') \to B\Big) $$
Let $\bthree:= \bigvee\{ \TTH^{|B|_{\{a\}}}_{\{a\}}(\omega')\mid \omega'\in S\}$; since $\Gamma$ is closed, we deduce then that 
$(\bthree\to B)^{\omega}\in \Gamma$ and that $\mu(\bthree)\geq q$; again, by closure, this implies $(\B C^{q}_{a}B)^{\omega}\in \Gamma$.  

For the converse direction, suppose $(\B C^{q}_{a}B)^{\omega}\in \Gamma$ and let $\Gamma'\supseteq \Gamma$; since $\Gamma$ is $\B C$-closed,  
there exists a Boolean formula $\bthree$ with $\mu(\bthree)\geq q$ such that $(\bthree\to B)^{\omega}\in \Gamma\subseteq \Gamma'$. This implies that for all $\omega'\in \model{\bthree}$, $B^{\omega+\omega'}\in \Gamma'$; 
hence, by the I.H.~for all $\omega'\in \model{\bthree}$, $\Gamma', \omega+\omega'\Vdash^{X\cup\{a\}}_{\CC T}B$.
Since $\mu(\bthree)\geq q$, the set of $\omega'$ such that $\Gamma', \omega+\omega'\Vdash^{X\cup\{a\}}_{\CC T}B$ has measure greater than $q$, and we can conclude $\Gamma, \omega \Vdash^{X} _{\CC T}A$.
%
\end{varitemize}
\end{proof}

\subsection{Permutative Rules in $\ICPLL$.}\label{app:permu}

We complete the picture of the normalization rules of $\ICPLL$ by considering permutative rules for $(\mathsf m$), illustrated in Fig.~\ref{fig:normaperma}, where we let $D_{a}^{i}(\bone, \btwo)$ be an abbreviation for
$( \bvar_{a}^{i}\land \bone)\lor (\lnot \bvar_{a}^{i}\land \btwo)$. 

It is easily checked that if $\Pi \redall \Sigma$, $t^{\Pi} \redalll t^{\Sigma}$, since the rules closely correspond to the permutation rules for $\oplus$ in $\EVLL$.

\begin{figure}
\fbox{
\begin{minipage}{0.98\textwidth}
\adjustbox{scale=0.5}{
\begin{minipage}{\textwidth}
\begin{tabular}{c c c}
$
\AXC{$\Pi$}
\noLine
\UIC{$\Gamma \vdash \btwo \Pto A$}
\AXC{$\Pi$}
\noLine
\UIC{$\Gamma \vdash \btwo \Pto A$}
\AXC{$\bone \vDash D_{a}^{i}(\btwo, \btwo)$}
\RL{($\mathsf m$)}
\TIC{$\Gamma \vdash \bone \Pto A$}
\DP
$
&
$\redperm$
& 
$
\AXC{$\Pi[\bone \mapsto \btwo]$}
\noLine
\UIC{$\Gamma \vdash \bone \Pto A$}
\DP
$
\\ & & \\
$
\AXC{$\Pi$}
\noLine
\UIC{$\Gamma \vdash \btwo \Pto A$}
\AXC{$\Sigma$}
\noLine
\UIC{$\Gamma \vdash \bthree \Pto A$}
\AXC{$\bone' \vDash D_{a}^{i}(\btwo,\bthree)$}
\RL{($\mathsf m$)}
\TIC{$\Gamma \vdash \bone' \Pto A$}
\AXC{$\Theta$}
\noLine
\UIC{$\Gamma \vdash \bfour \Pto A$}
\AXC{$\bone \vDash D_{a}^{i}(\bone', \bfour)$}
\RL{($\mathsf m$)}
\TIC{$\Gamma \vdash \bone \Pto A$}
\DP
$
&
$\redperm$
& 
$
\AXC{$\Pi$}
\noLine
\UIC{$\Gamma \vdash \btwo \Pto A$}
\AXC{$\Theta$}
\noLine
\UIC{$\Gamma \vdash \bfour \Pto A$}
\AXC{$\bone \vDash D_{a}^{i}(\btwo, \bfour)$}
\RL{($\mathsf m$)}
\TIC{$\Gamma \vdash \bone \Pto A$}
\DP
$
\\ & & \\
$
\AXC{$\Pi$}
\noLine
\UIC{$\Gamma \vdash \btwo \Pto A$}
\AXC{$\Sigma$}
\noLine
\UIC{$\Gamma \vdash \bthree \Pto A$}
\AXC{$\Theta$}
\noLine
\UIC{$\Gamma \vdash \bfour \Pto A$}
\AXC{$\bone' \vDash D_{a}^{i}(\bthree,\bfour)$}
\RL{($\mathsf m$)}
\TIC{$\Gamma \vdash \bone' \Pto A$}
\AXC{$\bone \vDash D_{a}^{i}(\btwo, \bone')$}
\RL{($\mathsf m$)}
\TIC{$\Gamma \vdash \bone \Pto A$}
\DP
$
&
$\redperm$
& 
$
\AXC{$\Pi$}
\noLine
\UIC{$\Gamma \vdash \btwo \Pto A$}
\AXC{$\Theta$}
\noLine
\UIC{$\Gamma \vdash \bfour \Pto A$}
\AXC{$\bone \vDash D_{a}^{i}(\btwo, \bfour)$}
\RL{($\mathsf m$)}
\TIC{$\Gamma \vdash \bone \Pto A$}
\DP
$
\\ & & \\
$
\AXC{$\Pi$}
\noLine
\UIC{$\Gamma, A \vdash \btwo \Pto B$}
\AXC{$\Sigma$}
\noLine
\UIC{$\Gamma,A \vdash \bthree \Pto B$}
\AXC{$\bone \vDash D_{a}^{i}(\btwo, \bthree)$}
\RL{($\mathsf m$)}
\TIC{$\Gamma,A \vdash \bone\Pto B$}
\RL{($\to$I)}
\UIC{$\Gamma \vdash \bone\Pto (A\to B)$}
\DP
$
& $\redperm$ &
$
\AXC{$\Pi$}
\noLine
\UIC{$\Gamma, A \vdash \btwo \Pto B$}
\RL{($\to$I)}
\UIC{$\Gamma \vdash \btwo\Pto (A\to B)$}
\AXC{$\Sigma$}
\noLine
\UIC{$\Gamma,A \vdash \bthree \Pto B$}
\RL{($\to$I)}
\UIC{$\Gamma \vdash \bthree\Pto (A\to B)$}
\AXC{$\bone \vDash D_{a}^{i}(\btwo, \bthree)$}
\RL{($\mathsf m$)}
\TIC{$\Gamma \vdash \bone\Pto (A\to B)$}
\DP
$
\\ & & \\
$
\AXC{$\Pi$}
\noLine
\UIC{$\Gamma \vdash \btwo \Pto(A\to B)$}
\AXC{$\Sigma$}
\noLine
\UIC{$\Gamma \vdash \bthree \Pto (A\to B)$}
\AXC{$\bone \vDash D_{a}^{i}(\btwo, \bthree)$}
\RL{($\mathsf m$)}
\TIC{$\Gamma \vdash \bone\Pto (A\to B)$}
\AXC{$\Theta$}
\noLine
\UIC{$\Gamma \vdash \bone \Pto A$}
\RL{($\to$E)}
\BIC{$\Gamma \vdash \bone\Pto B$}
\DP
$
& $\redperm$ &
$
\AXC{$\Pi[\btwo\mapsto \bone\land \btwo]$}
\noLine
\UIC{$\Gamma \vdash \bone\land \btwo \Pto(A\to B)$}
\AXC{$\Theta[\bone\mapsto \bone\land \btwo]$}
\noLine
\UIC{$\Gamma \vdash \bone\land \btwo \Pto A$}
\RL{($\to$E)}
\BIC{$\Gamma \vdash \bone\land \btwo\Pto B$}
\AXC{$\Sigma[\btwo\mapsto \bone\land \bthree]$}
\noLine
\UIC{$\Gamma \vdash \bone \land\bthree \Pto(A\to B)$}
\AXC{$\Theta[\bone\mapsto \bone\land \bthree]$}
\noLine
\UIC{$\Gamma \vdash \bone\land \bthree \Pto A$}
\RL{($\to$E)}
\BIC{$\Gamma \vdash \bone\land \bthree\Pto B$}
\AXC{$\bone \vDash D_{a}^{i}(\bone\land\btwo, \bone\land\bthree)$}
\RL{($\mathsf m$)}
\TIC{$\Gamma \vdash \bone\Pto  B$}
\DP
$
\\ & & \\
$
\AXC{$\Pi$}
\noLine
\UIC{$\Gamma \vdash \bone \Pto(A\to B)$}
\AXC{$\Sigma$}
\noLine
\UIC{$\Gamma \vdash \btwo \Pto A$}
\AXC{$\Theta$}
\noLine
\UIC{$\Gamma \vdash \bthree \Pto A$}
\AXC{$\bone \vDash D_{a}^{i}(\btwo, \bthree)$}
\RL{($\mathsf m$)}
\TIC{$\Gamma \vdash \bone\Pto A$}
\RL{($\to$E)}
\BIC{$\Gamma \vdash \bone\Pto B$}
\DP
$
& $\redperm$ &
$
\AXC{$\Pi[\bone\mapsto \bone\land \btwo]$}
\noLine
\UIC{$\Gamma \vdash \bone\land \btwo \Pto(A\to B)$}
\AXC{$\Sigma[\btwo\mapsto \bone\land \btwo]$}
\noLine
\UIC{$\Gamma \vdash \bone \land\btwo \Pto A$}
\RL{($\to$E)}
\BIC{$\Gamma \vdash \bone\land \btwo\Pto B$}
\AXC{$\Pi[\bone\mapsto \bone\land \bthree]$}
\noLine
\UIC{$\Gamma \vdash \bone\land \bthree \Pto(A\to B)$}
\AXC{$\Theta[\bthree \mapsto \bone \land \bthree]$}
\noLine
\UIC{$\Gamma \vdash\bone \land \bthree \Pto A$}
\RL{($\to$E)}
\BIC{$\Gamma \vdash \bone\land \bthree\Pto B$}
\AXC{$\bone \vDash D_{a}^{i}(\bone\land\btwo, \bone\land\bthree)$}
\RL{($\mathsf m$)}
\TIC{$\Gamma \vdash \bone\Pto B$}
\DP
$
\\ & & \\
\AXC{$\Pi$}
\noLine
\UIC{$\Gamma \vdash \btwo\land \bfour_{1}\Pto A$}
\AXC{$\Sigma$}
\noLine
\UIC{$\Gamma \vdash  \bthree \land \bfour_{2}\Pto A$}
\AXC{$\bone\land \bfour\vDash D_{a}^{i}(\btwo\land \bfour_{1}, \bthree\land \bfour_{2})$}
\RL{($\mathsf m$)}
\TIC{$\Gamma \vdash \bone \land \bfour \Pto A$}
\AXC{$\mu(\bfour)\geq q$}
\RL{($\B C$I)}
\BIC{$\Gamma \vdash \bone \Pto (\B C^{q}_{b}A)$}
\DP
& $\redperm$&
 $
 \AXC{$\Pi[\btwo\land \bfour_{1} \mapsto (\bone \land \btwo) \land \bfour]$}
\noLine
\UIC{$\Gamma \vdash (\bone \land \btwo)\land \bfour\Pto A$}
\AXC{$\mu(\bfour)\geq q$}
\RL{($\B C$I)}
\BIC{$\Gamma \vdash \bone\land \btwo \Pto (\B C^{q}_{b}A)$}
 \AXC{$\Sigma[\bthree\land \bfour_{2} \mapsto (\bone \land \bthree) \land \bfour]$}
\noLine
\UIC{$\Gamma \vdash (\bone \land \bthree)\land \bfour\Pto A$}
\AXC{$\mu(\bfour)\geq q$}
\RL{($\B C$I)}
\BIC{$\Gamma \vdash \bone\land \bthree \Pto (\B C^{q}_{b}A)$}
\AXC{$\bone \vDash D_{a}^{i}(\bone \land \btwo, \bone \land \bthree)$}
\RL{($\mathsf m$)}
\TIC{$\Gamma \vdash \bone  \Pto (\B C^{q}_{b}A)$}
\DP
 $
 \\ & & \\
 $
 \AXC{$\Pi$}
 \noLine
 \UIC{$\Gamma \vdash \btwo \Pto (\B C^{q}_{b}A)$}
  \AXC{$\Sigma$}
 \noLine
 \UIC{$\Gamma \vdash \bthree \Pto (\B C^{q}_{b}A)$}
 \AXC{$\bone \vDash D_{a}^{i}(\btwo,\bthree)$}
 \RL{($\mathsf m$)}
 \TIC{$\Gamma\vdash \bone \Pto (\B C^{q}_{b}A)$}
 \AXC{$\Theta$}
 \noLine
 \UIC{$\Gamma, A\vdash \bone \Pto B$}
 \RL{($\B C$E$_{2}$)}
 \BIC{$\Gamma \vdash \bone \Pto (\B C^{qs}_{b}B)$}
 \DP$
 & $\redall$ & 
 $
  \AXC{$\Pi[\btwo \mapsto \bone\land \btwo]$}
 \noLine
 \UIC{$\Gamma \vdash \bone\land \btwo \Pto (\B C^{q}_{b}A)$}
  \AXC{$\Theta[\bone\mapsto \bone \land \btwo]$}
 \noLine
 \UIC{$\Gamma, A\vdash \bone\land \btwo \Pto B$}
 \RL{($\B C$E$_{2}$)}
 \BIC{$\Gamma \vdash \bone\land \btwo \Pto (\B C^{qs}_{b}B)$}
  \AXC{$\Sigma[\bthree \mapsto \bone\land \bthree]$}
 \noLine
 \UIC{$\Gamma \vdash \bone\land \bthree \Pto (\B C^{q}_{b}A)$}
  \AXC{$\Theta[\bone\mapsto \bone \land \bthree]$}
 \noLine
 \UIC{$\Gamma, A\vdash \bone\land \bthree \Pto B$}
 \RL{($\B C$E$_{2}$)}
 \BIC{$\Gamma \vdash \bone\land \bthree \Pto (\B C^{qs}_{b}B)$}
 \AXC{$\bone \vDash D_{a}^{i}(\bone\land\btwo,\bone\land \bthree)$}
 \RL{($\mathsf m$)}
 \TIC{$\Gamma\vdash \bone \Pto (\B C^{qs}_{b}B)$}
 \DP
$ \\ & & \\
$
 \AXC{$\Pi$}
 \noLine
 \UIC{$\Gamma \vdash \bone \Pto (\B C^{q}_{b}A)$}
 \AXC{$\Sigma$}
 \noLine
 \UIC{$\Gamma, A\vdash \btwo \Pto B$}
   \AXC{$\Theta$}
 \noLine
 \UIC{$\Gamma,A \vdash \bthree \Pto B$}
 \AXC{$\bone \vDash D_{a}^{i}(\btwo,\bthree)$}
 \RL{($\mathsf m$)}
 \TIC{$\Gamma,A\vdash \bone \Pto B$}
 \RL{($\B C$E$_{2}$)}
 \BIC{$\Gamma \vdash \bone \Pto (\B C^{qs}_{b}B)$}
 \DP$
 & $\redall$ & 
 $
  \AXC{$\Pi[\bone \mapsto \bone\land \btwo]$}
 \noLine
 \UIC{$\Gamma \vdash \bone\land \btwo \Pto (\B C^{q}_{b}A)$}
  \AXC{$\Sigma[\btwo \mapsto \bone \land \btwo]$}
 \noLine
 \UIC{$\Gamma, A\vdash\bone \land  \btwo \Pto B$}
 \RL{($\B C$E$_{2}$)}
 \BIC{$\Gamma \vdash \bone\land \btwo \Pto (\B C^{qs}_{b}B)$}
  \AXC{$\Pi[\bone \mapsto \bone\land \bthree]$}
 \noLine
 \UIC{$\Gamma \vdash \bone\land \bthree \Pto (\B C^{q}_{b}A)$}
  \AXC{$\Theta[\bthree \mapsto \bone \land \bthree]$}
 \noLine
 \UIC{$\Gamma, A\vdash\bone \land  \bthree \Pto B$}
 \RL{($\B C$E$_{2}$)}
 \BIC{$\Gamma \vdash \bone\land \bthree \Pto (\B C^{qs}_{b}B)$}
  \AXC{$\bone \vDash D_{a}^{i}(\btwo,\bthree)$}
 \RL{($\mathsf m$)}
 \TIC{$\Gamma,A\vdash \bone \Pto (\B C^{qs}_{b}B)$}
 \DP
$
\end{tabular}
\end{minipage}
}
\end{minipage}
}
\caption{Permutative rules of $\NDCPLL$.}
\label{fig:normaperma}
\end{figure}

\subsection{A ``CbN'' Proof-System.}\label{app:CbNProofSystem}

The CHC described in Section \ref{section5} relates the proof-system $\NDCPLL$ with the type system $\LCPL$. 
In this subsection we describe a ``CbN'' variant $\NDCPLLV$ of the proof-system $\NDCPLL$, for which it is possible to describe a CHC with the type system $\TCI$.

In the correspondence from Section \ref{section5}, CbV application $\{t\}u$ plays a fundamental role, as it translates the elimination rule of the counting quantifier. To obtain a translation into $\EVL$ we thus need to restrict the rule ($\B C$E). 

A sequent of $\NDCPLLV$ is of the form $\Phi;\Gamma \vdash \bone \Pto A$, where $\Phi$ and $\Gamma$ are two sets of formulas, and $\Phi$ contains \emph{at most} one formula. 
The fundamental intuition is that the formula, if any, in $\Phi$, has to be used \emph{linearly} in the proof.

The rules of $\NDCPLLV$ are illustrated in Fig.~\ref{fig:rulescbn}. 
Observe that, contrarily to $\NDCPLL$ (see Section \ref{section3}), the rules include a ``multiplication rule'' ($\B C\times$) to passes from $\B C^{q}\B C^{s}A$ to $\B C^{qs}A$.

$\NDCPLLV$ proves \emph{less} formulas than $\NDCPLL$. Indeed, the restricted ($\B C$E)-rule allows to deduce
$\B C^{q}B$ from $\B C^{q}A$ only when $B$ can be deduced \emph{linearly} from $A$. For example, one cannot reproduce in $\NDCPLLV$ the proof of $\B C^{q}(A\to A)\to (A\to \B C^{q}A)$ illustrated in Section \ref{section3}, since the hypothesis
$\B C^{q}(A\to A)$, which is used as major premiss in a ($\B C$E)-rule, should be used \emph{twice}.
From the programming perspective, this means that one cannot encode in $\NDCPLLV$ the non-linear ``CbV'' function
$\lambda y.\lambda x.\{ \lambda f.f(fx)\}y$.
For similar reasons, it seems that one cannot prove $\B C^{q}A\to (A \to A \to B) \to \B C^{q}B$ in $\NDCPLLV$, while this can be proved in $\NDCPLL$ as shown in Fig.~\ref{fig:proofnew}. From the programming perspective, this means that one cannot encode the 
non-linear ``CbV'' function $\lambda x. \lambda y.\{\lambda y. yxx\}x$.

\begin{figure}
\fbox{
\begin{minipage}{0.98\textwidth}
\adjustbox{scale=0.9}{
\begin{minipage}{\textwidth}
\begin{center}
$
\AXC{$\B C^{q}A, A\to A\to A \vdash \B C^{q}A$}
\AXC{$A, A\to A\to A\vdash A\to A\to B$}
\AXC{$A, A\to A\to A\vdash A$}
\RL{($\to $E)}
\BIC{$A, A\to A\to A\vdash A\to B$}
\AXC{$A, A\to A\to A\vdash A$}
\RL{($\to $E)}
\BIC{$A, A\to A\to A\vdash B$}
\RL{($\B C$E)}
\BIC{$\B C^{q}A, A\to A\to A\vdash \B C^{q}B$}
\doubleLine
\RL{($\to$I)}
\UIC{$\vdash \B C^{q}A\to (A\to A\to A)\to \B C^{q}B$}
\DP
$
\end{center}
\end{minipage}
}
\end{minipage}
}
\caption{Proof in $\NDCPL$ of $\B C^{q}A\to (A \to A \to A) \to \B C^{q}A$.}
\label{fig:proofnew}
\end{figure}

The normalization steps are as in $\NDCPLL$, except for ($\B C$I/$\B C$E), which is illustrated in Fig.~\ref{fig:celimcbn} and exploits the admissibility of the following rule:
\begin{align*}
\AXC{$; \Gamma \vdash \bone \Pto  \B C^{s_{1}*\dots *s_{n}} A$}
\AXC{$A; \Gamma\vdash \bone \Pto C$}
\RL{(subst$^{*}$)}
\BIC{$; \Gamma \vdash \bone \Pto  \B C^{s_{1}*\dots *s_{n}}C$}
\DP 
\end{align*}

Notice that the ($\B C$I/$\B C$E)-step now includes a finite number of internal ``multiplications'' ($\B C\times$).

\begin{figure}[t]
\fbox{
\begin{minipage}{0.98\textwidth}
\adjustbox{scale=0.85}{
\begin{minipage}{\textwidth}
\begin{center}
Identity Rules 
$$
\AXC{}
\RL{(id)}
\UIC{$; \Gamma,A \vdash \bone \Pto A$}
\DP
\qquad\qquad
\AXC{}
\RL{(id$_{\mathrm{lin}}$)}
\UIC{$ A;\Gamma  \vdash \bone \Pto A$}
\DP
$$

\bigskip

Structural Rules
$$
\AXC{$\bone \vDash \BOT$}
\RL{($\BOT$)}
\UIC{$;\Gamma \vdash \bone \Pto A$}
\DP
$$

$$
\AXC{$;\Gamma \vdash \btwo \Pto A$}
\AXC{$;\Gamma \vdash \bthree \Pto A$}
\AXC{$\bone \vDash (\btwo \land \bvar_{a}^{i})\lor (\bthree \land \lnot \bvar_{a}^{i})$}
\RL{($\mathsf m$)}
\TIC{$;\Gamma \vdash \bone \Pto A$}
\DP
$$

\bigskip

Logical Rules
$$
\AXC{$\Phi; \Gamma, A \vdash \bone \Pto B$}
\RL{($\to$I)}
\UIC{$\Phi;\Gamma \vdash \bone \Pto (A\to B)$}
\DP
\qquad\qquad 
\AXC{$\Phi; \Gamma \vdash \bone \Pto (A\to B)$}
\AXC{$; \Gamma \vdash \bone \Pto A$}
\RL{($\to$E)}
\BIC{$\Phi;\Gamma \vdash \bone \Pto B$}
\DP
$$

\bigskip

Counting Rules
$$
\AXC{$;\Gamma \vdash \bone \land \bthree\Pto A$}
\AXC{$\mu(\bthree)\geq q$}
\RL{($\BOX$I)}
\BIC{$;\Gamma \vdash \bone \Pto \BOX^{q}A$}
\DP
\qquad\qquad
\AXC{$\Phi;\Gamma \vdash \bone \Pto \BOX^{q}A$}
\AXC{$ A;\Gamma \vdash \bone \Pto B$}
\RL{($\BOX$E)}
\BIC{$\Phi;\Gamma \vdash \bone \Pto \BOX^{q}B$}
\DP
$$

$$
\AXC{$ \Phi;\Gamma \vdash \bone \Pto \BOX^{q}\BOX^{s}A$}
\RL{($\BOX_{\times}$)}
\UIC{$\Phi;\Gamma \vdash \bone \Pto \BOX^{qs}A$}
\DP
$$
%

\end{center}
\end{minipage}
}
\end{minipage}
}
\caption{Rules of $\NDCPLLV$.}
\label{fig:rulescbn}
\end{figure}

\begin{figure*}[t]
%
%
\fbox{
\begin{subfigure}{\textwidth}
\begin{minipage}{0.98\textwidth}
\adjustbox{scale=0.75}{
\begin{tabular}{c  c  c}
\AXC{$\Sigma$}
\noLine
\UIC{$;\Gamma \vdash \bone \land \bthree  \Pto \B C^{s_{1}*\dots *s_{n}}A$}
\AXC{$\mu(\bthree  )\geq q$}
\RL{($\B C$I)}
\BIC{$;\Gamma \vdash \bone \Pto \B C^{q*s_{1}*\dots*s_{n}}A$}
\doubleLine
\RL{($\B C\times$)}
\UIC{$;\Gamma \vdash \bone \Pto \B C^{q\prod_{i}s_{i}}A$}
\AXC{$\Pi$}
\noLine
\UIC{$A;\Gamma\vdash \bone \Pto B$}
\RL{($\B C$E)}
\BIC{$;\Gamma \vdash \bone \Pto \B C^{q\prod_{i}s_{i}}B$}
\DP
&
$\leadsto$
&
\AXC{$\Sigma$}
\noLine
\UIC{$;\Gamma \vdash \bone \land \bthree  \Pto  \B C^{s_{1}*\dots *s_{n}} A$}
\AXC{$\Pi${\small$[\bone\mapsto \bone\land \bthree]$}}
\noLine
\UIC{$A;\Gamma \vdash \bone\land \bthree  \Pto B$}
\RL{(subst$^{*}$)}
\BIC{$;\Gamma \vdash \bone \land \bthree  \Pto  \B C^{s_{1}*\dots *s_{n}} B$}
\AXC{$\mu(\bthree)\geq q$}
\RL{($\B C$I)}
\BIC{$;\Gamma \vdash \bone\Pto  \B C^{q*s_{1}*\dots *s_{n}}B$}
\doubleLine
\RL{($\B C\times$)}
\UIC{$;\Gamma \vdash \bone \Pto \B C^{q\prod_{i}s_{i}}B$}
\DP
\end{tabular}
}
\end{minipage}
\caption{($\B C$I)/($\B C$E).}
\label{fig:cutelim2n}
\end{subfigure}
}
\caption{Normalization steps of $\NDCPLLV$.}
\label{fig:celimcbn}
\end{figure*}

The translation from $\NDCPLLV$ to $\TCI$ relies on properties of
\emph{head contexts} in $\EVL$. These are defined by the grammar below:
$$
\TT H[\ ]::= [\ ]\mid \lambda x.\TT H[\ ]\mid  \TT H[\ ]u 
$$

The fundamental property of head contexts is that they naturally behave as CbV functions, due to the following lemma:

\begin{lemma}\label{lemma:headcontext}
For any head context $\TT H[\ ]$ and term $t$, 
$
\TT H[\nu a.t] \ \redperm^{*} \ \nu a. \TT H[t]
$.
\end{lemma}
\begin{proof}
By induction on $\TT H[\ ]$:
\begin{varitemize}
\item if $\TT H[\ ]=[\ ]$, the claim is immediate;
\item if $\TT H[\ ]=\lambda x.\TT H'[\ ]$, then we have 
$\TT H[\nu a.t]= \lambda x.\TT H'[\nu a.t] \stackrel{\text{IH}}{\redperm^{*}} \lambda x.\nu a.\TT H'[t]\redperm \nu a.\lambda x.\TT H'[t]=\nu a.\TT H[t]$;

\item if $\TT H[\ ]=\TT H'[\ ]u$, then we have 
$\TT H[\nu a.t]= \TT H'[\nu a.t]u \stackrel{\text{IH}}{\redperm^{*}} (\nu a.\TT H'[t])u\redperm \nu a.\TT H'[t]u=\nu a.\TT H[t]$.

\end{varitemize}
\end{proof}

In other words, whenever $t$ is a function of the form $\lambda x.\TT H[x]$ for some head context $\TT H[\ ]$, CbN and CbV application of $t$ coincide, since
$t (\nu a.v)$ and $\nu a.tv$ have the same normal form. 

This property is reflected in the following:
%
%
\begin{lemma}\label{lemma:headsubst}
For any head context $\TT H[\ ]$, the following rule is derivable in $\TCI$:
$$
\AXC{$ \Gamma \vdash t: \bone \Pto \B C^{qs}\sigma$}
\AXC{$x: \B C^{q}\sigma, \Gamma \vdash \TT H[x]: \bone\Pto \B C^{r}\tau$}
\RL{(head-subst)}
\BIC{$\Gamma \vdash \TT H[t]: \bone \Pto\B C^{rs}\tau$}
\DP
$$
\end{lemma}
\begin{proof}
By induction on $\TT H[\ ]$:
\begin{varitemize}
\item if $\TT H[\ ]=[\ ]$, then it must be $q=r$, and the claim is then immediate;
\item if $\TT H[\ ]= \lambda y.\TT H'[\ ]$, then $\tau= \FF t\To \tau'$ and we must have 
$x:\B C^{q}\sigma, \Gamma,  y: \FF t \vdash \TT H'[x]:\bone \Pto \B C^{q}\tau'$. Then by IH we deduce
$\Gamma, y:\FF t\vdash \TT H'[t]: \bone \Pto \B C^{rs}\tau'$, and finally
$\Gamma\vdash \TT H[t]: \bone \Pto \B C^{rs}\tau$.

\item if $\TT H[\ ]= \TT H'[\ ]u$, then we must have 
$x:\B C^{q}\sigma, \Gamma \vdash \TT H'[x]:\bone \Pto \B C^{q}(\FF t\To \tau)$ and 
$x:\B C^{q}\sigma, \Gamma \vdash u:\bone \Pto \FF t$. Then by IH we deduce
$\Gamma\vdash \TT H'[t]: \bone \Pto \B C^{rs}(\FF t\To \tau)$, and we can thus conclude
$\Gamma\vdash \TT H[t]: \bone \Pto \B C^{rs}\tau$.
\end{varitemize}
\end{proof}

To any formula $A$ of $\ICPLL$, we associate a non-quantified type $\sigma_{A}$ and a positive real $|A|\in(0,1]\cap \BB Q$ of $\TCI$ as follows:
\begin{align*}
\sigma_{\CC p}& :=  o &   |o|:= 1 \\
\sigma_{A\to B}&:= (\B C^{|A|}\sigma_{A})\To \sigma_{B}  & |A\to B|= |B| \\
\sigma_{\B C^{q}A}&:= \sigma_{A} & |\B C^{q}A|= q\cdot |A|
\end{align*}
We let then $\FF s_{A}:= \B C^{|A|}\sigma_{A}$ (observe that $\FF s_{\B C^{q}A}= \B C^{q\cdot |A|}\sigma_{A}$).

The translation of a derivation $\Pi$ of $\Phi; \Gamma \vdash \bone \Pto A$ into a typing derivation 
$D^{\Pi}$ of $\FF s_{\Phi}, \FF s_{\Gamma}\vdash t^{\Pi}: \bone \Pto \FF s_{A}$ in $\TCI$ is illustrated in Fig.~\ref{fig:chcn}, where we exploit the fact that if $\Phi=\{A\}$ is non-empty, 
then $t^{\Pi}=t^{\Pi}[x: \FF s_{A}]$ is a head context (as it can be checked by induction on the construction).
We omit the case of the rule ($\B C_{\times}$), as it follows immediately from the induction hypothesis, since $\FF s_{\B C^{q}\B C^{s}A}=\FF s_{\B C^{qs}A}$.

The stability of the translation under normalization is easily checked using Lemma \ref{lemma:headcontext}.

\begin{figure*}
\fbox{
\begin{minipage}{0.98\textwidth}
\adjustbox{scale=0.75}{
\begin{minipage}{\textwidth}
\begin{center}
\begin{tabular}{c c c}
$
\AXC{}
\RL{(id)}
\UIC{$;\Gamma, A \vdash \bone\Pto A$}
\DP$
 &
 $\mapsto$
  & 
 $
 \AXC{}
 \RL{(id)}
 \UIC{$\FF s_{\Gamma}, y: \FF s_{A} \vdash y:\bone\Pto \FF s_{A}$}
 \DP
 $
  \\ & & \\
  $
  \AXC{}
\RL{(id$_{\mathrm{lin}}$)}
\UIC{$A;\Gamma  \vdash \bone\Pto A$}
\DP$
 &
 $\mapsto$
  & 
 $
 \AXC{}
 \RL{(id)}
 \UIC{$x:\FF s_{A}, \FF s_{\Gamma} \vdash x:\bone\Pto \FF s_{A}$}
 \DP
 $
  \\ & & \\
  $
\AXC{$\bone \vDash \BOT$}
\RL{($\BOT$)}
\UIC{$\ ;\Gamma \vdash \BOT\Pto A$}
\DP$
 &
 $\mapsto$
  & 
 $
 \AXC{$\bone \vDash \BOT$}
 \RL{($\vee$)}
 \UIC{$\FF s_{\Gamma} \vdash \B c:\bone\Pto \FF s_{A}$}
 \DP
 $
  \\ & & \\
\AXC{$\Pi$}
\noLine
\UIC{$; \Gamma \vdash \btwo \Pto A$}
\AXC{$\bone \vDash \btwo$}
\RL{($\vDash$)}
\BIC{$;  \Gamma \vdash \bone \Pto A$}
\DP
& 
$\mapsto $
&
\AXC{$D^{\Pi}$}
\noLine
\UIC{$\FF s_{\Gamma}\vdash t^{\Pi}: \btwo \Pto \FF s_{A}$}
\AXC{$\bone \vDash \btwo$}
\RL{($\vDash$)}
\BIC{$ \FF s_{\Gamma}\vdash t^{\Pi}: \bone \Pto \FF s_{A}$}
\DP
  \\ & & \\
\AXC{$\Pi$}
\noLine
\UIC{$;\Gamma \vdash \btwo \Pto A$}
\AXC{$\Pi'$}
\noLine
\UIC{$;\Gamma \vdash \bthree \Pto A$}
\AXC{$\bone \vDash (\btwo \land \bvar_{a}^{i})\lor (\bthree \land \lnot \bvar_{a}^{i})$}
\RL{($\mathsf m$)}
\TIC{$; \Gamma \vdash \bone \Pto A$}
\DP
& 
$\mapsto $
&
\AXC{$D^{\Pi}$}
\noLine
\UIC{$ \FF s_{\Gamma}\vdash t^{\Pi}: \btwo \Pto A$}
\AXC{$D^{\Pi'}$}
\noLine
\UIC{$\FF s_{\Gamma}\vdash t^{\Pi'}: \bthree \Pto A$}
\AXC{$\bone \vDash (\btwo \land \bvar_{a}^{i})\lor (\bthree \land \lnot \bvar_{a}^{i})$}
\RL{($\oplus$)}
\TIC{$ \FF s_{\Gamma}\vdash t^{\Pi}\oplus_{a}^{i}t^{\Pi'}: \bone \Pto \FF s_{A}$}
\DP
\\ & & \\
\AXC{$\Pi$}
\noLine
\UIC{$\Phi;\Gamma,A \vdash \btwo \Pto B$}
\RL{($\to$I)}
\UIC{$\Phi; \Gamma \vdash \bone \Pto (A\to B)$}
\DP
& 
$\mapsto $
&
\AXC{$D^{\Pi}$}
\noLine
\UIC{$\FF s_{\Phi}, \FF s_{\Gamma},y:\FF s_{A}\vdash t^{\Pi}: \btwo \Pto \B C^{|B|}\sigma_{B}$}
\RL{($\lambda$)}
\UIC{$\FF s_{\Phi},\FF s_{\Gamma}\vdash \lambda y.t^{\Pi}: \bone \Pto \B C^{|B|}(\FF s_{A}\To \sigma_{B})$}
\DP
\\ & & \\
\AXC{$\Pi$}
\noLine
\UIC{$\Phi;\Gamma\vdash \btwo \Pto (A\to B)$}
\AXC{$\Sigma$}
\noLine
\UIC{$\Phi;\Gamma\vdash \bone \Pto A$}
\RL{($\to$E)}
\BIC{$\Phi; \Gamma \vdash \bone \Pto  B$}
\DP
& 
$\mapsto $
&
\AXC{$D^{\Pi}$}
\noLine
\UIC{$\FF s_{\Phi};\FF s_{\Gamma}\vdash t^{\Pi}: \bone \Pto\B C^{|B|} (\FF s_{A}\to \sigma_{B})$}
\AXC{$D^{\Sigma}$}
\noLine
\UIC{$\FF s_{\Phi};\FF s_{\Gamma}\vdash t^{\Sigma}: \bone \Pto \FF s_{A}$}
\RL{($@$)}
\BIC{$\FF s_{\Phi}; \FF s_{\Gamma}\vdash t^{\Pi}t^{\Sigma}: \bone \Pto \B C^{|B|} \sigma_{B}$}
\DP
\\ & & \\
\AXC{$\Pi$}
\noLine
\UIC{$;\Gamma\vdash \bone\land \bthree \Pto A$}
\AXC{$\mu(\bthree)\geq q$}
\RL{($\B C$I)}
\BIC{$;\Gamma \vdash \bone \Pto \B C^{q}A$}
\DP
& 
$\mapsto $
&
\AXC{$D^{\Pi}$}
\noLine
\UIC{$ \FF s_{\Gamma}\vdash t^{\Pi}: \bone \land \bthree \Pto \FF s_{A}$}
\AXC{$\mu(\bthree)\geq q$}
\RL{($\mu$)}
\BIC{$\FF s_{\Gamma}\vdash \nu a.t^{\Pi}: \bone \Pto \FF s_{\B C^{q}A}$}
\DP
\\ & & \\
\AXC{$\Pi$}
\noLine
\UIC{$\Phi;\Gamma\vdash \bone \Pto \B C^{q}A$}
\AXC{$\Sigma$}
\noLine
\UIC{$A;\Gamma\vdash \bone \Pto B$}
\RL{($\B C$E)}
\BIC{$\Phi;\Gamma \vdash \bone \Pto  \B C^{q}B$}
\DP
& 
$\mapsto $
&
\AXC{$D^{\Pi}$}
\noLine
\UIC{$\FF s_{\Phi}, \FF s_{\Gamma}\vdash t^{\Pi}: \bone \Pto \B C^{q\cdot |A|} \sigma_{A}$}
\AXC{$D^{\Sigma}$}
\noLine
\UIC{$x: \B C^{|A|}\sigma_{A},\FF s_{\Gamma}\vdash t^{\Sigma}[x]: \bone \Pto  \B C^{ |B|}\sigma_{B}$}
\RL{(head-subst)}
\BIC{$\FF s_{\Phi}, \FF s_{\Gamma}\vdash  t^{\Sigma}[t^{\Pi}]: \bone \Pto  \B C^{q\cdot|B|}\sigma_{B}$}
\DP
\end{tabular}
\end{center}
\end{minipage}
}
\end{minipage}
}
\caption{Translation $\Pi\mapsto D^{\Pi}$ from $\NDCPLLV$ to $\TCI$.}
\label{fig:chcn}
\end{figure*}

\section{Details about the Probabilistic Event $\lambda$-Calculus.}\label{app:PElambda}

\subsection{The Language $\EVL$.}\label{subs:evl}

Let us first comment on how Theorem \ref{thm:confluence} (i.e.~confluence and strong normalization of $\redperm$, and confluence of $\redall$) follows from 
confluence of $\redall$ in the calculus from \cite{DLGH}.
The $\lambda$-calculus in \cite{DLGH} slightly differs from our presentation of $\EVL$, since choice operators do not depend on indexes $i \in \mathbb N$. However, if $\varphi$ is any bijection from $\mathbb N^{2}$ to $\mathbb N$, 
one can define an invertible embedding $t\mapsto t^{\varphi}$ from $\EVL$ to the calculus in \cite{DLGH} by replacing 
$t\oplus^{i}_{a}u$ with $t^{\varphi} \stackrel{\varphi(a,i)}{\oplus} u^{\varphi}$ and 
$(\nu a.t)^{\varphi}= \nu \varphi(a,0).\dots.\nu \varphi(a,\mathsf{ord}_{a}(t)).t^{\varphi}$, where $\mathsf{ord}_{a}(t)$ is the 
maximum $i$ s.t. $\oplus_{a}^{i}$ occurs in $t$.
Since the permutative rules in \cite{DLGH} translate into those in Fig.~\ref{fig:permutations} under this translation, the results from \cite{DLGH} can be transported to our language.

We now discuss PNFs in $\EVL$.
\begin{definition}\label{def:satree}
Let $S$ be a set of name-closed $\lambda$-terms and $a$ be a name.
For all $i\in \mathbb N\cup\{-1\}$, the set of \emph{$(S,a)$-trees of level $i$} is defined as follows:
\begin{varitemize}
\item any $t\in S$ is a $(S,a)$-tree of level $-1$;
\item if $t,u$ are $(S,a)$-trees of level $j$ and $k$, respectively, and $j,k<i$, then $t\oplus_{a}^{i}u$ is a $(S,a)$-tree of level $i$.
\end{varitemize}
%
The \emph{support} of a $(S,a)$-tree $t$, indicated as $\mathrm{Supp}(t)$, is the finite set of terms in $S$ which are leaves of $t$.
\end{definition}

%

\begin{definition}\label{def:pnf}
The sets $\mathcal T$ and $\mathcal V$ of name-closed terms are defined inductively as follows:
\begin{varitemize}
\item all variable $x\in \mathcal V$;

\item if $t\in \mathcal V$, $\lambda x.t\in \mathcal V$;
\item if $t\in \mathcal V$ and $u\in \mathcal T$, then $tu\in \mathcal V$;

\item if $t\in \mathcal V$, then $t\in \mathcal T$;


\item if $ t$ is a $(\mathcal T, a)$-tree, then $\nu a. t\in \mathcal T$.

\end{varitemize}
\end{definition}

\begin{lemma}\label{lemma:enne}
For all name-closed term $t\in \mathcal T$, $t\in \mathcal V$ iff it does not start with $\nu$.
\end{lemma}
\begin{proof}
First observe that if $t=\nu a.b$, then $t\notin \mathcal V$. For the converse direction, we argue by induction on $t\neq \nu a.b$:
	\begin{varitemize}
	\item if $t=x$ then $t\in \mathcal V$;
	\item if $t= \lambda x.u$ then by IH $u\in \mathcal V$, so $t\in \mathcal V$;
	\item if $t=uv$, then the only possibility is that $t\in \mathcal V$;
	\item if $t=u\oplus_{a}^{i}v$, then $t$ is not name-closed, again the hypothesis.
	\end{varitemize}
\end{proof}


For all term $t$ and $a\in \FN(t)$, we let $I_{a}(t)$ be the maximum index $i$ such that $\oplus_{a}^{i}$ occurs in $t$, and $I_{a}(t)=-1$ if $\oplus_{a}^{i}$ does not occur in $t$ for all index $i$.

We now show that Definition \ref{def:pnf} precisely captures PNF.
\begin{lemma}\label{lemma:PNF}
A name-closed term $t$ is in PNF iff  $t\in \mathcal T$.

\end{lemma}
\begin{proof}
\begin{description}
\item[($\To$)]
We argue by induction on $t$. 
If $t$ has no bound name, then $t$ is obviously in $\mathcal T$. Otherwise:
\begin{varitemize}
\item if $t= \lambda x.u$, then $u$ is also name-closed. Hence, by IH $u\in \mathcal T$; observe that $u$ cannot start  with $\nu$ (as $t$ would not be normal) so by Lemma \ref{lemma:enne}, $u\in \mathcal V$, which implies $t\in \mathcal T$; 


\item if $t=uv$, then $u$ and $v$ are both name-closed, and so by induction $u,v\in \mathcal T$; if $u$ started with $\nu$, then $t$ would not be normal, hence, by Lemma \ref{lemma:enne}, $u\in \mathcal V$, and we conclude then that  $t\in \mathcal T$.

\item if $t=u\oplus_{a}^{i}v$, then it cannot be name-closed;

\item if $t= \nu a.u$, then we show, by a sub-induction on $u$, that $u$ is a $(\mathcal T, a)$-tree of level $I_{a}(u)$: first note that $I_{a}(u)$ cannot be $-1$, since otherwise $\nu a.u \redperm u$, so $u$ would not be normal. We now consider all possible cases for $u$:

		\begin{varitemize}
		\item $u$ cannot be a variable, or $I_{a}(u)$ would be $-1$;
		\item If $u=\lambda x.v$, then $I_{a}(v)=I_{a}(u)$ and so by sub-IH $v$ is a $(\mathcal T, a)$-tree of level $I_{a}(u)$, which implies that $u= \lambda x. v_{1}\oplus_{a}^{I_{a}(u)} v_{2}$, which is not normal. Absurd.
		
		\item if $u= vw$, then let $J=I_{a}(u)= \max\{ I_{a}(v),I_{a}(w)\}$ cannot be $-1$, since otherwise $I_{a}(u)$ would be $-1$. Hence $J\geq 0$ is either $I_{a}(v)$ or $I_{a}(w)$. We consider the two cases separatedly:
			\begin{varitemize}
			\item if $J=I_{a}(v)$, then by sub-IH, $v= v_{1}\oplus_{a}^{J}v_{2}$, and thus 
			$u= ( v_{1}\oplus_{a}^{J}v_{2})w$ is not normal;
			\item if $J=I_{a}(w)$, then by sub-IH, $w= w_{1}\oplus_{a}^{J}w_{2}$, and thus 
			$u= v( w_{1}\oplus_{a}^{J}w_{2})$ is not normal.
			
			\end{varitemize}
		In any case we obtain an absurd conclusion.
		
		\item if $u= u_{1}\oplus_{a}^{i}u_{2}$, then if $u_{1},u_{2}$ are both in $\mathcal V$, we are done, since  $u$ is a $(\mathcal T, a)$-tree of level $i=I_{a}(u)$. Otherwise, if $i< I_{a}(u)$, then $J=I_{a}(u)= \max\{ I_{a}(v), I_{a}(w)\}$, so we consider 	two cases:
		\begin{varitemize}
		\item if $J=I_{a}(v)>i$, then by sub-IH, $v= v_{1}\oplus_{a}^{J}v_{2}$, and thus 
		$u= ( v_{1}\oplus_{a}^{J}v_{2})\oplus_{a}^{i}u_{2}$ is not normal;
		\item if $J=I_{a}(w)>i$, then by sub-IH, $w= w_{1}\oplus_{a}^{J}w_{2}$, and thus 
		$u= v\oplus_{a}^{i}( w_{1}\oplus_{a}^{J}w_{2})$ is not normal.
			
		\end{varitemize}
		We conclude then that $i=I_{a}(u)$;  we must then show that $I_{a}(v),I_{a}(w)<i$. Suppose first $I_{a}(v)\geq i$, then $( v_{1}\oplus_{a}^{J}v_{2})\oplus_{a}^{i}u_{2}$ is not normal. In a similar way we can show that 
		$I_{a}(w)<i$. 

		\item if $u=\nu b.v$, then $I_{a}(u)=I_{a}(v)$, so by sub-IH $v= v_{1}\oplus_{a}^{I_{a}(u)}v_{2}$, and we conclude that $u=\nu b.  v_{1}\oplus_{a}^{i}v_{2}$ is not normal.
		
		\end{varitemize}
	
\end{varitemize}

\item[($\Leftarrow$)] It suffices to check by induction on $t\in \mathcal T$ that it is in PNF.

\end{description}
\end{proof}

\begin{corollary}\label{cor:pnf}
A name-closed term of the form $\nu a.t$ is in PNF iff $t$ is a $(\mathcal T, a)$-tree of level $I_{a}(t)\geq 0$.
\end{corollary}

Let us conclude by discussing \emph{randomized trees} and \emph{head-reduction}. 
With any term $t$ we associate a labeled finitely branching tree $RBT(t)$ as follows:

\begin{varitemize}

\item if $t\redall^{*}_{\mathsf h}u$, where $u=\lambda x_{1}.\dots.\lambda x_{n}.yu_{1}\dots u_{n}$ is a HNV, then $RBT(t)$ only consists of one node labeled $u$;

\item if $t\redall^{*}_{\mathsf h}u$, where $u$ is a PNF of the form $\nu a.T$, $T$ is a $(\CC T, a)$-tree and $\mathrm{supp}(T)=\langle u_{1},\dots, u_{n}\rangle$, then $RBT(t)$ has root labeled $\nu a.T$ and coincides with the syntactic tree of $T$, with leaves replaced by $RBT(u_{1}),\dots, RBT(u_{n})$;


\item otherwise, $RBT(t)$ has a root labeled $\Omega$ and no sub-trees.

\end{varitemize}

%
%
%
\begin{proposition}
If $t$ reduces to $u$ by either $\redperm$ or head $\beta$-reduction, then $RBT(t)\sqsubseteq RBT(u)$.
\end{proposition}

A \emph{randomized path in $RBT(t)$} is any path $\pi$ in $RBT(t)$ starting from the root and ending in a leaf.
%
%
A randomized path in $RBT(t)$ is thus of one of the following two forms:
\begin{enumerate}
\item  a finite path 
$\pi = \langle\nu a_{1}.T_{1}, \nu a_{2}.T_{2},\dots, \nu a_{N}.T_{N}, \lambda x_{1}\dots x_{n}.y\rangle$, where 
$v\leadsto \nu a_{i+1}.T_{i+1}$ holds for some $v\in \mathrm{supp}(T_{i})$;
\item  an infinite path 
$\pi = \langle\nu a_{1}.T_{1}, \nu a_{2}.T_{2},\dots, \nu a_{n}.T_{n},\dots \rangle$, where 
$v\leadsto \nu a_{i+1}.T_{i+1}$ holds for some $v\in \mathrm{supp}(T_{i})$.

\end{enumerate}
With any such path $\pi$ we can associate an (either finite or infinite) list of words $b^{\pi}_{i}\in \{0,1\}^{*}$, where $b^{\pi}_{i}$ is the list of choices leading from $\nu a.T_{i}$ to the unique term $v\in \mathrm{supp}(T_{i})$ such that $v\leadsto \nu a_{i+1}.T_{i+1}$.

If $\pi$ is a finite randomized path in $RBT(t)$ ending in some $u\in \CC V$, we say that $\pi$ is a randomized path \emph{from $t$ to $u$} (noted $\pi: t \mapsto u$).

Let us consider the (unique) Borel $\sigma$-algebra $\mu$ on $2^{\BB N\times \BB N}$ satisfying $\mu(C_{i,j})=1/2$, where $C_{i,j}$ is the cylinder $\{\omega\mid \omega(i,j)=1\}$.

Any randomized path $\pi$ yields a Borel set $B_{\pi}:= \bigcap_{i,j}C_{i, (b^{\pi}_{i})_{j}}$, so that two distinct paths $\pi\neq \pi'$ are such that $B_{\pi}\cap B_{\pi'}=\emptyset$.
Moreover, if $\pi$ is finite, the Borel set $B_{\pi}$ is captured by the Boolean formula
$
\bone_{\pi}:= \bigwedge_{i}\bone_{\pi}^{i}$, where $\bone_{\pi}^{i}= \bigwedge_{j} (\lnot)^{1+(b^{\pi}_{i})_{j}}\bvar_{a_{i}}^{j}$.

\begin{lemma}
For all terms $t\in \CC T$ and $u\in \CC V$, 
$\CC D_{t}(u)= 
\mu\left ( \bigcup\{B_{\pi}\mid \pi : t\mapsto u\}\right)=
\sum_{\pi: t\mapsto u}\mu(\bone_{\pi})
$.
\end{lemma}
%
%
%

%
%
%
%
%
%

%
%
\begin{lemma}\label{lemma:hformrbt}
For all $t\in \CC T$, the following are equivalent:
\begin{varitemize}
\item[(i.)] $t$ has a head-normal form;
\item[(iii.)] $RBT(t)$ is finite.
\end{varitemize}
\end{lemma}
\begin{proof}
(i.)$\To$ (iii.) is proved by induction on the length of a reduction of $t$ to a head-normal form. 
For 
(iii.)$\To$ (i.), observe that $RBT(t)$ contains an infinite path only if $t$ admits an infinite head-reduction.
From (iii.) 
we deduce then that all head reductions of $t$ are finite, so $t$ has a head-normal form.%
\end{proof}

\subsection{The Language $\EVLL$.}\label{subs:evll}

In this section we study some properties of reduction $\redalll$ in the calculus $\EVLL$. Most arguments closely resemble arguments from \cite{DLGH}, so we omit several details.

\begin{lemma}\label{lemma:snperm}
$\redpermm$ is strongly normalizing.
\end{lemma}
\begin{proof}
We follow the proof of [\cite{DLGH}, Lemma 7], we define the partial order $\prec_{M}$ by 
\begin{align*}
\oplus_{a}^{i} & \prec_{M} \oplus_{b}^{j}  &  \text{ if }(a,i)<_{M}(b,j) \\
\oplus_{a}^{i} & \prec_{M}\nu b & \text{ for all labels }a,b \\
\nu a & \prec_{M} @, \lambda x & \text{ fo any label } b\\
@, \lambda x & \prec_{M} \{\}
\end{align*}
One can check then that the well-founded recursive path ordering $<$ defined by 
\begin{align*}
t< u \ \Leftrightarrow \ 
\begin{cases}
[t_{1},\dots, t_{n}] < [u_{1},\dots, u_{m}] & \text{ if }f=g \\
[t_{1},\dots, t_{n}]< [u] & \text{ if }f\prec_{M}g\\
[t]\leq [u_{1},\dots, u_{m}] & \text{ if }f\not\prec_{B}g
\end{cases}
\end{align*}
where $t= f(t_{1},\dots, t_{n})$ and $u=g(u_{1},\dots, u_{n})$, with $f,g\in \{\oplus_{a}^{i}, \nu a.,@,\lambda x, \{\}\}$, is such that $t\redpermm u$ implies $u<t$. 
\end{proof}

\begin{lemma}
$\redpermm$ is confluent.
\end{lemma}
\begin{proof}
By Lemma \ref{lemma:snperm} and Newman's Lemma, it suffices to check local confluence. All rules  from \eqref{eq:c1} to \eqref{eq:cnu} from Fig.~\ref{fig:permutations}, as well as rules \eqref{eq:e1} and \eqref{eq:e2}, can be written under the general form 
\begin{align}
\TT C[t\oplus_{a}^{i}u] \redpermm \TT C[t]\oplus_{a}^{i}\TT C[u] \tag{$\oplus\star$}\label{eq:star}
\end{align}
where $\TT C[\ ]$ is defined by the grammar
$$
\TT C[\ ]::= [\ ] \mid \lambda x.\TT C[\ ]\mid \TT C[\ ]u\mid t\TT C[\ ]\mid \TT C[\ ]\oplus_{a}^{i}u \mid t\oplus_{a}^{i}\TT C[\ ]\mid \nu a.\TT C[\ ]\mid\{ \TT C[\ ]\} u\mid \{t\} \TT C[\ ]
$$
We consider all rules against each other. All cases involving \eqref{eq:star} can be treated as in the proof of  [\cite{DLGH}, Lemma 9]. Beyond these, the only new case with respect to those treated in [\cite{DLGH}, Lemma 9] is the one below:
$$
\begin{tikzcd}
\{ t\} \nu a. u\oplus_{b}^{i} v \ar{r}{\eqref{eq:cnu}} \ar{d}{\eqref{eq:e0}} &   \{t\} (\nu a.u)\oplus_{b}^{i}(\nu a.v) \ar{d}{\eqref{eq:e1}}\\
\nu a. t(u\oplus_{b}^{i}v)  
\ar{d}{\eqref{eq:oa}}  &   (\{t\} \nu a.u ) \oplus_{b}^{i}(\{t\} \nu a.v) \ar{d}{\eqref{eq:e0}}  \\
\nu a. (tu) \oplus_{b}^{i} (tv) \ar{r}[below]{\eqref{eq:cnu}} &  (\nu a.tu) \oplus_{b}^{i} (\nu a. tv) 
\end{tikzcd}
$$
\end{proof}

\begin{theorem}
$\redalll$ is confluent.
\end{theorem}
\begin{proof}
The argument closely follows the one from [\cite{DLGH}, pp.~9-12], using the observation that 
 none of the permutations \eqref{eq:e0}-\eqref{eq:e2} can either block a $\beta$-redex or figure in a critical pair with a $\beta$-redex (in other words, $\beta$-reduction commutes with \eqref{eq:e0}-\eqref{eq:e2}).
\end{proof}

To conclude, let us study head-reduction $\redallh$ in $\EVLL$.

Randomized contexts $\TT R[\ ]$ for $\EVLL$ are defined as for  $\EVL$.
\emph{Head-contexts} $\TT H[\ ]$ for $\EVL$ are defined by the grammar below:
$$
\TT H[\ ]:= [\ ]\mid \lambda x.\TT H[\ ]\mid \TT H[\ ]u \mid \{\TT H[\ ]\} u
$$
Head-reductions $t\redallh u$ are defined inductively as being either a $\redalll$-reduction or a $\beta$-reduction of one of the following two forms
\begin{align*}
\TT R[\TT H[(\lambda y.t)u]] & \redbeta \TT R[\TT H[t[u/y]]] \\
\TT R[\TT H[\{t\}u]] & \redall \TT R[\TT H[\{t\}u']] 
\end{align*}
where $u\redallh u'$.

Let $\CC T^{\{\}}$ be the set of PNF with respect to $\redpermm$.
A pseudo-value $t$ is a PNF which is either a $\lambda$, a variable, an application $tu$ or a CbV application $\{t\}u$.
We let $\CC V^{\{\}}$ indicate the set of pseudo-values.

The definition of $\pi^{\omega}_{X}(t)$ is extended from $\EVL$ to $\EVLL$ simply by adding the condition
$\pi^{\omega}_{X}(\{t\}u)=\{\pi^{\omega}_{X}(t)\}\pi^{\omega}_{X}(u)$. 
In this way, a sub-distribution $\CC D_{t}: \CC V^{\{\}}\to [0,1]$ can be defined as in Section \ref{section3}.

\begin{lemma}\label{lemma:hnfecvl}
A pseudo-value $t\in \CC V^{\{\}}$ is in head normal form iff either $t= \TT H[x]$ or $t=\TT H[\{t\}u]$, where 
$u\in \CC V^{\{\}}$ and is $\redalll$-normal.
\end{lemma}

A head normal value is defined as in $\EVL$ as a pseudo-value in head normal form. We let HNV$^{\{\}}$ indicate the set of head normal values.
Using Lemma \ref{lemma:hnfecvl} we can define, in analogy with the case of $\EVL$, the functions $\HNF(t):= \sum_{v\in \mathrm{HNV}^{\{\}}}\CC D_{t}(v)$, and 
 $\NHNF(t):= \sup\{ \HNF(u)\mid t\redallh^{*}u\}$.

\section{Details about $\LCPL$.}\label{sectionC}

\subsection{Subject Reduction.}
The goal of this section is to establish the following result:
\begin{proposition}[Subject Reduction]\label{prop:subject}
If $\Gamma\vdash^{X}t:\bone \Pto \FF s$ and $t\redalll u$, then
$\Gamma\vdash^{X}u:\bone \Pto \FF s$.
\end{proposition}

With the goal of making proof slightly simpler, in the formulation of $\LCPL$ we replace the rule ($\oplus$) with the two rules ($\oplus l$) and ($\oplus r$) below:

\begin{center}
\adjustbox{scale=0.9}{$
\AXC{$\Gamma \vdash^{X\cup\{a\}}t: \btwo \Pto \FF s$}
\AXC{$\bone\vDash \btwo \land \bvar_{a}^{i}$}
\RL{($\oplus l$)}
\BIC{$\Gamma \vdash^{X\cup\{a\}}t\oplus_{a}^{i} u: \bone \Pto \FF s$}
\DP
\qquad\qquad
\AXC{$\Gamma \vdash^{X\cup\{a\}}u: \btwo \Pto \FF s$}
\AXC{$\bone\vDash \btwo \land \lnot\bvar_{a}^{i}$}
\RL{($\oplus r$)}
\BIC{$\Gamma \vdash^{X\cup\{a\}}t\oplus_{a}^{i} u: \bone \Pto \FF s$}
\DP
$
}
\end{center}

It is easily checked that, in presence of the rule ($\vee$), having the rule ($\oplus$) is equivalent to having the rules ($\oplus l$) and ($\oplus r$).

To establish the subject reduction property, we first need to establish a few auxiliary lemmas.

\begin{lemma}\label{lemma:weak2}
If $\Gamma \vdash^{X}t: \bone \Pto \FF s$ and $X\subseteq Y$, then 
$\Gamma \vdash^{Y}t: \bone\Pto \FF s$.
\end{lemma}
\begin{proof}
By induction on a type derivation of $t$.
\end{proof}

\begin{lemma}\label{lemma:weak}
If $\Gamma \vdash^{X}t: \bone\Pto \FF s$ holds and $\btwo\vDash^{X} \bone$, then 
$\Gamma \vdash^{X}t: \btwo\Pto \FF s$ is derivable by a derivation of the same length.
\end{lemma}
\begin{proof}
By induction on a type derivation of $t$.
\end{proof}

\begin{lemma}[substitution lemma]\label{lemma:betasub}
The following rule is derivable:
$$
\AXC{$\Gamma, x: \FF s \vdash^{X} t: \btwo \Pto  \FF t$}
\AXC{$\Gamma \vdash^{X} u: \bthree \Pto \FF s$}
\AXC{$\bone\vDash\btwo \land \bthree$}
\RL{(subst)}
\TIC{$\Gamma \vdash^{X} t[u/x]: \bone \Pto \FF t$}
\DP
$$
\end{lemma}
\begin{proof}
We argue by induction on the typing derivation of $t$:
\begin{varitemize}

%
%
\item if the last rule is

\begin{prooftree}
\AXC{$\FN(\bone)\subseteq X$}
\RL{(id)}
\UIC{$\Gamma, x: \FF s \vdash^{X}x: \btwo \Pto \FF s$}
\end{prooftree}
then $t[u/x]=u$, so the claim can be deduced using Lemma \ref{lemma:weak}.

\item if the last rule is

\begin{prooftree}
\AXC{$\left \{\Gamma, x: \FF s \vdash^{X}t: \btwo_{i} \Pto \FF t\right\}_{i}$}
\AXC{$\btwo \vDash^{X}\bigvee_{i} \btwo_{i}$}
\RL{$(\lor)$}
\BIC{$\Gamma, x: \FF s \vdash^{X}x: \btwo \Pto \FF t$}
\end{prooftree}

Then, by IH, we deduce $\Gamma, x: \FF s \vdash^{X}t[u/x]: \btwo_{i}\land \bthree \Pto \FF t$, and since 
$\bone \vDash^{X} \bigvee_{i}(\bone_{i}\land \bthree)$ we  conclude by applying an instance of ($\lor$).

\item if the last rule is 

\begin{prooftree}
\AXC{$\Gamma, x:\FF s\vdash^{X\cup\{a\}} t_{1}:\btwo'\Pto \FF t$}
\AXC{$
\btwo\vDash^{X\cup\{a\}}\bvar_{a}^{i}\land \btwo'
$}
\RL{($\oplus l$)}
\BIC{$\Gamma, x:\FF s\vdash^{X\cup\{a\}} t_{1}\oplus^{i}_{a}t_{2}:  \btwo\Pto \FF t$}
\end{prooftree}

Then, by IH, we deduce 
$\Gamma\vdash^{X\cup\{a\}, q} t_{1}[u/x]: \btwo'\land \bthree$.
From $t[u/x]= (t_{1}[u/x])\oplus_{a}^{i}(t_{2}[u/x])$ and the fact that  
$
\bone\vdash^{X\cup\{a\}}\bvar_{a}^{i}\land (\btwo'\land \bthree)
$, we deduce the claim by an instance of the same rule.

\item if the last rule is 

\begin{prooftree}
\AXC{$\Gamma, x:\FF s \vdash^{X\cup\{a\}} t_{2}:\btwo'\Pto \FF t$}
\AXC{$\btwo\vDash^{X\cup\{a\}}\lnot\bvar_{a}^{i}\land \btwo'$}
\RL{$\TR$}
\BIC{$\Gamma, x:\FF s \vdash^{X\cup\{a\}} t_{1}\oplus^{i}_{a}t_{2}:  \btwo\Pto \FF t$}
\end{prooftree}

we can argue similarly to the previous case.

\item if the last rule is 

\begin{prooftree}
\AXC{$\Gamma, y:\FF t', x: \FF s \vdash^{X} t': \btwo\Pto\B C^{\vec q}\tau$}
\RL{$\TLA$}
\UIC{$\Gamma, x:\FF s \vdash^{X} \lambda y.t':\btwo\Pto  \B C^{\vec q}(\FF t'\To\tau) $}
\end{prooftree}

Then, by IH, we deduce $\Gamma, y:\FF t' \vdash^{X}t'[u/x]: \bone \Pto \B C^{\vec q}\tau$ and since $t[u/x]=(\lambda y.t')[u/x]=\lambda y.t'[u/x]$ we conclude by applying an instance of the same rule.

\item if the last rule is 

\begin{prooftree}
\AXC{$\Gamma, x:\FF s \vdash^{X} t_{1}: \btwo_{1}\Pto \B C^{\vec q}(\FF t' \To \tau)$}
\AXC{$\Gamma,x:\FF s \vdash^{X} t_{2}: \btwo_{2} \Pto \FF t'$}
\AXC{$\btwo\vDash^{X} \btwo_{1}\land \btwo_{2}$}
\RL{($@$)}
\TIC{$\Gamma,x:\FF s\vdash^{X} t_{1}t_{2}: \btwo \Pto \B C^{\vec q}\tau$}
\end{prooftree}

Then, by IH, we deduce $\Gamma \vdash^{X}t_{1}[u/x]: \btwo_{1}\land \bthree \Pto \B C^{\vec q}(\FF t'\To \tau)$ and 
$\Gamma \vdash^{X}t_{2}[u/x]: \btwo_{2}\land \bthree\Pto \FF t'$, and since $(t_{1}t_{2})[u/x]= (t_{1}[u/x])(t_{2}[u/x])$ and  
$\bone \vDash^{X} (\btwo_{1}\land \bthree)\land (\btwo_{2}\land \bthree)$, 
 we conclude by applying an instance of the same rule.

\item if the last rule is 

\begin{prooftree}
\AXC{$\Gamma, x:\FF s \vdash^{X} t_{1}: \btwo_{1}\Pto \B C^{\vec q}(\FF t' \To \tau)$}
\AXC{$\Gamma,x:\FF s \vdash^{X} t_{2}: \btwo_{2} \Pto \B C^{s}\FF t'$}
\AXC{$\btwo\vDash^{X} \btwo_{1}\land \btwo_{2}$}
\RL{($\{\}$)}
\TIC{$\Gamma,x:\FF s\vdash^{X} \{t_{1}\}t_{2}: \btwo \Pto \B C^{s}\B C^{\vec q}\tau$}
\end{prooftree}

Then, by IH, we deduce $\Gamma \vdash^{X}t_{1}[u/x]: \btwo_{1}\land \bthree \Pto \B C^{\vec q}(\FF t'\To \tau)$ and 
$\Gamma \vdash^{X}t_{2}[u/x]: \btwo_{2}\land \bthree\Pto \B C^{s}\FF t'$, and since $\{t_{1}\}t_{2})[u/x]= \{t_{1}[u/x]\}(t_{2}[u/x])$ and  
$\bone \vDash^{X} (\btwo_{1}\land \bthree)\land (\btwo_{2}\land \bthree)$, 
 we conclude by applying an instance of the same rule.

\item if the last rule is

\begin{prooftree}
\AXC{$\Gamma, x:\FF s \vdash^{X\cup\{a\}} t:\btwo'\land \bfour\Pto \FF t$}
\AXC{$\mu(\bfour)\geq s$}
\AXC{$\btwo \vDash^{X} \btwo'$}
\RL{$\TN$}
\TIC{$\Gamma, x:\FF s \vdash^{X} \nu a.t: \btwo\Pto  \B C^{s}\FF t$}
\end{prooftree}

Then, by IH,
$\Gamma \vdash^{X\cup\{a\}} t[u/x]:
(\btwo'\land \bthree)\land \bfour_{i}
\Pto \FF t$.
%
Hence, from the fact that  
$\bone \vDash^{X}\btwo \land \bthree$ and that $a$ cannot occur in $\bthree$, and since 
 $(\nu a.t)[u/x]=\nu a.t[u/x]$, we can deduce the claim by applying an instance of the same rule. 
%
%
%
%
\end{varitemize}

\end{proof}

We now have all ingredients to establish the subject reduction property of $\LCPL$.

%
%
\begin{proof}[Proof of Proposition \ref{prop:subject}]
First observe that if the typing derivation $D$ of $t$ ends by a $(\lor)$-rule, it suffices to establish the property for the immediate sub-derivations of $D$ and then apply an instance of $(\lor)$-rule to the resulting derivations. So we will always suppose that the typing derivation of $D$ does not end by a $(\lor)$-rule.

For the case of $\beta$-reduction it suffices to check the claim when $t$ is a redex $(\lambda x.t_{1})t_{2}$ and $u$ is $t_{1}[t_{2}/x]$. From 
$\Gamma\vdash^{X}t:\bone \Pto \FF s$ we can suppose w.l.o.g.~that
the typing derivation is as below: 
\begin{center}

\begingroup\makeatletter\def\f@size{10}\check@mathfonts
\begin{lrbox}{\mypti}
\begin{varwidth}{\linewidth}
$
\AXC{$\Gamma,x:\FF t \vdash^{X}t_{1}: \btwo_{i}\Pto  \B C^{\vec q}\sigma$}
\UIC{$\Gamma \vdash^{X}\lambda x.t_{1}: \btwo_{i}\Pto\B C^{\vec q}(\FF t\To  \sigma)$}
\DP
$
\end{varwidth}
\end{lrbox}

\resizebox{\textwidth}{!}{$
\AXC{$\left\{\usebox{\mypti}
\right\}_{i}$}
\AXC{$\bone_{1}\vDash \bigvee_{i}\btwo_{i}$}
\RL{$(\lor)$}
\BIC{$\Gamma \vdash^{X}\lambda x.t_{1}: \bone_{1}\Pto  \B C^{\vec q}(\FF t\To \sigma)$}
\AXC{$\left\{\Gamma \vdash^{X}t_{2}: \bthree_{j}\Pto \FF t\right\}_{j}$}
\AXC{$\bone_{2}\vDash \bigvee_{j}\bthree_{j}$}
\RL{$(\lor)$}
\BIC{$\Gamma \vdash^{X}t_{2}: \bone_{2}\Pto\FF t$}
\AXC{$\bone \vDash \bone_{1}\land \bone_{2}$}
\TIC{$\Gamma \vdash^{X}t: \bone\Pto  \B C^{\vec q}\sigma$}
\DP
$}
\endgroup
\end{center}

From Lemma \ref{lemma:betasub} we deduce the existence of derivations of  
$\Gamma \vdash^{X} u: \btwo_{i}\land \bthree_{j}\Pto \B C^{\vec q}\sigma$;
from $\bone \vDash^{X}\bone_{1}\land \bone_{2}$ we deduce then
$\bone \vDash^{X} \bigvee_{i,j}\btwo_{i}\land \bthree_{j}$, using the fact that 
$\bigvee_{i}\btwo_{i} \land \bigvee_{j}\bthree_{j}\equiv \bigvee_{i,j}\btwo_{i}\land \bthree_{j}$. We can thus conclude as follows:

$$
\AXC{$\left\{\Gamma \vdash^{X} u: \btwo_{i}\land \bthree_{j}\Pto \B C^{\vec q}\sigma\right\}_{i,j}$}
\AXC{$\bone\vDash^{X} \bigvee_{i,j}\btwo_{i}\land \bthree_{j}$}
\RL{$(\lor)$}
\BIC{$\Gamma \vdash^{X}u: \bone \Pto \B C^{\vec q}\sigma$}
\DP
$$

%
%
%
%
%

%
%
For the case of $\redpermm$ wee consider reduction rules one by one. 
\begin{description}

\item[($t\oplus_{a}^{i}t\redpermm t$)]  
The last rule of $t$ is either
\begin{prooftree}
\AXC{$ \Gamma \vdash^{ X} t: \bone'\Pto \FF s$}
\AXC{$ \bone \vDash^{X} \bvar_{a}^{i}\land \bone'$}
\RL{($\oplus l$)}
\BIC{$ \Gamma \vdash^{X} t\oplus^{a}_{i}t:  \bone\Pto \FF s$}
\end{prooftree}

or

\begin{prooftree}
\AXC{$ \Gamma \vdash^{ X} t: \bone'\Pto \FF s$}
\AXC{$ \bone \vDash^{X}\lnot \bvar_{a}^{i}\land \bone'$}
\RL{($\oplus r$)}
\BIC{$ \Gamma \vdash^{X} t\oplus^{a}_{i}t:  \bone\Pto \FF s$}
\end{prooftree}
Then, in either case, from $ \Gamma \vdash^{ X} t: \bone'\Pto \FF s$, $\bone \vDash^{X} \bone'$, using Lemma \ref{lemma:weak} we deduce $ \Gamma \vdash^{ X} t: \bone\Pto \FF s$.

\bigskip

\item[($ (t\oplus_{a}^{i}u)\oplus_{a}^{i}v \redpermm t\oplus_{a}^{i}v$)] 
There are three possible sub-cases:
	\begin{enumerate}
	\item the type derivation is as follows:
	
\begin{prooftree}
	\AXC{$\Gamma\vdash^{X} t: \bone''\Pto \FF s$}
	\AXC{$\bone'\vDash^{X} \bvar_{a}^{i}\land \bone''$}
	\RL{$\TL$}
	\BIC{$\Gamma \vdash^{X} t\oplus_{a}^{i} u :\bone'\Pto \FF s$}
	\AXC{$\bone\vDash^{X} \bvar_{a}^{i}\land \bone'$}
	\RL{$\TL$}
	\BIC{$\Gamma \vdash^{X} (t\oplus_{a}^{i} u)\oplus_{a}^{i} :\bone\Pto \FF s$}
\end{prooftree}

	Then from $\Gamma\vdash^{X} t: \bone''\Pto \FF s$ and since we have $\bone\vdash^{X} \bvar_{a}^{i}\land \bone''$ we deduce $\Gamma \vdash^{X}t\oplus_{a}^{i}v : \bone \Pto \FF s$.
	
	\item the type derivation is as follows:
\begin{prooftree}
	\AXC{$\Gamma\vdash^{X} u: \bone''\Pto \FF s$}
	\AXC{$\bone'\vDash^{X} \lnot\bvar_{a}^{i}\land \bone''$}
	\RL{$\TR$}
	\BIC{$\Gamma \vdash^{X} t\oplus_{a}^{i} u :\bone'\Pto \FF s$}
	\AXC{$\bone\vDash^{X} \bvar_{a}^{i}\land \bone'$}
	\RL{$\TL$}
	\BIC{$\Gamma \vdash^{X} (t\oplus_{a}^{i} u)\oplus_{a}^{i}v :\bone\Pto \FF s$}
\end{prooftree}

	Then, from $\bone\vDash^{X} \bvar_{a}^{i}\land \bone'$ and $\bone'\vDash^{X} \lnot\bvar_{a}^{i}\land \bone''$ we deduce $\bone \vDash^{X}\BOT$, so we conclude $\Gamma \vdash^{X} (t\oplus_{a}^{i}v): \bone \Pto \FF s$ using one of the initial rules.
	
	\item the type derivation is as follows:
	
\begin{prooftree}
	\AXC{$\Gamma\vdash^{X} v: \bone''\Pto \FF s$}
	\AXC{$\bone\vDash^{X} \lnot \bvar_{a}^{i}\land \bone''$}
	\RL{$\TR$}
	\BIC{$\Gamma \vdash^{X} (t\oplus_{a}^{i} u)\oplus_{a}^{i}v :\bone\Pto \FF s$}
\end{prooftree}

	Then from $\Gamma\vdash^{X} u: \bone'\Pto \FF s$ and  $\bone\vDash^{X} \lnot \bvar_{a}^{i}\land \bone''$ we deduce $\Gamma \vdash^{X}t\oplus_{a}^{i}v : \bone \Pto \FF s$.
	\end{enumerate}

	\bigskip

\item[($t\oplus_{a}^{i}(u\oplus_{a}^{i}v)\redpermm t\oplus_{a}^{i}v$)] Similar to the case above.

\bigskip

\item[($\lambda x.(t\oplus_{a}^{i}u)\redpermm (\lambda x.t)\oplus_{a}^{i}(\lambda x.u)$)]
There are two possible sub-cases:
	\begin{varitemize}
	\item[1.]
	
\begin{prooftree}
	\AXC{$\Gamma, x:\FF s \vdash^{X} t: \bone'\Pto \B C^{\vec q}\tau$}
	\AXC{$\bone\vDash^{X} \bvar_{a}^{i}\land \bone'$}
	\RL{$\TL$}
	\BIC{$\Gamma, x:\FF s \vdash^{X} t\oplus_{a}^{i}u: \bone\Pto \B C^{\vec q}\tau$}
	\RL{$\TLA$}
	\UIC{$\Gamma \vdash^{X} \lambda x.(t\oplus_{a}^{i}u): \bone \Pto\B C^{\vec q} (\FF s\To \tau)$}
\end{prooftree}

	Then, we deduce 
	
\begin{prooftree}
	\AXC{$\Gamma, x:\FF s \vdash^{X} t: \bone'\Pto \B C^{\vec q}\tau$}
	\RL{$\TLA$}
	\UIC{$\Gamma \vdash^{X} \lambda x.t: \bone' \Pto\B C^{\vec q} (\FF s\To \tau)$}
	\AXC{$\bone\vDash^{X} \bvar_{a}^{i}\land \bone'$}
	\RL{$\TL$}
	\BIC{$\Gamma \vdash^{X} (\lambda x.t)\oplus_{a}^{i}(\lambda x.u): \bone \Pto \B C^{\vec q} (\FF s\To \tau)$}
\end{prooftree}

	\item[2.]
	
\begin{prooftree}
	\AXC{$\Gamma, x:\FF s \vdash^{X}u: \bone'\Pto \B C^{\vec q}\tau$}
	\AXC{$\bone\vDash^{X} \lnot\bvar_{a}^{i}\land \bone'$}
	\RL{$\TR$}
	\BIC{$\Gamma, x:\FF s \vdash^{X} t\oplus_{a}^{i}u: \bone\Pto \B C^{\vec q}\tau$}
	\RL{$\TLA$}
	\UIC{$\Gamma \vdash^{X} \lambda x.(t\oplus_{a}^{i}u): \bone \Pto \B C^{\vec q}(\FF s\To \tau)$}
\end{prooftree}

	Then, we can argue similarly to the previous case.
	
	\end{varitemize}

	\item[($(t\oplus_{a}^{i}u)v \redpermm (tv)\oplus_{a}^{i}(uv)$)] 
	There are two possible sub-cases:
	\begin{enumerate}
	\item 
	\footnotesize{
\begin{prooftree}
	\AXC{$\Gamma \vdash^{X} t: \bone''\Pto \B C^{\vec q}(\FF s\To \tau)$}
	\AXC{$\bone'\vDash^{X} \bvar_{a}^{i}\land \bone''$}
	\RL{$\TL$}
	\BIC{$\Gamma\vdash^{X} t\oplus_{a}^{i} u: \bone' \Pto \B C^{\vec q}(\FF s\To\tau)$}
	\AXC{$\Gamma \vdash^{X}v: \btwo \Pto \FF s$}
	\AXC{$\bone \vDash^{X} \bone'\land \btwo$}
	\RL{$\TA$}
	\TIC{$\Gamma\vdash^{X}(t\oplus_{a}^{i}u)v: \bone \Pto \B C^{\vec q}\tau$} 
\end{prooftree}
}
\normalsize
	Then, we deduce 
	\footnotesize{
\begin{prooftree}
	\AXC{$\Gamma \vdash^{X} t: \bone''\Pto  \B C^{\vec q}(\FF s\To\tau)$}
	\AXC{$\Gamma \vdash^{X}v: \btwo \Pto \FF s$}
	\AXC{$\bone'\land \btwo \vDash^{X} \bone''\land \btwo$}
	\RL{$\TA$}
	\TIC{$\Gamma \vdash^{X} tv : \bone'\land \btwo \Pto  \B C^{\vec q}\tau$}
	\AXC{$\bone \vDash^{X} \bvar_{a}^{i}\land \bone'\land \btwo $}
	\RL{$\TL$}
	\BIC{$\Gamma \vdash^{X} (tv)\oplus_{a}^{i} (uv): \bone \Pto  \B C^{\vec q}\tau$}
\end{prooftree}
}
\normalsize

	\item 
	\footnotesize
\begin{prooftree}
	\AXC{$\Gamma \vdash^{X} u: \bone''\Pto  \B C^{\vec q}(\FF s\To\tau)$}
	\AXC{$\bone'\vDash^{X} \lnot\bvar_{a}^{i}\land \bone''$}
	\RL{$\TR$}
	\BIC{$\Gamma\vdash^{X} t\oplus_{a}^{i} u: \bone' \Pto  \B C^{\vec q}(\FF s\To\tau)$}
	\AXC{$\Gamma \vdash^{X}v: \btwo \Pto \FF s$}
	\AXC{$\bone \vDash^{X} \bone'\land \btwo$}
	\RL{$\TA$}
	\TIC{$\Gamma\vdash^{X}(t\oplus_{a}^{i}u)v: \bone \Pto  \B C^{\vec q}\tau$} 
\end{prooftree}
\normalsize

	Then, we can argue similarly to the previous case.
	\end{enumerate}

\bigskip

	\item[($\{t\oplus_{a}^{i}u\}v \redpermm (\{t\}v)\oplus_{a}^{i}(\{u\}v)$)]
	Similar to the previous case. 

\item[{($t(u\oplus_{a}^{i}v)\redpermm (tu)\oplus_{a}^{i}(tv)$)}]
	There are two sub-cases:
	\begin{enumerate}
	\item 
	
	\footnotesize{
\begin{prooftree}
	\AXC{$\Gamma \vdash^{X} t: \bone'\Pto \B C^{\vec q}(\FF s\To\tau)$}
	\AXC{$\Gamma \vdash^{X}u: \btwo' \Pto \FF s$}
	\AXC{$\btwo\vDash^{X} \bvar_{a}^{i}\land \btwo'$}
	\RL{$\TL$}
	\BIC{$\Gamma \vdash^{X} u\oplus_{a}^{i} v: \btwo \Pto \FF s$}
	\AXC{$\bone \vDash^{X} \bone'\land \btwo$}
	\RL{$\TA$}
	\TIC{$\Gamma\vdash^{X}t(u\oplus_{a}^{i}v): \bone \Pto  \B C^{\vec q}\tau$} 
\end{prooftree}
}
\normalsize
	Then, we deduce that 
\footnotesize{	
\begin{prooftree}
	\AXC{$\Gamma \vdash^{X} t: \bone'\Pto  \B C^{\vec q}(\FF s\To\tau)$}
	\AXC{$\Gamma \vdash^{X}u: \btwo' \Pto \FF s$}
	\AXC{$\bone \vDash^{X} \bone'\land \btwo'$}
	\RL{$\TA$}
	\TIC{$\Gamma\vdash^{X}tu: \bone \Pto  \B C^{\vec q}\tau$} 
	\AXC{$\bone\vDash^{X} \bvar_{a}^{i}\land \bone$}
	\RL{$\TR$}
	\BIC{$\Gamma \vdash^{X} (tu)\oplus_{a}^{i}(tv): \bone \Pto  \B C^{\vec q}\tau$}
\end{prooftree}
}
\normalsize
	\item 
\footnotesize{	
\begin{prooftree}
	\AXC{$\Gamma \vdash^{X} t: \bone'\Pto  \B C^{\vec q}(\FF s\To\tau)$}
	\AXC{$\Gamma \vdash^{X}v: \btwo' \Pto \FF s$}
	\AXC{$\btwo\vDash^{X}\lnot \bvar_{a}^{i}\land \btwo'$}
	\RL{$\TR$}
	\BIC{$\Gamma \vdash^{X} u\oplus_{a}^{i} v: \btwo \Pto \FF s$}
	\AXC{$\bone \vDash^{X} \bone'\land \btwo$}
	\RL{$\TA$}
	\TIC{$\Gamma\vdash^{X}t(u\oplus_{a}^{i}v): \bone \Pto  \B C^{\vec q}\tau$} 
\end{prooftree}
}
\normalsize
	Then, we can argue similarly to the previous case.
	\end{enumerate}

\item[{($\{t\}(u\oplus_{a}^{i}v)\redpermm (\{t\}u)\oplus_{a}^{i}(\{t\}v)$)}]
Similar to the previous case.

\bigskip

\item[($(t\oplus_{a}^{i}u)\oplus_{b}^{j}v \redpermm (t\oplus_{b}^{j}v)\oplus_{a}^{i}(u\oplus_{b}^{j}v)$)]  
We suppose here $a\neq b$ or $i< j$. There are three sub-cases:
		\begin{enumerate}
		\item 
		
\begin{prooftree}

		\AXC{$\Gamma \vdash^{X} t: \bone''\Pto \FF s$}
		\AXC{$\bone'\vDash^{X} \bvar_{a}^{i} \land \bone''$}
		\RL{$\TL$}
		\BIC{$\Gamma \vdash^{X}t\oplus_{a}^{i}u: \bone'\Pto \FF s$}
		\AXC{$\bone\vDash^{X} \bvar_{b}^{j} \land \bone'$}
		\RL{$\TL$}
		\BIC{$\Gamma \vdash^{X}(t\oplus_{a}^{i}u)\oplus_{b}^{j}v: \bone'\Pto \FF s$}
\end{prooftree}

		Then, we deduce 		
\begin{prooftree}
		\AXC{$\Gamma \vdash^{X} t: \bone''\Pto \FF s$}
		\AXC{$\bone\vDash^{X} \bvar_{b}^{j} \land \bone''$}
		\RL{$\TL$}
		\BIC{$\Gamma \vdash^{X}t\oplus_{b}^{j}v: \bone'\Pto \FF s$}
		\AXC{$\bone\vDash^{X} \bvar_{a}^{i} \land \bone$}
		\RL{$\TL$}
		\BIC{$\Gamma \vdash^{X}(t\oplus_{b}^{j}u)\oplus_{a}^{i}(u\oplus_{b}^{j}v): \bone'\Pto \FF s$}
\end{prooftree}

		\item 
		
\begin{prooftree}
		\AXC{$\Gamma \vdash^{X} u: \bone''\Pto \FF s$}
		\AXC{$\bone'\vDash^{X} \lnot\bvar_{a}^{i} \land \bone''$}
		\RL{$\TR$}
		\BIC{$\Gamma \vdash^{X}t\oplus_{a}^{i}u: \bone'\Pto \FF s$}
		\AXC{$\bone\vDash^{X} \bvar_{b}^{j} \land \bone'$}
		\RL{$\TL$}
		\BIC{$\Gamma \vdash^{X}(t\oplus_{a}^{i}u)\oplus_{b}^{j}v: \bone'\Pto \FF s$}
\end{prooftree}

		Then, we can argue similarly to the previous case.
		\item 
		
\begin{prooftree}
		\AXC{$\Gamma \vdash^{X} v: \bone'\Pto \FF s$}
		\AXC{$\bone\vDash^{X} \bvar_{b}^{j} \land \bone'$}
		\RL{$\TR$}
		\BIC{$\Gamma \vdash^{X}(t\oplus_{a}^{i}u)\oplus_{b}^{j}v: \bone\Pto \FF s$}
\end{prooftree}
		Then, we deduce (using the fact that $\bone\equiv_{X} (\bvar_{a}^{i}\land \bone)\vee(\lnot \bvar_{a}^{i}\land \bone)$)

		\medskip
	
		\begin{center}
		\resizebox{0.95\textwidth}{!}{
		$
		\AXC{$\Gamma \vdash^{X} v: \bone'\Pto \FF s$}
		\AXC{$\bone\vDash^{X} \bvar_{b}^{j} \land \bone'$}
		\RL{$\TR$}
		\BIC{$\Gamma \vdash^{X}t\oplus_{b}^{j}v: \bone\Pto \FF s$}
		\AXC{$\bvar_{a}^{i}\land \bone \vDash^{X} \bvar_{a}^{i}\land \bone$}
		\RL{$\TR$}
		\BIC{$\Gamma \vdash^{X} (t\oplus_{b}^{j}v)\oplus_{a}^{i}(u\oplus_{b}^{j}v): \bvar_{a}^{i}\land \bone\Pto \FF s$}
		\AXC{$\Gamma \vdash^{X} v: \bone'\Pto \FF s$}
		\AXC{$\bone\vDash^{X} \bvar_{b}^{j} \land \bone'$}
		\RL{$\TR$}
		\BIC{$\Gamma \vdash^{X}u\oplus_{b}^{j}v: \bone\Pto \FF s$}
		\AXC{$\lnot\bvar_{a}^{i}\land \bone \vDash^{X} \bvar_{a}^{i}\land \bone$}
		\RL{$\TR$}
		\BIC{$\Gamma \vdash^{X} (t\oplus_{b}^{j}v)\oplus_{a}^{i}(u\oplus_{b}^{j}v): \lnot\bvar_{a}^{i}\land \bone\Pto\FF s$}
		\AXC{$\bone\vDash^{X} (\bvar_{a}^{i} \land \bone)\lor (\lnot \bvar_{a}^{i}\land \bone)$}
		\RL{$\TU$}
		\TIC{$\Gamma \vdash^{X} (t\oplus_{b}^{j}v)\oplus_{a}^{i}(u\oplus_{b}^{j}v): \bone\Pto \FF s$}
		\DP
		$}
		\end{center}
		\end{enumerate}

\item[($t\oplus_{b}^{j}(u\oplus_{a}^{i}v) \redpermm (t\oplus_{b}^{j}u)\oplus_{a}^{i}(t\oplus_{b}^{j}v)$)]  

Similar to the case above.

\bigskip

\item[{($\nu b. (t\oplus_{a}^{i}u)\redpermm (\nu b.t)\oplus_{a}^{i}(\nu b.u)$)}] 
We suppose $a\neq b$.
There are two sub-cases:
	\begin{enumerate}
	\item

	$$
\AXC{$\Gamma \vdash^{X\cup\{a,b\}}t: \bthree \Pto  \FF s$}
\AXC{$\bone'\land \btwo\vDash \bvar_{a}^{i}\land \bthree$}
\RL{($\oplus l$)}
\BIC{$\Gamma \vdash^{X\cup\{a,b\}}t\oplus_{a}^{i}u: \bone' \land \btwo\Pto  \FF s$}
	\AXC{$\mu(\btwo)\geq s$}
	\AXC{$\bone \vDash \bone'$}
	\RL{($\mu$)}
	\TIC{$\Gamma \vdash^{X\cup\{a\}}\nu b.t\oplus_{a}^{i}u :\bone\Pto\B C^{s} \FF s$}
	\DP	
	$$
	
From $\bone'\land \btwo_{i}\vDash  \bthree$, by Lemma \ref{lemma:weak}, we deduce the existence of a derivation of $\Gamma \vdash^{X\cup\{a,b\}}t: \bone'\land \btwo\Pto  \FF s$. Moreover, from $\bone'\land \btwo_{i}\vDash \bvar_{a}^{i}\land \bthree_{i}$ and the fact that $\mathrm{FN}(\btwo_{i})\subseteq \{b\}$ it follows that	$\bone'\vDash \bvar_{a}^{i}$, so we can construct the following derivation:

$$
\AXC{$\Gamma \vdash^{X\cup\{a,b\}}t: \bone'\land\btwo\Pto  \FF s$}
	\AXC{$\mu(\btwo_{i})\geq s$}
		\RL{($\mu$)}
	\BIC{$\Gamma \vdash^{X\cup\{a\}}\nu b.t: \bone' \Pto \B C^{s}\FF s$}
	\AXC{$\bone \vDash \bvar_{a}^{i}\land \bone'$}
	\RL{($\oplus l$)}
\BIC{$\Gamma \vdash^{X\cup\{a\}}(\nu b.t)\oplus_{a}^{i}(\nu b.u): \bone \Pto \B C^{s}\FF s$}
\DP
$$

	\item
	
	$$
\AXC{$\Gamma \vdash^{X\cup\{a,b\}}u: \bthree \Pto  \FF s$}
\AXC{$\bone'\land \btwo\vDash\lnot \bvar_{a}^{i}\land \bthree$}
\RL{($\oplus r$)}
\BIC{$\Gamma \vdash^{X\cup\{a,b\}}t\oplus_{a}^{i}u: \bone' \land \btwo\Pto  \FF s$}
	\AXC{$\mu(\btwo)\geq s$}
	\AXC{$\bone \vDash \bone'$}
		\RL{($\mu$)}
	\TIC{$\Gamma \vdash^{X\cup\{a\}}\nu b.t\oplus_{a}^{i}u :\bone\Pto\B C^{s} \FF s$}
	\DP	
	$$

\normalsize
	The we can argue similarly to the previous case.
	\end{enumerate}

\bigskip 
\item[($\lambda x.\nu a.t\redpermm \nu a.\lambda x.t$ )]  
We have

\begin{prooftree}
\AXC{$\Gamma, x: \FF s \vdash^{X\cup\{a\}} t: \bone'\land \btwo \Pto \B C^{\vec q}\sigma$}
\AXC{$\mu(\btwo)\geq s$}
\AXC{$\bone \vDash^{X} \bone'$}
\RL{$\TN$}
\TIC{$\Gamma, x:\FF s\vdash^{X}\nu a.t : \bone \Pto \B C^{s}\B C^{\vec q}\sigma$}
\RL{$\TLA$}
\UIC{$\Gamma \vdash^{X}\lambda x.\nu a.t: \bone \Pto \B C^{s}\B C^{\vec q}( \FF s\To \sigma)$}
\end{prooftree}
from which we deduce:

$$
\AXC{$
\Gamma, x: \FF s \vdash^{X\cup\{a\}} t: \bone'\land \btwo\Pto \B C^{\vec q}\sigma$}
\RL{$\TLA$}
\UIC{$\Gamma \vdash^{X}\lambda x.t: \bone'\land \btwo\Pto \B C^{\vec q}(\FF s\To \sigma)$}
\AXC{$\mu(\btwo)\geq s$}
\AXC{$\bone \vDash^{X} \bone'$}
\RL{$\TN$}
\TIC{$\Gamma \vdash^{X}\nu a.\lambda x.t: \bone \Pto\B C^{s}\B C^{\vec q}( \FF s\To \sigma)$}
\DP
$$

\item[($ (\nu a.t)u \redpermm \nu a.(tu)$)] 
We have

$$
\AXC{$\Gamma\vdash^{X\cup\{a\},} t: \bone'\land \btwo \Pto \B C^{\vec q}(\FF s\To\sigma)$}
\AXC{$\mu(\btwo)\geq s$}
\AXC{$\bone''\vDash \bone'$}
\RL{($\mu$)}
\TIC{$\Gamma\vdash^{X}\nu a. t: \bone'' \Pto \B C^{s}\B C^{\vec q}(\FF s\To\sigma)$}
\AXC{$ \Gamma \vdash^{X} u: \bthree\Pto \FF s$}
\AXC{$\bone\vDash \bone''\land \bthree$}
\RL{($@$)}
\TIC{$\Gamma\vdash^{X}(\nu a.t)u: \bone \Pto \B C^{s}\B C^{\vec q}\sigma$}
\DP
$$


from which we deduce

$$
\AXC{$\Gamma\vdash^{X\cup\{a\},} t: \bone'\land \btwo \Pto \B C^{\vec q}(\FF s\To\sigma)$}
\AXC{$ \Gamma \vdash^{X} u: \bthree\Pto \FF s$}
\AXC{$\bone'\land \btwo\vDash (\bone'\land \btwo)\land \bthree$}
\RL{($@$)}
\TIC{$\Gamma\vdash^{X}tu: \bone'\land \btwo \Pto \B C^{\vec q}\sigma$}
\AXC{$\mu(\btwo)\geq s$}
\AXC{$\bone''\vDash \bone'$}
\RL{($\mu$)}
\TIC{$\Gamma\vdash^{X}\nu a.tu: \bone \Pto \B C^{s}\B C^{\vec q}\sigma$}
\DP
$$

\item[($ \{t\}\nu a.u \redpermm \nu a.(tu)$)] 
We have 
$$
\AXC{$\Gamma\vdash^{X} t: \bone'\land \btwo \Pto \B C^{\vec q}(\FF s\To\sigma)$}
\AXC{$ \Gamma \vdash^{X\cup\{a\}} u: \btwo' \land \bthree\Pto \FF s$}
\AXC{$\mu(\bthree)\geq s$}
\AXC{$\btwo'\vDash \btwo$}
\RL{($\mu$)}
\TIC{$\Gamma\vdash^{X}\nu a. u: \btwo \Pto \B C^{s}\FF s$}
\AXC{$\bone\vDash \bone'\land \btwo$}
\RL{($\{\}$)}
\TIC{$\Gamma\vdash^{X}\{t\}\nu a.u: \bone \Pto \B C^{s}\B C^{\vec q}\sigma$}
\DP
$$

from which we deduce (using Lemma \ref{lemma:weak2})

\begin{center}
\adjustbox{scale=0.92}{
$
\AXC{$\Gamma\vdash^{X\cup\{a\}} t: \bone'\land \btwo \Pto \B C^{\vec q}(\FF s\To\sigma)$}
\AXC{$ \Gamma \vdash^{X\cup\{a\}} u: \btwo' \land \bthree\Pto \FF s$}
\AXC{$\bone'\land \bthree\vDash \bone'\land (\btwo\land \bthree)$}
\RL{($@$)}
\TIC{$\Gamma\vdash^{X}tu: \bone'\land \bthree  \Pto \B C^{\vec q}\sigma$}
\AXC{$\mu(\bthree)\geq s$}
\AXC{$\bone\vDash \bone'$}
\RL{($\mu$)}
\TIC{$\Gamma\vdash^{X}\nu a. tu: \btwo \Pto \B C^{s}\B C^{\vec q}\sigma$}
\DP
$}
\end{center}

%

%
%
%

\end{description}
\end{proof}

\section{Details about $\TCINT$.}

\subsection{Subject Reduction.}

In this subsection we show that subject reduction holds for $\TCINT$ (and a fortiori for $\TCI$).
\begin{proposition}[Subject Reduction]\label{prop:subjectb}
If $\Gamma\vdash^{X}t:\bone \Pto \FF s$ and $t\redall u$, then
$\Gamma\vdash^{X}u:\bone \Pto \FF s$.
\end{proposition}

As for $\LCPL$ (see Section \ref{sectionC}), we replace the rule ($\oplus$) by the two rules ($\oplus l$) and ($\oplus r$). 
Moreover, we will ignore the rules ($\HNORM$) and ($\NORM$), as the result extends immediately to them.

To show the reduction property of $\TCINT$ we need to establish a few lemmas, some of which are analogous to results for $\LCPL$, and proved in a similar way:

\begin{lemma}\label{lemma:weak2b}
If $\Gamma \vdash^{X}t: \bone \Pto \FF s$ and $X\subseteq Y$, then 
$\Gamma \vdash^{Y}t: \bone\Pto \FF s$.
\end{lemma}

\begin{lemma}\label{lemma:weakb}
If $\Gamma \vdash^{X}t: \bone\Pto \FF s$ holds and $\btwo\vDash^{X} \bone$, then  
$\Gamma \vdash^{X}t: \btwo\Pto \FF s$ is derivable by a derivation of the same length.
\end{lemma}

The next lemmas are new:

\begin{lemma}\label{lemma:order}
The following rule is admissible in $\TCINT$:
$$
\AXC{$\Gamma \vdash^{X}t: \bone \Pto \FF s$}
\AXC{$\FF s \preceq \FF t$}
\RL{$(\preceq)$}
\BIC{$\Gamma \vdash^{X}t: \bone \Pto \FF t$}
\DP
$$
\end{lemma}
\begin{proof}
We will show the admissibility of a more general rule, namely
$$\AXC{$\Gamma \vdash^{X}t: \bone \Pto \FF t$}
\AXC{$\Delta \preceq \Gamma,\FF s \preceq \FF t$}
\RL{$(\preceq^{*})$}
\BIC{$\Delta \vdash^{X}t: \bone \Pto \FF t$}
\DP
$$
where $\Delta \preceq \Gamma$ holds when $\Gamma=\{x_{1}:\FF M_{1},\dots, x_{n}:\FF M_{n}\}$, 
$\Delta=\{x_{1}:\FF N_{1},\dots, x_{n}:\FF N_{n}\}$ and $\FF M_{i}\preceq\FF N_{i}$.

We argue by induction on a typing derivation $D$ of $t$:
\begin{varitemize}
\item if $D$ is 
$$
\AXC{$\mathrm{FN}(\bone)\subseteq X$}
\AXC{$\FF s_{i}\preceq \FF t$}
\BIC{$\Gamma, x: [\FF s_{1},\dots, \FF s_{n}]\vdash^{X}x:\bone \Pto \FF t$}
\DP 
$$
then from $\Delta \preceq \Gamma, x: [\FF s_{1},\dots, \FF s_{n}]$ we deduce that $\Delta $ contains $ x: \FF M$, with $\FF M\preceq  [\FF s_{1},\dots, \FF s_{n}]$. This implies that $\FF M$ contains $\FF u_{i}$, where $\FF u_{i}\preceq \FF s_{i}$; by transitivity of $\preceq$, from $\FF u_{i}\preceq\FF s_{i}\preceq \FF t$, we deduce $\FF u_{i}\preceq \FF t$, and thus we can construct the derivation below
$$
\AXC{$\mathrm{FN}(\bone)\subseteq X$}
\AXC{$\FF u_{i}\preceq \FF t$}
\BIC{$\Delta\vdash^{X}x:\bone \Pto \FF t$}
\DP 
$$

\item if $D$ ends by any of the rules $(\lor)$, $(\oplus l)$, $(\oplus r)$ or $(\mu_{\Sigma})$, then we can directly conclude by applying the I.H.

\item if $D$ is 
$$
\AXC{$\vdots$}
\noLine
\UIC{$\Gamma, x: \FF M \vdash^{X}t: \bone \Pto \B C^{q}\sigma$}
\RL{$(\lambda)$}
\UIC{$\Gamma \vdash^{X}\lambda x.t: \bone \Pto\B C^{q} (\FF M\To \sigma)$}
\DP
$$
then from $\B C^{q}(\FF M\To \sigma)  \preceq \FF t$ we deduce $\FF t=\B C^{s}(\FF M'\To \sigma')$, with $s\leq q$, 
$\FF M\preceq \FF M'$ and $\sigma\preceq \sigma''$, and from $\Delta \preceq \Gamma$, we deduce $\Delta, x:\FF M'\preceq \Gamma,x:\FF M$. So by the I.H. we deduce the existence of a derivation of $\Delta, x:\FF M'\vdash^{X}t: \bone \Pto \B C^{s}\sigma'$ and we can conclude by applying an instance of $(\lambda)$.

\item if $D$ is 

\begingroup\makeatletter\def\f@size{10}\check@mathfonts
\begin{lrbox}{\mypti}
\begin{varwidth}{\linewidth}
$
\AXC{$\vdots$}
\noLine
\UIC{$\Gamma \vdash^{X}u: \bthree_{i}\Pto \FF s_{i}$}
\DP
$
\end{varwidth}
\end{lrbox}

$$
\AXC{$\vdots$}
\noLine
\UIC{$\Gamma \vdash^{X}t: \btwo \Pto\B C^{q} (\FF M\To \sigma)$}
\AXC{$\left \{ \usebox{\mypti}\right\}_{i}$}
\AXC{$\bone \vDash \btwo \land \bigwedge_{i}\bthree_{i}$}
\RL{$@_{\cap}$}
\TIC{$\Gamma \vdash^{X}tu: \bone \Pto \B C^{q}\sigma$}
\DP
$$
\endgroup

where $\FF M=[\FF s_{1},\dots, \FF s_{n}]$.
Then from $\B C^{q}\sigma \preceq \FF s$ we deduce $\tau= \B C^{s}\sigma'$, with $s\leq q$ and $\sigma \preceq \sigma'$, and thus $\B C^{q}(\FF M \To \sigma)\preceq \B C^{s}(\FF M\To \sigma')$, 
 so by the I.H.~applied to the left-hand side sub-derivation we obtain a the existence of a derivation of $\Gamma \vdash^{X}t: \btwo \Pto\B C^{s} (\FF M\To \sigma')$, and thus we can conclude by applying an instance of $(@_{\cap})$.
\end{varitemize}
\end{proof}

\begin{lemma}\label{lemma:disj}
Let $\bone, \bone_{1},\dots, \bone_{n}$ and $\btwo,\btwo_{1},\dots, \btwo_{n}$ be such that $\FN({\bone}), \FN({\bone_{i}})\subseteq{X}$ and $\FN(\btwo), \FN({\btwo_{i}})\subseteq\{a\}$, where $a\notin X$.
If $\btwo$ is satisfiable, then if $\bone\land \btwo \vDash^{X\cup \{a\}}\bigvee_{i}^{n} \bone_{i}\land \btwo_{i}$ holds, 
also $\bone\vDash^{X}\bigvee_{i}^{n}\bone_{i}$ holds.

\end{lemma}
\begin{proof}
Let $v\in 2^{X}$ be a model of $\bone$. Since $\btwo$ is satisfiable, $v$ can be extended to a model $v'\in 2^{X\cup\{a\}}$ of $\bone \land \btwo$. By hypothesis, then $v'$ satisfies $\bigvee_{i}^{n} \bone_{i}\land \btwo_{i}$, so for some $i_{0}\leq n$, it satisfies $\bone_{i_{0}}\land \btwo_{i_{0}}$. We deduce then that $v$ satisfies $\bone_{i_{0}}$, and thus $v$ satisfies $\bigvee_{i}^{n}\bone_{i}$.
\end{proof}

\begin{lemma}\label{lemma:nota}

If $\Gamma \vdash^{X\cup Y}t: \bone \land \btwo \Pto \FF s$ is derivable, where $X\cap Y=\emptyset$, $\mathrm{FN}(t)\subseteq X$, 
$\mathrm{FN}(\bone)\subseteq X$, $\mathrm{FN}(\btwo)\subseteq Y$, and $\btwo$ is satisfiable, then $\Gamma \vdash^{X}t: \bone  \Pto \FF s$ is also derivable.
\end{lemma}
\begin{proof}
By induction on a typing derivation of $t$:
\begin{varitemize}
%
%
%

\item if the last rule is

\begin{prooftree}
\AXC{$\FF s_{i}\preceq \FF t$}
\AXC{$\FN(\bone\land \btwo)\subseteq X\cup Y$}
\RL{(id$_{\cap}$)}
\BIC{$\Gamma , x:[\FF s_{1},\dots, \FF s_{n}]\vdash^{X\cup\{a\}}x: \bone \land \btwo\Pto \FF t$}
\end{prooftree}
then the claim can be deduced by an instance of the same rule.

\item if the last rule is
\begin{prooftree}
\AXC{$\left\{\Gamma \vdash^{X\cup Y}t: \bone_{i} \Pto \FF s\right\}_{i}$}
\AXC{$\bone\land \btwo \vDash^{X\cup Y}\bigvee_{i} \bone_{i}$}
\RL{$(\lor)$}
\BIC{$\Gamma \vdash^{X\cup Y}x: \bone\land \btwo \Pto \FF s$}
\end{prooftree}
Let $\bigvee_{j}\bone_{ij}\land \bthree_{ij}$ be weak $Y$-decompositions of the $\bone_{i}$ (see [\cite{AntonelliArXiv}, p.~17]); 
since $\bone_{ij}\land \bthree_{ij}\vDash \bone_{i}$, by Lemma \ref{lemma:weakb} we deduce 
that $\Gamma \vdash^{X\cup Y}t: \bone_{ij}\land \bthree_{ij} \Pto \FF s$ is derivable for all $i$ and $j$, by a derivation of same length as the corresponding derivation of $\Gamma \vdash^{X\cup Y}t: \bone_{i} \Pto \FF s$; 
hence we can apply the I.H.~to such derivations, yielding derivations of 
 $\Gamma \vdash^{X}t: \bone_{ij}\Pto \FF s$. 

Using Lemma \ref{lemma:disj}, from $\bone \land \btwo \vDash \bigvee_{i}\bone_{i}$ and $\bigvee_{i}\bone_{i}\equiv \bigvee_{ij}\bone_{ij}\land \btwo_{ij}$, we deduce
$\bone \vDash \bigvee_{ij}\bone_{ij}$, so we conclude 
\begin{prooftree}
\AXC{$\left\{\Gamma \vdash^{X}t: \bone_{ij} \Pto \FF s\right\}_{ij}$}
\AXC{$\bone \vDash^{X}\bigvee_{ij} \bone_{ij}$}
\RL{$(\lor)$}
\BIC{$\Gamma \vdash^{X}x: \bone \Pto \FF s$}
\end{prooftree}
%
%
%
%

\item if the last rule is 

\begin{prooftree}
\AXC{$\Gamma\vdash^{X\cup Y} t_{1}:\bone'\Pto \FF s$}
\AXC{$\bone\land \btwo\vDash\bvar_{i}^{b}\land \bone'
$}
\RL{$\TL$}
\BIC{$\Gamma\vdash^{X\cup Y} t_{1}\oplus^{i}_{b}t_{2}:  \bone\land \btwo\Pto \FF s$}
\end{prooftree}
Then, let $\bigvee_{j}\bthree_{j}\land \btwo_{j}$ be a $Y$-decomposition of $\bone'$;
by Lemma \ref{lemma:weakb} there exist derivations of $\Gamma \vdash^{X\cup Y}t_{1}: \bthree_{j}\land \btwo_{j}\Pto \FF s$ of same length as the derivation $\Gamma\vdash^{X\cup Y} t_{1}:\bone'\Pto \FF s$, so by I.H.~we obtain derivations of
$\Gamma \vdash^{X\cup Y}t_{1}: \bthree_{j}\Pto \FF s$.

From $\bone\land \btwo\vDash\bvar_{i}^{b}\land \bone'$ we deduce
$\bone\land \btwo\vDash\bigvee_{j}(\bvar_{i}^{b}\land \bthree_{j} )\land \btwo_{j}$, and
using Lemma \ref{lemma:disj} we deduce
$\bone\vDash \bigvee_{j}\bvar_{i}^{b}\land \bthree_{j}\equiv \bvar_{i}^{b}\land \bigvee_{j}\bthree_{j}$. We can thus  conclude as follows:

$$
\AXC{$\left \{ \Gamma \vdash^{X}t_{1}: \bthree_{j}\Pto \FF s\right\}_{j}$}
\RL{$(\lor)$}
\UIC{$\Gamma \vdash^{X}t_{1}: \bigvee_{j}\bthree_{j}\Pto \FF s$}
\AXC{$\bone\vDash^{X}\bvar_{i}^{b}\land \bigvee_{j}\bthree_{j}
$}
\RL{$\TL$}
\BIC{$\Gamma\vdash^{X} t_{1}\oplus^{i}_{a}t_{2}:  \bone\Pto \FF s$}
\DP
$$

%
%

\item if the last rule is 

\begin{prooftree}
\AXC{$\Gamma \vdash^{X\cup Y} t_{2}:\bone'\Pto \FF s$}
\AXC{$\bone\land \btwo\vDash\lnot\bvar_{i}^{b}\land \bone'$}
\RL{$\TR$}
\BIC{$\Gamma \vdash^{X\cup Y} t_{1}\oplus^{i}_{b}t_{2}:  \bone\land \btwo\Pto \FF s$}
\end{prooftree}

then we can argue similarly to the previous case.

\item if the last rule is 

\begin{prooftree}
\AXC{$\Gamma, y:\FF M \vdash^{X\cup Y} t: \bone\land\btwo\Pto \B C^{q}\sigma$}
\RL{$\TLA$}
\UIC{$\Gamma \vdash^{X\cup Y} \lambda y.t:\bone\land\btwo\Pto  \B C^{q}(\FF M\To\sigma)$}
\end{prooftree}
Then, the claim follows from the I.H.~by applying an instance of the same rule.

\item if the last rule is 

\begin{prooftree}
\AXC{$\Gamma \vdash^{X\cup Y} t_{1}: \bone_{1}\Pto \B C^{q}(\FF M\To \sigma)$}
\AXC{$\Big\{\Gamma \vdash^{X\cup Y} t_{2}: \bone_{i} \Pto \FF s_{i}\Big\}_{i}$}
\AXC{$\bone\land \btwo\vDash^{X\cup\{a\}} \bone_{1}\land\left ( \bigwedge_{i} \bone_{i}\right)$}
\RL{($@_{\cap}$)}
\TIC{$\Gamma\vdash^{X\cup Y} t_{1}t_{2}: \bone\land \btwo \Pto \B C^{q}\sigma$}
\end{prooftree}
Then 
let $\bigvee_{j}\btwo'_{j}\land \bthree'_{j}$ and $\bigvee_{k}\btwo''_{ik}\land \bthree''_{ik}$ be weak $Y$-decompositions of $\bone_{1}$ and $\bone_{i}$.
By Lemma \ref{lemma:weakb} there exist then derivations of 
$\Gamma \vdash^{X\cup Y} t_{1}: \btwo'_{j}\land \bthree'_{j}\Pto \B C^{q}(\FF M\To \sigma)$ and
$\Gamma \vdash^{X\cup Y} t_{2}: \btwo''_{ik}\land \bthree''_{ik} \Pto \FF s_{i}$ of same length as the corresponding derivations. 

By applying the I.H.~to such derivations we obtain then derivations of 
$\Gamma \vdash^{X} t_{1}: \btwo'_{j}\Pto\B C^{q}(\FF M\To \sigma)$ and
$\Gamma \vdash^{X} t_{2}: \btwo''_{ik} \Pto \FF s_{i}$, respectively.

Moreover, from $\bone \land \btwo \vDash \bone_{1}\land\bigwedge_{i} \bone_{i}$ we deduce 
$\bone \land \btwo \vDash \bigvee_{j}\btwo'_{j}\land \bthree'_{j}$ and 
$\bone \land \btwo \vDash \bigvee_{ik}\btwo''_{ik}\land \bthree''_{ik}$, so by Lemma \ref{lemma:disj} we deduce
$\bone \vDash \bigvee_{j}\btwo'_{j}$ and 
$\bone  \vDash\bigwedge_{i} \bigvee_{ik}\btwo''_{ik}$.

Thus, we can conclude as follows: 
 
\begin{lrbox}{\mypti}
\begin{varwidth}{\linewidth}
$
\AXC{$\left\{ 
\Gamma \vdash^{X}t_{2}: \btwo''_{ik}\Pto \FF s_{i}
\right\}_{k}$}
\RL{($\lor$)}
\UIC{$\Gamma \vdash^{X}t_{2}:\bigvee_{ik} \btwo''_{ik}\Pto \FF s_{i}$}
\DP
$
\end{varwidth}
\end{lrbox}

\begin{center}
\resizebox{0.95\textwidth}{!}{
$\AXC{$\left\{ 
\Gamma \vdash^{X}t_{1}: \btwo'_{j}\Pto \B C^{q}(\FF M\To \sigma)
\right\}_{j}$}
\RL{($\lor$)}
\UIC{$\Gamma \vdash^{X}t:\bigvee_{j} \btwo'_{j}\Pto \B C^{q}(\FF M\To \sigma)$}
\AXC{$\left\{ \usebox{\mypti}\right\}_{i}$}
\AXC{$\bone\vDash^{X}\left (\bigvee_{j} \btwo'_{j}\right)\land \left(\bigwedge_{i}\bigvee_{ik}\btwo''_{ik}\right)$}
\RL{$\TA$}
\TIC{$\Gamma\vdash^{X} t_{1}t_{2}: \bone \Pto \B C^{q}\sigma$}
\DP
$}
\end{center}

%
%
%
%

\item if the last rule is

\begin{prooftree}
\AXC{$\left\{\Gamma \vdash^{X\cup Y\cup\{b\},q_{u}} t:\bone'\land \mathscr f_{u}\Pto \B C^{q_{u}}\sigma\right\}_{u}$}
\AXC{$\mu(\mathscr f_{u})\geq s_{u}$}
\AXC{$\bone\land \btwo \vDash \bone'$}
	\RL{($\mu_{\Sigma}$)}
\TIC{$\Gamma \vdash^{X\cup Y} \nu b.t: \bone\land \btwo\Pto \B C^{r} \sigma$}
\end{prooftree}

where $r=\sum_{u}q_{u}s_{u}$, then let $ \bigvee_{i}\btwo_{i}\land \bthree_{i}$, be a weak $a$- decomposition of $\bone$. 
By Lemma \ref{lemma:weak} there exist derivations of 
$\Gamma \vdash^{X\cup Y\cup\{b\}} t: \btwo_{i}\land \bthree_{i}\land \mathscr f_{u}\Pto \B C^{q_{u}}\sigma$ of same length as the corresponding derivation of $\Gamma \vdash^{X\cup Y\cup\{b\}} t:\bone'\land \mathscr f_{u}\Pto \B C^{q_{u}}\sigma$. 
Hence, by the I.H.~there exist derivations of $\Gamma \vdash^{X\cup\{b\}} t:\btwo_{i}\land \mathscr f_{u}\Pto \B C^{q_{u}}\sigma$. 

From $\bone \land \btwo\vDash^{X\cup\{a\}} \bone'$ we deduce by Lemma \ref{lemma:disj} $\bone \vDash \bigvee_{i}\btwo_{i}$ and so we can conclude as follows:

\begingroup\makeatletter\def\f@size{10}\check@mathfonts
\begin{lrbox}{\mypti}
\begin{varwidth}{\linewidth}
$
\AXC{$\left\{\Gamma \vdash^{X\cup \{b\}} t:\btwo_{i}\land \mathscr f_{u}\Pto \B C^{q_{u}}\sigma\right\}_{i}$}
\RL{$(\lor)$}
\UIC{$\Gamma \vdash^{X\cup\{b\}} t:\bigvee_{i}\btwo_{i}\land \mathscr f_{u}\Pto \B C^{q_{u}}\sigma$}
\DP
$
\end{varwidth}
\end{lrbox}

$$
\AXC{$\left\{ \usebox{\mypti} \right\}_{u}$}
\AXC{$\mu(\mathscr f_{u})\geq s_{u}$}
\AXC{$\bone \vDash \bigvee_{i}\btwo_{i}$}
	\RL{($\mu_{\Sigma}$)}
\TIC{$\Gamma \vdash^{X} \nu b.t: \bone\Pto\B C^{r} \sigma$}
\DP
$$
\endgroup

%
%
%
%
\end{varitemize}
\end{proof}

The proof of the substitution lemma below is analogous to the proof of Lemma \ref{lemma:betasub}.
\begin{lemma}[substitution lemma]\label{lemma:betasubb}
The following rule is derivable:
$$
\AXC{$\Gamma, x: [\FF s_{1},\dots, \FF s_{n}] \vdash^{X} t: \btwo \Pto  \B C^{q}\tau$}
\AXC{$\{\Gamma \vdash^{X} u: \bthree_{i} \Pto \FF s_{i}\}_{i=1,\dots, n}$}
\AXC{$\bone\vDash\btwo \land \left(\bigwedge_{i}\bthree_{i}\right)$}
\RL{(subst$_{\cap}$)}
\TIC{$\Gamma \vdash^{X} t[u/x]: \bone \Pto \B C^{q}\tau$}
\DP
$$
\end{lemma}

We now have all elements to establish the subject reduction property.

\begin{proof}[Proof of Proposition \ref{prop:subjectb}]
As in the proof of Prop.~\ref{prop:subject}, observe that if the typing derivation $D$ of $t$ ends by a $(\lor)$-rule, it suffices to establish the property for the immediate sub-derivations of $D$ and then apply an instance of $(\lor)$-rule to the resulting derivations. So we will always suppose that the typing derivation of $D$ does not end by a $(\lor)$-rule.

The argument for $\beta$-reduction works similarly to the one in the Prop.~\ref{prop:subject}, using Lemma \ref{lemma:betasubb} in place of Lemma \ref{lemma:betasub}. (subst$_{\cap}$) %

For the case of $\redperm$ wee consider reduction rules one by one. Most cases are analogous to those from the proof of Prop.~\ref{prop:subject}. We limit ourselves to the case of the permutation rule \eqref{eq:notnu} (not considered before), and to 
those permutations that involve the rule ($\mu_{\Sigma}$):    
\begin{description}

\item[{($\nu b. (t\oplus_{a}^{i}u)\redperm (\nu b.t)\oplus_{a}^{i}(\nu b.u)$)}] 
We suppose $a\neq b$. As in the proof of Prop.~\ref{prop:subject} there are two sub-cases, we only consider the first one (the second one being treated similarly):

\begingroup\makeatletter\def\f@size{10}\check@mathfonts
\begin{lrbox}{\mypti}
\begin{varwidth}{\linewidth}
$
\AXC{$\Gamma \vdash^{X\cup\{a,b\}}t: \bthree_{i}\Pto  \B C^{q_{i}} \sigma$}
\AXC{$\bone'\land \btwo_{i}\vDash \bvar_{a}^{i}\land \bthree_{i}$}
\RL{($\oplus l$)}
\BIC{$\Gamma \vdash^{X\cup\{a,b\}}t\oplus_{a}^{i}u: \bone' \land \btwo_{i}\Pto  \B C^{q_{i}} \sigma$}
\DP
$
\end{varwidth}
\end{lrbox}

	$$
	\AXC{$\left\{ \usebox{\mypti}\right\}_{i}$}
	\AXC{$\mu(\btwo_{i})\geq s_{i}$}
	\AXC{$\bone \vDash \bone'$}
	\RL{($\mu_{\Sigma}$)}
	\TIC{$\Gamma \vdash^{X\cup\{a\}}\nu b.t\oplus_{a}^{i}u \Pto \B C^{\sum_{i}q_{i}s_{i}} \sigma$}
	\DP	
	$$
	\endgroup
	
From $\bone'\land \btwo_{i}\vDash  \bthree_{i}$, by Lemma \ref{lemma:weak}, we deduce the existence of a derivation of $\Gamma \vdash^{X\cup\{a,b\}}t: \bone'\land \btwo_{i}\Pto  \B C^{q_{i}}\sigma$. Moreover, from $\bone'\land \btwo_{i}\vDash \bvar_{a}^{i}\land \bthree_{i}$ and the fact that $\mathrm{FN}(\btwo_{i})\subseteq \{b\}$ it follows that	$\bone'\vDash \bvar_{a}^{i}$, so we can construct the following derivation:

$$
\AXC{$\left\{\Gamma \vdash^{X\cup\{a,b\}}t: \bone'\land\btwo_{i}\Pto  \B C^{q_{i}}\sigma\right\}_{i}$}
	\AXC{$\mu(\btwo_{i})\geq s_{i}$}
	\RL{($\mu_{\Sigma}$)}
	\BIC{$\Gamma \vdash^{X\cup\{a\}}\nu b.t: \bone' \Pto\B C^{\sum_{i}q_{i}s_{i}} \sigma$}
	\AXC{$\bone \vDash \bvar_{a}^{i}\land \bone'$}
	\RL{($\oplus l$)}
\BIC{$\Gamma \vdash^{X\cup\{a\}}(\nu b.t)\oplus_{a}^{i}(\nu b.u): \bone \Pto \B C^{\sum_{i}q_{i}s_{i}}\sigma$}
\DP
$$

\item[($\nu a.t \redperm t$)] 
We have 
\begin{prooftree}
\AXC{$\left\{\Gamma \vdash^{X\cup\{a\}}t: \bone'\land \btwo_{i}\Pto\B C^{q_{i}} \sigma\right\}_{i}$}
\AXC{$\mu(\model{\btwo_{i}}_{\{a\}})\geq s_{i}$}
\AXC{$\bone \vDash^{X} \bone'$}
\RL{($\mu_{\Sigma}$)}
\TIC{$\Gamma\vdash^{X}\nu a.t:\bone \Pto\B C^{\sum_{i}q_{i}s_{i}} \sigma$}
\end{prooftree}
Since the $\btwo_{i}$ are pairwise disjoint, from $\mu({\btwo_{i}})\geq s_{i}$ it follows that $\sum_{i}s_{i}\leq 1$. Hence, by letting $Q=\max_{i}\{q_{i}\}$, we have 
 $\sum_{i}q_{i}s_{i}\leq \sum_{i}Qs_{i}=Q\cdot \sum_{i}s_{i}\leq Q$. 
Let $i_{0}$ be the index such that $Q=q_{i_{0}}$. 
Since $a\notin \mathrm{FN}(t)$, by applying Lemma \ref{lemma:nota}
to the derivation of $\Gamma \vdash^{X\cup\{a\}}t: \bone'\land \btwo_{i_{0}}\Pto\B C^{Q} \sigma$,  
 we deduce $\Gamma \vdash^{X}t:\bone' \Pto\B C^{Q} \sigma$, and 
 since $ \B C^{Q}\sigma\preceq \B C^{\sum_{i}q_{i}s_{i}}\sigma$ and, we deduce, by Lemma \ref{lemma:order},
 $\Gamma \vdash^{X}t:\bone' \Pto\B C^{\sum_{i}q_{i}s_{i}} \sigma$, and by Lemma \ref{lemma:weakb},
  $\Gamma \vdash^{X}t:\bone \Pto\B C^{\sum_{i}q_{i}s_{i}} \sigma$.

\item[($\lambda x.\nu a.t\redperm \nu a.\lambda x.t$ )]  
We have

\begin{prooftree}
\AXC{$\left\{\Gamma, x: \FF M \vdash^{X\cup\{a\}} t: \bone'\land \btwo_{i} \Pto \B C^{q_{i}}\sigma\right\}_{i}$}
\AXC{$\mu({\btwo_{i}})\geq s_{i}$}
\AXC{$\bone \vDash^{X} \bone'$}
\RL{($\mu_{\Sigma}$)}
\TIC{$\Gamma, x:\FF M\vdash^{X}\nu a.t : \bone \Pto \B C^{\sum_{i}q_{i}s_{i}}\sigma$}
\RL{$\TLA$}
\UIC{$\Gamma \vdash^{X}\lambda x.\nu a.t: \bone \Pto\B C^{\sum_{i}q_{i}s_{i}}( \FF M\To \sigma)$}
\end{prooftree}
from which we deduce:

\begingroup\makeatletter\def\f@size{10}\check@mathfonts
\begin{lrbox}{\mypti}
\begin{varwidth}{\linewidth}
$
\AXC{$
\Gamma, x: \FF M \vdash^{X\cup\{a\}} t: \bone'\land \btwo_{i} \Pto\B C^{q_{i}} \sigma$}
\RL{$\TLA$}
\UIC{$\Gamma \vdash^{X}\lambda x.t: \bone'\land \btwo_{i}\Pto \B C^{q_{i}}(\FF M\To \sigma)$}
\DP
$
\end{varwidth}
\end{lrbox}

$$\AXC{$\left\{\usebox{\mypti}\right\}_{i}$}
\AXC{$\mu({\btwo_{i}})\geq s_{i}$}
\AXC{$\bone \vDash^{X} \bone'$}
\RL{($\mu_{\Sigma}$)}
\TIC{$\Gamma \vdash^{X}\nu a.\lambda x.t: \bone \Pto\B C^{\sum_{i}q_{i}s_{i}}( \FF M\To \sigma)$}
\DP
$$\endgroup

\item[($ (\nu a.t)u \redperm \nu a.(tu)$)] 
We have
\begin{center}
%

\resizebox{0.95\textwidth}{!}{
$
\AXC{$\left\{
\Gamma\vdash^{X\cup\{a\}} t: \btwo'\land \bthree_{i} \Pto \B C^{q_{i}}(\FF M\To \tau)\right\}_{i}$}
\AXC{$\mu(\bthree_{i})\geq s_{i}$}
\AXC{$\btwo\vDash \btwo'$}
\RL{($\mu_{\Sigma}$)}
\TIC{$\Gamma\vdash^{X}\nu a. t: \btwo \Pto\B C^{\sum_{i}q_{i}s_{i}} (\FF M\To \tau)$}
\AXC{$\Big\{ \Gamma \vdash^{X} u: \bfour_{j}\Pto \FF s_{i} \Big\}_{j}$}
\AXC{$\bone\vDash \btwo \land \left (\bigwedge_{j}\bfour_{j}\right)$}
\RL{($@_{\cap}$)}
\TIC{$\Gamma\vdash^{X}(\nu a.t)u: \bone \Pto \B C^{ \sum_{i}q_{i}s_{i}}\tau$}
\DP
$}

\end{center}

where $\FF M=[\FF s_{1},\dots, \FF s_{n}]$, 
from which we deduce

\begin{center}
	
\begingroup\makeatletter\def\f@size{10}\check@mathfonts
\begin{lrbox}{\mypti}
\begin{varwidth}{\linewidth}
$
\AXC{$
\Gamma\vdash^{X\cup\{a\}} t: \btwo'\land \bthree_{i} \Pto \B C^{q_{i}}(\FF M\To \tau)$}
\AXC{$\Big\{ \Gamma \vdash^{X} u: \bfour_{j}\Pto \FF s_{i} \Big\}_{j}$}
\AXC{$\bone\land \bthree_{i}\vDash (\btwo'\land \bthree_{i}) \land \left (\bigwedge_{j}\bfour_{j}\right)$}
\RL{($@_{\cap}$)}
\TIC{$\Gamma\vdash^{X}tu: \bone\land \bthree_{i} \Pto \B C^{ q_{i}}\tau$}
\DP
$
\end{varwidth}
\end{lrbox}

%

\resizebox{0.95\textwidth}{!}{
$\AXC{$\left\{\usebox{\mypti}\right\}_{i}$}
\AXC{$\mu(\bthree_{i})\geq s_{i}$}
\RL{($\mu_{\Sigma}$)}
\BIC{$\Gamma\vdash^{X}\nu a. tu: \bone \Pto\B C^{\sum_{i}q_{i}s_{i}}\tau$}
\DP
$
}
\endgroup

\end{center}

\end{description}
\end{proof}

\subsection{Subject Expansion.}

The goal of this section is to establish the following result for $\TCINT$:

\begin{proposition}[subject expansion]\label{thm:subjexp}
If $\Gamma \vdash^{X}t:\bone\Pto \FF s$ and $u\redall t$, then $\Gamma \vdash^{X}t:\bone\Pto \FF s$.
\end{proposition}

As for subject reduction, we will ignore the rules ($\HNORM$) and ($\NORM$), as the result extends immediately to them.

We will consider first the (more laborious case) of $\beta$-reduction, and then the more direct case of permutative reduction.

\subsubsection{Subject $\beta$-Expansion}

To establish the subject expansion property we need to develop a finer analysis of typing derivations. In particular we will need to prove that derivations can be turned into a canonical form. The main idea is to reduce to a case where $(\lor)$-rules are always applied ``as late as possible'', that is, either at the end of the derivation or right before a counting rule. 
Moreover, we will require that if a $(\lor)$-rule occurs before a counting rule for some name $a$, then these two rule follow a specific pattern, defined below, which ensures that no information is loss about other names. 


\begin{definition}[$\lor\mu$-pattern]
For any name $a\in \BB N$ and typing derivation $D$, a \emph{$\vee\mu$-pattern on $a$ in $D$} is given by the occurrence in $D$ of $(\mu)$-rule preceded by multiple occurrences of $(\vee)$-rules in a configuration as illustrated below:

\begingroup\makeatletter\def\f@size{10}\check@mathfonts
\begin{lrbox}{\mypti}
\begin{varwidth}{\linewidth}
$
\AXC{$
\Big\{ \Gamma \vdash^{X\cup \{a\}}t: \bone\land \bthree_{ij}\Pto \B C^{q_{i}}\sigma\Big\}_{j=1,\dots,L_{i}}$}
\AXC{$\bone \land \bthree_{i} \vDash^{X}\bigvee_{l}\bone\land\bthree_{ij}$}
\RL{$(\lor)$}
\BIC{$\Gamma \vdash^{X\cup \{a\}}t: \bone \land \bthree_{i}\Pto \B C^{q_{i}} \sigma$}
\DP
$
\end{varwidth}
\end{lrbox}

\begin{equation}\label{eq:vmu}
\AXC{$
\left \{ \usebox{\mypti}\right\}_{i}$}
\AXC{$\mu (\bthree_{i} )\geq s_{i}$}
\RL{$(\mu_{\Sigma})$}
\BIC{$\Gamma\vdash^{X} \nu a.t: \bone\Pto\B C^{\sum_{i}q_{i}s_{i}} \sigma$}
\DP
\end{equation}
\endgroup
where $\mathrm{FN}(\btwo)\subseteq X$, $\mathrm{FN}(\bthree_{ij})\subseteq \{a\}$, and $\btwo$ and the $\bthree_{ij}$ are conjunctions of literals.
\end{definition}

\begin{definition}[canonical and pseudo-canonical derivation]
A typing derivation $D$ is said \emph{canonical} when the following hold:
\begin{varitemize}

\item[a.] all Boolean formulas occurring in $D$ are conjunctions of literals, except for those which occur either in the conclusion of a $(\lor)$-rule or in the major premiss of a $(\mu_{\Sigma})$-rule;

\item[b.] any occurrence of the $(\lor)$ rule in $D$ occurs in a $\lor\mu$-pattern.

\end{varitemize}

A typing derivation $D$ is said \emph{pseudo-canonical} when $D$ is obtained by one application of the $(\lor)$-rule where all premises are conclusions of canonical sub-derivations.

\end{definition}
The fundamental property of canonical derivations that we will exploit is the following:

\begin{lemma}\label{lemma:canonical}
Suppose $D$ is a canonical derivation of conclusion $\Gamma \vdash^{X}t: \bone\Pto \FF s$. 
If $\btwo$ is a Boolean formula $\btwo$ with $\mathrm{FN}(\btwo)\subseteq X$ occurring in an axiom of $D$, then $\bone\vDash^{X} \btwo$.
\end{lemma}
\begin{proof}
By induction on $D$:
\begin{varitemize}
\item if $D$ only consists in an axiom, the claim is trivial;
\item if $D$ ends by any rule of the form
$$
\AXC{$\{\Delta \vdash ??: \bfour_{i}\Pto ?? \}_{i}$}
\AXC{other conditions}
\BIC{$\Delta\vdash ??: \bone\Pto ??$}
\DP
$$ 
where $\bone\vDash \bfour_{i}$ (this includes the rules for $\oplus_{a}^{i}$, $\lambda$ and application), then from the induction hypothesis $\bfour_{i}\vDash \btwo$ holds, where $i$ is the index of the derivation where $\btwo$ occurs, and we can thus conclude $\bone\vDash \btwo$.

\item if $D$ ends by a $\lor\mu$-pattern as in Eq.~\eqref{eq:vmu}, then 
by the induction hypothesis
$\bone \land \bthree_{ij} \vDash \btwo$, where $i$ and $j$ are chosen so that $\btwo$ occurs in the associated sub-derivation. Since $\mathrm{FN}(\bthree_{ij})\cap \mathrm{FN}(\btwo)=\emptyset$ and the $\bthree_{ij}$ are all satisfiable, we can conclude $\bone\vDash \btwo$ by Lemma \ref{lemma:disj}.
\end{varitemize}
\end{proof}

\begin{remark}
It is not difficult to see that the property of Lemma \ref{lemma:canonical} fails for non-canonical derivations. For example, consider the non-canonical derivation below:

\begin{center}
\resizebox{\textwidth}{!}{
$
\AXC{$x: \B C^{q}\sigma \vdash^{\{a,b\}} x: \bvar_{a}^{1}\Pto \B C^{q}\sigma$}
\RL{$(\oplus l)$}
\UIC{$x: \B C^{q}\sigma \vdash^{\{a,b\}} x\oplus_{b}^{0} x:\bvar_{b}^{0}\land \bvar_{a}^{1}\Pto \B C^{q}\sigma$}
\AXC{$\mu(\bvar_{b}^{0})\geq \frac{1}{2}$}
\RL{$(\mu)$}
\BIC{$x: \B C^{q}\sigma \vdash^{\{a,b\}} \nu b.x\oplus_{b}^{0} x: \bvar_{a}^{1}\Pto \B C^{\frac{q}{2}}\sigma$}
\AXC{$x: \B C^{q}\sigma \vdash^{\{a,b\}} x: \lnot\bvar_{a}^{1}\Pto \B C^{q}\sigma$}
\RL{$(\oplus r)$}
\UIC{$x: \B C^{q}\sigma \vdash^{\{a,b\}} x\oplus_{b}^{0} x:\lnot\bvar_{b}^{0}\land \lnot \bvar_{a}^{1}\Pto \B C^{q}\sigma$}
\AXC{$\mu(\lnot\bvar_{b}^{0})\geq \frac{1}{2}$}
\RL{$(\mu)$}
\BIC{$x: \B C^{q}\sigma \vdash^{\{a,b\}} \nu b.x\oplus_{b}^{0} x: \lnot\bvar_{a}^{1}\Pto \B C^{\frac{q}{2}}\sigma$}
\AXC{$\TOP \vDash \bvar_{a}^{1}\land \lnot \bvar_{a}^{1}$}
\RL{$(\lor)$}
\TIC{$x: \B C^{q}\sigma \vdash^{\{a\}} \nu b.x\oplus_{b}^{0} x:\TOP\Pto \B C^{\frac{q}{2}}\sigma$}
\DP
$}
\end{center}
\noindent
where $\bthree=(\bvar_{b}^{0}\land \bvar_{a}^{1})\vee(\lnot\bvar_{b}^{0}\land\lnot \bvar_{a}^{1})$.
While the literal $\bvar_{a}^{1}$ occurs in an axiom, it is not true that $\TOP\vDash \bvar_{a}^{1}$.

Observe that the derivation above is pseudo-canonical. Instead, the derivation below is neither canonical nor pseudo-canonical, and similarly violates the property of Lemma \ref{lemma:canonical}.
\begin{center}
\resizebox{\textwidth}{!}{
$
\AXC{$x: \B C^{q}\sigma \vdash^{\{a,b\}} x: \bvar_{a}^{1}\Pto \B C^{q}\sigma$}
\RL{$(\oplus l)$}
\UIC{$x: \B C^{q}\sigma \vdash^{\{a,b\}} x\oplus_{b}^{0} x:\bvar_{b}^{0}\land \bvar_{a}^{1}\Pto \B C^{q}\sigma$}
\AXC{$x: \B C^{q}\sigma \vdash^{\{a,b\}} x: \lnot\bvar_{a}^{1}\Pto \B C^{q}\sigma$}
\RL{$(\oplus l)$}
\UIC{$x: \B C^{q}\sigma \vdash^{\{a,b\}} x\oplus_{b}^{0} x:\bvar_{b}^{0}\land \lnot \bvar_{a}^{1}\Pto \B C^{q}\sigma$}
\AXC{$ \bvar_{b}^{0}\vDash (\bvar_{a}^{1}\land \bvar_{b}^{0})\lor(\lnot \bvar_{a}^{1}\land \bvar_{b}^{0})$}
\RL{$(\lor)$}
\TIC{$
x: \B C^{q}\sigma \vdash^{\{a,b\}} x\oplus_{b}^{0} x:\bvar_{b}^{0}\Pto \B C^{q}\sigma$}
\AXC{$\mu(\bvar_{b}^{0})\geq \frac{1}{2}$}
\RL{$(\mu)$}
\BIC{$x: \B C^{q}\sigma \vdash^{\{a\}} \nu b.x\oplus_{b}^{0} x:\TOP\Pto \B C^{\frac{q}{2}}\sigma$}
\DP
$}
\end{center}
\end{remark}

The following is the fundamental structural result we need to establish subject $\beta$-expansion.

\begin{theorem}\label{conje}
Any typing derivation can be transformed into a pseudo-canonical derivation.
\end{theorem}


We need two preliminary lemmas.

\begin{lemma}\label{lemma:cool1}
Suppose $D$ is pseudo-canonical and ends as follows:
$$
\AXC{$\Big\{\Gamma \vdash^{X\cup\{a\}}t: \bone_{i}\land \bfour_{i}\Pto \FF s\Big\}_{i=1,\dots,n}$}
\AXC{$\bone\vDash^{X\cup\{a\}} \bigvee_{i}\bone_{i}\land \bfour_{i}$}
\RL{$(\lor)$}
\BIC{$\Gamma \vdash^{X\cup\{a\}}t: \bone\Pto \FF s$}
\DP
$$
where $\mathrm{FN}(\bone_{i})\subseteq X$ and $\mathrm{FN}(\bfour_{i})\subseteq\{a\}$.
 Then there exists literals $\btwo_{ij}$ with free names in $X$ such that 
 \begin{varitemize}
 \item[(i.)] if $i\neq i'$, then for all $j,j'$ either $\btwo_{ij}\equiv\btwo_{i'j'}$ or $\btwo_{ij}\land \btwo_{i'j'}\vDash \BOT$;

 \item[(ii.)] 
 $\bone_{i}\equiv \bigvee_{j}\btwo_{ij}$, and more generally 
 $\bigvee_{i,j}\btwo_{ij} \equiv \bigvee_{i=1}^{n}\bone_{i}$.
 \end{varitemize}
 Moreover, $D$ can be turned into a pseudo-canonical derivation $D'$ ending as follows:
$$
\AXC{$\Big\{\Gamma \vdash^{X\cup\{a\}}t: \btwo_{ij}\land \bfour_{i}\Pto \FF s\Big\}_{i,j}$}
\AXC{$\bone\vDash^{X\cup\{a\}} \bigvee_{i,j}\btwo_{ij}\land \bfour_{i}$}
\RL{$(\lor)$}
\BIC{$\Gamma \vdash^{X\cup\{a\}}t: \bone\Pto \FF s$}
\DP
$$
\end{lemma}
\begin{proof}
%
%

For any of the conjunction of literals $\bone_{i}$, let $W_{i}\subseteq \BB N\times \BB N$ contain all pairs $(j,b)$ such that
at least one of the literals $\bvar_{b}^{j},\lnot \bvar_{b}^{j}$ occurs in $\bigvee_{i}\bone_{i}$ but
neither of them occurs in $\bone_{i}$. For any such $\bone_{i}$, 
let $\btwo_{i1},\dots, \btwo_{i2^{\sharp W_{i}}}$ be the conjunctions of literals obtained by adding to $\bone_{i}$ all possible consistent conjunctions of literals chosen from $W_{i}$ (i.e. all possible choices between $\bvar_{b}^{j}$ and $\lnot\bvar_{b}^{j}$, for $(j,b)\in W_{i}$). It is clear that $\bone_{i}\equiv \bigvee_{j}\btwo_{ij}$ and that all $\btwo_{ij}$ are disjoint, so $ii.$ holds. 

Moreover, let $i\neq i'$ and suppose $\btwo_{ij}=\bone_{i}\land \bthree$ and $\btwo_{i'j'}= \bone_{j}\land \bfour$ are not equivalent; since $\btwo_{ij}$ and $\btwo_{i'j'}$ are conjunctions of literals, and each of them contains either $\bvar_{b}^{j}$ or $\lnot \bvar_{b}^{j}$ for $(j,b)\in \bigcup_{i}W_{i}$, the only possibility is that for some pair $(j,b)$, $\btwo_{ij}$ contains either $\bvar_{b}^{j}$ or $\lnot \bvar_{b}^{j}$ and $\btwo_{i'j'}$ contains its negation, so $\btwo_{ij}\land \btwo_{i'j'}\vDash \BOT$. Hence $i.$ also holds.

Finally, the pseudo-canonical derivation $D'$ is obtained as follows: for any canonical sub-derivation $D_{i}$
 of $D$ of conclusion $\Gamma \vdash^{X}t: \bone_{i}\land \bfour_{i}\Pto \FF s$
 (whose conclusion contains the Boolean formula $\bone_{i}$) and for each choice of $ j$, 
define a sub-derivation $D_{ij}$ 
of conclusion $\Gamma \vdash^{X}t: \btwo_{ij}\land \bfour_{i}\Pto \FF s$ 
by replacing each Boolean formula $\bfive$ occurring in $D_{i}$ by $\bfive\land \bthree$, where $\btwo_{ij}=\bone_{i}\land \bthree$.
Since $\bthree$ is a conjunction of literals, one can check by induction on $D_{i}$ that $D_{ij}$ is also canonical.
This is immediate for all rules except for the $(\lor)$-rule. However, since $D_{i}$ is canonical, $(\lor)$-rules occur in $D_{i}$ only in $\lor\mu$-patterns, and one can check that one obtains then a $\lor\mu$-pattern by replacing each Boolean formula $\CC a$ by $\CC a \land \bthree$.

We finally obtain $D'$ as follows:

\begingroup\makeatletter\def\f@size{10}\check@mathfonts
\begin{lrbox}{\mypti}
\begin{varwidth}{\linewidth}
$
\AXC{$D_{ij}$}
\noLine
\UIC{$\Gamma \vdash^{X\cup\{a\}}t: \btwo_{ij}\land \bfour_{i} \Pto \FF s$}
\DP
$
\end{varwidth}
\end{lrbox}

$$
\AXC{$\left\{ \usebox{\mypti} \right\}_{i,j} $}
\AXC{$\bone \vDash \bigvee_{i,j}\btwo_{ij}\land \bfour_{i}$}
\RL{$(\lor)$}
\BIC{$\Gamma\vdash^{X\cup\{a\}}t: \bone\Pto \FF s$}
\DP
$$
\endgroup

%
%
\end{proof}

\begin{lemma}\label{lemma:cool2}
If $D$ is a derivation ending with a rule whose premises are conclusions of pseudo-canonical derivations, then $D$ can be turned into a pseudo-canonical derivation $D'$.
\end{lemma}
\begin{proof}
We consider all possible rules:
\begin{varitemize}

\item if $D$ is as follows:

\begingroup\makeatletter\def\f@size{10}\check@mathfonts
\begin{lrbox}{\mypti}
\begin{varwidth}{\linewidth}
$
\AXC{$D_{j}$}
\noLine
\UIC{$\Gamma \vdash^{X\cup\{a\}}t:\btwo_{j}\Pto \FF s$}
\DP
$
\end{varwidth}
\end{lrbox}

$$
\AXC{$\left\{\usebox{\mypti}\right \}_{j}$}
\AXC{$\bone\vDash \bigvee_{i}\btwo_{j}$}
\RL{$(\lor)$}
\BIC{$\Gamma\vdash^{X\cup\{a\}}t: \bone\Pto \FF s$}
\RL{$(\oplus l)$}
\UIC{$\Gamma \vdash^{X\cup\{a\}}t\oplus_{a}^{i}u: \bvar_{a}^{i}\land \bone \Pto \FF s$}
\DP
$$
\endgroup
then $D'$ is as follows:

\begingroup\makeatletter\def\f@size{10}\check@mathfonts
\begin{lrbox}{\mypti}
\begin{varwidth}{\linewidth}
$
\AXC{$D_{j}$}
\noLine
\UIC{$\Gamma \vdash^{X\cup\{a\}}t:\btwo_{j}\Pto \FF s$}
\RL{$(\oplus l)$}
\UIC{$\Gamma \vdash^{X\cup\{a\}}t\oplus_{a}^{i}u: \bvar_{a}^{i}\land \btwo_{j} \Pto \FF s$}
\DP
$
\end{varwidth}
\end{lrbox}

$$
\AXC{$\left\{ \usebox{\mypti}\right\}_{j}$}
\AXC{$\bvar_{a}^{i}\land \bone\vDash \bigvee_{j}\bvar_{a}^{i}\land \btwo_{j}$}
\RL{$(\lor)$}
\BIC{$\Gamma\vdash^{X\cup\{a\}}t\oplus_{a}^{i}u: \bvar_{a}^{i}\land\bone\Pto \FF s$}
\DP
$$
\endgroup

The symmetric rule for $\oplus_{a}^{i}$ is treated similarly.


\item If $D$ is as follows:

\begingroup\makeatletter\def\f@size{10}\check@mathfonts
\begin{lrbox}{\mypti}
\begin{varwidth}{\linewidth}
$
\AXC{$D_{j}$}
\noLine
\UIC{$\Gamma, x:\FF M \vdash^{X}t:\btwo_{j}\Pto \B C^{q}\sigma$}
\DP
$
\end{varwidth}
\end{lrbox}

$$
\AXC{$\left\{ \usebox{\mypti}\right\}_{j}$}
\AXC{$ \bone\vDash \bigvee_{j} \btwo_{j}$}
\RL{$(\lor)$}
\BIC{$\Gamma,x:\FF M\vdash^{X}t: \bone\Pto \B C^{q}\sigma$}
\RL{$(\lambda)$}
\UIC{$\Gamma \vdash^{X}\lambda x.t: \bone \Pto \B C^{q}(\FF M\To \sigma)$}
\DP
$$
\endgroup

then $D'$ is as follows:

\begingroup\makeatletter\def\f@size{10}\check@mathfonts
\begin{lrbox}{\mypti}
\begin{varwidth}{\linewidth}
$
\AXC{$D_{j}$}
\noLine
\UIC{$\Gamma, x:\FF M \vdash^{X}t:\btwo_{j}\Pto \B C^{q}\sigma$}
\RL{$(\lambda)$}
\UIC{$\Gamma \vdash^{X}\lambda x.t: \btwo_{j}\Pto\B C^{q} (\FF M\To \sigma)$}
\DP
$
\end{varwidth}
\end{lrbox}

$$
\AXC{$\left\{ \usebox{\mypti}\right\}_{j}$}
\AXC{$ \bone\vDash \bigvee_{j} \btwo_{j}$}
\RL{$(\lor)$}
\BIC{$\Gamma \vdash^{X}\lambda x.t: \bone \Pto\B C^{q} (\FF M\To \sigma)$}
\DP
$$
\endgroup

\item if $D$ is as follows:

\begin{center}
\begingroup\makeatletter\def\f@size{10}\check@mathfonts
\begin{lrbox}{\mypti}
\begin{varwidth}{\linewidth}
$
\AXC{$D_{j}$}
\noLine
\UIC{$\Gamma \vdash^{X}t:\btwo_{j}\Pto\B C^{q}(\FF M\To \sigma)$}
\DP
$
\end{varwidth}
\end{lrbox}

\begin{lrbox}{\myptu}
\begin{varwidth}{\linewidth}
$
\AXC{$E_{l}$}
\noLine
\UIC{$\Gamma \vdash^{X}u:\bthree_{l}\Pto \FF s_{l}$}
\DP
$
\end{varwidth}
\end{lrbox}

\resizebox{0.9\textwidth}{!}{
$
\AXC{$\left\{ \usebox{\mypti}\right\}_{j}$}
\AXC{$ \btwo\vDash \bigvee_{j} \btwo_{j}$}
\RL{$(\lor)$}
\BIC{$\Gamma\vdash^{X}t: \btwo\Pto\B C^{q}(\FF M\To \sigma)$}
\AXC{$\left\{\usebox{\myptu}\right\}_{l}$}
\AXC{$\bone \vDash \btwo \land \Big(\bigwedge_{l}d_{l}\Big)$}
\RL{$(@_{\cap})$}
\TIC{$\Gamma \vdash^{X}tu: \bone \Pto \B C^{q} \sigma$}
\DP
$}
\endgroup
\end{center}

\noindent where $\FF M=[\FF s_{1},\dots, \FF s_{n}]$ and each derivation $E_{l}$ is of the form

\begin{center}
\begingroup\makeatletter\def\f@size{10}\check@mathfonts
\begin{lrbox}{\mypti}
\begin{varwidth}{\linewidth}
$
\AXC{$E_{lk}$}
\noLine
\UIC{$\Gamma \vdash^{X}t:\bfour_{lk}\Pto\FF s_{l}$}
\DP
$
\end{varwidth}
\end{lrbox}

$
\AXC{$\left\{ \usebox{\mypti}\right\}_{k}$}
\AXC{$ \bthree_{l}\vDash \bigvee_{k} \bfour_{lk}$}
\RL{$(\lor)$}
\BIC{$\Gamma\vdash^{X}t: \bthree_{l}\Pto \FF s_{l}$}
\DP
$
\endgroup
\end{center}

Then $D'$ is as follows:

\begin{center}

\begingroup\makeatletter\def\f@size{10}\check@mathfonts

\begin{lrbox}{\myptu}
\begin{varwidth}{\linewidth}
$
\AXC{$E_{l g(l)}$}
\noLine
\UIC{$\Gamma \vdash^{X}u:\bfour_{l g(l)}\Pto \FF s_{l}$}
\DP
$
\end{varwidth}
\end{lrbox}

\begin{lrbox}{\mypti}
\begin{varwidth}{\linewidth}
$
\AXC{$D_{j}$}
\noLine
\UIC{$\Gamma \vdash^{X}t:\btwo_{j}\Pto\B C^{q}(\FF M\To \sigma)$}
\AXC{$\left\{\usebox{\myptu}\right\}_{l}$}
\AXC{$\bone_{jg} \vDash \btwo_{j} \land \Big(\bigwedge_{l}\bfour_{lg(l)}\Big)$}
\RL{$(@_{\cap})$}
\TIC{$\Gamma \vdash^{X}tu: \bone_{jg} \Pto  \B C^{q}\sigma$}
\DP
$
\end{varwidth}
\end{lrbox}

\resizebox{0.95\textwidth}{!}{
$
\AXC{$\left\{\usebox{\mypti}\right\}_{j, g}$}
\AXC{$\bone \vDash \bigvee_{j,g}\bone_{jg}$}
\RL{$(\lor)$}
\BIC{$\Gamma \vdash^{X}tu: \bone \Pto\B C^{q}  \sigma$}
\DP
$}
\endgroup
\end{center}

\noindent where $\bone_{jg}=\btwo_{j}\land \Big(\bigwedge_{l}\bfour_{lg(l)}\Big)$ and $g$ ranges over all possible choice functions associating each $l$ with a suitable value $k$ such that $\bfour_{lk}$ exists, using the fact that 

\begin{align*}
\bone & \vDash \btwo\land \Big (\bigwedge_{l}\bthree_{l}\Big) \equiv \Big(\bigvee_{j}\btwo_{j}\Big)\land \Big(\bigwedge_{l}\bigvee_{k}\bfour_{lk}\Big) \equiv \bigvee_{j}\Big( \btwo_{j}\land \big(\bigwedge_{l}\bigvee_{k}\bfour_{lk}\big)\Big) 
\equiv
\bigvee_{jg} \btwo_{j}\land \big( \bigwedge_{l}\bfour_{lg(l)}\big)
\end{align*}

\item if $D$ is as follows:

\begin{center}
%
%

\begingroup\makeatletter\def\f@size{10}\check@mathfonts
\begin{lrbox}{\mypti}
\begin{varwidth}{\linewidth}
$
\AXC{$D_{ij}$}
\noLine
\UIC{$\Gamma \vdash^{X\cup\{a\}}t:\bthree_{ij}\Pto\B C^{q_{i}}\sigma$}
\DP
$
\end{varwidth}
\end{lrbox}

\begin{lrbox}{\myptu}
\begin{varwidth}{\linewidth}
$
\AXC{$\left\{ \usebox{\mypti}\right\}_{j}$}
\AXC{$ \bone'\land \btwo_{i}\vDash \bigvee_{j} \bthree_{ij}$}
\RL{$(\lor)$}
\BIC{$\Gamma\vdash^{X\cup\{a\}}t: \bone'\land \btwo_{i}\Pto \B C^{q_{i}}\sigma$}
\DP
$
\end{varwidth}
\end{lrbox}

\resizebox{0.9\textwidth}{!}{
$
\AXC{$\left\{ \usebox{\myptu}\right\}_{i}$}
\AXC{$\mu(\btwo_{i})\geq s_{i}$}
\AXC{$\bone \vDash \bone'$}
\RL{$(\mu)$}
\TIC{$\Gamma \vdash^{X}\nu a.t: \bone \Pto \B C^{\sum_{i}q_{i}s_{i}} \sigma$}
\DP
$}
\endgroup

\end{center}

\noindent where 
$\mathrm{FN}(\bone,\bone')\subseteq X$ and $\mathrm{FN}(\btwo_{i})\subseteq \{a\}$, then, since the $D_{ij}$ are canonical, each $\bthree_{ij}$ is a conjunction of literals and can thus be written
as $\bfour_{ij}\land\bfive_{ij}$, where $\mathrm{FN}(\bfour_{ij})\subseteq X$ and $\mathrm{FN}(\bfive_{ij})\subseteq \{a\}$.

Observe that we can suppose w.l.o.g.~that none of either the $\bfour_{ij}$ or the $\bfive_{ij}$ is equivalent to $\BOT$.
Moreover, we can also suppose w.l.o.g.~that $\bone$ and $\bone'$ coincide and are satisfiable. 

Now, from $\bone\land \btwo_{i}\vDash \bigvee_{j}\bfour_{ij}\land \bfive_{ij}$ (using the fact that the $\btwo_{i}$ are satisfiable) we deduce by Lemma \ref{lemma:disj} that $\bone\vDash \bigvee_{j}\bfour_{ij}$ and $\btwo_{i}\vDash \bigvee_{j} \bfive_{ij}$.

For each $i$, by applying Lemma \ref{lemma:cool1} to the family $(\bfour_{ij})_{j}$,
we obtain a family $\bfour_{ijk}$ of conjunctions of literals, all satisfiable, with $\mathrm{FN}(\bfour_{ijk})\subseteq X$, and such that 
$\bfour_{ij}\equiv \bigvee_{k}\bfour_{ijk}$, and for $j\neq j'$, either $\bfour_{ijk}\equiv \bfour_{ij'k'}$ or $\bfour_{ijk}\land \bfour_{ij'k'}\vDash \BOT$.
Moreover, we also deduce the existence of a family $D_{ijk}^{\dag}$ of canonical derivations of conclusion
$\Gamma \vdash^{X\cup\{a\},q_{i}}t: \bfour_{ijk}\land \bfive_{ij}\Pto \FF s$.

%

%
%

For each $i,j,k$, let $U_{jk}^{i}$ be the set of pairs $(j',k')$ such that $\bfour_{ij'k'}\equiv \bfour_{ijk}$. Observe that for any pair $(j'',k'')\notin U_{jk}^{i}$, $\bfour_{ijk}\land \bfour_{ij''k''}\vDash \BOT$. 

%
%
%

Let $\approx^{i}$ be the equivalence relation on pairs $(j,k)$ defined by $(j,k)\approx^{i} (j',k')$ iff $\bfour_{ijk}\equiv \bfour_{ij'k'}$ (i.e.~iff $(j',k')\in U_{jk}^{i}$), and let $W^{i}$ contain a chosen representative of $[(j,k)]_{\approx^{i}}$ for each equivalence class. 
Then we have that 
\begin{equation}\label{eq:iuvk}
\bigvee_{j}\bfour_{ij}\land \bfive_{ij}\equiv
\bigvee_{jk}\bfour_{ijk}\land \bfive_{ij} \equiv \bigvee_{(j,k)\in W^{i}}\bfour_{ijk}\land \Big( \bigvee_{j'k'\in U^{i}_{jk}}\bfive_{ij'}\Big)
\end{equation}
Let $\widetilde{\bfive_{ijk}}=  \bigvee_{j'k'\in U^{i}_{jk}}\bfive_{ij'}$.
Let $\bone \equiv \bigvee_{u}\bone_{u}$, where the $\bone_{u}$ are pairwise disjoint conjunctions of literals, all being satisfiable.
From $\bone \land \btwo_{i}\vDash \bigvee_{j}\bthree_{ij}\equiv \bigvee_{(j,k)\in W^{i}}\bfour_{ijk}\land\widetilde{\bfive_{ijk}}$
we deduce
by Lemma \ref{lemma:disj} and the satisfiability of $\btwo_{i}$ that $\bone \vDash \bigvee_{(j,k)\in W^{i}}\bfour_{ijk}$; 
since the $\bfour_{ijk}$, with $(j,k)$ varying in $W^{i}$, are pairwise disjoint, 
 this implies the existence of functions $J,K$ such that 
 $\bone_{u}\vDash \bfour_{iJ(u)K(u)}$, and thus a fortiori such that  
 $\bone_{u}\land \btwo_{i}\vDash \bfour_{iJ(u)K(u)}\land \widetilde{\bfive_{iJ(u)K(u)}}$. Since the $\bone_{u}$ are satisfiable, by Lemma \ref{lemma:disj}, this implies 
$\btwo_{i}\vDash \widetilde{\bfive_{iJ(u)K(u)}}$. 

For each $u$, we can now construct a derivation $D_{u}^{*}$ of $\Gamma \vdash^{X}\nu a.t: \bfour_{iJ(u)K(u)}\Pto \B C^{\sum_{i}q_{i}s_{i}}\sigma$ as shown below:

\begin{center}
\begingroup\makeatletter\def\f@size{10}\check@mathfonts
\begin{lrbox}{\mypti}
\begin{varwidth}{\linewidth}
$
\AXC{$D^{\ddag}_{ij'k'}$}
\noLine
\UIC{$\Gamma \vdash^{X\cup\{a\}}t:\bone_{u}\land \bfive_{ij'}\Pto \B C^{q_{i}}\sigma$}
\DP
$
\end{varwidth}
\end{lrbox}

\begin{lrbox}{\myptu}
\begin{varwidth}{\linewidth}
$
\AXC{$\left\{ \usebox{\mypti}\right\}_{(j',k')\in U^{i}_{J(u)K(u)}}$}
\AXC{$\bone_{u}\land \btwo_{i}  \vDash \bigvee_{j'k'\in U_{J(u)K(u)}} \bone_{u}\land \bfive_{ij'}$}
\RL{$(\lor)$}
\BIC{$\Gamma\vdash^{X\cup\{a\}}t: \bone_{u}\land\btwo_{i}\Pto \B C^{q_{i}}\sigma$}
\DP
$
\end{varwidth}
\end{lrbox}

\resizebox{0.9\textwidth}{!}{
$\AXC{$\left\{ \usebox{\myptu} \right\}_{i}$}
\AXC{$\mu\left(\btwo_{i}\right)\geq s_{i}$}
\RL{$(\mu)$}
\BIC{$\Gamma \vdash^{X}\nu a.t: \bone_{u}\Pto\B C^{\sum_{i}q_{i}s_{i}}  \sigma$}
\DP
$}

\endgroup
\end{center}

\noindent where $D^{\ddag}_{ij'k'}$ is obtained from $D_{ij'k'}$ by replacing everywhere $\bfour_{ij'k'}$ by 
$\bone_{u}$, using the fact that $\bone_{u}\vDash \bfour_{iJ(u)K(u)}\equiv \bfour_{ij'k'}$.
Observe that the one above is a $\vee\mu$-pattern, and thus $D^{*}_{u}$ is canonical.

We can now conclude as follows:

\begingroup\makeatletter\def\f@size{10}\check@mathfonts
\begin{lrbox}{\mypti}
\begin{varwidth}{\linewidth}
$
\AXC{$D^{*}_{u}$}
\noLine
\UIC{$\Gamma \vdash^{X}\nu a.t:\bone_{u}\Pto \B C^{\sum_{i}q_{i}s_{i}}\sigma$}
\DP
$
\end{varwidth}
\end{lrbox}


$$
\AXC{$\left\{ \usebox{\mypti}\right\}_{u}$}
\AXC{$ \bone  \vDash \bigvee_{u}\bone_{u}$}
\RL{$(\lor)$}
\BIC{$\Gamma\vdash^{X}\nu a.t: \bone\Pto \B C^{\sum_{i}q_{i}s_{i}}\sigma$}
\DP
$$
\endgroup

\end{varitemize}
\end{proof}

We can now prove Theorem \ref{conje}.

\begin{proof}[Proof of Theorem \ref{conje}]
First, for any axiom 
$$
\AXC{$\FF s_{i}\preceq \FF t$}
\AXC{$\mathrm{FN}(\bone)\subseteq X$}
\RL{$(\mathrm{id}_{\preceq})$}
\BIC{$\Gamma, x:[\FF s_{1},\dots, \FF s_{n}] \vdash^{X}x: \bone\Pto \FF t$}
\DP$$
transform $\bone$ into an equivalent disjunctive normal form $\bigvee_{j}\bone_{j}$ (with the $\bone_{j}$ conjunctions of literals), and replace the axiom by the derivation below

\begingroup\makeatletter\def\f@size{10}\check@mathfonts
\begin{lrbox}{\mypti}
\begin{varwidth}{\linewidth}
$
\AXC{$\FF s_{i}\preceq \FF t$}
\AXC{$\mathrm{FN}(\bone)\subseteq X$}
\RL{$(\mathrm{id}_{\preceq})$}
\BIC{$\Gamma, x:[\FF s_{1},\dots, \FF s_{n}]  \vdash^{X}x: \bone_{j}\Pto \FF t$}
\DP
$
\end{varwidth}
\end{lrbox}

$$
\AXC{$\left\{\usebox{\mypti}\right\}_{j}$}
\AXC{$\bone\vDash^{X}\bigvee_{j}\bone_{j}$}
\RL{$(\lor)$}
\BIC{$\Gamma, x:[\FF s_{1},\dots, \FF s_{n}]  \vdash^{X}x: \bone\Pto \FF t$}
\DP
$$
\endgroup
 Now, starting from the axioms, for each rule occurring in the construction of the derivation, use Lemma \ref{lemma:cool2} to progressively permute $(\lor)$-rules downwards, hence turning the derivation into a pseudo-canonical one.
\end{proof}

Using Theorem \ref{conje} we can finally prove the subject $\beta$-expansion property.
We will make use of the following property (which is an easy consequence of Lemma \ref{lemma:nota} and Lemma \ref{lemma:weakb}):

\begin{lemma}\label{lemmatodo}
If $\Gamma ,\Delta\vdash^{X\cup Y}t: \bone\Pto \FF s$ is derivable, where $\mathrm{FV}(t)\subseteq \Gamma$ and $\mathrm{FN}(t)\subseteq X$, then there exists a Boolean formula $\bone'$ with $\mathrm{FN}(\bone')\subseteq X$ such that $\bone\vDash \bone'$ and 
$\Gamma \vdash^{X}t:\bone'\Pto \FF s$ is derivable.

\end{lemma}

\begin{proposition}[subject $\beta$-expansion]\label{prop:betaexpa}
If $\Gamma\vdash^{X}t[u/x]: \bone\Pto \FF s$ is derivable then
$\Gamma \vdash^{X}(\lambda x.t)u:\bone\Pto \FF s$ is also derivable.

%
%
\end{proposition}
\begin{proof}
A typing derivation of $\Gamma\vdash^{X}t[u/x]: \bone\Pto \FF s$ can be depicted as follows:
$$
\AXC{$D_{1}$}
\noLine
\UIC{$\Gamma, \Delta_{1}\vdash^{X\cup Y_{1}}u:\bfive_{1}\Pto \FF s_{1}$}
\AXC{$\dots$}
\AXC{$D_{N}$}
\noLine
\UIC{$\Gamma, \Delta_{N}\vdash^{X\cup Y_{N}}u:\bfive_{N}\Pto \FF s_{N}$}
\noLine
\TIC{$D$}
\noLine
\UIC{$\Gamma \vdash^{X}t[u/x]: \bone\Pto \FF s$}
\DP
$$

Observe that neither the variables in $\Delta_{i}$ nor the names in $Y_{i}$ can occur in $u$. Hence using Lemma \ref{lemmatodo} we obtain derivations $D^{*}_{i}$ of $\Gamma \vdash^{X}u: \bthree_{i}\Pto \FF s_{i}$, with $\bfive_{i}\vDash \bthree_{i}$.

Moreover, from $D$ we can deduce a derivation $D^{*}$ of $\Gamma, x: [\FF s_{1},\dots, \FF s_{N}] \vdash^{X}t: \bone'\Pto \FF s$, where each of the $D_{i}$ is replaced by an axiom
$$
\AXC{$\FF s_{i}\preceq\FF s_{i}$}
\AXC{$\mathrm{FN}(\bthree_{i})\subseteq X$}
\RL{$(\mathrm{id}_{\preceq})$}
\BIC{$
\Gamma,\Delta_{i}, x: [\FF s_{1},\dots, \FF s_{N}]\vdash^{X}x: \bthree_{i}\Pto \FF s_{i}
$}
\DP
$$
and $\bone'$ is such that $\bone\vDash \bone'$.
 
Using Theorem \ref{conje} we can turn $D^{*}$ into a pseudo-canonical derivation $D^{*\dag}$, which is thus of the form

\begingroup\makeatletter\def\f@size{10}\check@mathfonts
\begin{lrbox}{\mypti}
\begin{varwidth}{\linewidth}
$
\AXC{$E_{l}$}
\noLine
\UIC{$\Gamma,x: [\FF s_{1},\dots, \FF s_{N}]\vdash^{X}t: \btwo_{l}\Pto \FF s$}
\DP
$
\end{varwidth}
\end{lrbox}


$$
D^{*\dag} \quad = \quad 
\AXC{$\left\{ \usebox{\mypti}\right\}_{l}$}
\AXC{$ \bone'  \vDash \bigvee_{l}\btwo_{l}$}
\RL{$(\lor)$}
\BIC{$\Gamma,x: [\FF s_{1},\dots, \FF s_{N}]\vdash^{X}t: \bone'\Pto \FF s$}
\DP
$$
\endgroup

\noindent where each $E_{l}$ is canonical and contains all axioms of the form
$$
\AXC{$\FF s_{i}\preceq\FF s_{i}$}
\AXC{$\mathrm{FN}(\bthree_{ig_{l}(i)})\subseteq X$}
\RL{$(\mathrm{id}_{\preceq})$}
\BIC{$
\Gamma,\Delta_{i}, x: [\FF s_{1},\dots, \FF s_{N}]\vdash^{X}x: \bthree_{ig_{l}(i)}\Pto \FF s_{i}$}
\DP
$$
where $\bigvee_{k}\bthree_{ik}\equiv \bthree_{i}$ is a disjunctive normal form and $g_{l}$ is some choice function
choosing one value $k$ for each $i$.

Since $\bthree_{ik}\vDash \bigvee_{k}\bthree_{ik}\equiv \bthree_{i}$ holds we deduce by Lemma \ref{lemma:weak} the existence of derivations $D_{ik}$ of $\Gamma\vdash^{X}u:\bthree_{ik}\Pto \FF s_{i}$.  

Moreover, since the $E_{l}$ are canonical, by Lemma \ref{lemma:canonical} we deduce that $\btwo_{l}\vDash \bigwedge_{i}\bthree_{ig_{l}(i)}$.

Using $\bone\vDash \bone'$ and $\bone'\vDash \bigvee_{l}\btwo_{l}$, we can thus obtain a derivation of $\Gamma \vdash^{X}(\lambda x.t)u: \bone\Pto \FF s$ as follows (where $\FF M=[\FF s_{1},\dots, \FF s_{N}]$ and $\FF s=\B C^{q}\sigma$):

\begingroup\makeatletter\def\f@size{10}\check@mathfonts

\begin{lrbox}{\myptu}
\begin{varwidth}{\linewidth}
$
\AXC{$D_{ig_{l}(i)}$}
\noLine
\UIC{$\Gamma \vdash^{X}u:\bthree_{ig_{l}(i)}\Pto \FF s_{i}$}
\DP
$
\end{varwidth}
\end{lrbox}

\begin{lrbox}{\mypti}
\begin{varwidth}{\linewidth}
$
\AXC{$E_{l}$}
\noLine
\UIC{$\Gamma,x: \FF M\vdash^{X}t: \btwo_{l}\Pto \FF s$}
\UIC{$\Gamma\vdash^{X}\lambda x.t: \btwo_{l}\Pto \B C^{q}(\FF M\To \sigma)$}
\AXC{$\left \{ \usebox{\myptu}\right\}_{i}$}
\AXC{$\btwo_{l}\vDash \bigwedge_{i}\bthree_{ig_{l}(i)}  $}
\TIC{$\Gamma \vdash^{X}(\lambda x.t)u: \btwo_{l}\Pto \FF s$}
\DP
$
\end{varwidth}
\end{lrbox}

$$
\AXC{$\left\{ \usebox{\mypti}\right\}_{l}$}
\AXC{$ \bone  \vDash \bigvee_{l}\btwo_{l}$}
\RL{$(\lor)$}
\BIC{$\Gamma\vdash^{X}(\lambda x.t)u: \bone\Pto \FF s$}
\DP
$$
\endgroup
%
%

\end{proof}

\subsubsection{Subject Permutative Expansion}

The subject expansion property with respect to permutative reduction can be established by a direct inspection of permutative rules.

\begin{proposition}
If $\Gamma \vdash^{X} t: \bone\Pto \FF s$ and $u\redperm t$, then $\Gamma \vdash^{X}u: \bone\Pto \FF s$.
\end{proposition}
\begin{proof}
As in the proof of Proposition \ref{prop:subject}, if the typing derivation $D$ of $t$ ends by a $(\lor)$-rule, it suffices to establish the property for the immediate sub-derivations of $D$ and then apply an instance of $(\lor)$-rule to the resulting derivations. So we will always suppose that the typing derivation of $t$ does not end by a $(\lor)$-rule.

\begin{description}

\item[($t\oplus_{a}^{i}t\redperm t$)]

From $\Gamma \vdash^{X\cup\{a\}}t: \bone\Pto \FF s$ we deduce
$$
\AXC{$\Gamma \vdash^{X\cup\{a\}}t: \bone\Pto \FF s$}
	\RL{($\oplus l$)}
\UIC{$\Gamma \vdash^{X\cup\{a\}}t\oplus_{a}^{i}t: \bvar_{a}^{i}\land\bone\Pto \FF s$}
\AXC{$\Gamma \vdash^{X\cup\{a\}}t: \bone\Pto \FF s$}
	\RL{($\oplus r$)}
\UIC{$\Gamma \vdash^{X\cup\{a\}}t\oplus_{a}^{i}t:\lnot \bvar_{a}^{i}\land\bone\Pto \FF s$}
\AXC{$\bone\vDash (\bvar_{a}^{i}\land\bone)\lor(\lnot\bvar_{a}^{i}\land\bone)$}
\RL{($\vee$)}
\TIC{$\Gamma \vdash^{X\cup\{a\}}t\oplus_{a}^{i}t:\bone\Pto \FF s$}
\DP
$$

\item[($ (t\oplus_{a}^{i}u)\oplus_{a}^{i}v \redperm t\oplus_{a}^{i}v$)]
There are two possible sub-cases:
	\begin{enumerate}
	\item the type derivation is as follows:
	$$
	\AXC{$\Gamma\vdash^{X\cup\{a\}} t: \bone'\Pto \FF s$}
	\AXC{$\bone'\vDash \bvar_{a}^{i}\land \bone'$}
		\RL{($\oplus l$)}
	\BIC{$\Gamma \vdash^{X\cup\{a\}} t\oplus_{a}^{i} v :\bone\Pto \FF s$}
		\DP
	$$
	Then we deduce
	$$
	\AXC{$\Gamma \vdash^{X\cup\{a\}}t: \bone'\Pto \FF s$}
\AXC{$\bone\vDash \bvar_{a}^{i}\land \bone'$}
	\RL{($\oplus l$)}
	\BIC{$\Gamma \vdash^{X\cup\{a\}} t\oplus_{a}^{i} u :\bone\Pto \FF s$}
\AXC{$\bone\vDash \bvar_{a}^{i}\land \bone$}
	\RL{($\oplus l$)}
	\BIC{$\Gamma \vdash^{X\cup\{a\}} (t\oplus_{a}^{i} u)\oplus_{a}^{i}v :\bone\Pto \FF s$}
\DP
	$$
	
	\item the type derivation is as follows:
	$$
	\AXC{$\Gamma\vdash^{X\cup\{a\}} v: \bone'\Pto \FF s$}
	\AXC{$\bone'\vDash\lnot \bvar_{a}^{i}\land \bone'$}
		\RL{($\oplus l$)}
	\BIC{$\Gamma \vdash^{X\cup\{a\}} t\oplus_{a}^{i} v :\bone\Pto \FF s$}
		\DP
	$$
	Then we deduce
	$$
	\AXC{$\Gamma \vdash^{X\cup\{a\}}v: \bone'\Pto \FF s$}
\AXC{$\bone\vDash\lnot \bvar_{a}^{i}\land \bone'$}
	\RL{($\oplus l$)}
	\BIC{$\Gamma \vdash^{X\cup\{a\}} (t\oplus_{a}^{i} u)\oplus_{a}^{i}v :\bone\Pto \FF s$}
\DP
	$$
	\end{enumerate}

\item[($t\oplus_{a}^{i}(u\oplus_{a}^{i}v)\redperm t\oplus_{a}^{i}v$)] Similar to the case above.

\item[($\lambda x.(t\oplus_{a}^{i}u)\redperm (\lambda x.t)\oplus_{a}^{i}(\lambda x.u)$)]
There are two possible sub-cases, treated similarly. We only consider the first one:
	\begin{enumerate}
	\item 
		$$
	\AXC{$\Gamma, x:\FF M \vdash^{X} t: \bone'\Pto \B C^{q}\sigma$}
	\RL{($\lambda$)}
	\UIC{$\Gamma \vdash^{X} \lambda x.t: \bone' \Pto \B C^{q}(\FF M\To \sigma)$}
	\AXC{$\bone\vDash \bvar_{a}^{i}\land \bone'$}
		\RL{($\oplus l$)}
	\BIC{$\Gamma \vdash^{X} (\lambda x.t)\oplus_{a}^{i}(\lambda x.u): \bone \Pto\B C^{q} (\FF M\To \sigma)$}
	\DP
	$$
		Then we deduce 
	$$
	\AXC{$\Gamma, x:\FF M \vdash^{X} t: \bone'\Pto \B C^{q}\sigma$}
	\AXC{$\bone\vDash \bvar_{a}^{i}\land \bone'$}
		\RL{($\oplus l$)}
	\BIC{$\Gamma, x:\FF M \vdash^{X} t\oplus_{a}^{i}u: \bone\Pto \B C^{q}\sigma$}
	\RL{($\lambda$)}
	\UIC{$\Gamma \vdash^{X} \lambda x.(t\oplus_{a}^{i}u): \bone \Pto\B C^{q} (\FF M\To \sigma)$}
	\DP
	$$

	\end{enumerate}

	\item[($(t\oplus_{a}^{i}u)v \redperm (tv)\oplus_{a}^{i}(uv)$)]
		There are two possible sub-cases, treated similarly.
	 We only consider the first one (with $\FF M=[\FF s_{1},\dots, \FF s_{n}]$:
	\begin{enumerate}
	\item 
	$$
	\AXC{$\Gamma \vdash^{X} t: \bone''\Pto \B C^{q}(\FF M\To \sigma)$}
	\AXC{$\Big\{\Gamma \vdash^{X}v: \btwo_{i} \Pto \FF s_{i}\Big\}_{i}$}
	\AXC{$\bone' \vDash \bone''\land\left(\bigwedge_{i} \btwo_{i}\right)$}
	\RL{($@_{\cap}$)}
	\TIC{$\Gamma \vdash^{X} tv : \bone' \Pto \B C^{q}\sigma$}
	\AXC{$\bone \vDash \bvar_{a}^{i}\land \bone' $}
	\RL{($\oplus l$)}
	\BIC{$\Gamma \vdash^{X} (tv)\oplus_{a}^{i} (uv): \bone \Pto \B C^{q}\sigma$}
	\DP
	$$
	Then we deduce 

	$$
	\AXC{$\Gamma \vdash^{X} t: \bone''\Pto \B C^{q}(\FF M\To \sigma)$}
	\AXC{$\bvar_{a}^{i}\land\bone''\vDash \bvar_{a}^{i}\land \bone''$}
		\RL{($\oplus l$)}
	\BIC{$\Gamma\vdash^{X} t\oplus_{a}^{i} u: \bvar_{a}^{i}\land\bone'' \Pto\B C^{q}(\FF M\To \sigma)$}
	\AXC{$\Big\{\Gamma \vdash^{X}v: \btwo_{i} \Pto \FF s_{i}\Big\}_{i}$}
	\AXC{$\bone \vDash(\bvar_{a}^{i}\land \bone'')\land\left(\bigwedge_{i} \btwo_{i}\right)$}
		\RL{($@_{\cap}$)}
	\TIC{$\Gamma\vdash^{X}(t\oplus_{a}^{i}u)v: \bone \Pto \B C^{q}\sigma$} 
	\DP
	$$
		\end{enumerate}

\item[{($t(u\oplus_{a}^{i}v)\redperm (tu)\oplus_{a}^{i}(tv)$)}]
There are two sub-cases, treated similarly.  We only consider the first one:
	\begin{enumerate}
	\item 
	$$
	\AXC{$\Gamma \vdash^{X} t: \bone''\Pto \B C^{q}(\FF M\To \sigma)$}
	\AXC{$\Big\{\Gamma \vdash^{X}u: \btwo_{i} \Pto \FF s_{i}\Big\}_{i}$}
	\AXC{$\bone' \vDash \bone''\land \left (\bigwedge_{i}\btwo_{i}\right)$}
	\RL{($@_{\cap}$)}
	\TIC{$\Gamma\vdash^{X}tu: \bone' \Pto \B C^{q}\sigma$} 
	\AXC{$\bone\vDash \bvar_{a}^{i}\land \bone'$}
		\RL{($\oplus l$)}
	\BIC{$\Gamma \vdash^{X} (tu)\oplus_{a}^{i}(tv): \bone \Pto\B C^{q}\sigma$}
	\DP
	$$
		Then we deduce that 
		
\begingroup\makeatletter\def\f@size{10}\check@mathfonts

\begin{lrbox}{\mypti}
\begin{varwidth}{\linewidth}
$
\AXC{$\Gamma \vdash^{X}u: \btwo_{i} \Pto \FF s_{i}$}
\AXC{$\bvar_{a}^{i}\land\btwo_{i}\vDash \bvar_{a}^{i}\land \btwo_{i}$}
\RL{($\oplus l$)}
\BIC{$\Gamma \vdash^{X} u\oplus_{a}^{i} v: \bvar_{a}^{i}\land\btwo_{i} \Pto \FF s_{i}$}
\DP
$
\end{varwidth}
\end{lrbox}
	$$
	\AXC{$\Gamma \vdash^{X} t: \bone''\Pto\B C^{q}(\FF M\To \sigma)$}
	\AXC{$\left\{ \usebox{\mypti}\right \}_{i}$}
	\AXC{$\bone \vDash \bone''\land\left(\bigwedge_{i}(\bvar_{a}^{i}\land \btwo_{i})\right)$}
	\RL{($@_{\cap}$)}
	\TIC{$\Gamma\vdash^{X}t(u\oplus_{a}^{i}v): \bone \Pto \B C^{q}\sigma$} 
	\DP
	$$
\endgroup
	\end{enumerate}

\item[($(t\oplus_{a}^{i}u)\oplus_{b}^{j}v \redperm (t\oplus_{b}^{j}v)\oplus_{a}^{i}(u\oplus_{b}^{j}v)$)]
We suppose here $a\neq b$ or $i< j$. There are four sub-cases, all treated similarly. We only consider the first one:
		\begin{enumerate}
		\item 
		$$
		\AXC{$\Gamma \vdash^{X} t: \bone''\Pto \FF s$}
		\AXC{$\bone'\vDash \bvar_{b}^{j} \land \bone''$}
		\RL{($\oplus l$)}
		\BIC{$\Gamma \vdash^{X}t\oplus_{b}^{j}v: \bone'\Pto \FF s$}
		\AXC{$\bone\vDash \bvar_{a}^{i} \land \bone'$}
				\RL{($\oplus l$)}
		\BIC{$\Gamma \vdash^{X}(t\oplus_{b}^{j}u)\oplus_{a}^{i}(u\oplus_{b}^{j}v): \bone\Pto \FF s$}
		\DP
		$$
		Then we deduce
		$$
		\AXC{$\Gamma \vdash^{X} t: \bone''\Pto \FF s$}
		\AXC{$\bone\vDash \bvar_{a}^{i} \land \bone''$}
				\RL{($\oplus l$)}
		\BIC{$\Gamma \vdash^{X}t\oplus_{a}^{i}u: \bone\Pto \FF s$}
		\AXC{$\bone\vDash \bvar_{b}^{j} \land \bone$}
				\RL{($\oplus l$)}
		\BIC{$\Gamma \vdash^{X}(t\oplus_{a}^{i}u)\oplus_{b}^{j}v: \bone\Pto \FF s$}
		\DP
		$$

		\end{enumerate}

\item[($t\oplus_{b}^{j}(u\oplus_{a}^{i}v) \redperm (t\oplus_{b}^{j}u)\oplus_{a}^{i}(t\oplus_{b}^{j}v)$)]
Similar to the case above.

\item[{($\nu b. (t\oplus_{a}^{i}u)\redperm (\nu b.t)\oplus_{a}^{i}(\nu b.u)$)}] We suppose $a\neq b$.
There are two sub-cases, treated similarly. We only consider the first one:
	\begin{enumerate}
	\item
		$$
	\AXC{$ \left\{\Gamma\vdash^{X\cup\{a,b\}} t: \bone''\land \btwo_{i}\Pto \B C^{q_{i}}\sigma\right\}_{i}$}
	\AXC{$ \mu(\btwo_{i})\geq s_{i}$}
	\AXC{$\bone' \vDash \bone''$}
	\RL{($\mu_{\Sigma}$)}
	\TIC{$\Gamma \vdash^{X\cup\{a\}} \nu b. t: \bone' \Pto \B C^{\sum_{i}q_{i}s_{i}}\sigma$}
	\AXC{$\bone\vDash \bvar_{a}^{i}\land \bone'$}
	\RL{($\oplus l$)}
	\BIC{$\Gamma \vdash^{X\cup\{a\}}( \nu b. t)\oplus_{a}^{i}(\nu b.u): \bone \Pto \B C^{\sum_{i}q_{i}s_{i}}\sigma$}
	\DP
	$$
	From $\bone \vDash \bone''$ we deduce by Lemma \ref{lemma:weak} the existence of derivations of 
	$\Gamma\vdash^{X\cup\{b\}} t: \bone\land \btwo_{i}\Pto \B C^{q_{i}}\sigma$.
	Then we can construct the following type derivation for $u$:

	\begingroup\makeatletter\def\f@size{10}\check@mathfonts
\begin{lrbox}{\mypti}
\begin{varwidth}{\linewidth}
$
\AXC{$\Gamma\vdash^{X\cup\{a,b\}} t: \bone\land \btwo_{i}\Pto \B C^{q_{i}}\sigma$}
\RL{($\oplus l$)}
\UIC{$\Gamma \vdash^{X\cup\{a,b\}}t\oplus_{a}^{i}u: \bone\land \btwo_{i}\Pto \B C^{q_{i}}\sigma$}
\DP
$
\end{varwidth}
\end{lrbox}

$$
\AXC{$\left \{ \usebox{\mypti}\right\}_{i}$}
	\AXC{$ \mu(\btwo_{i})\geq s_{i}$}
	\RL{($\mu_{\Sigma}$)}
\BIC{$\Gamma \vdash^{X\cup\{a\}} \nu b. t\oplus_{a}^{i}u: \bone \Pto \B C^{\sum_{i}q_{i}s_{i}}\sigma$}
\DP
$$
\endgroup
	
%
%
%
\end{enumerate}

\item[($\nu a.t\redperm t$)]
We can suppose that $a\notin \mathrm{FN}(\bone)$, so 
from $\Gamma \vdash^{X}t:\bone\Pto \FF s$ and $\bone\land \TOP\vDash \bone$ we deduce 
$\Gamma \vdash^{X\cup\{a\}}t: \bone\land \TOP\Pto \FF s$ and since $\mu(\TOP)\geq 1$, 
$\Gamma \vdash^{X}\nu a.t: \bone\Pto \FF s$, since $\FF s=\B C^{q}\sigma= \B C^{q\cdot 1}\sigma$.  

\item[($\lambda x.\nu a.t\redperm \nu a.\lambda x.t$ )]

\begingroup\makeatletter\def\f@size{10}\check@mathfonts
\begin{lrbox}{\mypti}
\begin{varwidth}{\linewidth}
$
\AXC{$\Gamma, x: \FF M \vdash^{X\cup\{a\}} t: \bone'\land \btwo_{i} \Pto \B C^{q_{i}}\sigma$}
\RL{($\lambda$)}
\UIC{$\Gamma \vdash^{X\cup\{a\}}\lambda x.t: \bone' \land \btwo_{i} \Pto\B C^{q_{i}} (\FF M\To \sigma)$}
\DP
$
\end{varwidth}
\end{lrbox}

$$
\AXC{$\left \{\usebox{\mypti}\right\}_{i}$}
\AXC{$\mu(\btwo_{i})\geq s_{i}$}
\AXC{$\bone \vDash \bone'$}
\RL{($\mu_{\Sigma}$)}
\TIC{$\Gamma\vdash^{X}\nu a.\lambda x.t : \bone \Pto \B C^{\sum_{i}q_{i}s_{i}}( \FF M\To\sigma)$}
\DP
$$
\endgroup
from which we deduce
$$
\AXC{$\{\Gamma, x: \FF M \vdash^{X\cup\{a\}} t: \bone'\land \btwo_{i} \Pto \B C^{q_{i}}\sigma\}_{i}$}
\AXC{$\mu(\btwo_{i})\geq s_{i}$}
\AXC{$\bone \vDash \bone'$}
\RL{($\mu_{\Sigma}$)}
\TIC{$\Gamma, x: \FF M\vdash^{X}\nu a.t : \bone \Pto\B C^{\sum_{i}q_{i}s_{i}}\sigma$}
\RL{($\lambda$)}
\UIC{$\Gamma \vdash^{X}\lambda x.\nu a.t: \bone \Pto\B C^{\sum_{i}q_{i}s_{i}} (\FF M\To \sigma)$}
\DP
$$

\item[($ (\nu a.t)u \redperm \nu a.(tu)$)]
We have
\begin{center}
\begingroup\makeatletter\def\f@size{10}\check@mathfonts
\begin{lrbox}{\mypti}
\begin{varwidth}{\linewidth}
$
\AXC{$\Gamma\vdash^{X\cup\{a\}} t: \btwo_{i}\Pto \B C^{q_{i}}(\FF M_{i}\To \sigma)$}
\AXC{$\left\{\Gamma \vdash^{X\cup\{a\}}u: \bthree_{ij}\Pto \FF t_{ij}\right\}_{j}$}
\AXC{$\bone'\land \bfour_{i}\vDash\btwo_{i}\land \bigwedge_{j} \bthree_{ij}$}
\RL{($@_{\cap}$)}
\TIC{$\Gamma \vdash^{X\cup\{a\}}tu: \bone'\land \bfour_{i}\Pto \B C^{q_{i}}\sigma$}
\DP
$
\end{varwidth}
\end{lrbox}

\resizebox{0.95\textwidth}{!}{
$
\AXC{$\left\{ \usebox{\mypti}\right\}_{i}$}
\AXC{$\mu(\bfour_{i})\geq s_{i}$}
\AXC{$\bone \vDash \bone'$}
\RL{($\mu_{\Sigma}$)}
\TIC{$\Gamma \vdash^{X}\nu a.tu: \bone \Pto\B C^{\sum_{i}q_{i}s_{i}}\sigma$}
\DP
$}
\endgroup
\end{center}
where $\FF M_{i}= [\FF t_{i1},\dots, \FF t_{ip_{i}}]$. 
Let $\FF M$ be the concatenation of $\FF M_{1},\dots, \FF M_{k+1}$. 
Since $ \FF M\preceq \FF M_{i}$, we have $\B C^{q}(\FF M_{i}\To \sigma)\preceq \B C^{q}(\FF M\To \sigma)$, so 
 from the derivations of $\Gamma\vdash^{X\cup\{a\}} t: \btwo_{i}\Pto \B C^{q_{i}} (\FF M_{i}\To \sigma)$, using Lemma \ref{lemma:order} we deduce derivations of $\Gamma\vdash^{X\cup\{a\}} t: \btwo_{i}\Pto \B C^{q}(\FF M\To \sigma)$.
Moreover, from $\bone\land \bfour_{i}\vDash \btwo_{i}$,
 using Lemma \ref{lemma:weakb}, we deduce the existence of derivations of 
$\Gamma\vdash^{X\cup\{a\}} t: \bone\land \bfour_{i}\Pto \B C^{q}(\FF M\To \sigma)$, and from $\bone\land \bfour_{i}\vDash  \bigwedge_{j}\bthree_{ij}$, since $\mathrm{FN}(\bfour_{i})\cap \mathrm{FN}(\bthree_{ij})=\emptyset$, we deduce $\bone\vDash \bigwedge_{j}\bthree_{ij}$. 
We can thus construct a type derivation for $(\nu a.t)u$ as follows:

$$
\AXC{$\left\{\Gamma\vdash^{X\cup\{a\}} t: \bone\land\bfour_{i}\Pto\B C^{q} (\FF M\To \sigma)\right\}_{i}$}
\AXC{$\mu(\bfour_{i})\geq s_{i}$}
\RL{($\mu_{\Sigma}$)}
\BIC{$\Gamma \vdash^{X}\nu a.t: \bone \Pto\B C^{\sum_{i}q_{i}s_{i}}(\FF M\To \sigma)$}
\AXC{$\left \{\Gamma \vdash^{X,r_{i}}u: \bthree_{ij}\Pto \FF t_{ij}\right\}_{i,j}$}
\AXC{$\bone \vDash \bigwedge_{ij}\bthree_{ij}$}
\RL{($@_{\cap}$)}
\TIC{$\Gamma \vdash^{X}(\nu a.t)u: \bone \Pto \B C^{\sum_{i}q_{i}s_{i}}\sigma$}
\DP
$$
Observe that this is the unique case in the proof in which we use intersection types in an essential way.

\end{description}
\end{proof}

\subsection{Completeness.}

\subsubsection{Deterministic Completeness}\label{subs:detcomple}

The goal of this subsection is to establish Proposition \ref{prop:detcomple}.

\detcomple*

%
%

%
%
%
%
%
%

We need a few preliminary lemmas:

 \begin{lemma}\label{lemma:hnorm}
 For any closed term $t$:
 \begin{varitemize}
 \item[(i.)] if $t$ is a HNV, then $\vdash_{\lnot \vee} t: \TOP \Pto \B C^{1}\sigma$ holds for some type;
  \item[(ii.)] if $t$ is a HNV and is $\redall$-normal, then $\vdash_{\lnot \vee} t: \TOP \Pto \B C^{1}\sigma$ holds for some safe type.
 \end{varitemize}
 \end{lemma}
 \begin{proof}
 If  $t=\lambda x_{1}\dots x_{n}.x_{i}u_{1}\dots u_{p}$, then 
we can let $\sigma= \FF s_{1}\To \dots \To \FF s_{n}\To o$, where 
for $j\neq i$, $\FF s_{j}=[\ ]$, and $\FF s_{i}=\B C^{1} [\B C^{1}[\ ]\To \dots \To \B C^{1}[\ ]\To  o] $. 
The second claim is proved by induction on normal forms, using a similar construction.
 \end{proof}

 \begin{lemma}\label{lemma:strongnorm}
 If $\Gamma \vdash_{\lnot \vee} t[u/x]: \bone \Pto \FF s$ is derivable without using $[]$, and $t$ contains at least one occurrence of $x$, then 
there exists $\FF M\neq []$ such that 
$\Gamma,x:\FF M \vdash_{\lnot \vee} t: \bone \Pto \FF s$ and 
$\Gamma  \vdash_{\lnot \vee} u: \bone \Pto \FF M$ are derivable without using $[]$. \end{lemma}
\begin{proof}
It suffices to look at the proof of Proposition \ref{prop:betaexpa}, and observe that (1) since $u$ occurs at least once in $t$, $\FF M$ cannot be $[]$, and (2) 
since the original derivation has no occurrence of $(\vee)$, it must be canonical, and thus to construct the derivation of $\Gamma \vdash (\lambda x.t)u:\bone \Pto \FF s$ we do not need to make use of any application of ($\vee$). 
\end{proof}

Using Lemma \ref{lemma:hnorm} together with the subject expansion property we can prove this first half of Proposition \ref{prop:detcomple}. The second half comes from Corollary \ref{cor:headnorm}.

\begin{proof}[Proof of Proposition \ref{prop:detcomple}]
For (i.), if $t$ is head normalizable, with head-normal form $u$, then by Lemma \ref{lemma:hformrbt} the tree $RBT(t)=RBT(u)$ is finite and its leaves are HNV $u_{1},\dots, u_{N}$. By Lemma \ref{lemma:hnorm} 
we can deduce $\vdash u_{i}:\TOP \Pto \B C^{1}\HNORM$; we can then climb up $RBT(t)$ using the rules
($\oplus l$), ($\oplus r$) and ($\mu_{\Sigma}$), hence proving  $\vdash u:\TOP \Pto \B C^{1}\HNORM$, and finally
$\vdash t: \TOP \Pto \B C^{1}\HNORM$ by subject expansion. 
The converse direction follows from Corollary \ref{cor:headnorm} (i.).

The argument for (ii.) is similar.

%
%

For (iii.), let $t$ be strongly normalizable, we argue by induction on the maximum length $N$ of a reduction of $t$.
If $N=0$, then $t$ is in normal form, so by an argument similar to (i.) we can deduce $\vdash_{\not \vee}t: \TOP \Pto \B C^{1}\NORM$.
If $N>0$ and $t\redall u$, then by IH $\vdash_{\not\vee}u:\TOP\Pto \B C^{1}\NORM$ holds with a $[]$-free derivation, and $t=\TT C[(\lambda x.t')t'']$, $u=\TT C[t'[t''/x]]$. If $x$ occurs in $t'$, then we deduce from Lemma \ref{lemma:strongnorm}
that $\vdash_{\not\vee}t:\TOP\Pto \B C^{1}\NORM$ holds with a $[]$-derivation.
Otherwise, since $t''$ must be strong normalizable, by IH  $\vdash_{\not\vee}t'':\TOP\Pto \B C^{1}\NORM$, and we deduce, again by IH, that $\lambda x.t'$ can be typed by declaring $x: \B C^{1}\NORM$. We can then conclude
$\vdash_{\not\vee}t:\TOP\Pto \B C^{1}\NORM$ also in this case.

The converse claim follows from Corollary \ref{cor:headnorm} (iii.)
\end{proof}

\subsubsection{Probabilistic Completeness}

The goal of this subsection is to establish the completeness part of  Theorem \ref{thm:completenessa}:

\completenessa*
%
%
%
%
%

More precisely, we will show the $\leq$-part of the equations above, while soundness (i.e.~the $\geq$-part of the equations above) will follow from Theorem \ref{thm:normalization} and Theorem \ref{thm:normalization2}, proved in the next section.

First, we need to study the structure of the trees $RBT(t)$ (see Subsection \ref{subs:evl}) further.

\begin{definition}
A \emph{randomized multi-context} $\TT C$ is a term constructed from the grammar below
$$
\TT C::= [\ ]_{i} \ (i\in \BB N) \mid \TT C \oplus_{a}^{i}\TT C\mid \nu a.\TT C
$$
where for each two occurrences $[\ ]_{i}$ and $[\ ]_{j}$, $i\neq j$.
The rank $\mathsf r(\TT C)$ is the maximum $k$ such that $[\ ]_{k}$ occurs in $\TT C$.

For any randomized multi-context $\TT C$, we let $\TT C^{\lnot\nu}$ be the randomized multi-context obtained by deleting all occurrences of $\nu$-binders. 
\end{definition}

If $\TT C$ is a randomized multi-context of rank $k$, for all terms $t_{1},\dots, t_{k}$, we let $\TT C[t_{1},\dots, t_{k}]$ indicate the term obtained by replacing each hole $[\ ]_{i}$ in $\TT C$ by $t_{i}$.

An \emph{initial segment} $T$ of $RBT(t)$, noted $T\sqsubset RBT(t)$ is any sub-tree of $RBT(t)$ which contains its root.

Observe that a finite initial segment $T\sqsubset RBT(t)$ is the same thing as a term of the form 
$\TT C_{T}[u_{1},\dots, u_{k}]$ for some randomized multi-context $\TT C_{T}$.

\begin{lemma}\label{lemma:mcox1}
For any term $t$ and HNF $u_{1},\dots, u_{k}$, if there exist randomized paths from $t$ to $u_{1},\dots, u_{k}$, then there exists $N\in \BB N$, a randomized multi-context $\TT C$ of rank $N+k$ and terms $t_{1},\dots, t_{N}$ such that  
$t \redperm^* \TT C[u_{1},\dots, u_{k}, t_{1},\dots, t_{N}]$
and $RBT(\TT C[u_{1},\dots, u_{k}, t_{1},\dots, t_{N}])\sqsubseteq RBT(t)$.
\end{lemma}
\begin{proof}
We argue by induction on the maximum length $K$ of the randomized paths considered. 
If $K=0$, $k=1$ and $t=u$, so we just let $N=0$ and $\TT C [\ ]=[\ ]$.

If $K\geq 1$, then in a finite number of reduction steps we get $t\redperm^* \nu a.T$, where $T$ is a $(\CC T,a)$-tree such that for each $i=1,\dots,k$, there exists $v_{i}\in \mathsf{supp}(T)$ with a randomized path from $v_{i}$ to $u_{i}$ of length $\leq K-1$; observe that we can write $T$ as $\TT D[v_{1},\dots, v_{q}]$ for some randomized multi-context $\TT C$. 

For any $j=1,\dots,q$, let $I_{j}$ be the set of values $i\leq k$ such that there is a randomized path from $v_{j}$ to $u_{i}$. The $I_{j}$ are disjoint sets with $\bigcup_{j}I_{j}=\{1,\dots,k\}$.
By the I.H. we then deduce that there exist randomized multicontexts $\TT C_{j}$ or rank $N_{j}+| I_{j}|$ and terms $w_{j1},\dots, w_{jN_{j}}$ such that 
$v_{j} \redperm^* \TT C[u_{i_{1}},\dots, u_{i_{M_{j}}},w_{j1},\dots, w_{jN_{j}}]$, where $I_{j}=\{i_{1},\dots, i_{M_{j}}\}$. We now have that 
\begin{align*}
t \redperm^* \nu a.T= \nu a.\TT C[v_{1},\dots,v_{q}] \redperm^*
\nu a. \TT D[ \TT C_{1}[\vec u,\vec w],\dots, \TT C_{q}[\vec u, \vec w]]
\end{align*}
so we can let $N=\sum_{j}N_{j}$, 
$\TT C=\nu a.\TT D[\TT C_{1},\dots, \TT C_{q}]$ and $t_{1},\dots, t_{N}=w_{11},\dots, w_{qM_{q}}$.
\end{proof}

\begin{lemma}\label{lemma:mcox2}
For any randomized multi-context $\TT C$ with $n$ bound names $a_{1},\dots, a_{n}$, 
\begin{equation}\label{eq:nuzz0}
\nu a_{1}.\dots. \nu a_{n}. \TT C^{\lnot\nu} \ \redperm^{*} \ \TT C
\end{equation}
\end{lemma}
\begin{proof}
We will prove the following fact: for all terms $\nu b.t, t_{1},\dots, t_{k}$, it $\TT C$ has no $\nu$-binders, then
\begin{equation}\label{eq:nuzz}
\nu b. \TT C[t_{1},t_{2},\dots, t_{k},t] \redperm^{*} \TT C[  t_{1},\dots, t_{k}, \nu b.t]
\end{equation}
From this fact we can deduce the main claim as follows: let $b$ be a bound variable in $\TT C$ whose binder $\nu$ has no other binder in its scope; then
$\TT C$ splits as $\TT C[t_{0},t_{1},\dots, t_{k}]= 
 \TT  E[t_{i_{1}},\dots, t_{i_{p}},\nu b. \TT D[t_{j_{1}},\dots, t_{j_{q}}]]$, where 
 $I=\{i_{1},\dots, i_{p}\},J=\{j_{1},\dots, j_{q}\}\subseteq\{1,\dots,k+1\}$ satisfy $I\cap J=\emptyset$ and $p+q=\mathsf r(\TT C)$ (we are here supposing that $r(\TT C)\geq 1$, because otherwise the claim is trivial).
 
 Then, by \eqref{eq:nuzz} we deduce that $\nu b.\TT  E^{\lnot\nu}[t_{i_{1}},\dots, t_{i_{p}}, \TT D[t_{j_{1}},\dots, t_{j_{q}}]]\redperm \TT  E^{\lnot\nu}[t_{i_{1}},\dots, t_{i_{p}},\nu b. \TT D[t_{j_{1}},\dots, t_{j_{q}}]]$; by arguing in this way on all binders $\nu$ we finally obtain \eqref{eq:nuzz0}.

We prove \eqref{eq:nuzz} by induction on the construction of $\TT C$:
\begin{varitemize}
\item if $\TT C=[\ ]_{i}$, then the claim trivially holds;

\item if $\TT C= \TT C_{1}\oplus_{a}^{j}\TT C_{2}$, then $[\ ]_{k+1}$ occurs in either $\TT C_{1}$ or $\TT C_{2}$, and not both; let us say $[\ ]_{k+1}$ occurs in $\TT C_{1}$; we then have that 
\begin{align*}
\nu b.\TT C[t_{1},\dots, t_{k},t] &= \nu b. \TT C_{1}[t_{i_{1}},\dots, t_{i_{p}},t]\oplus_{a}^{j}\TT C_{2}[t_{j_{1}},\dots, t_{j_{q}}] \\
& \redperm
 \big(\nu b. \TT C_{1}[t_{i_{1}},\dots, t_{i_{p}},t]\big) \oplus_{a}^{i} \big ( \nu b.\TT C_{2}[t_{j_{1}},\dots, t_{j_{q}}]\big) \\
 & \redperm
\big(\nu b. \TT C_{1}[t_{i_{1}},\dots, t_{i_{p}},t]\big) \oplus_{a}^{i}\TT C_{2}[t_{j_{1}},\dots, t_{j_{q}}] \\
& \stackrel{\tiny\text{I.H.}}{\redperm^{*}}
\TT C_{1}[t_{i_{1}},\dots, t_{i_{p}}, \nu b.t] \oplus_{a}^{i}
\TT C_{2}[t_{j_{1}},\dots, t_{j_{q}}] 
=
 \TT C[  t_{1},\dots, t_{k}, \nu b.t]
\end{align*}
\end{varitemize}

\end{proof}

Still one technical lemmas before proceeding to the proof of Theorem \ref{thm:completenessa}.

\begin{lemma}\label{lemma:delirio}

Let $N,K\geq 1$, $X=\{a_{1},\dots, a_{K}\}$ be distinct names, and for all $i\in N$ and $j\in K$, $\bone_{ij}$ be a conjunction of literals of name $a_{j}$ such that for all $i\neq i'$:
\begin{varitemize}
 \item[a.] for all $k\leq K$, if $\bigwedge_{j=1}^{k}\bone_{ij}\neq \bigwedge_{j=1}^{k}\bone_{i'j}$, then 
 $\left(\bigwedge_{j=1}^{k}\bone_{ij}\right)\land \left(\bigwedge_{j=1}^{k}\bone_{ij}\right)\vDash \BOT$;
 \item[b.] $\left(\bigwedge_{j=1}^{K}\bone_{ij}\right)\land\left( \bigwedge_{j=1}^{K}\bone_{i'j}\right)\vDash \BOT$.
 \end{varitemize}
 
Suppose $\Gamma\vdash^{X,q_{i}}t: \btwo\land\left( \bigwedge_{j=1}^{K}\bone_{ij}\right)\Pto \sigma$ is derivable for all $i\in \{1,\dots,N\}$. Then there exists a derivation of $\Gamma \vdash^{\emptyset}\nu a_{1}.\dots.\nu a_{K}.t:\btwo \Pto \B C^{r}\sigma$, where $r= \sum_{i=1}^{N}\left(q_{i}\cdot\prod_{j=1}^{K} \mu(\bone_{ij})\right)$.

\end{lemma}
\begin{proof}
We argue by induction on $K$. If $K=1$ then $i\neq i'$ implies $\bone_{i}\land \bone_{i'}\vDash \BOT$, so we can conclude as follows:

\begingroup\makeatletter\def\f@size{10}\check@mathfonts
\begin{lrbox}{\mypti}
\begin{varwidth}{\linewidth}
$
\AXC{$D_{i}$}
\noLine
\UIC{$\Gamma\vdash^{\{a\}}t: \btwo\land \bone_{i}\Pto\B C^{q_{i}} \sigma$}
\DP
$
\end{varwidth}
\end{lrbox}


$$
\AXC{$\left\{ \usebox{\mypti}\right\}_{i=1,\dots, N}$}
\RL{$(\mu_{\Sigma})$}
\UIC{$\Gamma\vdash^{\emptyset}\nu a.t: \btwo\Pto\B C^{\sum_{i=1}^{N}q_{i}\cdot \mu(\bone_{i})} \sigma$}
\DP
$$
\endgroup

If $K>1$, then 
for all $i,i'\in\{1,\dots,N\}$, let $i\sim i'$ if $\bone_{iK}\land \bone_{i'K}\vDash \BOT$
 and for all $j\in\{1,\dots K-1\}$, $\bone_{ij}\equiv \bone_{i'j}$. 
Let $\mathsf{Cl}(N_{\sim})$ indicate the set of cliques of the relation $\sim$, and for each $U\in \mathsf{Cl}(N_{\sim})$, fix some $i_{U}\in U$. Then we claim that the following equation holds: 
\begin{equation}\label{eq:cliques}
\sum_{i=1}^{N}q_{i}\cdot\prod_{j=1}^{K}\mu(\bone_{ij})=
\sum_{U\in \mathsf{Cl}_{N_{\sim}}}\left(\prod_{j=1}^{K-1}
\mu(\bone_{i_{U}j})\right)\cdot \left(\sum_{i\in U}q_{i}\cdot\mu(\bone_{iK})\right)
\end{equation}

Indeed we can compute
\begin{align*}
\sum_{i=1}^{N}\left(q_{i}\cdot\prod_{j=1}^{K} \mu(\bone_{ij})\right) &=
\sum_{U\in \mathsf{Cl}_{N_{\sim}}}\sum_{i\in U}\left (q_{i}\cdot\prod_{j=1}^{K}\mu(\bone_{ij}) \right)\\
&=
\sum_{U\in \mathsf{Cl}_{N_{\sim}}}\sum_{i\in U}\left(\prod_{j=1}^{K-1}\mu(\bone_{i_{U}j})\right) \cdot\Big (q_{i}\cdot \mu(\bone_{iK})\Big) \\
&=
\sum_{U\in \mathsf{Cl}_{N_{\sim}}} \left(\prod_{j=1}^{K-1}\mu(\bone_{i_{U}j})\right) \cdot\left(\sum_{i\in U} q_{i}\cdot\mu(\bone_{iK})\right)
\end{align*}

Let us now show how to construct, for all $U\in \mathsf{Cl}_{N_{\sim}}$, a derivation $D_{U}$ of 
$\Gamma \vdash^{\{a_{1},\dots, a_{K-1}\}}\nu a_{K}.t: \btwo\land \left (\prod_{j=1}^{K-1}\bone_{i_{U}j}\right)\Pto \B C^{q_{U}}\sigma$, where $q_{U}=\sum_{i\in U}q_{i}\cdot \mu(\bone_{iK})$:

\begingroup\makeatletter\def\f@size{10}\check@mathfonts
\begin{lrbox}{\mypti}
\begin{varwidth}{\linewidth}
$
\AXC{$D_{i_{U}}$}
\noLine
\UIC{$\Gamma \vdash^{X}t: \btwo\land \left (\bigwedge_{j=1}^{K-1}\bone_{i_{U}j}\right) \land \bone_{i_{U}K} \Pto \B C^{q_{i}}\sigma$}
\DP
$
\end{varwidth}
\end{lrbox}

$$
D_{U} \quad = \quad
\AXC{$\left\{ \usebox{\mypti}\right\}_{j\in U}$}
\RL{$(\mu_{\Sigma})$}
\UIC{$\Gamma \vdash^{X-\{a_{K}\}}\nu a_{K}.t: \btwo\land \left (\bigwedge_{j=1}^{K-1}\bone_{i_{U}j}\right)  \Pto \B C^{q_{U}}\sigma$}
\DP
$$
\endgroup

Let $\bthree_{U}= \bigwedge_{j=1}^{K-1}\bone_{i_{U}j}$, $t'= \nu a_{K}.t$ and observe that we now have derivations $D_{U}$ of 
$\Gamma \vdash^{\{a_{1},\dots, a_{K-1}\}, q_{U}} t': \btwo\land \bthree_{U}\Pto \sigma$, where for all 
$U\neq V$, $\bthree_{U}\land \bthree_{V}\vDash \BOT$. 
Moreover, for all $U\neq V$, from $\bthree_{U}\not\equiv\bthree_{V}$ we deduce by condition (a.) that $\bthree_{U}\land \bthree_{V}\vDash \BOT$; finally condition (b.) for the $\bthree_{U}$ follows immediately from condition (b.) for the $\bigwedge_{j=1}^{K}\bone_{ij}$.

Hence, by applying the induction hypothesis, we deduce the existence of a derivation of 
$\Gamma \vdash^{\emptyset}\nu X.t: \btwo\Pto \B C^{q}\sigma$, where 
\begin{align*}
q& =\sum_{U\in \mathsf{Cl}_{N_{\sim}}}q_{U}\cdot \mu(\bthree_{U}))  =
\sum_{U\in \mathsf{Cl}_{N_{\sim}}}\left (q_{U}\cdot \prod_{j=1}^{K-1}\mu(\bone_{i_{U}j}))\right) \\
& =  
\sum_{U\in \mathsf{Cl}_{N_{\sim}}}\left (\sum_{i\in U}q_{i}\cdot \mu(\bone_{iK})\right)\cdot \left( \prod_{j=1}^{K-1}\mu(\bone_{i_{U}j}))\right) \\
& \stackrel{\eqref{eq:cliques}}{=}
\sum_{i=1}^{N} \left (q_{i}\cdot \prod_{j=1}^{K}\mu(\bone_{ij})\right)
\end{align*}
so we are done.
\end{proof}

The technical lemma just proved can be used to establish Proposition \ref{cor:mustar}, i.e.~the admissibility of the generalized counting rule $(\mu^{*})$.

\begin{proposition}\label{cor:mustar}
The following generalized counting rule is derivable in $\TCINT$:
$$
\AXC{$\Gamma\vdash^{\{a_{1},\dots, a_{K}\}} t: \bone\Pto \B C^{q}\sigma $}
\RL{$(\mu^{*})$}
\UIC{$\Gamma \vdash^{\emptyset}\nu a_{1}.\dots.\nu a_{K}.t:\TOP \Pto \B C^{ q\cdot \mu(\bone)}\sigma$}
\DP
$$
\end{proposition}
\begin{proof}
First, let us transform $\bone$ into a disjunctive normal form of the form $\bigvee_{i=1}^{N}\bigwedge_{j=1}^{K}\btwo_{ij}$, where each $\btwo_{ij}$ is a conjunction of literals of name $a_{j}$. Now, as it was done in the proof of Lemma \ref{lemma:cool1}, the $\btwo_{ij}$ can be turned into conjunctions of literals $\bfour_{ij}$, with $i\leq N'$, for some $N'\geq N$, and $j\leq K$, where for all $i\neq i'$ and $j,j'$, either $\bfour_{ij}=\bfour_{i'j'}$ or $\bfour_{ij}\land \bfour_{i'j'}\vDash \BOT$, and 
$\bigvee_{i=1}^{N'}\bigwedge_{j=1}^{K}\bfour_{ij}\equiv\bigvee_{i=1}^{N}\bigwedge_{j=1}^{K}\btwo_{ij}$.

Observe that the $\bfour_{ij}$ now satisfy conditions (a.) and (b.) of Lemma \ref{lemma:delirio}. 
Moreover, by letting $\bfive_{i}=\bigwedge_{j=1}^{K}\bfour_{ij}$, from $\bfive_{i} \vDash \bone$ we deduce that for all $i=1,\dots,N'$, 
$\Gamma \vdash^{\{a_{1},\dots, a_{K}\}, q}t: \bfive_{i}\Pto \sigma$ is derivable.
Hence, by applying the Lemma, we can conclude that 
$\Gamma \vdash^{\emptyset,  r} \nu a_{1}.\dots.\nu a_{K}.t: \TOP\Pto \sigma$, where 
$r= \sum_{i=1}^{N'}q\cdot \mu(\bfive_{i})=q \cdot \sum_{i=1}^{N'} \mu(\bfive_{i})=q\cdot \mu(\bone)$. 
\end{proof}

We can finally turn to the proof of the completeness part Theorem \ref{thm:completenessa}.

\begin{proposition}
For any closed term $t$, 
\begin{align*}
\NHNF(t) & \leq \sup \{ q\mid\  \vdash t: \TOP \Pto \B C^{q}\HNORM\} \\
\NNF(t) & \leq \sup \{ q\mid \  \vdash t: \TOP \Pto\B C^{q} \NORM\} 
\end{align*}

\end{proposition}
\begin{proof}First observe that (with the notations from Subsection \ref{subs:evl})
\begin{align*}
\NHNF(t)&=\sup\left \{ \sum_{\pi:v\mapsto w\in \HNF}\mu(\bone^{v}_{\pi})\mid t\redall_{\mathsf h}^{*} v\right \}\\
&
=
\sup\left \{ \sum_{i=1}^{N}\mu(\bone^{u}_{\pi_{i}})\mid N\in \BB N, u_{1},\dots, u_{N}\in \HNF, \ \pi_{1}:v\mapsto u_{1},\dots, \pi_{N}:t\mapsto u_{N}, t\redall_{\mathsf h}^{*}v\right\}
\end{align*}
so it is enough to show that for any $n\in \BB N$, term $v$ such that $t\redall_{\mathsf h}^{*}v$ and terms $u_{1},\dots, u_{n}\in \HNF$ together with randomized paths $\pi_{i}:v\mapsto u_{i}$, 
we can prove that $\vdash t: \TOP \Pto\B C^{q} \HNORM$ is derivable, with $q=\sum_{i=1}^{n}\mu(\bone_{\pi_{i}})$.
Notice that, by subject expansion, it suffices to show $\vdash v: \TOP \Pto\B C^{q} \HNORM$.

Let $\pi_{1}:v\mapsto u_{1},\dots, \pi_{k}:t\mapsto u_{k}$ be randomized paths, with $u_{1},\dots, u_{n}\in \HNF$.
By Lemma \ref{lemma:mcox1} there exists $N$, terms $v_{1},\dots, v_{N}$ and a randomized multi-context $\TT C$ of rank $k+N$ such that $v\redperm^* \TT C[u_{1},\dots, u_{k}, v_{1},\dots, v_{N}]$.
Moreover, by Lemma \ref{lemma:mcox2} there exists a permutative reduction
$\nu a_{1}\dots \nu a_{K}. \TT C^{\lnot\nu}[u_{1},\dots, u_{k}, v_{1},\dots, v_{N}]\redperm^{*}
\TT C[u_{1},\dots, u_{k}, v_{1},\dots, v_{N}]$.

Let $v^{*}$ be the term $\TT C^{\lnot\nu}[u_{1},\dots, u_{k}, v_{1},\dots, v_{N}]$.
For each HNV $u_{i}$, by Lemma \ref{lemma:hnorm}, we have that $\vdash u_{i}:\TOP \Pto \B C^{1}\HNORM$, so in particular 
$\vdash^{\{a_{1},\dots, a_{K}\}, 1}u_{i}:\TOP\Pto \B C^{1}\HNORM$ is derivable; using this fact one can construct by induction on $\TT C$ derivations $D_{i}$ of 
$\vdash^{\{a_{1},\dots, a_{K}\}}v^{*}: \bone_{\pi_{i}}\Pto \B C^{1}\HNORM$
(using the fact that $RBT(\TT C[u_{1},\dots, u_{k}, v_{1},\dots, v_{N}])\sqsubseteq RBT(t)$)
:
\begin{varitemize}

\item if $\TT C^{\lnot\nu}= [\ ]_{i}$, then $k=1$, $u^{*}=u$, and $\bone_{\pi_{1}}=\TOP$; then by the hypothesis we have 
$\vdash v^{*}: \TOP\Pto\B C^{1} \HNORM$;

\item if $\TT C^{\lnot\nu}= \TT C_{1} \oplus_{a}^{l}\TT C_{2}$, then $u^{*}=
\TT C_{1}[u_{i_{1}},\dots, u_{i_{p}}, v_{i'_{1}},\dots, v_{i'_{r}}]\oplus_{a}^{l}
\TT C_{2}[u_{j_{1}},\dots, u_{j_{q}}, v_{j'_{1}},\dots, v_{j'_{s}}]$, where $p+q=k$, $r+s=N$ and 
$I\cap J=\{i_{1},\dots, i_{p}\}\cap \{j_{1},\dots, j_{q}\}=\emptyset$, 
$I'\cap J'=\{i'_{1},\dots, i'_{r}\}\cap \{j'_{1},\dots, j'_{s}\}=\emptyset$. 
Moreover, we have that for all $m\in I$, $\bone_{\pi_{l}}=\bvar_{a}^{l}\land \bone_{\pi'_{m}}$, where $\pi'_{m}$ is the path of $u_{m}$ in $\TT C_{1}$; similarly, for all $m\in J$, $\bone_{\pi_{m}}=\lnot\bvar_{a}^{l}\land \bone_{\pi'_{m}}$, where $\pi'_{m}$ is the path of $u_{l}$ in $\TT C_{2}$.

By the induction hypothesis we then have derivations of $\vdash^{\{a_{1},\dots, a_{K}\}}
\TT C_{1}[u_{i_{1}},\dots, u_{i_{p}}, v_{i'_{1}},\dots, v_{i'_{r}}]  : \bone_{\pi'_{m}}\Pto \B C^{1}\HNORM$ and 
$\vdash^{\{a_{1},\dots, a_{K}\}}
\TT C_{2}[u_{j_{1}},\dots, u_{j_{q}}, v_{j'_{1}},\dots, v_{j'_{s}}]: \bone_{\pi'_{m}}\Pto \B C^{1}\HNORM$.
Since $\bone_{\pi_{i}}$ is either 
$\bvar_{a}^{l}\land \bone_{\pi'_{i}}$ or $\lnot\bvar_{a}^{l}\land \bone_{\pi'_{i}}$,
 on $i$, we obtain then in any case a derivation of 
$\vdash^{\{a_{1},\dots, a_{K}\}}v^{*}:\bone_{\pi_{i}}\Pto\B C^{1} \HNORM$ using either the $(\oplus l)$- or the $(\oplus r)$-rule. 

\end{varitemize}

Using the derivations $D_{i}$ we thus obtain a derivation of $\vdash^{\emptyset} \nu a_{1}.\dots.\nu a_{K}.v^{*}: \TOP \Pto \B C^{ \sum_{i}\mu(\bone_{\pi_{i}})}\HNORM$ by applying the derivable rule $(\mu^{*})$ (Corollary \ref{cor:mustar}) as shown below:

\medskip

\begingroup\makeatletter\def\f@size{10}\check@mathfonts
\begin{lrbox}{\mypti}
\begin{varwidth}{\linewidth}
$
\AXC{$D_{i}$}
\noLine
\UIC{$\vdash^{\{a_{1},\dots, a_{K}\}}v^{*}:\bone_{\pi_{i}}\Pto\B C^{1} \HNORM$}
\DP
$
\end{varwidth}
\end{lrbox}


$$\AXC{$\left\{ \usebox{\mypti}\right\}_{i=1,\dots,k}$}
\AXC{$ \bigvee_{i}\bone_{\pi_{i}} \vDash \bigvee_{i}\bone_{\pi_{i}}$}
\RL{$(\lor)$}
\BIC{$\vdash^{\{a_{1},\dots,a_{K}\}}v^{*}:\bigvee_{i}\bone_{\pi_{i}}\Pto\B C^{1} \HNORM$}
\RL{$(\mu^{*})$}
\UIC{$\vdash^{\emptyset} \nu a_{1}.\dots.\nu a_{K}.v^{*}: \TOP \Pto\B C^{\mu(\bigvee_{i}\bone_{\pi_{i}})} \HNORM$}
\DP
$$
\endgroup

%
%
%
%
%
%
%
%
%

\medskip

Since $\mu(\bigvee_{i}\bone_{\pi_{i}})= \sum_{i}\mu(\bone_{\pi_{i}})$, using subject reduction and subject expansion we can finally deduce 
$\vdash^{\emptyset}v: \TOP\Pto\B C^{\sum_{i}\mu(\bone_{\pi_{i}})} \HNORM $.

For the case of $\NNF(t)$, we establish the result first for $\NF(t)$, arguing by induction on normal forms with an argument similar to the one above, exploiting the fact that if 
$q_{i}=\sup S_{i}$, then 
$\prod_{i=1}^{k}q_{i}= \sup\{ \prod_{i=1}^{k}s_{i}\mid s_{i}\in S_{i}\}$.
\end{proof}

\section{Normalization Results}\label{app:normalization}

In this section we establish a few normalization results for $\LCPL$ and $\TCINT$. 
All results are based on some variant of the standard reducibility predicates technique.

First, we prove Theorem \ref{thm:detnormalization} by adapting a result from \cite{DLGH} to the language $\EVLL$.
We then turn to probabilistic normalization. In Subsection \ref{subs:norma1} we develop a theory of reducibility candidates for probabilistic head normal forms and probabilistic normal forms for $\EVL$ and for a general family of types which comprises both those of $\LCPL$ an $\TCINT$. This will allows us to establish the normalization theorem for $\TCINT$, but will also provide most of the steps needed to establish the normalization theorem for $\LCPL$. In Subsection \ref{subs:norma2} we discuss reducibility candidates for $\EVLL$, by adapting most of results from the previous subsection, and we complete the proof of Theorem \ref{thm:normalization}.

\subsection{Deterministic Normalization.}\label{app:detnorm}

\subsubsection{Deterministic Normalization of $\LCPL$.}
In this section we adapt the proof of [\cite{DLGH}, Theorem 1] to $\LCPL$ , in order to prove Theorem \ref{thm:detnormalization}:

\detnormalization*

We introduce a simply typed $\lambda$-calculus $\STLC$ for $\EVLL$ and we show that typable terms are strongly normalizable under $\redalll$. Simple types are defined by the grammar below
$$
\sigma::= o \mid \sigma \To \sigma
$$
A judgement of $\STLC$ is of the form $\Gamma \vdash^{\theta}t: \sigma$, where $\theta$ is a \emph{list} of names and $\FN(t)\subseteq \theta$ (observe that we consider here lists and not sets). 
The rules of $\STLC$ are illustrated in Fig.~\ref{fig:rulesimple}.

The following is easily established by inspecting all reduction rules:
\begin{proposition}[subject reduction]
If $\Gamma \vdash^{\theta} t:\sigma $ and $t\redalll u$, then $\Gamma \vdash^{\theta} u:\sigma$.
\end{proposition}

For any type $\FF s$ of $\LCPL$, let us define the corresponding simple type $|\FF s|$ by letting
$|o|=o$, $|\FF s\To \sigma|=|\FF s|\To |\sigma|$ and $|\B C^{q}\FF s|=|\FF s|$. The following is easily checked by induction:

\begin{proposition}\label{prop:lcpltostlc}
If $\Gamma \vdash^{X}_{\lnot \vee}t: \bone \Pto \FF s$ holds $\LCPL$, then
$|\Gamma |\vdash^{\theta} t: |\FF s|$ is derivable in $\STLC$, for any suitable $\theta$.
\end{proposition}

\begin{figure}
\fbox{
\begin{minipage}{0.98\textwidth}
$$
\AXC{}
\UIC{$\Gamma, x:\sigma \vdash^{\theta} x:\sigma$}
\DP
$$

$$
\AXC{$\Gamma, x:\sigma \vdash^{\theta} t:\tau$}
\UIC{$\Gamma \vdash^{\theta} \lambda x.t:\sigma \To \tau$}
\DP
\qquad\qquad
\AXC{$\Gamma \vdash^{\theta} t: \sigma\To \tau$}
\AXC{$\Gamma \vdash^{\theta} u:\sigma$}
\BIC{$\Gamma \vdash^{\theta} tu: \tau$}
\DP
$$

$$
\AXC{$\Gamma \vdash^{\theta} t:\sigma$}
\AXC{$\Gamma \vdash^{\theta} u:\sigma$}
\AXC{$a\in \theta$}
\TIC{$\Gamma \vdash^{ \theta} t\oplus_{a}^{i}u:\sigma$}
\DP
\qquad\qquad
\AXC{$\Gamma \vdash^{ \theta} t: \sigma$}
\UIC{$\Gamma \vdash^{\theta-\{a\}} \nu a.t:\sigma$}
\DP
\qquad\qquad
\AXC{$\Gamma \vdash^{\theta} t: \sigma\To \tau$}
\AXC{$\Gamma \vdash^{\theta} u:\sigma$}
\BIC{$\Gamma \vdash^{\theta} \{t\} u: \tau$}
\DP
$$

\end{minipage}
}
\caption{Typing rules for the simply typed $\lambda$-calculus $\STLC$.}
\label{fig:rulesimple}
\end{figure}

Let $\HLnu$ indicate the sets of $\lambda$-terms of $\EVLL$, and for all finite set of names $X$, let $\HLnu^{X}$ indicate the $\lambda$-terms with free names in $X$.

For any term $t$ with $\FN(t)\subseteq \theta= a_{1}\dots a_{k}$, we let $\nu^{\theta}t$ indicate the name-closed term
$\nu a_{1}.\dots. \nu a_{k}.t$.
For any list of names $\theta$, let 
$\SN^{\theta}= \{ t\in \HLnu^{X}\mid \nu^{\theta} t \text{ is strongly normalizing}\}$, and for any $t\in \SN^{\theta}$, let 
$\sn^{\theta}(t)$ be the maximum length of a reduction of $\nu^{\theta}t$.

We let a \emph{$\{\}$-redex} be a redex for any of the tree reductions rules 
\eqref{eq:e0}-\eqref{eq:e2}.

%

\begin{lemma}\label{lemma:cases}
\begin{varitemize}
\item[(i.)] if $t v_{1}\dots v_{n}\in \SN^{\theta}$, $ut v_{1}\dots v_{n}\in \SN^{\theta}$, and $a$ occurs in $\theta$; then 
$(t\oplus_{a}^{i} u) v_{1}\dots v_{n}\in \SN^{\theta}$;

\item[(ii.)] if $t v_{1}\dots v_{n}\in \SN^{a\cdot \theta}$, where $\FN(v_{i})\subseteq X$, then 
$(\nu a.t) v_{1}\dots v_{n}\in \SN^{\theta}$;

\item[(iii.)] if $v_{i}\in\SN^{\theta}$, then 
$x v_{1}\dots v_{n}\in \SN^{\theta}$;

\item[(iv.)] if $(t[u/x]) v_{1}\dots v_{n}\in \SN^{\theta}$ and $u\in \SN^{\theta}$, then 
$(\lambda x.t)u v_{1}\dots v_{n}\in \SN^{\theta}$;

\end{varitemize}
\end{lemma}
\begin{proof}
\begin{varitemize}
\item[(i.)] Same argument as [\cite{DLGH}, Lemma 24],
 observing that a $\{\}$-redex in $(t\oplus_{a}^{i} u)v_{1}\dots v_{n}$ can only occur inside one of $t,u,v_{1},\dots, v_{n}$.

\item[(ii.)] Same argument as [\cite{DLGH}, Lemma 25], again observing that a $\{\}$-redex in $(\nu a.t)v_{1}\dots v_{n}$ can only occur inside one of $t,u,v_{1},\dots, v_{n}$.

\item[(iii.)] Immediate.

\item[(iv.)] Same argument as [\cite{DLGH}, Lemma 27], again observing that a $\{\}$-redex in 
$(\lambda x.t)uv_{1}\dots v_{n}$
can only occur inside one of $t,u,v_{1},\dots, v_{n}$.

\end{varitemize}

\end{proof}
\begin{lemma}\label{lemma:snx}
Let $t^{\bullet}$ be obtained from $t$ by deleting a finite number of $\nu$-binders $\nu a_{1},\dots, \nu a_{k}$.
If $t\in \SN^{\theta}$, then $t^{\bullet}\in \SN^{\theta'}$, for any suitable $\theta'$.
\end{lemma}
\begin{proof}

We argue by induction on $|t|+\sn^{\theta}(t)$:

\begin{varitemize}

\item if $t=x$ or $t=y$, the claim is immediate;
\item if $t=\lambda y.t'$, the claim follows by IH;
\item if $t= t_{1}t_{2}$, then $t^{\bullet}= t_{1}^{\bullet}t_{2}^{\bullet}$; let us consider all possible reductions of $t^{\bullet}$:
\begin{varitemize} 
\item a reduction of either $t_{1}^{\bullet}$ or $t_{2}^{\bullet}$, in which case we apply the IH;

\item a reduction $(\lambda y. u^{\bullet})t_{2}^{\bullet}\redbeta u^{\bullet}[t_{2}^{\bullet}/y]$, where $t_{1}=\lambda y.u$, in which case we also apply the IH;

\item a reduction of the form $u^{\bullet}t_{2}^{\bullet} \redalll v^{\bullet}$, where $t_{1}=\nu a.u$, then  $t_{1}t_{2}\redpermm \nu a.ut_{2}$, so we can apply the IH to $ut_{2}$.

\end{varitemize}

\item if $t=\{t_{1}\}t_{2}$, then $t^{\bullet}=\{ t_{1}^{\bullet}\}t_{2}^{\bullet}$; then a reduction is either a reduction of either $t_{1}^{\bullet}$ or $t_{2}^{\bullet}$, in which case we apply the IH, either a $\eqref{eq:e1},\eqref{eq:e2}$-reduction, in which case we also apply the IH, or a reduction 
of the form $\{t_{1}^{\bullet}\} \nu a.u^{\bullet}\redpermm \nu a. t_{1}^{\bullet}u^{\bullet}$, where $t_{2}=\nu a.u$, in which case we also have $\{t_{1}\}t_{2}\redpermm \nu a. t_{1}u$, so we can apply the IH.

\item if $t= t_{1}\oplus_{a}^{i}t_{2}$, then a reduction of $t^{\bullet}$ is either a reduction of $t_{1}^{\bullet}$ or $t_{2}^{\bullet}$, in which case we apply the IH, or a reduction of the form
$t_{1}^{\bullet}\oplus_{a}^{i}(u^{\bullet}\oplus_{b}^{j}v^{\bullet})\redpermm (t_{1}^{\bullet}\oplus_{a}^{i}u^{\bullet})\oplus_{b}^{j}(t_{1}^{\bullet}\oplus_{a}^{i}v^{\bullet})$, where $t_{2}=v\oplus_{a}^{i}w$, in which case $t \redpermm (t_{1}\oplus_{a}^{i}u)\oplus_{b}^{j}(t_{2}\oplus_{a}^{i}v)$, so we can apply the IH;

\item if $t=\nu b.u$, then if $t^{\bullet}=u^{\bullet}$, we can argue by the IH; if $t^{\bullet}= \nu b.u^{\bullet}$, then a reduction of $t^{\bullet}$ is either a reduction of $u^{\bullet}$, in which case we can apply the IH, or a permutation $\nu b. u_{1}^{\bullet}\oplus_{a}^{i}u_{2}^{\bullet}\redpermm (\nu b.u_{1}^{\bullet})\oplus_{a}^{i}(\nu b.u_{2}^{\bullet})$, where $u=u_{1}\oplus_{a}^{i}u_{2}$, in which case $t\redpermm (\nu b.u_{1})\oplus_{a}^{i}(\nu b.u_{2})$ so we can apply the IH. 
%
%
%

\end{varitemize}

\end{proof}

\begin{lemma}\label{lemma:snxx}
\begin{varitemize}
\item[(i.)] If $ (\nu a.tu)v_{1}\dots v_{n} \in \SN^{\theta}$, then $\{t\}(\nu a.u)v_{1}\dots v_{n}\in \SN^{\theta}$.
\item[(ii.)] If $ t[\nu a.u/x]v_{1}\dots v_{n}\in \SN^{\theta}$, then $ t[u/x]v_{1}\dots v_{n} \in \SN^{ \theta'}$, for any suitable $\theta'$.
\item[(iii.)] If $(tu)v_{1}\dots v_{n}\in \SN^{\theta}$, then $(\{t\}u)v_{1}\dots v_{n}\in \SN^{\theta}$.
\end{varitemize}
\end{lemma}
\begin{proof}
\begin{varitemize}
\item[(i.)]
We argue by induction on $|t|+\sn^{\theta}(t)+\sn^{\theta}(v)+\sum_{i}\sn^{\theta}(v_{i})$.

A reduction of $\{t\}(\nu a.u)v_{1}\dots v_{n}$ can only be a reduction in either of $t,u,v_{1}\dots,v_{n}$, a reduction of the form
$\{t_{1}\oplus_{a}^{i}t_{2}\}(\nu a.u)v_{1}\dots v_{n}
\redpermm (\{t_{1}\}(\nu a.u)v_{1}\dots v_{n})\oplus_{a}^{i}  (\{t_{1}\}(\nu a.u)v_{1}\dots v_{n})
$, or a reduction
$\{t\}(\nu a.u)v_{1}\dots v_{n}\redpermm (\nu a.tu)v_{1}\dots v_{n}$. In the first case we can use the IH; in the second case by the IH we deduce  $(\{t_{i}\}(\nu a.u)v_{1}\dots v_{n})\in \SN^{\theta}$ and we can conclude by Lemma \ref{lemma:cases} (i.). In the third case we conclude by the hypothesis.

\item[(ii.)] 
Immediate application of Lemma \ref{lemma:snx}.
\item[(iii.)]

We argue by induction on $n+ |t|+|u|+ \sn^{\theta}(t)+\sn^{\theta}(u)+ \sum_{i}\sn^{\theta}(v_{i})$.

A reduction of $(\{t\}u)v_{1}\dots v_{n}$ is of one of the following forms:
\begin{varitemize}
\item a reduction of $t$ or of $u$, $v_{1}$,\dots, $v_{n}$ in which case we argue by the IH;
\item a reduction of the form $(\{t\}\nu a.u' )v_{1}\dots v_{n}\redpermm (\nu a.tu')v_{1}\dots v_{n}$, where $u= \nu a.u'$; then by Lemma \ref{lemma:snx}, from $(tu)v_{1}\dots v_{n}\in \SN^{\theta}$ we deduce $(tu')v_{1}\dots v_{n}\in \SN^{X\cup\{a\}}$, and by Lemma \ref{lemma:cases} (ii.), 
$(\nu a.tu')v_{1}\dots v_{n}\in\SN^{\theta}  $;

\item a reduction of the form $(\{t_{1}\oplus_{a}^{i}t_{2}\}u)v_{1}\dots v_{n}\redpermm
((\{t_{1}\}u)\oplus_{a}^{i}(\{t_{2}\}u))v_{1}\dots v_{n}$,where $t=t_{1}\oplus_{a}^{i}t_{2}$; then by IH
$(\{t_{i}\}u)v_{1}\dots v_{n}\in \SN^{\theta}$, so we conclude $((\{t_{1}\}u)\oplus_{a}^{i}(\{t_{2}\}u))v_{1}\dots v_{n}$ by Lemma \ref{lemma:cases} (i.).

\item a reduction of the form $(\{t\}(u_{1}\oplus_{a}^{i}u_{2}))v_{1}\dots v_{n}\redpermm
((\{t\}u_{1})\oplus_{a}^{i}(\{t\}u_{2}))v_{1}\dots v_{n}$; 
then by IH
$(\{t\}u_{i})v_{1}\dots v_{n}\in \SN^{\theta}$, so we conclude $((\{t\}u_{1})\oplus_{a}^{i}(\{t\}u_{2}))v_{1}\dots v_{n}$ by Lemma \ref{lemma:cases} (i.).

\item a reduction of the form $((\{t\}u)v_{1}\dots v_{k})(v_{k+1}^{1}\oplus_{a}^{i}v_{k+1}^{2}) v_{k+2}\dots v_{n}\redpermm
 ((\{t\}u)v_{1}\dots v_{k}v_{k+1}^{1})\oplus_{a}^{i}((\{t\}u)v_{1}\dots v_{k}v_{k+1}^{2})v_{k+2}\dots v_{n}$; then by IH 
$(\{t\}u)v_{1}\dots v_{k}v_{k+1}^{1}\in \SN^{\theta}$,
$(\{t\}u)v_{1}\dots v_{k}v_{k+1}^{2}\in \SN^{\theta}$, and by Lemma \ref{lemma:cases} (i.) we conclude
$ ((\{t\}u)v_{1}\dots v_{k}v_{k+1}^{1})\oplus_{a}^{i}((\{t\}u)v_{1}\dots v_{k}v_{k+1}^{2})v_{k+2}\dots v_{n}\in \SN^{\theta}$.
\end{varitemize}
\end{varitemize}
\end{proof}

The family of reducibility predicates $\RRED_{\sigma}^{\theta}$ of $\STLC$  is defined as follows:
\begin{align*}
\RRED^{\theta}_{o}& := \SN^{\theta}\\
\RRED^{\theta}_{\sigma\To \tau}& :=
\{t\in \HLnu^{X}\mid \forall u\in \RRED^{\theta}_{\sigma}, \ tu\in \RRED^{\theta}_{\tau}\}
\end{align*}

\begin{lemma}
\begin{enumerate}
\item If $t\in \RRED^{\theta}_{\sigma}$, then $t\in \SN^{\theta}$;
\item if $\Gamma \vdash xu_{1}\dots u_{n}$ and $u_{1},\dots, u_{n}\in \SN^{\theta}$, then 
$ xu_{1}\dots u_{n}\in \RRED^{\theta}_{\sigma}$;

\item If $t[u/x]v_{1}\dots v_{n}\in \RRED^{\theta}_{\sigma}$, 
%
and $u\in \SN^{\theta}$, then
$ (\lambda x.t)uv_{1}\dots v_{n}\in \RRED^{\theta}_{\sigma}$;

\item if $tv_{1}\dots v_{n}\in \RRED^{\theta}_{\sigma}$ and 
 $ uv_{1}\dots v_{n}\in \RRED^{\theta}_{\sigma}$, and $a\in \theta$, then 
 $( t\oplus_{a}^{i}u)v_{1}\dots v_{n}\in \RRED^{\theta}_{\sigma}$

\item if $ tv_{1}\dots v_{n}\in \RRED^{a\cdot \theta}_{\sigma}$ and $\FN(v_{i})\subseteq X$, then
$ (\nu a.t)v_{1}\dots v_{n}\in \RRED^{\theta}_{\sigma}$.
\end{enumerate}

\end{lemma}
\begin{proof}
The proof is similar to the proof of [\cite{DLGH}, Lemma 28].
\end{proof}

\begin{lemma}\label{lemma:permuzza}
If $ t\in \RRED^{\theta}_{\sigma\To \tau}$ and $u\in \RRED^{\theta}_{\sigma}$, then 
$\{t\}u\in \RRED^{\theta}_{\tau}$.
\end{lemma}
\begin{proof}
Let $\tau=\sigma_{1}\To \dots \To \sigma_{n}\To o$ and let $v_{1}\in \RRED^{\theta}_{\sigma_{1}}, \dots, v_{n}\in \RRED^{\theta}_{\sigma_{n}}$.
Then $(tu)v_{1}\dots v_{n}\in \SN^{\theta}$ which, by Lemma \ref{lemma:snxx}, implies $\{t\}uv_{1}\dots v_{n}\in \SN^{\theta}$, from which we conclude $\{t\}u\in \RED^{X}_{\tau}$.
\end{proof}

\begin{proposition}
If $\Gamma \vdash^{\theta} t:\sigma$ is derivable in $\STLC$, where $\Gamma=\{x_{1}:\sigma_{1},\dots, x_{n}:\sigma_{n}\}$, then for all $ u_{i}\in \RRED_{\sigma_{i}}^{\theta}$, 
$ t[u_{1}/x_{1},\dots, u_{n}/x_{n}]\in \RRED_{\sigma}^{\theta}$.
\end{proposition}
\begin{proof}

The argument is by induction on $t$. All inductive cases are treated as in the proof of [\cite{DLGH}, Proposition 29], using the lemmas proved above; the only new case is the one below:
\begin{varitemize}
\item if $t$ is $\{t_{1}\}t_{2}$ then the last rule is
$$
\AXC{$\Gamma \vdash^{\theta} t_{1}: \sigma \To \tau$}
\AXC{$\Gamma \vdash^{\theta} t_{2}: \sigma$}
\BIC{$\Gamma \vdash^{\theta} t:\tau$}
\DP
$$
then by IH for all $u_{i}\in \RRED_{\sigma_{i}}^{X}$, 
$t_{1}[u_{1}/x_{1},\dots, u_{n}/x_{n}]\in \RRED_{\sigma\To \tau}^{\theta}$ and 
$t_{2}[u_{1}/x_{1},\dots, u_{n}/x_{n}]\in \RRED_{\sigma}$, so by Lemma \ref{lemma:permuzza} 
$
t[u_{1}/x_{1},\dots, u_{n}/x_{n}]=
\{ t_{1}[u_{1}/x_{1},\dots, u_{n}/x_{n}]\}(t_{2}[u_{1}/x_{1},\dots, u_{n}/x_{n}])\in \RRED_{\tau}^{\theta}
$.
\end{varitemize}

\end{proof}

\begin{theorem}
If $\Gamma \vdash^{\theta} t:\sigma$ is derivable in $\STLC$, then $t\in \SN^{\theta}$.
\end{theorem}

From this, using Proposition \ref{prop:lcpltostlc} we immediately deduce:

\begin{corollary}
If $\Gamma \vdash^{X} t:\bone \Pto \sigma$ is derivable in $\STLC$ with no 0-ary instance of the ($\vee$)-rule, then $t\in \SN^{\theta}$, for any suitable $\theta$.
\end{corollary}
%
%
%
%
%
%
%

\subsubsection{Deterministic Normalization of $\TCINT$.}

We now adapt the proof of [\cite{DLGH}, Theorem 1] to $\TCINT$.

For this we consider an intersection type system $\STLCINT$ for $\EVL$ with the following types
$$
\sigma::= [] \mid o \mid \HNORM\mid\NORM\mid \sigma \To \sigma \mid \sigma\land \sigma 
$$
The rules of $\STLCINT$ include all those of $\STLC$ except the one for $\{\}$, together with the rules in Fig.~\ref{fig:rulesimpleint}.

\begin{figure}
\fbox{
\begin{minipage}{0.98\textwidth}
\begin{center}
\begin{tabular}{c c c}
$\AXC{$\Gamma \vdash^{\theta}t: \sigma$}
\UIC{$\Gamma\vdash^{\theta}t: \HNORM$}
\DP$
& & 
$\AXC{$\Gamma \vdash^{\theta}t: \sigma$}
\RL{\small($\sigma$ $\{[],\HNORM\}$-free)}
\UIC{$\Gamma\vdash^{\theta}t: \NORM$}
\DP$
\\ & & \\
$
\AXC{$\mathsf{FV}(t)\subseteq \Gamma$}
\UIC{$\Gamma \vdash^{\theta} t: []$}
\DP$
& & 
$
\AXC{$\Gamma \vdash^{\theta} t:\sigma $}
\AXC{$\Gamma \vdash^{\theta} t:\tau $}
\BIC{$\Gamma \vdash^{\theta} t: \sigma\land \tau$}
\DP
$
\end{tabular}
\end{center}

\end{minipage}
}
\caption{Typing rules for the simply typed $\lambda$-calculus $\STLCINT$.}
\label{fig:rulesimpleint}
\end{figure}

For any type $\FF s$ of $\TCINT$, we define the corresponding type $|\FF s|$ of $\STLCINT$ by 
$|\sigma|=\sigma$, for $\sigma\in\{\HNORM,\NORM\}$, $|\FF M\To \sigma|= |\FF M|\To |\sigma$, where 
$|[]|=[]$, $|[\FF s_{1},\dots, \FF s_{n+1}]|= |\FF s_{1}|\land \dots \land |\FF s_{n+1}|$, and 
$|\B C^{q}\sigma|=|\sigma|$.


The following is easily checked by induction:
\begin{proposition}\label{prop:tcinttostlcint}
If $\Gamma \vdash^{X}_{\lnot \vee}t: \bone \Pto \FF s$ is derivable in $\TCINT$, then 
$|\Gamma |\vdash^{\theta} t: |\FF s|$ is derivable in $\STLCINT$, for any suitable $\theta$.
\end{proposition}

Let $\Lnu$ indicate the sets of $\lambda$-terms of $\EVL$, and for all finite set of names $X$, let $\Lnu^{X}$ indicate the $\lambda$-terms with free names in $X$. 
For any finite list of names $\theta$ contained in $X$, we let 
\begin{align*}
\HNN^{\theta} &:=\{t\in \Lnu^{X}\mid \nu^{\theta}t\text{ is head-normalizing}\}
\\
\NNN^{\theta} &:=\{t\in \Lnu^{X}\mid\nu^{\theta} t\text{ is normalizing}\}
\\
\SN^{\theta} &:=\{t\in \Lnu^{X}\mid \nu^{\theta}t\text{ is strongly normalizing}\}
\end{align*}

We define two families of reducibility predicates $\NNRED_{\sigma}^{\theta}, \HNRED_{\sigma}^{\theta}\subseteq \Lnu^{X}$ as follows:
\begin{align*}
 \HNRED_{[]}^{\theta} & := \Lnu^{X}
 & \NNRED_{[]}^{\theta} & := \NNRED_{\HNORM}= \Lnu^{X} \\
\HNRED_{o}^{\theta}& :=\HNRED_{\HNORM}^{\theta}=
\HNRED_{\HNORM}^{\theta}=
\HNN^{\theta} & \NNRED_{o}^{\theta}& :=
\NNRED_{\NORM}^{\theta}=
 \NN^{\theta} \\
 \HNRED_{\sigma\To \tau}^{\theta}& := \{t\in \Lnu^{X}\mid \forall u \in \HNRED_{\sigma}^{\theta}, \ tu\in \HNRED_{\tau}^{\theta}\}& 
\NNRED_{\sigma\To \tau}^{\theta}& := \{t\in \Lnu^{X}\mid \forall u \in \NNRED_{\sigma}^{\theta}, \ tu\in \NNRED_{\tau}^{\theta}\}\\
\HNRED_{\sigma\land \tau}^{\theta}&:= \HNRED_{\sigma}^{\theta}\cap \HNRED_{\tau}^{\theta} 
&
\NNRED_{\sigma\land \tau}^{\theta}&:= \NNRED_{\sigma}^{\theta}\cap \NNRED_{\tau}^{\theta}  
\end{align*}

Moreover, we define a third family of reducibility predicates $\SNRED_{\sigma}^{\theta}\subseteq \Lnu^{X}$, indexed on all $\{[],\HNORM\}$-free types as follows:
\begin{align*}
\SNRED_{o}^{\theta}&=
\SNRED_{\NORM}^{\theta}=
 \SN^{\theta} \\
\SNRED_{\sigma\To \tau}^{\theta}& := \{t\in \Lnu^{X}\mid \forall u \in \NNRED_{\sigma}^{\theta}, \ tu\in \NNRED_{\tau}^{\theta}\}
\\
\SNRED_{\sigma\land \tau}^{\theta}&:= \NNRED_{\sigma}^{\theta}\cap \NNRED_{\tau}^{\theta}  
\end{align*}

%

%
%
\begin{definition}[non-trivial types]\label{def:nontrivial}
The set of \emph{non-trivial types} is defined by induction as follows:
\begin{varitemize}

\item $o, \HNORM, \NORM$ are non-trivial;
\item if $\tau$ is non-trivial, $\sigma\To \tau$ is non-trivial;
\item if either $\sigma$ or $\tau$ are non-trivial, $\sigma\land \tau$ is non-trivial.

\end{varitemize}
%
%
%
%
%
%
%
%
\end{definition}

\begin{lemma}\label{lemma:norms}
For any type $\sigma$, 
\begin{varitemize}
\item[(i.)] if $\sigma$ is non-trivial then  $\HNRED_{\sigma}^{\theta}\subseteq \HNN^{\theta}$;
\item[(ii.)] if $\sigma$ is $\{[],\HNORM\}$-free, then $\NNRED_{\sigma}^{\theta}\subseteq \NNN^{\theta}$;
\item[(iii.)] if $\sigma$ is $\{[],\HNORM\}$-free, then  $\SNRED_{\sigma}^{\theta}\subseteq \SN^{\theta}$.

\end{varitemize}
\end{lemma}
\begin{proof}
(i.) is proved in a standard way by induction on all types $\sigma$, also proving at each step 
that for any term $t$ of the form 
$xt_{1}\dots t_{n}$, $t\in \HNRED_{\sigma}^{\theta}$.

(ii.) is proved in a standard way by induction on all types $\sigma$, also proving at each step 
that for any term $t$ of the form 
$xt_{1}\dots t_{n}$,  where $t_{1},\dots, t_{n}\in \NNN^{\theta}$, $t\in \NNRED_{\sigma}^{\theta}$.

(iii.) is proved in a standard way by induction on all $\{[],\HNORM\}$-free types $\sigma$, also proving at each step 
that for any term $t$ of the form 
$xt_{1}\dots t_{n}$,  where $t_{1},\dots, t_{n}\in \SN^{\theta}$, $t\in \SNRED_{\sigma}^{\theta}$.

\end{proof}

\begin{proposition}
If $\Gamma \vdash^{\theta} t:\sigma$ is derivable in $\STLCINT$, where $\Gamma=\{x_{1}:\sigma_{1},\dots, x_{n}:\sigma_{n}\}$, and $\FN(t)\subseteq X$, then 
\begin{varitemize}

\item[(i.)] for all $ u_{i}\in \HNRED_{\sigma_{i}}^{\theta}$, 
$ t[u_{1}/x_{1},\dots, u_{n}/x_{n}]\in \HNRED_{\sigma}$;

\item[(ii.)] for all $ u_{i}\in \NNRED_{\sigma_{i}}^{\theta}$, 
$ t[u_{1}/x_{1},\dots, u_{n}/x_{n}]\in \NNRED_{\sigma}$;
\item[(iii.)]
if the types $[]$ and $\HNORM$ never occurs in the derivation, then 
for all $ u_{i}\in \SNRED_{\sigma_{i}}^{\theta}$, 
$ t[u_{1}/x_{1},\dots, u_{n}/x_{n}]\in \SNRED_{\sigma}$;

\end{varitemize}
\end{proposition}
\begin{proof}

Both arguments are by induction on the typing derivation of $t$. All inductive cases are treated as in the proof of [\cite{DLGH}, Proposition 29], except for those below (we consider here only case (ii.)):
\begin{varitemize}

\item if the last rule is
$$
\AXC{$\Gamma \vdash^{\theta}t: \sigma$}
\UIC{$\Gamma \vdash^{\theta} t: \HNORM$}
\DP
$$
then for all $u_{i}\in \NNRED_{\sigma_{i}}^{\theta}$, 
$t[u_{1}/x_{1},\dots, u_{n}/x_{n}]\in \Lnu^{X}= \NNRED_{\HNORM}^{\theta}$.

\item if the last rule is
$$
\AXC{$\Gamma \vdash^{\theta}t: \sigma$}
\UIC{$\Gamma \vdash^{\theta} t: \NORM$}
\DP
$$
where $\sigma$ is $\{[],\HNORM\}$-free,
then for all $u_{i}\in \NNRED_{\sigma_{i}}^{\theta}$, 
$t[u_{1}/x_{1},\dots, u_{n}/x_{n}]\in \NNRED_{\sigma}^{\theta}$, and by Lemma \ref{lemma:norms}
$\NNRED_{\sigma}^{\theta}\subseteq \NNN^{\theta}=\NNRED_{\NORM}^{\theta}$, so 
$t[u_{1}/x_{1},\dots, u_{n}/x_{n}]\in \NNRED_{\NORM}^{\theta}$.

\item if the last rule is
$$
\AXC{$\FN(t)\subseteq X$}
\UIC{$\Gamma \vdash^{\theta} t: []$}
\DP
$$
then for all $u_{i}\in \NNRED_{\sigma_{i}}^{\theta}$, 
$t[u_{1}/x_{1},\dots, u_{n}/x_{n}]\in \Lnu^{X}= \NNRED_{[]}^{\theta}$.

\item if the last rule is 
$$
\AXC{$\Gamma\vdash^{\theta} t:\sigma$}
\AXC{$\Gamma\vdash^{\theta}t:\tau$}
\BIC{$\Gamma \vdash^{\theta} t:\sigma\land \tau$}
\DP
$$
then for all $u_{i}\in \NNRED_{\sigma_{i}}^{\theta}$,  by IH 
$t[u_{1}/x_{1},\dots, u_{n}/x_{n}]\in  \NNRED_{\sigma}^{\theta}$ and
$t[u_{1}/x_{1},\dots, u_{n}/x_{n}]\in  \NNRED_{\tau}^{\theta}$, whence
$t[u_{1}/x_{1},\dots, u_{n}/x_{n}]\in  \NNRED_{\sigma}^{\theta}\cap \NNRED_{\tau}^{\theta}= \NNRED_{\sigma\land \tau}^{\theta}$.

\end{varitemize}

\end{proof}

\begin{proposition}
If $\Gamma \vdash^{\theta} t: \sigma$ is derivable in $\STLCINT$, then:
\begin{varitemize}
\item[(i.)] if $[]$ is non-trivial, then $\nu^{\theta}t$ is head-normalizable;
\item[(ii.)] if $\Gamma$ and $\sigma$ are $\{[],\HNORM\}$-free, then $\nu^{\theta}t$ is normalizable;
\item[(iii.)] if $[]$ and $\HNORM$ never occur in the derivation, then $\nu^{\theta}t$ is strongly normalizable.

\end{varitemize}
\end{proposition}

Using Proposition \ref{prop:tcinttostlcint}

\begin{corollary}\label{cor:headnorm}
If $\Gamma \vdash^{X}_{\lnot\vee} t:\bone \Pto \FF s$ in $\TCINT$, then for all suitable $\theta$:
\begin{varitemize}
\item[(i.)] $\nu^{\theta}t$ is head-normalizable;
\item[(ii.)] if $\Gamma$ and $\sigma$ are $\{[],\HNORM\}$-free, then $\nu^{\theta}t$ is normalizable;
\item[(iii.)] if $[]$ and $\HNORM$ never occur in the derivation, then $\nu^{\theta}t$ is strongly normalizable.

\end{varitemize}
\end{corollary}

\subsection{Probabilistic Reducibility Candidates for $\EVL$.}\label{subs:norma1}

In this subsection and the following one we work with a general class of types which comprises both those of $\LCPL$ and those of $\TCINT$. 
These types are defined by the grammar below:
$$
\sigma::= [] \mid o \mid \HNORM\mid \NORM \mid \sigma\To \sigma \mid \sigma \land \sigma \mid  \B C^{q}\sigma
$$

It is clear that any type of $\LCPL$ is a type of the grammar above. For the types of $\TCINT$ it suffices to see
$\CC M=[\FF s_{1},\dots, \FF s_{n}]$ as $(\dots(\FF s_{1}\land \FF s_{2})\land \dots \land  \FF s_{n-1})\land \FF s_{n}$, where $n>0$, and as $[]$ otherwise.

For any type $\sigma$, we let the positive real $\srank\sigma\in [0,1]$ be defined by $\srank{[]}=0$, $\srank o=\srank\HNORM=\srank\NORM=1$, $\srank{\sigma\To \tau}=\srank\sigma \cdot \srank \tau$, $\srank{\sigma\land \tau}= \max\{\srank{\sigma}, \srank{\tau}\}$, and $\srank{\B C^{q}\sigma}=q\cdot \srank\sigma$.

One can check that for all type $\sigma$ coming from either $\LCPL$ or $\TCINT$, $\srank{\sigma}>0$.

%
%
%

Let us establish a few preliminary lemmas.

\begin{lemma}\label{lemma:t}
For all $t\in \CC T$ and variables $x,y$, with $y$ not occurring in $t$, 
\begin{varitemize}
\item[(i.)] $ \HNF(t)=\HNF(t[y/x])$; 
\item[(ii.)] $\NF(t)=\NF(t[y/x])$.
\end{varitemize}
\end{lemma}
\begin{proof}
We only prove case (ii.), the other one being proved in a similar way.
Observe that if $t$ is not a strong head normal form and $\CC D_{t}(u)> 0$, it must be that $u$ is of size strictly smaller that $t$. So we can argue by induction on $t$:
\begin{varitemize}
\item if $t=\lambda \vec x.zt_{1}\dots t_{n}$, then $\NF(t)=\prod_{i}\NF(t_{i})\stackrel{\small[\text{I.H.}]}{=}\prod_{i}\NF(t_{i}[y/x])=\NF(t[y/x])$;

\item otherwise, since $\CC D_{t}(u)= \CC D_{t[y/x]}(u[y/x])$, we have that $\NF(t)=\sum_{u\in \HNF}\CC D_{t}(u)\cdot \NF(u)=\sum_{u\in \HNF}\CC D_{t[y/x]}(u[y/x])\cdot \NF(u)\stackrel{\small[\text{I.H.}]}{=}
\sum_{u\in \HNF}\CC D_{t[y/x]}(u[y/x])\cdot \NF(u[y/x])
=\sum_{u\in \HNF}\CC D_{t[y/x]}(u)\cdot \NF(u)
= \NF(t[y/x])$.
\end{varitemize}

\end{proof}

\begin{lemma}\label{lemma:lamb}
For any name-closed term $t$,
\begin{varitemize}
\item[(i.)] $ \HNF(t)=\HNF(\lambda y.t)$;
\item[(ii.)] $ \NF(t)=\NF(\lambda y.t)$.
\end{varitemize}
\end{lemma}
\begin{proof}
We only prove case (ii.), the other one being proved in a similar way.
We argue by induction on $t$:
\begin{varitemize}
\item if $t\redpermm^{*}\lambda \vec x.zt_{1}\dots t_{n}$, then $\lambda y.t\redpermm^{*}\lambda y\lambda \vec x.zt_{1}\dots t_{n}$, so by definition $\NF(t)=\prod_{i}\NF(t_{i})=\NF(\lambda y.t)$;

\item otherwise, if $t \redpermm^{*} \nu a.t'$, then $\lambda y.t \redpermm^{*} \nu a.t^{\lambda y}$, and one can show by induction on a $\redperm$-reduction path that 
$\CC D_{t'}(u)= \CC D_{t^{\lambda y}}(\lambda y.u)$. So we have that 
$\NF(t)=\sum_{u\in \HNF}\CC D_{t'}(u)\cdot \NF(u) = \sum_{u\in \HNF}\CC D_{t^{\lambda y}}  (\lambda y.u)\cdot \NF(u) \stackrel{\small\text{[I.H.]}}{=} \sum_{u\in \HNF}\CC D_{t^{\lambda y}}(\lambda y.u)\cdot \NF(\lambda y.u)=\sum_{u\in \HNF}\CC D_{t^{\lambda y}}(u)\cdot \NF(u)=\NF(\lambda y.t)$.
\end{varitemize}
\end{proof}

\begin{lemma}\label{lemma:x}
For any name-closed term $t$,
\begin{varitemize}
\item[(i.)] $\HNF(t)\geq \HNF(tx)$;
\item[(ii.)] $\NF(t)\geq \NF(tx)$.
\end{varitemize}
\end{lemma}
\begin{proof}
We only prove case (ii.), the other one being proved in a similar way.
Suppose $tx\redall^{*} u$ and $\NF(u)\geq q$. We will show that $\NF(t)\geq q$. Two cases arise:
\begin{varitemize}
\item $u=u'x$ where $t\redall^{*}u'$. Then $u'$ cannot start with either $\nu$ (otherwise it would not be a PNF, since $(\nu a.v)x\redpermm \nu a.vx$), $\oplus$ (since $(v_{1}\oplus_{a}^{i}v_{2})x \redpermm (v_{1}x)\oplus_{a}^{i}(v_{2}x)$), or $\lambda$ (since $(\lambda y.v)x$ is a $\beta$-redex, which would imply $\NF(u)=0$). Hence $u'$ is of the form
$y u_{1}\dots u_{n}$, and thus $\NF(u)= \NF(u')\geq q$, which implies $\NF(t)\geq q$. 

\item $tx\redalll^{*}(\lambda y.t')x$, and $t'[y\mapsto x]\redalll^{*}u$. Then 
$t$ reduces to the PNF $\lambda y. u[x\mapsto y]$, and we can conclude by Lemma \ref{lemma:t} and Lemma \ref{lemma:lamb}.
\end{varitemize}

\end{proof}

In the following, we will sometimes write $t\in \NHNF^{r}$ (resp.~$t\in \NNF^{r}$) for $\NHNF(t)\geq r$ (resp.~$\NNF(t)\geq r$).

We will define two different classes of \emph{reducibility predicates}, one related to probabilistic head normal forms and the other related to probabilistic normal forms.
 For any type $\sigma$, finite set $X$ of names, $r\in[0,1]$, and $S\subseteq(2^{\BB N})^{X} $, we define the sets
$\HRED_{\sigma}^{X,r}(S),\RED_{\sigma}^{X,r}(S) \subseteq \Lnu^{X}$ by induction as follows:  

\begin{align*}
\HRED_{[]}^{X,r}(S)&= \Lnu^{X}\\ 
\HRED_{o}^{X,r}(S)&= \HRED_{\HNORM}^{X,r}(S)=
\HRED_{\NORM}^{X,r}(S)=\{ t \in \Lnu^{X} \mid   \forall \omega \in S   \ \pi^{\omega}_{X}(t)\in \NHNF^{r}\}\\ 
\HRED_{\sigma\To \tau}^{X,r}(S)& =\{t\in \Lnu^{X} \mid   
\forall S'\subseteq S, \forall s\in (0,1], \ \forall u\in \HRED_{\sigma}^{X,1}(S'), \ tu\in \HRED_{\tau}^{X, r}(S')\}\\
\HRED_{\sigma\land \tau}^{X,r}(S)& = \HRED_{\sigma}^{X,r}(S)\cap \HRED_{\tau}^{X,r}(S)\\
\HRED_{\B C^{q}\sigma}^{X,r}(S) & = \HRED_{\sigma}^{X,qr}(S)
\end{align*}

\begin{align*}
\RED_{[]}^{X,r}(S)&=\RED_{\HNORM}^{X,r}(S)= \Lnu^{X}\\ 
\RED_{o}^{X,r}(S)&=\RED_{\NORM}^{X,r}(S)= \{ t \in\Lnu^{X}\mid   \forall \omega \in S   \ \pi^{\omega}_{X}(t)\in \NNF^{r}\}\\ 
\RED_{\sigma\To \tau}^{X,r}(S)& =\{t\in\Lnu^{X} \mid   
\forall S'\subseteq S, \forall s\in (0,1], \ \forall u\in \RED_{\sigma}^{X,1}(S'), \ tu\in \RED_{\tau}^{X, r}(S')\}\\
\HRED_{\sigma\land \tau}^{X,r}(S)& = \RED_{\sigma}^{X,r}(S)\cap \RED_{\tau}^{X,r}(S)\\
\RED_{\B C^{q}\sigma}^{X,r}(S) & = \RED_{\sigma}^{X,qr}(S)
\end{align*}

%
%
%
%
%

We will now establish a few properties of the families $\HRED_{\sigma}^{X,r}(S)$ and $\RED_{\sigma}^{X,r}(S)$. Since the proofs for the two cases are similar, we will provide details only for the second case.

\begin{lemma}\label{lemma:leq}
For any type $\sigma$ and $q,r\in [0,1]$, if $q\geq r$, then 
\begin{varitemize}
\item[(i.)] $\HRED_{\sigma}^{q}(S)\subseteq \HRED_{\sigma}^{r}(S)$.
\item[(ii.)] 
$\RED_{\sigma}^{q}(S)\subseteq \RED_{\sigma}^{r}(S)$.
\end{varitemize}

\end{lemma}
\begin{proof}
We only prove case (ii.), the other one being proved in a similar way.
By induction on $\sigma$:
\begin{varitemize}
\item if $\sigma\in \{[],\HNORM\}$, the claim is immediate;
\item if $\sigma\in \{o,, \NORM\}$, then $t\in \RED_{\sigma}^{q}(S)$ iff for all $\omega \in S$, $\pi^{\omega}_{X}(t)\in \NNF^{q}\subseteq \NNF^{r}$, so
$t\in \RED_{\sigma}^{r}(S)$;

\item if $\sigma=(\tau\To\rho)$ and $t\in \RED_{\tau}^{X,q}(S)$  then for all $S'\subseteq S$, $u\in \RED_{\tau}^{X,1}(S')$, 
$tu\in \RED_{\rho}^{q}(S')\stackrel{\small\text{[I.H.]}}{\subseteq} \RED_{\rho}^{r}(S')$
; hence we can conclude $t\in \RED_{\sigma}^{r}(S)$;

\item if $\sigma=\tau\land \rho$ and $t\in \RED_{\sigma}^{X,q}(S)$, then by IH $t\in \RED_{\tau}^{X,r}(S)$ and $t\in \RED_{\rho}^{X,r}(S)$, so $t\in \RED_{\sigma}^{X,r}(S)$;

\item if $\sigma= \B C^{s}\tau$, then 
$\RED_{\sigma}^{X,q}(S) = \RED_{\tau}^{X,sq}(S)
\stackrel{\small\text{[I.H.]}}{\subseteq} \RED_{\tau}^{X,sr}(S)=\RED_{\sigma}^{X,r}(S)$.
\end{varitemize}
\end{proof}

\begin{lemma}\label{lemma:zero}
\begin{varitemize}
\item[(i.)] $\HRED_{\sigma}^{X,r}(\emptyset)=\{t\in \Lnu\mid \FN(t)\subseteq X\}$;
\item[(ii.)]
$\RED_{\sigma}^{X,r}(\emptyset)= \{t\in \Lnu\mid \FN(t)\subseteq X\}$.
\end{varitemize}
\end{lemma}
\begin{proof}
We only prove (ii.) as (i.) is proved similarly.
Let $\Lnu^{X}=\{t\in \Lnu\mid \FN(t)\subseteq X\}$. 
Observe that the inclusion $\RED_{\sigma}^{X,r}(\emptyset)\subseteq \Lnu^{X}$ is immediate. For the converse direction, we argue by induction on $\sigma$:
\begin{varitemize}
\item if $\sigma\in \{[],\HNORM\}$, the claim is immediate;

\item If $\sigma\in\{o,\NORM\}$, then trivially for all $\omega\in \emptyset$ and $t\in \Lnu^{X}$, $\pi^{\omega}_{X}(t)\in \NNF^{r}$, so $\Lnu^{X}\subseteq \RED_{\sigma}^{X,r}(\emptyset)$. 
 
\item if $\sigma =\tau \To \rho$, then for all $t\in \Lnu^{X}$, and $u\in \RED_{\tau}^{X,1}(\emptyset)$, $tu\in \Lnu^{X}$, so by IH $tu\in \RED_{\rho}^{X,r}(\emptyset)$. 
We can thus conclude that $t\in \RED_{\sigma}^{X,r}(\emptyset)$. 

\item if $\sigma=\tau\land \rho$, then for all $t\in \Lnu^{X}$, by IH $t\in \RED_{\tau}^{X,r}(\emptyset)$ and $t\in \RED_{\rho}^{X,r}(\emptyset)$, so $t\in \RED_{\sigma}^{X,r}(\emptyset)$.

\item 
If $\sigma=\B C^{q}\tau$, then by the I.H.~$ \Lnu^{X}\subseteq \RED_{\tau}^{X,qr}(\emptyset)=\RED_{\sigma}^{X,q}(\emptyset)$.

\end{varitemize}
%
\end{proof}

\begin{lemma}\label{lemma:efs}
\begin{varitemize}
\item[(i.)] $t\in \HRED_{\sigma}^{X,r}(S)$ iff for all $\omega \in S$, $\pi^{\omega}_{X}(t)\in \HRED_{\sigma}^{\emptyset, r}$;
\item[(ii.)] $t\in \RED_{\sigma}^{X,r}(S)$ iff for all $\omega \in S$, $\pi^{\omega}_{X}(t)\in \RED_{\sigma}^{\emptyset, r}$.
\end{varitemize}
\end{lemma}
\begin{proof}
We only prove (ii.) as (i.) is proved similarly.
We argue by induction on $\sigma$:
\begin{varitemize}
\item if $\sigma\in \{[],\HNORM\}$, the claim is immediate;

\item If $\sigma\in\{o,\NORM\}$, the claim follows from the definition.

\item If $\sigma=\tau\To \rho$, suppose $t\in \RED_{\sigma}^{X,r}(S)$, let $\omega \in S$ and $u\in \RED^{\emptyset,1}_{\tau}$. Then $u$ is name-closed hence for all $g\in S$, $\pi^{g}(u)=u$, which by the IH, implies $u\in \RED_{\tau}^{X,1}(S)$; we deduce then that $tu\in \RED_{\rho}^{X,r}(S)$, so by IH $\pi^{\omega}_{X}(tu)=\pi^{\omega}_{X}(t)u\in \RED_{\rho}^{\emptyset, r}$. By a similar argument we can show that for all $v\in \RED_{\tau}^{\emptyset ,s}$, $\pi^{\omega}_{X}(\{t\}v)=\{\pi^{\omega}_{X}(t)\}v\in \RED_{\rho}^{\emptyset, rs}$. 
We can thus conclude that $\pi^{\omega}_{X}(t)\in \RED_{\sigma}^{\emptyset,r}$.

Conversely, suppose that for all $\omega \in S$, $\pi^{\omega}_{X}(t)\in \RED_{\sigma}^{\emptyset, r}$, let $S'\subseteq S$ and $u\in \RED_{\tau}^{X,1}(S')$; if $\omega \in S'$, then by IH $\pi^{\omega}_{X}(u)\in \RED_{\tau}^{\emptyset,1}$, so $\pi^{\omega}_{X}(tu)=\pi^{\omega}_{X}(t)\pi^{\omega}_{X}(u)\in \RED_{\rho}^{\emptyset, r}$; 
We have thus proved that for all $\omega \in S'$, $\pi^{\omega}_{X}(tu) \in \RED_{\rho}^{\emptyset, r}$ 
, which by IH implies that $tu\in \RED_{\rho}^{X,r}(S')$
We conclude then that $t\in \RED_{\sigma}^{X,r}(S)$.

\item if $\sigma=\tau\land \rho$ then $t\in \RED_{\sigma}^{X,r}(S)$ iff $t\in \RED_{\tau}^{X,r}(S)$ and $t\in \RED_{\rho}^{X,r}(S)$ iff (by the IH) for all $\omega \in S$, $\pi^{\omega}_{X}(t)\in \RED_{\tau}^{X,r}(\emptyset)$ and $\pi^{\omega}_{X}(t)\in \RED_{\rho}^{X,r}(\emptyset)$, iff for all $\omega \in S$, $\pi^{\omega}_{X}(t)\in \RED_{\sigma}^{X,r}(\emptyset)$.

\item If $\sigma=\B C^{q}\tau$, then 
$
t\in \RED_{\sigma}^{X,r}(S)\text{ iff }t\in \RED_{\tau}^{X,rq}(S)\stackrel{\text{[I.H.]}}{\text{ iff }}
\forall \omega \in S \ \pi^{\omega}_{X}(t)\in \RED_{\tau}^{X,rq}\stackrel{\text{[I.H.]}}{\text{ iff }}
\forall \omega \in S \ \pi^{\omega}_{X}(t)\in \RED_{\sigma}^{X,r}
$.
\end{varitemize}
\end{proof}
%

\begin{lemma}\label{lemma:perm}
\begin{varitemize}
\item[(i.)]
If $t\in \HRED_{\sigma}^{X,r}(S)$ and $t' \redperm t$, then $t'\in \HRED_{\sigma}^{X,r}(S)$.

\item[(ii.)]
If $t\in \RED_{\sigma}^{X,r}(S)$ and $t'\redperm t$, then $t'\in \RED_{\sigma}^{X,r}(S)$.

\end{varitemize}
\end{lemma}
\begin{proof}By induction on $\sigma$.
\end{proof}

\begin{lemma}\label{lemma:nu1}
For all terms $t, u_{1},\dots, u_{n}$ with $\FN(t)\subseteq X\cup \{a\}$, $\FN(u_{1}),\dots,\FN(u_{n})\subseteq X$, and measurable sets $S_{1},\dots, S_{k+1}\subseteq(2^{\BB N})^{X\cup\{a\}}$ and $S'\subseteq(2^{\BB N})^{X}$, if 
\begin{enumerate}
\item the $S_{i}$ are pairwise disjoint;
\item for all $\omega \in S_{i}$, $\pi^{\omega}_{X}(tu_{1}\dots u_{n})\in  \NHNF^{r_{i}}$ (resp.~$\pi^{\omega}_{X}(tu_{1}\dots u_{n})\in  \NNF^{r_{i}}$);
\item 
for all $\omega \in S'$,
 $\mu( \Pi^{\omega}(S_{i}))\geq s_{i}$;
\end{enumerate}
then for all $\omega \in S'$, $\pi^{\omega}_{X}((\nu a.t)u_{1}\dots u_{n}) \in \NHNF^{ \sum_{i=1}^{k+1}r_{i}s_{i}}$
(resp.~$\pi^{\omega}_{X}((\nu a.t)u_{1}\dots u_{n}) \in \NNF^{ \sum_{i=1}^{k+1}r_{i}s_{i}}$).

%
%
%
\end{lemma}
\begin{proof}
We only prove the claim in the case of $\NNF$, the other case being proved similarly.
Let us first consider the case where $n=0$.
Let $\omega \in S'$. Since any term reduces to a (unique) PNF, we can suppose w.l.o.g.~that the name-closed term $\pi^{\omega}_{X}(\nu a.t)=\nu a.t^{*}$ is in PNF. Then by Corollary \ref{cor:pnf} $t^{*}$ is a $(\mathcal T, a)$-tree of level $I_{a}(t^{*})$. 
Observe that for all $\omega'\in (2^{\BB N})^{\{a\}}$, 
$\pi^{\omega'}_{\{a\}}(t^{*})=\pi^{\omega'}_{\{a\}}(\pi^{\omega}_{X}(t))= \pi^{\omega+\omega'}_{X\cup\{a\}}(t)$ . 
%
%

If $N$ is the cardinality of $\mathrm{Supp}(t^{*})$, by 3.~and the fact that $s_{i}>0$ we deduce that for all $i=1,\dots,k+1$ there exists a finite number $K_{i}$ of elements $w_{i1},\dots, w_{iK_{i}}$ of $\mathrm{Supp}(t^{*})$ such that 
$\pi^{\omega_{ij}}_{\{a\}}(t^{*})=w_{ij}\in \NNF^{r_{i}}$ for some $\omega_{ij}\in (2^{\BB N})^{\{a\}}$ such that $\omega+\omega_{ij} \in S_{i}$, and using 3.~we deduce that  
$$
 \sum_{w\in \NNF^{r_{i}}}
\mu\{ \omega\mid \pi^{\omega}_{\{a\}}(t^{*}) =w\}\geq 
\sum_{{j}=1}^{K_{i}} 
\mu\{ \omega\mid \pi^{\omega}_{\{a\}}(t^{*}) =w_{ij}\}\geq s_{i}$$


Now, by reducing, for all $i\leq k+1$, each such term $w_{ij}$ inside $\nu a.t$, to some PNF $w'_{ij}\in \NF^{r_{i}}$, we obtain a new PNF $\nu a.t^{\sharp}$ and we can compute (using the fact that $\sum_{u\in \HNF}\CC D_{w'_{ij}}(u)\cdot \NF(u) \geq r_{i}$):

\resizebox{0.95\linewidth}{!}{
\begin{minipage}{\linewidth}
\begin{align*}
\NF(\nu a.t^{\sharp}) =
\sum_{u\in \HNF}\mathcal D_{\nu a.t^{\sharp}}(u)\cdot \NF(u) & =
\sum_{u\in \HNF}\left (\sum_{t'\in \mathrm{Supp}(t^{\sharp})}\mathcal D_{t'}(u)\cdot 
\mu\{ \omega\mid \pi^{\omega}_{\{a\}}(t^{\sharp}) =t'\} \right)\cdot \NF(u) \\
&= 
\sum_{u\in \HNF}\left (\sum_{t'\in \mathrm{Supp}(t^{\sharp})}\mathcal D_{t'}(u)\cdot \NF(u)\cdot 
\mu\{ \omega\mid \pi^{\omega}_{\{a\}}(t^{\sharp}) =t'\} \right) \\
&= 
\sum_{t'\in \mathrm{Supp}(t^{\sharp})}
\left (\sum_{u\in \HNF} \mathcal D_{t'}(u)\cdot \NF(u)\cdot 
\mu\{ \omega\mid \pi^{\omega}_{\{a\}}(t^{\sharp}) =t'\} \right) \\
&= 
\sum_{t'\in \mathrm{Supp}(t^{\sharp})}
\left (\sum_{u\in \HNF} \mathcal D_{t'}(u)\cdot \NF(u)\right)\cdot 
\mu\{ \omega\mid \pi^{\omega}_{\{a\}}(t^{\sharp}) =t'\}  \\
&\geq 
\sum_{i=1}^{k+1}\sum_{{j}=1}^{K_{i}}
\left (\sum_{u\in \HNF} \mathcal D_{w'_{ij}}(u)\cdot \NF(u)\right)\cdot 
\mu\{ \omega\mid \pi^{\omega}_{\{a\}}(t^{\sharp}) =w'_{ij}\}  \\
& \geq 
\sum_{i=1}^{k+1}\sum_{{j}=1}^{K_{i}}
%
r_{i}\cdot 
\mu\{ \omega\mid \pi^{\omega}_{\{a\}}(t^{\sharp}) =w'_{ij}\}  \\
& = 
\sum_{i=1}^{k+1}r_{i}\cdot\left( \sum_{{j}=1}^{K_{i}}
\mu\{\omega\mid \pi^{\omega}_{\{a\}}(t^{\sharp}) =w'_{ij}\}  \right)\\
& =
\sum_{i=1}^{k+1}r_{i}\cdot\left( \sum_{{j}=1}^{K_{i}}
\mu\{ \omega\mid \pi^{\omega}_{\{a\}}(t^{*}) =w'_{ij}\}  \right)\\
& \geq
\sum_{i=1}^{k+1}r_{i}\cdot s_{i}
\end{align*}
\end{minipage}
}
\medskip

Now, from $\nu a.t\redalll^{*} \nu a.t^{\sharp}\in \NF^{\sum_{i}r_{i}s_{i}}$, we conclude $\nu a.t\in \NNF^{\sum_{i}r_{i}s_{i}}$.

For the case in which $n>0$, we argue as follows:
from the hypotheses we deduce by the first point that 
$ \nu a.(tu_{1}u_{2}\dots u_{n})\in \NNF^{\sum_{i}r_{i}s_{i}}$. We can then conclude by observing that 
$(\nu a.t)u_{1}u_{2}\dots u_{n}\redpermm \nu a.(tu_{1}u_{2}\dots u_{n})$. 
\end{proof}

%

\begin{lemma}\label{lemma:nu2}
For all types $\sigma$, terms $t,u_{1},\dots, u_{n}$ with $\FN(t)\subseteq X\cup \{a\}$, $\FN(u_{i})\subseteq X$,
and measurable sets  $S_{1},\dots, S_{k+1}\subseteq(2^{\BB N})^{X\cup\{a\}}$ and $S'\subseteq (2^{\BB N})^{X}$, if
\begin{enumerate}
\item the $S_{i}$ are pairwise disjoint;
\item $tu_{1}\dots u_{n}\in \HRED_{\sigma}^{X\cup\{a\}, r_{i}}(S_{i})$ (resp.~$tu_{1}\dots u_{n}\in \RED_{\sigma}^{X\cup\{a\}, r_{i}}(S_{i})$), for all $i=1,\dots,k+1$;
\item for all $\omega \in S'$, $\mu( \Pi^{\omega}(S_{i}))\geq s_{i}$;
\end{enumerate}
then $(\nu a.t)u_{1}\dots u_{n}\in \HRED_{\sigma}^{X,\sum_{i=1}^{k+1}r_{i}\cdot s_{i}}(S')$
(resp.~$(\nu a.t)u_{1}\dots u_{n}\in \RED_{\sigma}^{X,\sum_{i=1}^{k+1}r_{i}\cdot s_{i}}(S')$).

%
%
%
%
\end{lemma}
\begin{proof}
We only prove the claim for $\RED_{\sigma}$, the other case being proved similarly. We argue by induction on $\sigma$:
\begin{varitemize}
\item if $\sigma\in \{[],\HNORM\}$, the claim is immediate;

\item if $\sigma\in\{o,\NORM\}$, then the claim follows from Lemma~\ref{lemma:nu1}.
\item if $\sigma =\tau \To \rho$, then let $S''\subseteq S'$ and $u\in \RED_{\tau}^{X,1}(S'')$. 
For all $i$, since $s_{i}>0$ for all $\omega \in S''\subseteq S'$ there exists at least one $\omega'\in (2^{\BB N})^{\{a\}}$, such that $\omega +\omega'\in S_{i}$.
Let $T_{i}=\{ \omega+\omega' \in S_{i} \mid \omega \in S''\}$. The $T_{i}$ are pairwise disjoint and measurable, since they are counter-images of a measurable set through a measurable function (the projection function  from $(2^{\BB N})^{X\cup\{a\}}$ to $ (2^{\BB N})^{X}$).

We have that $u\in \RED_{\tau}^{X,1}(T_{i})$, for all $i=1,\dots,k+1$: for all $\omega+\omega' \in S''$, $\pi^{\omega+\omega'}_{X\cup\{a\}}(u)=\pi^{\omega}_{X}(u)\in \RED_{\tau}^{\emptyset, 1}$, so by Lemma~\ref{lemma:efs}, $u\in \RED_{\tau}^{X,1}(T_{i})$.
Since $T_{i}\subseteq S_{i}$, using the hypothesis 2.~we deduce that $tu_{1}\dots u_{n}u\in \RED_{\rho}^{X,r_{i}}(T_{i})$; moreover, for all $\omega \in S''$, $\mu(\Pi^{\omega}(T_{i}))=\mu(\Pi^{\omega}(S_{i}))\geq s_{i}$. So, by the induction hypothesis $(\nu a.t)u_{1}\dots u_{n}u\in \RED_{\rho}^{X,\sum_{i}r_{i}s_{i}}(S'')=\RED_{\rho}^{X,\sum_{i}r_{i}s_{i}}(S'')$. 
We can thus conclude that $(\nu a.t)u_{1}\dots u_{n}\in \RED_{\sigma}^{X,\sum_{i}r_{i}s_{i}}(S')$. 

\item $\sigma=\tau\land \rho$, then from the hypotheses by IH we deduce $(\nu a.t)u_{1}\dots u_{n}\in \RED_{\tau}^{X,\sum r_{i}s_{i}}(S')$
 and $(\nu a.t)u_{1}\dots u_{n}\in \RED_{\rho}^{X,\sum r_{i}s_{i}}(S')$, and thus
 $(\nu a.t)u_{1}\dots u_{n}\in \RED_{\sigma}^{X,\sum r_{i}s_{i}}(S')$.

\item if $\sigma =\B C^{q}\tau$, the claim follows from the induction hypothesis.

\end{varitemize}
\end{proof}

\begin{lemma}\label{lemma:tre}
\begin{varitemize}
\item[(i.)] $t[u/x]u_{1}\dots u_{n}\in \HRED_{\sigma}^{X,r}(S)$ $\To $ $(\lambda x.t)uu_{1}\dots u_{n}\in \HRED_{\sigma}^{X,r}(S)$.
\item[(ii.)] $t[u/x]u_{1}\dots u_{n}\in \RED_{\sigma}^{X,r}(S)$ $\To $ $(\lambda x.t)uu_{1}\dots u_{n}\in \RED_{\sigma}^{X,r}(S)$.
\end{varitemize}
\end{lemma}
\begin{proof}
We only prove claim (ii.) as claim (i.) is proved similarly.
By induction on $\sigma$:
\begin{varitemize}
\item if $\sigma\in \{[],\HNORM\}$, the claim is immediate;

\item If $\sigma\in\{o,\NORM\}$, then if $t[u/x]u_{1}\dots u_{n}\in \RED_{\sigma}^{X,r}(S)$, then for all $\omega \in S$, 
 $\pi^{\omega}_{X}(t[u/x]u_{1}\dots u_{n})=\pi^{\omega}_{X}(t)[\pi^{\omega}_{X}(u)/x]\pi^{\omega}_{X}(u_{1})\dots \pi^{\omega}_{X}(u_{n})\in \NNF^{r}$. Since the reduction tree of $\pi^{\omega}_{X}(t'')$, where 
 $t''=(\lambda x.t)uu_{1}\dots u_{n}$, is contained in that of $\pi^{\omega}_{X}(t[u/x]u_{1}\dots u_{n})$, we deduce $\pi^{\omega}_{X}(t'')\in \NF^{r}$, and we conclude that 
 $t''\in \RED_{\sigma}^{X,r}(S)$. 
 
\item If $\sigma=\tau\To \rho$, then let $S'\subseteq S$ and $v\in \RED_{\tau}^{X, 1}(S')$; then $
t[u/x]u_{1}\dots u_{n}v \in \RED_{\rho}^{X,r}(S')$ so by  IH $(\lambda x.t)uu_{1}\dots u_{n}v\in \RED_{\rho}^{X,r}(S')$. This proves in particular that $(\lambda x.t)uu_{1}\dots u_{n}\in \RED_{\sigma}^{X,r}(S)$.

\item If $\sigma=\tau\land \rho$ and $t[u/x]u_{1}\dots u_{n}\in \RED_{\sigma}^{X,r}(S)$, then 
$t[u/x]u_{1}\dots u_{n}\in \RED_{\tau}^{X,r}(S)$ and $t[u/x]u_{1}\dots u_{n}\in \RED_{\tau}^{X,r}(S)$, so by IH
$(\lambda x.t)uu_{1}\dots u_{n}\in \RED_{\tau}^{X,r}(S)$ and $(\lambda x.t)uu_{1}\dots u_{n}\in \RED_{\rho}^{X,r}(S)$, 
which implies $(\lambda x.t)uu_{1}\dots u_{n}\in \RED_{\sigma}^{X,r}(S)$.

\item If $\sigma=\B C^{q}\tau$, then if $t[u/x]u_{1}\dots u_{n}\in \RED_{\sigma}^{X,r}(S)=\RED_{\tau}^{X,rq}(S)$, then by I.H.~
$(\lambda x.t)uu_{1}\dots u_{n}\in \RED_{\tau}^{X,rq}(S)=\RED_{\sigma}^{X,r}(S)$.
\end{varitemize}
\end{proof}

\begin{lemma}\label{lemma:quattro}
For all $t$ such that $\FN(t)\subseteq X$ and measurable sets $S, S'\subseteq (2^{\BB N})^{X}$,
\begin{varitemize}
\item[(i.)] If for all $S''\subseteq S'$ and 
$u\in \HRED_{\sigma}^{X,s}( S'')$, $t[u/x]\in \HRED_{\tau}^{X, rs}(S\cap S'')$, then $\lambda x.t\in \HRED_{\sigma\To \tau}^{X, r}( S\cap S')$.

\item[(ii.)] If for all $S''\subseteq S'$ and 
$u\in \RED_{\sigma}^{X,s}( S'')$, $t[u/x]\in \RED_{\tau}^{X, rs}(S\cap S'')$, then $\lambda x.t\in \RED_{\sigma\To \tau}^{X, r}( S\cap S')$.

\end{varitemize}
\end{lemma}
\begin{proof}
As usual, we only prove claim (ii.).
Let $T\subseteq S\cap S'$ and  $u \in \RED_{\sigma}^{X,s}( T)$. 
Then, by hypothesis, $t[u/x]\in \RED_{\tau}^{X,rs}(S\cap T)= \RED_{\tau}^{X,rs}(S\cap S')$ and we conclude by Lemma \ref{lemma:tre} that $(\lambda x.t)u\in \RED_{\tau}^{X,rs}(S\cap S')$.
This proves that $\lambda x.t\in \RED_{\sigma\To \tau}^{X,r}(S\cap S')$.

\end{proof}

\begin{lemma}\label{lemma:inclusion}
For any two types $\FF s, \FF t$ of $\TCINT$, if $\FF s \preceq \FF t$, then for all $S\subseteq (2^{\BB N})^{X}$,
\begin{varitemize}
\item[(i.)] $\HRED_{\sigma}^{X,r}(S)\subseteq \HRED_{\tau}^{X,r}(S)$;
\item[(ii.)]  $\RED_{\sigma}^{X,r}(S)\subseteq \RED_{\tau}^{X,r}(S)$.
\end{varitemize}
\end{lemma}
\begin{proof}
As usual, we only prove claim (ii.).
We argue by induction on $\FF s$: 
\begin{varitemize}
\item  if $\FF s\in\{[],o,\HNORM,\NORM\}$, then it must be $\FF s=\FF t$, so the claim is immediate.
\item If $\FF s=[\FF M_{1},\dots, \FF M_{n}]\To \FF s'$, then 
$\FF t=[\FF N_{1},\dots, \FF N_{m}]\To \FF t'$, where $\FF s'\preceq \FF t'$, and $\FF s_{i}\succeq\FF t_{f(i)}$ holds for some injective function $f:\{1,\dots,n\}\to \{1,\dots, m\}$. 
By the I.H.~we have then that $\RED_{\FF s'}^{X,r}(S)\subseteq \RED_{\FF t'}^{X,r}(S)$ and 
$\RED_{\FF t_{f(i)}}^{X,1}(S')\subseteq \RED_{\FF s_{i}}^{X,1}(S')$; by the definition of $\RED_{\FF s}^{X,r}(S)$ this implies then $\RED_{\FF s}^{X,r}(S)\subseteq \RED_{\FF t}^{X,r}(S)$. 

\item if $\FF s= \B C^{q}\FF s'$, then it must be $\FF t=\B C^{s}\FF t'$, with $\FF s'\preceq \FF t'$ and $q\geq s$; then by IH and Lemma \ref{lemma:leq} we have 
$\RED_{\FF s}^{X,r}(S)=\RED_{\FF s'}^{X,qr}(S)\subseteq \RED_{\FF s'}^{X,sr}(S)\subseteq \RED_{\FF t'}^{X,sr}(S)=
\RED_{\FF t}^{X,r}(S)$.
\end{varitemize}
\end{proof}

We now show how the two families of reducibility candidates differ on the kind of normalization properties they warrant.

 Let us introduce a (first) notion of \emph{neutral term} for $\EVL$:

\begin{definition}[neutral terms, version 1]

For any
 measurable set $S\subseteq (2^{\BB N})^{X}$, the set  $\NEUT(S)$ is defined by induction as follows:
\begin{varitemize}
\item for any variable $x$, $x\in \NEUT(S)$;


\item if $t\in \NEUT(S)$ then for all $u\in \Lnu^{X}$, $tu\in \NEUT(S)$.

\end{varitemize}

\end{definition}

It is easily checked that for all $t\in \NEUT(S)$ and $\omega\in S$, $\HNF(\pi^{\omega}_{X}(t))=1$.
%

Probabilistic neutral terms can be used to show that reducible terms reduce to a probabilistic head normal form.

\begin{definition}[non-trivial type]\label{def:nontrivial2}
The set of \emph{non-trivial types} is defined by induction as follows:
\begin{varitemize}

\item $o,\HNORM,\NORM$ are non-trivial;
\item if $\tau$ is non-trivial, $\sigma\To \tau$ is non-trivial;
\item if either $\sigma$ or $\tau$ are non-trivial, $\sigma\land \tau$ is non-trivial;
\item if $\sigma$ is non-trivial, $\B C^{q}\sigma$ is non-trivial.

\end{varitemize}
\end{definition}

\begin{lemma}\label{lemma:unobis}
For any type $\sigma$: 
\begin{enumerate}
\item if $\sigma$ is non-trivial and 
$t\in 
\HRED_{\sigma}^{X,r}(S)$ and $\omega \in S$, then $\pi^{\omega}_{X}(t)\in \NHNF^{\srank{\sigma}\cdot r}$;
\item $\NEUT(S)\subseteq \HRED_{\sigma}^{X, 1}(S)$.
\end{enumerate}
\end{lemma}
\begin{proof}
We argue by induction on $\sigma$:
\begin{varitemize}
\item if $\sigma=[]$, then both claims are immediate;
\item if $\sigma\in\{o,\HNORM,\NORM\}$ then by definition $t\in \HRED_{\sigma}^{X,r}(S)$ iff for any $\omega \in S$, $\pi^{\omega}_{X}(t)\in \NHNF^{r}=\NHNF^{\srank{\sigma}r}$.
Moreover, if $t\in \NEUT(S)$, then $\HNF(t)=1$, so $t\in \HRED_{\sigma}^{X,1}(S)$.

\item if $\sigma=\tau \To \rho$ and $t\in \HRED_{\sigma}^{X,r}(S)$, then by IH $x\in \HRED_{\tau}^{X,1}(S)$, hence
$tx\in \HRED_{\tau}^{X,r}(S)$, and again by IH for all $\omega \in S$ $\pi^{\omega}_{X}(tx)=\pi^{\omega}_{X}(t)x\in \NHNF^{\srank{\tau}\cdot r}=\NHNF^{\srank{\sigma}\cdot r}$, and thus $\pi^{\omega}_{X}(t)\in \NHNF^{\srank{\sigma}\cdot r}$.

Moreover, if $t\in \NEUT(S)$, let $S'\subseteq S$ and $u\in \HRED_{\tau}^{X,1}(S)$, then  by IH for all $\omega \in S'$, $\pi^{\omega}(u)\in \NHNF^{\srank \tau}$ and thus $tu\in \NEUT(S')$, so by IH $tu\in \HRED_{\rho}^{X,1}(S')\subseteq \HRED_{\rho}^{X,r}(S')$; we can thus conclude that $t\in \HRED_{\sigma}^{X,1}(S)$.

\item if $\sigma=\tau \land \rho$ and $t\in \HRED_{\sigma}^{X,r}(S)$, then $t\in \HRED_{\tau}^{X,r}(S)$ and $t\in \HRED_{\rho}^{X,r}(S)$, so by IH for all $\omega \in S$, $\pi^{\omega}_{X}(t)\in \NHNF^{\srank{\tau}r}$ and $\pi^{\omega}_{X}(t)\in \NHNF^{\srank{\rho}r}$; hence we deduce that for all $\omega \in S$ $\pi^{\omega}_{X}(t)\in \NHNF^{\max\{\srank{\tau},\srank{\rho}\}\cdot r }=\NHNF^{\srank{\sigma}\cdot r}$.

Moreover, if $t\in \NEUT(S)$, then by IH $t\in \HRED_{\tau}^{X,1}(S)$ and $t\in \HRED_{\rho}^{X,1}(S)$, so $t\in \HRED_{\sigma}^{X,1}(S)$.

\item if $\sigma=\B C^{s}\tau$ and $t\in \HRED_{\sigma}^{X,q}(S)=\HRED_{\tau}^{X,qs}(S)$, then by IH
for all $\omega \in S$, $\pi^{\omega}_{X}(t)\in \NHNF^{\srank{\tau}qs}= \NHNF^{\srank{\sigma}q}$.
Moreover, if $t\in \NEUT(S)$, then by IH $t\in \HRED_{\tau}^{X,1}(S)
\subseteq \HRED_{\tau}^{X,s}(S)=\HRED_{\sigma}^{X,1}(S)$.

\end{varitemize}
\end{proof}

To establish a result like Lemma \ref{lemma:unobis} for the family $\RED_{\sigma}^{X,r}(S)$, as discussed in Section \ref{section6}, we need to restrict to \emph{balanced} types. For our general class of types, these are defined as follows:

\begin{definition}
The set $\BAL$ of \emph{balanced types} is defined inductively as follows:
\begin{varitemize}

\item $\B C^{\vec q}\sigma$, for $\sigma\in\{o,\HNORM,\NORM\}$;

\item if $\sigma,\tau\in \BAL$, $\sigma\land \tau\in \BAL$;

\item if $\FF \sigma_{1},\dots, \FF \sigma_{n}\in \BAL$ and $\prod_{i=1}^{k}q_{i}\leq \prod_{j=1}^{n}\srank{\sigma_{j}}$, then 
$\B C^{q_{1}}\dots \B C^{q_{n}}(\sigma_{1}\To \dots \To \sigma_{n}\To o)\in \BAL$.

\end{varitemize}
%
\end{definition}

For instance, the type $\B C^{\frac{1}{2}}(\B C^{\frac{1}{2}}o\To o)$ is balanced, while
 $\B C^{1}(\B C^{\frac{1}{2}}o\To o)$ is not. 
 
 One can check that any balanced type is also non-trivial.

 To study normalization for balanced types we need a more refined notion of neutral term.
 \begin{definition}[neutral terms, version 2]

For any $q\in (0,1]$, and measurable set $S$, the set  $\NEUT^{q}(S)$ is defined by induction as follows:
\begin{varitemize}
\item for any variable $x$, $x\in \NEUT^{q}(S)$;


\item if $t\in \NEUT^{q}(S)$ and for all $\omega \in S$, $\NNF(\pi^{\omega}_{X}(u))\geq s$, then $tu\in \NEUT^{qs}(S)$.

\end{varitemize}

\end{definition}

It is easily checked that for all $t\in \NEUT^{q}(S)$ and $\omega\in S$, $\NNF(\pi^{\omega}_{X}(t))\geq q$.

\begin{lemma}\label{lemma:uno}
For any balanced and $\{[],\HNORM\}$-free type $\sigma$: 
\begin{enumerate}
\item if
$t\in 
\RED_{\sigma}^{X,r}(S)$ and $\omega \in S$, then $\pi^{\omega}_{X}(S)\in \NNF^{\srank{\sigma}\cdot r}$;
\item $\NEUT^{q}(S)\subseteq \RED_{\sigma}^{X, q}(S)$.
\end{enumerate}
\end{lemma}
\begin{proof}
We argue by induction on $\FF s$:
\begin{varitemize}
\item if $\sigma=\B C^{q_{1}}\dots \B C^{q_{k}}\tau$, with $\tau\in\{o,\NORM\}$, then by definition $t\in \RED_{\sigma}^{X,r}(S)$ iff for all $\omega \in S$, $\pi^{\omega}_{X}(t)\in \NNF^{q_{1}\dots q_{k}r}=\NF^{\srank{o}\cdot r}$. Moreover if $t\in \NEUT^{q}(S)$, then 
for all $\omega \in S$, $\NNF(\pi^{\omega}_{X}(t))\geq q$, hence $t\in \RED_{o}^{X,q}(S)\subseteq
\RED_{\sigma}^{X,{q}{\srank{\sigma}}}(S)$.

\item if $\sigma=\tau\land \rho$, then we can argue as in the corresponding point in the proof of Lemma \ref{lemma:unobis}.

\item if $\sigma=\B C^{q_{1}}\dots \B C^{q_{k}}(\sigma_{1}\To \dots \To \sigma_{n}\To o)$, where $\prod_{i}q_{i}\leq \prod_{j}\srank{\sigma_{j}}$, and $t\in \RED_{\sigma}^{X,r}(S)$, then by IH $x_{j}\in \NEUT^{1}(S)\subseteq \RED_{\sigma_{i}}^{X,1}(S)$, and thus 
$tx_{1}\dots x_{n}\in \RED_{o}^{rq_{1}\dots q_{k}}(S)$, so for all $\omega \in S$, 
$\pi^{\omega}_{X}(tx_{1}\dots x_{n})=\pi^{\omega}_{X}(t)x_{1}\dots x_{n}\in \NNF^{rq_{1}\dots q_{k}}\subseteq \NNF^{rq_{1}\dots q_{n}\prod_{j}\srank{\sigma_{j}}}=\NNF^{\srank{\sigma}r}$.
We deduce then $\pi^{\omega}_{X}(t)\in\NNF^{\srank{\sigma}r}$ by applying Lemma \ref{lemma:x} a finite number of times.

%
%
%

Moreover, 
if $t\in \NEUT^{q}(S)$, $S'\subseteq S$, $\omega \in S'$, and $u_{j}\in\RED_{\sigma_{j}}^{X,1}(S')$, then
by the IH for all $\omega \in S'$ $\pi^{\omega}_{X}(u_{j})\in \NNF^{\srank{\sigma_{j}}}$, so  
 $tu_{1}\dots u_{n} \in \NEUT^{q\prod_{k}\srank{\sigma_{j}}}(S')\subseteq \NEUT^{qq_{1}\dots q_{k}}(S')$;
 by the IH it follows then $tu_{1}\dots u_{n}\in 
\RED_{o}^{X,q_{1}\dots q_{k}s}(S')$. 
We can thus conclude 
$t\in 
 \RED_{\sigma_{1}\To \dots \To \sigma_{n}\To o}^{X,qq_{1}\dots q_{k}}(S)=
 \RED_{\sigma}^{X,q}(S)$.
  
%
\end{varitemize}
\end{proof}

The last ingredient of our reducibility argument is the following:

\begin{proposition}\label{prop:normalization}
If $\Gamma \vdash^{X}t: \bone \Pto \FF s$ is derivable in $\TCINT$, where $\Gamma=
\{x_{1}:\FF M_{1},\dots, x_{m}:\FF M_{n}\}$, then for all $S\subseteq \model{\bone}_{X}$, the following two claims hold:
\begin{varitemize}
\item[(i.)] for all $u_{i}\in \HRED_{\FF M_{i}}^{X,1}(S)$,  $t[u _{1}/x_{1},\dots, u_{m}/x_{m}]\in \HRED_{\FF s}^{X,1}(S)$;
\item[(ii.)] for all $u_{i}\in \RED_{\FF M_{i}}^{X,1}(S)$,  $t[u _{1}/x_{1},\dots, u_{m}/x_{m}]\in \RED_{\FF s}^{X,1}(S)$.
\end{varitemize}
\end{proposition}
\begin{proof}
We only prove claim (ii.).
The argument is by induction on a type derivation of $t$:
\begin{varitemize}

\item if $t=x$ and the last rule is  
\begin{prooftree}
\AXC{$\FF s_{i}\preceq \FF t$}
\AXC{$\FN(\bone)\subseteq X$}
\RL{(id$_{\cap}$)}
\BIC{$\Gamma, x:[\FF s_{1},\dots, \FF s_{n}] \vdash^{X} x: \bone\Pto \FF t$}
\end{prooftree}
Then for all $S\subseteq \model{\bone}_{X}$, and for all
 $u_{1}\in \RED^{X,1}_{\FF M_{1}}(S),\dots
u_{n}\in \RED^{X,1}_{\FF M_{n}}(S)$,  and 
$u\in \RED^{X,1}_{\FF M}(S)$, 
using Lemma \ref{lemma:inclusion} we deduce $u\in \RED_{\FF s_{i}}^{X,1}(S)\subseteq \RED_{\FF t}^{X,1}(S)$, whence
$t[u_{1}/x_{1},\dots, u_{n}/x_{n}, u/x]= u \in \RED_{\FF t}^{X, 1}(\model{\bone}_{X}\cap S)= \RED_{\FF s}^{X,1}(S)$.

\item if the last rule is  
\begin{prooftree}
\AXC{$\Big\{\Gamma \vdash^{X} t: \bone_{j}\Pto \FF s\Big \}_{j=1,\dots,k}$}
\AXC{$\bone\vDash^{X} \bigvee_{j}\bone_{j}$}
\RL{$\TU$}
\BIC{$\Gamma \vdash^{X} t: \bone\Pto \FF s$}
\end{prooftree}
then by IH, for all $j=1,\dots, k$, $S\subseteq \model{\bone_{j}}_{X}$, for all $u_{1}\in \RED_{\FF M_{1}}^{X,1}(S),\dots, u_{n}\in \RED_{\FF M_{n}}^{X,1}(S)$ and for all $\omega \in S$, \\
$\pi^{\omega}_{X}( t[u_{1}/x_{1},\dots, u_{n}/x_{n}])\in \RED_{\rho}^{\emptyset, 1}$.

If $k=0$, then $\bone \vdash \bigvee_{j}\bone_{j}=\BOT$, so the claim trivially holds. Let then $k>0$. 
Let $S\subseteq \model{\bone}_{X} \subseteq\model{\bigvee_{j}\bone_{j}}_{X}=\bigcup_{i} \model{\bone_{j}}_{X}$ and $\omega \in S$. Then $f\in   \model{\bone_{j}}_{X}$, for some $j\leq k$. Then first observe that from $u_{i}\in \RED_{\FF M_{i}}^{X,q_{i}}(S)$ it follows in particular that 
$u_{i}\in \RED_{\FF M_{i}}^{X,1}(S\cap \model{\bone_{j}}_{X})$: for all $\omega\in S\cap  \model{\bone_{j}}_{X}$, since $\omega\in S$, $\pi^{\omega}_{X}(u_{i})\in \RED_{\FF M_{i}}^{\emptyset,1}$, so by Lemma \ref{lemma:efs} we conclude that $u_{i}\in \RED_{\FF M_{i}}^{X,1}(S\cap  \model{\bone_{j}}_{X})$. 
Hence we deduce that 
$\pi^{\omega}_{X}( t[u_{1}/x_{1},\dots, u_{n}/x_{n}])\in \RED_{\FF s}^{\emptyset, q_{1}\dots q_{n}}$, and again by Lemma \ref{lemma:efs} we conclude $ t[u_{1}/x_{1},\dots, u_{n}/x_{n}]\in \RED_{\FF s}^{X, 1}(S)$.

\item if $t= t_{1}\oplus^{i}_{a}t_{2}$, and the last rule is 
\begin{prooftree}
\AXC{$ \Gamma \vdash^{X} t_{1}: \bone\Pto \FF s$}
\AXC{$ \btwo\vDash  (\lnot\bvar_{a}^{i}\land \bone)$}
\RL{$\TL$}
\BIC{$ \Gamma \vdash^{X}t_{1}\oplus^{i}_{a} t_{2}:\btwo\Pto \FF s$}
\end{prooftree}
By IH and Lemma~\ref{lemma:efs}, for all 
$S\subseteq \model{\bone}_{X}$, $\omega \in S$, 
$u_{1}\in \RED^{X,1}_{\FF M_{1}}(S),\dots
u_{n}\in \RED^{X,1}_{\FF M_{n}}(S)$, 
$\pi^{\omega}_{X}(t_{1}[u_{1}/x_{1},\dots, u_{n}/x_{n}]) \in \RED^{\emptyset, 1}_{\FF s}$.
Let now $S\subseteq \model{\btwo}_{X}\subseteq \model{\bvar_{a}^{i}}_{X}\land \model{\bone}_{X}$ and $u_{1}\in \RED^{X,1}_{\FF M_{1}}(S),$ $\dots,$
$u_{n}\in \RED^{X,1}_{\FF M_{n}}(S)$.
If now $\omega \in S\subseteq \model{\btwo}_{X\cup\{a\}}$, then $\omega(a)(i)=0$ so in particular $\pi^{\omega}_{X}(t[u_{1}/x_{1},\dots, u_{n}/x_{n}])= \pi^{\omega}_{X}(t_{1}[u_{1}/x_{1},\dots, u_{n}/x_{n}]) \in \RED^{\emptyset, 1}_{\FF s}$. 
Hence, using Lemma~\ref{lemma:efs} we conclude that $t[u_{1}/x_{1},\dots,u_{n}/x_{n}]$ $\in$ $\RED_{\FF s}^{X,1}(S)$.

\item if $t= t_{1}\oplus^{i}_{a}t_{2}$, and the last rule is 
\begin{prooftree}
\AxiomC{$ \Gamma \vdash^{X} t_{2}: \bone\Pto \FF s$}
\AxiomC{$ \btwo\vDash  (\lnot\bvar_{a}^{i}\land \bone)$}
\RightLabel{$\TR$}
\BinaryInfC{$ \Gamma \vdash^{X}t_{1}\oplus^{i}_{a} t_{2}:\btwo\Pto \FF s$}
\end{prooftree}
then we can argue similarly to the previous case.

\item if the last rule is 
$$
\AXC{$\Gamma \vdash^{X}t: \bone \Pto\B C^{q} \sigma $}
\RL{($\HNORM$)}
\UIC{$\Gamma \vdash^{X}t: \bone \Pto \B C^{q}\HNORM$}
\DP
$$
then by IH for all $S\subseteq \model{\bone}_{X}$, for all $u_{1}\in \RED_{\FF M_{1}}^{X,1}(S),\dots, u_{n}\in \RED_{\FF M_{n}}^{X,1}(S)$, $\pi^{\omega}_{X}(t[u_{1}/x_{1},\dots, t_{n}/x_{n}])\in \RED_{\sigma}^{\emptyset,q}$, so by Lemma 
\ref{lemma:unobis}, $\HNF(\pi^{\omega}_{X}(t[u_{1}/x_{1},\dots, t_{n}/x_{n}]))\geq q\cdot \srank \sigma=q$, which implies
$t[u_{1}/x_{1},\dots, t_{n}/x_{n}]\in \RED_{\HNORM}^{X,q}(S)$.

\item if the last rule is 
$$
\AXC{$\Gamma \vdash^{X}t: \bone \Pto\B C^{q} \sigma $}
\AXC{$\sigma$ safe}
\RL{($\NORM$)}
\BIC{$\Gamma \vdash^{X}t: \bone \Pto \B C^{q}\NORM$}
\DP
$$
then $\sigma$ is $\{[],\HNORM\}$-free and balanced; by IH for all $S\subseteq \model{\bone}_{X}$, for all $u_{1}\in \RED_{\FF M_{1}}^{X,1}(S),\dots, u_{n}\in \RED_{\FF M_{n}}^{X,1}(S)$, $\pi^{\omega}_{X}(t[u_{1}/x_{1},\dots, t_{n}/x_{n}])\in \RED_{\sigma}^{\emptyset,q}$; furthermore, by Lemma 
\ref{lemma:uno}, $\NNF(\pi^{\omega}_{X}(t[u_{1}/x_{1},\dots, t_{n}/x_{n}]))\geq q\cdot \srank \sigma=q$, which implies
$t[u_{1}/x_{1},\dots, t_{n}/x_{n}]\in \RED_{\NORM}^{X,q}(S)$.

\item if $t=\lambda x.u$ and the last rule is
\begin{prooftree}
\AXC{$ \Gamma, x: \FF M \vdash^{X} u: \bone\Pto \B C^{\vec s}\tau$}
\RL{$\TLA$}
\UIC{$\Gamma\vdash^{X} \lambda x.u: \bone\Pto \B C^{\vec s}(\FF M\To \tau)$}
\end{prooftree}
where $\vec s= s_{1},\dots, s_{l}$; then, by IH, for all $S\subseteq \model{\bone}_{X}$, for all 
$u_{1}\in \RED_{\FF M_{1}}^{X,1}(S),\dots, u_{n}\in \RED_{\FF M_{n}}^{X,1}(S)$, 
and $v\in \RED_{\FF M}^{X,1}(S)$, 
$u[u_{1}/x_{1},\dots, u_{n}/x_{n},v/x]\in \RED_{\B C^{\vec s}\tau}^{X,1}( S)= \RED_{\tau}^{X, s_{1}\dots s_{k}}(S)$.
In particular, for all choice of $u_{1},\dots, u_{n}$ and variable $y$ not occurring free in $u_{1},\dots, u_{n}$, by letting $t'=u[u_{1}/x_{1},\dots, u_{n}/x_{n},y/x]$, we have that for all
 $v\in \RED_{\FF M}^{X,1}(S\cap \model{\bone}_{X})$, 
 $t'[v/y]=u[u_{1}/x_{1},\dots, u_{n}/x_{n},v/x] \in \RED_{\tau}^{X,s_{1}\dots s_{k}}(  S)$.
 By Lemma \ref{lemma:quattro}, then, we can conclude that $\lambda y.t'=(\lambda y.t[y/x])[u_{1}/x_{1},\dots, u_{n}/x_{n}]= (\lambda x.t)[u_{1}/x_{1},\dots, u_{n}/x_{n}]\in \RED_{\FF M\To \sigma}^{X,s_{1}\dots s_{k}}(S)=
 \RED_{\B C^{\vec s}(\FF M\To \sigma)}^{X,1}(S)$.

\item $t= uv$ and the last rule is
\begin{prooftree}
\AXC{$ \Gamma \vdash^{X} u: \btwo\Pto\B C^{\vec s}(\FF M\To \sigma)$}
\AXC{$ \Big\{\Gamma\vdash^{X}v: \bthree_{i}\Pto \FF s_{i}\Big\}_{i}$}
\AXC{$\bone\vDash^{X}\btwo\land\left (\bigwedge_{i} \bthree_{i}\right)$}
\RL{$\TA$}
\TIC{$\Gamma\vdash^{X} uv: \bone\Pto \B C^{\vec s}\sigma$}
\end{prooftree}
where $\vec s=s_{1},\dots, s_{k}$ and $\CC M=[\FF s_{1},\dots, \FF s_{n}]$, then 
by IH for all $S\subseteq \model{\bone}_{X}\subseteq \model{\btwo}_{X}\cap \bigcap_{i}\model{\bthree_{i}}_{X}$ and 
$u_{1}\in \RED^{X,1}_{\FF M_{1}}(S),\dots
u_{n}\in \RED^{X,1}_{\FF M_{n}}(S)$,
$u[u_{1}/x,\dots, u_{n}/x] \in \RED_{\FF M\To \sigma}^{X,s_{1}\dots s_{k}}(S)$ and 
$v[u_{1}/x,\dots,u_{n}/x]\in \bigcap_{i}\RED_{\FF s_{i}}^{X,1}(S)=\RED_{\FF M}^{X,1}(S)$
(where, if $n=0$, $v[u_{1}/x,\dots,u_{n}/x]\in \RED_{\FF M}^{X,1}(S)$ does not follow from the induction hypothesis but from the fact that  $\RED_{[]}^{X,1}(S)=\Lnu^{X}$) so we deduce  
$u[u_{1}/x,\dots, u_{n}/x] v[u_{1}/x,\dots, u_{n}/x] \
=
(uv)[u_{1}/x,\dots, u_{n}/x] \
\in \RED_{ \sigma}^{X,s_{1}\dots s_{k}  }(S)=\RED_{\B C^{\vec s}\sigma}^{X,1}(S)$.

\item If $t=\nu a.u$ and the last rule is
\begin{prooftree}
\AxiomC{$\Big\{\Gamma\vdash^{X\cup\{a\}} u: \btwo\land \bthree_{i}\Pto \B C^{q_{j}}\sigma\Big\}_{j=1,\dots,n+1}$}
\AxiomC{$\vDash\mu({\bthree_{j}})\geq r_{j}$}
\AxiomC{$\bone \vDash \btwo$}
\RightLabel{$\TN$}
\TrinaryInfC{$\Gamma\vdash^{X} \nu a.u: \bone\Pto \B C^{\sum_{j}r_{j}q_{j}}\sigma$}
\end{prooftree}
where $a\notin FV(\Gamma)\cup FV(\btwo)$, then let $S\subseteq \model{\bone}_{X}\subseteq \model{\btwo}_{X}$ and $u_{1}\in \RED^{X,1}_{\FF M_{1}}(S),\dots
u_{n}\in \RED^{X ,1}_{\FF M_{n}}(S)$.
Let $T=\{\omega+\omega'\mid \omega \in S\}\subseteq (2^{\BB N})^{X\cup\{a\}}$, which is measurable since counter-image of a measurable set through the projection function; then since $\FN(u_{i})\subseteq X$, we deduce using Lemma \ref{lemma:efs} that $u_{i}\in \RED_{\FF M_{i}}^{X\cup\{a\},1}(T\cap\model{ \bthree_{j}}_{X\cup\{a\}})$. 
By IH and the hypothesis we deduce then that 
\begin{varitemize}
\item the sets $\model{\bthree_{j}}_{X\cup\{a\}}$ are pairwise disjoint;
\item $u[u_{1}/x_{1},\dots, u_{n}/x_{n}]\in \RED^{X\cup\{a\},1}_{\FF s}(T\cap\model{\bthree_{j}}_{X\cup\{a\}})$;
\item for all $\omega \in S$, $\mu(\Pi^{\omega}(T\cap\model{ \bthree_{j}}_{X\cup\{a\}}))\geq r_{j}$;

\end{varitemize}

Hence, by Lemma~\ref{lemma:nu2} we conclude that 
$\nu a.u[u_{1}/x_{1},\dots, u_{n}/x_{n}]= (\nu a.u)[u_{1}/x_{1},\dots, u_{n}/x_{n}] \in \RED_{\FF s}^{X,\sum_{j}r_{j}q_{j}}(S)=\RED_{\B C^{\sum_{j}r_{j}q_{j}}\FF s}^{X,1}(S)$.

%
%
%

\end{varitemize}
\end{proof}


From Proposition \ref{prop:normalization} (i.) and Lemma \ref{lemma:unobis}, by observing that any type of $\TCINT$ is non-trivial, we deduce
the following, which generalizes the first part of Theorem \ref{thm:normalization}.

\begin{theorem}
If $ \Gamma\vdash t: \bone \Pto \B C^{q}\sigma$ holds in $\TCINT$, then for all $\omega\in \model\bone$, $\NHNF(\pi^{\omega}_{X}(t))\geq q$.
\end{theorem}
%
%
%

From Proposition \ref{prop:normalization} (i.) and Lemma \ref{lemma:uno} we similarly deduce a generalization of Theorem \ref{thm:normalization2}:

\begin{theorem}
If $\Gamma \vdash^{X}t: \bone\Pto \B C^{q}\sigma$ holds in $\TCINT$, where $\Gamma$ and $\FF s$ are formed of balanced types, then for all $\omega\in \model\bone$, $ \NNF(\pi^{\omega}_{X}(t))\geq q$.
\end{theorem}

From the two results above we can finally deduce the soundness part of Theorem \ref{thm:completenessa}:

\begin{corollary}
For any closed term $t$, 
\begin{align*}
\NHNF(t) & \geq \sup \{ q\mid\  \vdash t: \TOP \Pto \B C^{q}\HNORM\} \\
\NNF(t) & \geq \sup \{ q\mid \  \vdash t: \TOP \Pto\B C^{q} \NORM\} 
\end{align*}

\end{corollary}

\subsection{Probabilistic Reducibility Candidates for $\EVLL$.}\label{subs:norma2}

We now show how to adapt the family of reducibility candidates $\HRED_{\sigma}^{X,1}(S)$ to $\EVLL$, in order to prove the second part of Theorem \ref{thm:normalization}.
We will still work with the general family of types from the previous section (hence suggesting that intersection types could be added to $\LCPL$, too).

%

For any $t\in \CC T$, let $\CC D^{1}_{t}: \CC T \to [0,1]$ be $\delta_{t}$ if $t\in \HNF$, and if $t=\nu a.t'$, let 
$$
\CC D^{1}_{t}(u)=\mu\{ \omega\in 2^{\BB N}\mid \pi^{\omega}_{\{a\}}(t')=u\}
$$

By recalling the definition of head-reduction in $\EVLL$ from Subsection \ref{subs:evll}, one can easily check the following useful properties:
\begin{lemma}\label{lemma:redallh}
\begin{varitemize}
\item[(i.)]If $t\redallh^{*} u$, then $\pi^{\omega}_{X}(t)\redallh^{*}\pi^{\omega}_{X}(u)$.
\item[(ii.)] if $t\redallh^{*} u$, then $\HNF(t)\geq \HNF(u)$.
\end{varitemize}

\end{lemma}
%
%

\begin{lemma}\label{lemma:dapp}
For all $t,u,w\in \CC T$, 
$\CC D^{1}_{t}(u) = \CC D^{1}_{tw}(uw)$.
\end{lemma}
\begin{proof}
If $t$ is a HNV, $tw$ is a pseudo-value, so $\CC D^{1}_{t}=\delta_{t}$ and $\CC D^{1}_{tw}=\delta_{tw}$, 
which implies $\CC D^{1}_{t}(u)=1$ iff $t=u$ iff $tw=uw$ iff $\CC D^{1}_{tw}=tw$.
If $t=\nu a.t'$, since we can suppose $a\notin \FN(w)$, from
$\pi^{\omega}_{\{a\}}(t'w)=\pi^{\omega}_{\{a\}}(t')w$ we deduce
$\{\omega \mid \pi^{\omega}_{\{a\}}(t'w)=uw\}= \{\omega\mid \pi^{\omega}_{\{a\}}(t')=u\}$, from which
the claim follows.
\end{proof}

The following is a simple consequence of Lemma \ref{lemma:redallh}:
\begin{lemma}\label{lemma:dduno}
Let $W$ be a set of terms such that $t\in W$ and $u\redallh^{*}t$ implies $u\in W$; then for any $t,u\in\CC T$,] if $t\redallh^{*}u$, then $\sum_{w\in W}\CC D^{1}_{t}(w) \geq \sum_{w\in W}\CC D^{1}_{u}(w)$.
\end{lemma}

To adapt reducibility predicates to $\EVLL$, we need to adapt the definition of reducibility for counting quantified types.
We define yet a new family of reducibility predicates
$\GHRED_{o}^{X,r}(S)\subseteq \HLnu^{X}$ as follows:

\begin{align*}
\GHRED_{[]}^{X,r}(S)&= \HLnu^{X}\\ 
\GHRED_{o}^{X,r}(S)&=\GHRED_{\HNORM}^{X,r}(S)=\GHRED_{\NORM}^{X,r}(S)= \{ t \mid   \forall \omega\in S   \ \NHNF(\pi^{\omega}_{X}(t))\geq r\}\\ 
\GHRED_{\sigma\To \tau}^{X,r}(S)& =\{t \mid   
\forall S'\subseteq S,  \ \forall u\in \HRED_{\sigma}^{X,1}(S'), \ tu\in \HRED_{\tau}^{X, r}(S')\}\\
\GHRED_{\sigma\land \tau}^{X,r}(S)&= \GHRED_{\sigma}^{X,r}(S) \cap \GHRED_{\tau}^{X,r}(S)\\
\GHRED_{\B C^{q}\sigma}^{X,r}(S) & = 
\left \{t\mid \forall \omega\in S \ \ \sum_{u\in \GHRED_{\B \sigma}^{\emptyset, r}} \CC D^{1}_{\pi^{\omega}_{X}(t)}(u) \geq q\right\}
\end{align*}

Most properties of $\GHRED_{\sigma}$ are proved as in the corresponding cases of $\HRED_{\sigma}$. We here only consider the cases in which the proofs are different.

\begin{lemma}\label{lemma:sbetazza}
If $t\in \GHRED_{\sigma}^{X,r}(S)$ and $u\redallh^{*}t$, then $u\in \GHRED_{\sigma}^{X,r}(S)$.
\end{lemma}
\begin{proof}
We argue by induction on $\sigma$:
\begin{varitemize}
\item if $\sigma\in \{o,\HNORM,\NORM\}$, then $t\in \GHRED_{\sigma}^{X,r}(S)$ iff for all $\omega\in S$, $\NHNF(\pi^{\omega}_{X}(t))\geq r$;
the claim then follows from Lemma \ref{lemma:redallh} (i.).
%

%

\item if $\sigma=\tau\To \rho$ and $t\in \GHRED_{\sigma}^{X,r}(S)$, then for all $S'\subseteq S$ and $v\in \GHRED_{\tau}^{X,1}(S')$, 
$tv\in \GHRED_{\rho}^{X,r}(S')$; since $uv\redallh^{*}tv$, by IH we deduce then $uv\in \GHRED_{\rho}^{X,r}(S')$; we can thus conclude $v\in \GHRED_{\sigma}^{X,r}(S)$;

\item if $\sigma=\tau\land \rho$, and $t\in \GHRED_{\sigma}^{X,r}(S)$, then by IH
$u\in \GHRED_{\tau}^{X,r}(S)$ and $u\in \GHRED_{\rho}^{X,r}(S)$, so $u\in \GHRED_{\sigma}^{X,r}(S)$. 

\item if $\sigma= \B C^{q}\tau$ and $t\in \GHRED_{\sigma}^{X,r}(S)$, then 
$\sum_{v\in \GHRED_{\tau}^{X,r}(S)}\CC D^{1}_{\pi^{\omega}_{X}(t)}(v)  \geq   q$, and we can conclude
$\sum_{v\in \GHRED_{\tau}^{X,r}(S)}\CC D^{1}_{\pi^{\omega}_{X}(u)}(v)  \geq   q$ from the IH and Lemma \ref{lemma:dduno} (i).

%

\end{varitemize}

\end{proof}

Using Lemma \ref{lemma:nu1} (i.), whose proof is adapted without difficulties to $\EVLL$, 
Lemma \ref{lemma:nu2} is adapted as follows:
\begin{lemma}\label{lemma:nu3}
For all types $\sigma$, terms $t,u_{1},\dots, u_{n}$ with $\FN(t)\subseteq X\cup \{a\}$, $\FN(u_{i})\subseteq X$,
and measurable set  $S\subseteq (2^{\BB N})^{X\cup\{a\}}$ and $S'\subseteq (2^{\BB N})^{X}$, if
\begin{enumerate}
\item $tu_{1}\dots u_{n}\in \GHRED_{\sigma}^{X\cup\{a\}, r}(S)$;
\item for all $\omega\in S'$, $\mu( \Pi^{\omega}(S))\geq s$;
\end{enumerate}
then $(\nu a.t)u_{1}\dots u_{n}\in \GHRED_{\B C^{s}\sigma}^{X,r}(S')$.

\end{lemma}

The following adaptation of Lemma \ref{lemma:tre} immediately follows from Lemma \ref{lemma:sbetazza}:
\begin{lemma}\label{lemma:trebis}
 $t[u/x]u_{1}\dots u_{n}\in \GHRED_{\sigma}^{X,r}(S)$ $\To $ $(\lambda x.t)uu_{1}\dots u_{n}\in \GHRED_{\sigma}^{X,r}(S)$.
\end{lemma}

%

\begin{lemma}\label{lemma:lambdasgraffa}
For all types $\sigma$ and $\tau$, $\GHRED_{\B C^{q_{1}}\dots\B C^{ q_{k}}(\sigma\To \tau)}^{X,r}(S)\subseteq \GHRED_{\sigma \To \B C^{q_{1}}\dots\B C^{ q_{k}}\tau}^{X,r}(S)$.
\end{lemma}
\begin{proof}

For any $j\leq k$, let $U_{j}:=\GHRED_{\B C^{q_{k-j+1}}\dots\B C^{ q_{k}}(\sigma\To \tau)}^{X,r}(S)$
and 
$V_{j}:= \GHRED_{\sigma \To \B C^{q_{k-j+1}}\dots\B C^{ q_{k}}\tau}^{X,r}(S)$.

We argue, by induction on $j\leq k$, that $U_{j}\subseteq V_{j}$. If $j=0$ then $U_{j}=\sigma\To \tau= V_{j}$.

Let then $j>0$ and $\omega \in S$.
Suppose 
$t\in U_{j}$ and let 
$w\in \GHRED_{\sigma}^{X,1}(S)$.

By IH we have $U_{j-1}\subseteq V_{j-1}$, from which we deduce  
$
q_{1}\leq \sum_{u\in U_{j-1}}\CC D^{1}_{\pi^{\omega}_{X}(t)}(u) \
\leq\sum_{u\in V_{j-1}}\CC D^{1}_{\pi^{\omega}_{X}(t)}(u)$ \\
$\stackrel{\text{[Lemma \ref{lemma:dapp}]}}{=}
\sum_{u\in V_{j-1}}\CC D^{1}_{\pi^{\omega}_{X}(t)\pi^{\omega}_{X}(w)}(u\pi^{\omega}_{X}(w))=
\sum_{u\in V_{j-1}}\CC D^{1}_{\pi^{\omega}_{X}(tw)}(u\pi^{\omega}_{X}(w))
$ which implies $t\in V_{j}$.

%

%
%

\end{proof}

We call a term $t$ \emph{$\nu$-safe} if $t$ never reduces to a term of the form $\nu a.t'$.

\begin{definition}[neutral terms for $\EVLL$.]

For any
 measurable set $S\subseteq (2^{\BB N})^{X}$, the set  $\GNEUT(S)$ is defined by induction as follows:
\begin{varitemize}
\item for any variable $x$, $x\in \GNEUT(S)$;

\item if $t\in \GNEUT(S)$ and for all $\omega\in S$ $\pi^{\omega}_{X}(u)$ is $\nu$-safe and $\NHNF(u)=1$, $\{t\}u\in \GNEUT(S)$.

\item if $t\in \GNEUT(S)$, then for all $u\in \HLnu^{X}$, $tu\in \GNEUT(S)$.

\end{varitemize}

\end{definition}

It is easily checked that for all $t\in \GNEUT(S)$ and $\omega\in S$, $\HNF(\pi^{\omega}_{X}(t))=1$.

%
\begin{lemma}\label{lemma:unobiss}
For any type $\sigma$, 
\begin{varitemize}
\item[(i.)] if $\sigma$ is non-trivial, then $t\in \GHRED_{\sigma}^{X,r}(S)$ and $\omega \in S$, then $\HNF(\pi^{\omega}_{X}(t))\geq \srank\sigma \cdot r$;
\item[(ii.)] $\GNEUT(S)\subseteq \GHRED_{\sigma}^{X,1}(S)$.

\end{varitemize}
\end{lemma}
\begin{proof}
By induction on $\sigma$:

\begin{varitemize}

\item if $\sigma=[]$, both claims are immediate.
\item if $\sigma\in \{o,\HNORM,\NORM\}$, claim (i.) holds by definition, and claim (ii.) follows from the fact that for all $t\in \GNEUT(S)$ and for all $\omega \in S$, $\HNF(\pi^{\omega}_{X}(t))=1$, so $t\in \GHRED_{\sigma}^{X,1}(S)$.

\item if $\sigma=\tau\To \rho$ and $t\in \GHRED_{\sigma}^{X,r}(S)$, then, since by IH $x\in \GHRED_{\tau}^{X,1}(S)$, $tx\in \GHRED_{\tau}^{X,1}(S)$, hence for all $\omega\in S$,
$\HNF(\pi^{\omega}_{X}(t))\geq \HNF(\pi^{\omega}_{X}(tx))\geq \srank\rho\cdot r= \srank\sigma\cdot r$.

Moreover, if $t\in \GNEUT(S)$, then for all $S'\subseteq S$, $u\in \GHRED_{\tau}^{X,1}(S')$ and $\omega\in S'$, 
$tu\in \GNEUT(S)$, so by IH $tu\in \GHRED_{\rho}^{X,r}(S')$, and we conclude then $t\in \GHRED_{\sigma}^{X,r}(S)$.

\item if $\sigma=\tau\land \rho$ and $t\in \GHRED_{\sigma}^{X,r}(S)$, then
$t\in \GHRED_{\tau}^{X,r}(S)$ and $t\in \GHRED_{\rho}^{X,r}(S)$, so by IH for all $\omega\in S$,
$\HNF(\pi^{\omega}_{X}(t))\geq \max\{r\cdot\srank \tau,r\cdot \srank \rho\}=
r\cdot \max\{\srank \tau, \srank\rho\}=r\cdot \srank \sigma$.

Moreover, it $t\in \GNEUT(S)$, by IH $t\in \GHRED_{\tau}^{X,1}(S)$ and
$t\in \GHRED_{\rho}^{X,1}(S)$, whence $t\in \GHRED_{\sigma}^{X,1}(S)$.

\item if $\sigma=\B C^{q}\tau$, and $t\in \GHRED_{\B C^{q}\tau}^{X,r}(S)$, let $\omega \in S$; then we consider two cases:
\begin{varitemize}

\item if $\pi^{\omega}_{X}(t)$ (that we can suppose being in $\CC T$) is a pseudo-value, then 
$\sum_{w\in \GHRED_{\tau}^{X,r}}\CC D^{1}_{t}(w)\geq q>0$ implies
 $\pi^{\omega}_{X}(t)\in \GHRED_{\tau}^{X,r}$, so by IH we deduce 
 $\HNF(\pi^{\omega}_{X}(t))\geq r\cdot \srank\tau \geq r\cdot q\cdot \srank \tau=r\cdot \srank\sigma$.
 
 \item if $\pi^{\omega}_{X}(t)=\nu a.t'$, then
 $\sum_{w\in \GHRED_{\tau}^{X,r}}\CC D^{1}_{t}(w)\geq q>0$ implies, by IH, 
that $\mu( \{ \omega'\in 2^{\BB N}\mid \pi^{\omega'}(t')\in \GHRED_{\tau}^{\emptyset,r}\}) \geq\mu( 
\{\omega'\in 2^{\BB N}\mid \HNF(\pi^{\omega'}(t))\geq {\srank{\tau}r} \}) \geq q$, and this implies
$\HNF(\pi^{\omega}_{X}(t))\geq q\cdot r\cdot \srank\tau=r\cdot \srank \sigma$.

\end{varitemize}

Moreover, if $t\in \GNEUT(S)$, then for all $\omega \in S$, $\CC D^{1}_{\pi^{\omega}_{X}(t)}=\delta_{\pi^{\omega}_{X}(t)}$, and since by the IH
$t\in \GHRED_{\tau}^{X,r}(S)$, 
we deduce  
$
\sum_{u\in \GHRED_{\tau}^{X,r}}\CC D^{1}_{\pi^{\omega}_{X}(t)}(u)= \delta_{\pi^{\omega}_{X}(t)}(\pi^{\omega}_{X}(t))=1\geq q
$, 
so $t\in \GHRED_{\B C^{q}\tau}^{X,r}(S)$. 
\end{varitemize}
\end{proof}

\begin{lemma}\label{lemma:hnv}
For any non-trivial type $\sigma$, and  term $t$ name-closed and $\nu$-safe, if $t\in \GHRED_{\sigma}^{\emptyset,r}$, then
$\NHNF(t)=1$.
\end{lemma}
\begin{proof}
By induction on $\sigma$:
\begin{varitemize}
\item if $\sigma\in \{o,\HNORM,\NORM\}$, from $t\in \GHRED_{\sigma}^{\emptyset,r}$ we deduce $\NHNF(t)\geq r>0$, and since for all $u$ such that $t\redallh^{*}u$, 
$\CC D_{u}=\delta_{u}$, 
we deduce that $\sup\{ \HNF(u) \mid t\redallh^{*}u\}=\sup\{\sum_{v\in \HNF^{\{\}}}\delta_{u}(v)\mid t\redallh^{*}u\}>0$, 
this implies that $t$ reduces to a HNV, whence $\NHNF(t)=1$.

\item if $\sigma=\tau\To \rho$ and $t\in \GHRED_{\sigma}^{\emptyset,r}$ then since by Lemma \ref{lemma:unobiss} $x\in \GHRED_{\tau}^{\emptyset,1}$, 
$tx\in \GHRED_{\rho}^{\emptyset,r}$; 
suppose $tx\redallh^{*} \nu a.t'$; then since $t$ is $\nu$-safe, the only possibility is that 
$t\redallh^{*} (\lambda x.u)$ and $u\redallh^{*} \nu a.t'$; but then $t\redallh^{*}\lambda x.\nu a.t'\redallh \nu a.\lambda x.t'$, against the hypothesis.
We conclude that also $tx$ must be $\nu$-safe, so by IH 
$\NHNF(tx)=1$, whence $\NHNF(t)=1$.

\item if $\sigma= \tau \land \rho$ and $t\in \GHRED_{\sigma}^{\emptyset,r}$, then $t\in \GHRED_{\tau}^{\emptyset,r}$, so by IH, 
$\NHNF(t)=1$.
\item if $\sigma= \B C^{s}\tau$ and $t\in \GHRED_{\sigma}^{\emptyset,r}$, then
since $\CC D_{t}=\delta_{t}$, from 
$\sum_{w\in \GHRED_{\tau}^{\emptyset,r}}\delta_{t}(w)\geq s>0$, it follows that $t\in \GHRED_{\tau}^{\emptyset,r}$, so by IH 
$\NHNF(t)=1$

\end{varitemize}
\end{proof}

\begin{lemma}\label{lemma:graff}
If $t\in \GHRED_{\sigma\To \tau}^{X,q}(S)$ and $u\in \GHRED_{\B C^{s}\sigma}^{X,1}(S)$, where $\sigma$ is non-trivial, then $\{t\}u\in \GHRED_{\B C^{s}\tau}^{X,q}(S)$.
\end{lemma}
\begin{proof}
Let $\omega\in S$. We distinguish two cases:
\begin{varitemize}

\item $\pi^{\omega}_{X}(u)\redallh^{*} \nu a.u'$;
then, letting 
$W_{\sigma}:=\GHRED_{\sigma}^{X,1}(S)$ and
$W_{\tau}:=\GHRED_{\tau}^{X,q}(S)$, we have that  
$\sum_{w\in W}\CC D_{\pi^{\omega}_{X}(\nu a.tu')}^{1}(w)=
\sum_{w\in W}\mu(\{\omega' \mid \pi^{\omega+\omega'}_{X\cup\{a\}}(tu')=w\}
=\mu\{ \omega'\mid \pi^{\omega+\omega'}_{X\cup\{a\}}(tu')\in W\}
=\mu\{ \omega'\mid t\pi^{\omega+\omega'}_{X\cup\{a\}}(u')\in W\}\geq 
\mu\{ \omega'\mid \pi^{\omega+\omega'}_{X\cup\{a\}}(u')\in U\}=
\sum_{w\in U}\CC D^{1}_{\nu a.u}(w)\geq  s
$, which proves that $\nu a.tu'\in W_{\tau}$. Since $\{t\}u\redallh^{*}\nu a.tu'$, we conclude then $\{t\}u\in W_{\tau}$ 
from Lemma \ref{lemma:sbetazza}.
%
%

\item $\pi^{\omega}_{X}(u)$ is $\nu$-safe; then from $\pi^{\omega}_{X}(u)\in \GHRED_{\sigma}^{\emptyset,q}$ we deduce, by Lemma \ref{lemma:hnv}, that $\NHNF(\pi^{\omega}_{X}(u))=1$. We deduce then that $\{t\}u\in \GNEUT(S)$, and by Lemma \ref{lemma:unobiss} we conclude $\{t\}u\in \GHRED_{\B C^{s}\sigma}^{X,q}(S)$.
%


\end{varitemize}
\end{proof}

\begin{proposition}\label{prop:normalizationbis}
If $\Gamma \vdash^{X}t: \bone \Pto \FF s$ is derivable in $\LCPL$, where $\Gamma=
\{x_{1}:\FF s_{1},\dots, x_{m}:\FF s_{n}\}$, then for all $S\subseteq \model{\bone}_{X}$, for all $u_{i}\in \GHRED_{\FF s_{i}}^{X,1}(S)$,  $t[u _{1}/x_{1},\dots, u_{m}/x_{m}]\in \GHRED_{\FF s}^{X,1}(S)$.
\end{proposition}
\begin{proof}
We can argue as in the proof of Proposition \ref{prop:normalization}, with the following two new cases:

\begin{varitemize}
\item if $t=\{t_{1}\}t_{2}$ and the last rule is
$$
\AXC{$\Gamma \vdash^{X} t_{1}: \btwo \Pto \B C^{\vec q}(\FF s\To \sigma)$}
\AXC{$\Gamma\vdash^{X} t_{2}: \bthree \Pto \B C^{s}\FF s$}
\AXC{$\bone \vDash \btwo \land\bthree$}
\RL{($\{\}$)}
\TIC{$\Gamma \vdash^{X}t: \bone\Pto \B C^{s}\B C^{\vec q}\sigma$}
\DP
$$
then by IH for all $u_{i}\in \GHRED_{\FF s_{i}}^{X,1}(S)$, $t_{1}[u_{1}/x_{1},\dots, u_{n}/x_{n}]\in \GHRED_{\B C^{\vec q}(\FF s\To \sigma)}^{X,1}(S)$ and 
$t_{2}[u_{1}/x_{1},\dots, u_{n}/x_{n}]\in \GHRED_{\B C^{s}\FF s}^{X,1}(S)$;
then by Lemma \ref{lemma:lambdasgraffa} and Lemma \ref{lemma:graff} (observing that $\FF s$ is non-trivial) we deduce \\ 
$(\{t_{1}\}t_{2})[u_{1}/x_{1},\dots, u_{n}/x_{n}]= \{t_{1}[u_{1}/x_{1},\dots, u_{n}/x_{n}]\}(t_{2}[u_{1}/x_{1},\dots, u_{n}/x_{n}])
\in \GHRED_{\B C^{\vec q}\sigma}^{X,1}(S)$.

\item If $t=\nu a.u$ and the last rule is
\begin{prooftree}
\AxiomC{$\Gamma\vdash^{X\cup\{a\}} u: \btwo\land \bthree\Pto \B C^{q_{j}}\sigma$}
\AxiomC{$\vDash\mu({\bthree})\geq r$}
\AxiomC{$\bone \vDash \btwo$}
\RightLabel{$\TN$}
\TrinaryInfC{$\Gamma\vdash^{X} \nu a.u: \bone\Pto \B C^{r}\sigma$}
\end{prooftree}
where $a\notin FV(\Gamma)\cup FV(\btwo)$, then let $S\subseteq \model{\bone}_{X}\subseteq \model{\btwo}_{X}$ and $u_{1}\in \GHRED^{X,1}_{\FF s_{1}}(S),\dots
u_{n}\in \GHRED^{X ,1}_{\FF s_{n}}(S)$.
Let $T=\{g+f\mid \omega \in S\}\subseteq (2^{\BB N})^{X\cup\{a\}}$, which is measurable since counter-image of a measurable set through the projection function; then since $\FN(u_{i})\subseteq X$, we deduce using Lemma \ref{lemma:efs} that $u_{i}\in \GHRED_{\FF s_{i}}^{X\cup\{a\},1}(T\cap\model{ \bthree_{j}}_{X\cup\{a\}})$. 
By IH and the hypothesis we deduce then that 
\begin{varitemize}
\item $u[u_{1}/x_{1},\dots, u_{n}/x_{n}]\in \GHRED^{X\cup\{a\},1}_{\FF s}(T\cap\model{\bthree_{j}}_{X\cup\{a\}})$;
\item for all $\omega \in S$, $\mu(\Pi^{\omega}(T\cap\model{ \bthree_{j}}_{X\cup\{a\}}))\geq r$;

\end{varitemize}

Hence, by Lemma~\ref{lemma:nu3} we conclude that 
$\nu a.u[u_{1}/x_{1},\dots, u_{n}/x_{n}]= (\nu a.u)[u_{1}/x_{1},\dots, u_{n}/x_{n}] \in \GHRED_{\FF s}^{X,\sum_{j}r_{j}q_{j}}(S)=\GHRED_{\B C^{r}\FF s}^{X,1}(S)$.

\end{varitemize}
\end{proof}


We can now deduce Theorem \ref{thm:normalization} from Proposition \ref{prop:normalizationbis} as in the case of $\TCINT$.

\begin{theorem}
If $\Gamma \vdash^{X}t: \bone \Pto \FF s$, then for all $\omega\in \model{\bone}$, $\NHNF(\pi^{\omega}_{X}(t))\geq \srank \sigma$.
\end{theorem}

\end{document}